\documentclass[10pt, letterpaper]{article} 

\usepackage[letterpaper,includehead,left=1.9cm,right = 1.9cm ,top=0cm,bottom=2.5cm]{geometry}

\usepackage{fancyhdr}

\fancypagestyle{plain}{
	\fancyhf{}

}

\renewcommand{\footnoterule}{%
  \kern -3pt
  \hrule width \textwidth height 0.5pt
  \kern 2pt
}

\usepackage[font=footnotesize,format=plain, bf]{caption} %

\usepackage{microtype}	%

\usepackage{imakeidx}

\makeatletter
\def\@idxitem{\par\hangindent 2ex}
\makeatother

\makeindex[intoc,options= -s index_style.ist]

\usepackage[nottoc]{tocbibind}

\usepackage[dvipsnames]{xcolor}

\usepackage{tcolorbox}
\tcbuselibrary{skins,breakable}

\newtcolorbox{definition_frame}{enhanced, colback = gray!50!blue!15, borderline = {0.5pt}{0mm}{gray!90!blue}, leftrule = 4pt, toprule = 0pt, bottomrule = 0pt, rightrule = 0 pt, beforeafter skip= 1em, breakable, after={\noindent}}
\newtcolorbox{theorem_frame}{enhanced, colback = gray!50!blue!25, borderline = {0.5pt}{0mm}{gray!90!blue},leftrule = 5pt, toprule = 0pt, bottomrule = 0pt, rightrule = 0 pt, beforeafter skip= 1em, breakable, after={\noindent}}
\newtcolorbox{lemma_frame}{enhanced, colback = gray!50!blue!10, borderline = {0.5pt}{0mm}{gray!90!blue}, leftrule = 2pt, toprule = 0pt, bottomrule = 0pt, rightrule = 0 pt ,beforeafter skip= 1em, breakable, after={\noindent}}
\newtcolorbox{corollary_frame}{enhanced, colback = gray!50!blue!6, borderline = {0.5pt}{0mm}{gray!90!blue}, leftrule = 0pt, toprule = 0pt, bottomrule = 0pt, rightrule = 0 pt, beforeafter skip= 1em, breakable, after={\noindent}}
\newtcolorbox{example_frame}{enhanced, colback = gray!50!blue!0, borderline = {0.5pt}{0mm}{gray!90!blue},leftrule = 0pt, toprule = 0pt, bottomrule = 0pt, rightrule = 0 pt , beforeafter skip= 1em, breakable, after={\noindent}}
\newtcolorbox{remark_frame}{enhanced, colback = gray!50!blue!0, borderline = {0.5pt}{0mm}{gray!90!blue},leftrule = 0pt, toprule = 0pt, bottomrule = 0pt, rightrule = 0 pt , beforeafter skip= 1em, breakable, after={\noindent}}

\usepackage[strict]{changepage}
\newenvironment{proof_frame}{\begin{adjustwidth}{12pt}{12pt}}{\end{adjustwidth}}

\usepackage[utf8]{inputenc}	
\usepackage[T1]{fontenc}

\usepackage{graphicx} %
\usepackage[font=footnotesize]{subcaption} %
\usepackage{wrapfig} %
\usepackage{rotating} %
\usepackage{tikz} %
\usetikzlibrary{positioning, calc, intersections, cd, babel} %

\usepackage{colortbl} %
\usepackage{longtable} %
\usepackage{multirow} %
\usepackage{paralist} %

\usepackage{amsmath,amsthm,amsfonts,amssymb,mathtools} %
\usepackage{bbm} %
\usepackage{units} %
\usepackage{cancel} %
\usepackage{resizegather} %
\usepackage{tensor} %
\usepackage{ellipsis} %
\usepackage{empheq} %

\usepackage{needspace} %
\usepackage{placeins} %
\usepackage[inline]{enumitem}

\usepackage[style=alphabetic]{biblatex}
\usepackage{authblk}
\usepackage{fontawesome}
\usepackage{scrextend}
\usepackage{setspace}

\usepackage{lmodern}

\usepackage[english]{babel}

\usepackage[colorlinks=true, linkcolor=Cerulean!75!black, citecolor=green!60!black, urlcolor=Cerulean!75!black, linktocpage=true]{hyperref}

\allowdisplaybreaks 

\makeatletter
\g@addto@macro\bfseries{\boldmath}
\makeatother

\newtheoremstyle{theorem}{\topsep}{\topsep}{\itshape}{}{\bf}{.\ \\}{.5ex}{}
\newtheoremstyle{definition}{\topsep}{\topsep}{}{}{\bf}{.\ \\}{.5ex}{}
\newtheoremstyle{info}{\topsep}{\topsep}{}{}{\bfseries \normalsize }{.\ \\}{.5ex}{}

\theoremstyle{theorem}
\newtheorem{theoremEnvironment}{Theorem}[section]
\newtheorem{corollaryEnvironment}[theoremEnvironment]{Corollary}
\newtheorem{lemmaEnvironment}[theoremEnvironment]{Lemma}

\newcommand{\varName}{varTheorem} 
\newtheorem{varTheorem}[theoremEnvironment]{\varName}
\newtheorem{varLemma}[theoremEnvironment]{\varName}
\newtheorem{varCorollary}[theoremEnvironment]{\varName}

\theoremstyle{definition}
\newtheorem{definitionEnvironment}[theoremEnvironment]{Definition}
\newtheorem{proofEnvironment}[theoremEnvironment]{Proof}
\newtheorem{varDefinition}[theoremEnvironment]{\varName}

\theoremstyle{info}
\newtheorem{remarkEnvironment}[theoremEnvironment]{Remark}
\newtheorem{exampleEnvironment}[theoremEnvironment]{Example}

\newenvironment{theorem}[1][]{\begin{theorem_frame}\begin{theoremEnvironment}[#1]}{\end{theoremEnvironment}\end{theorem_frame}}

\newenvironment{corollary}[1][]{\begin{corollary_frame}\begin{corollaryEnvironment}[#1]}{\end{corollaryEnvironment}\end{corollary_frame}}
\newenvironment{definition}[1][]{\begin{definition_frame}\begin{definitionEnvironment}[#1]}{\end{definitionEnvironment}\end{definition_frame}}
\newenvironment{remark}[1][]{\begin{remark_frame}\begin{remarkEnvironment}[#1]}{\end{remarkEnvironment}\end{remark_frame}}

\renewenvironment{proof}[1][]{\begin{proof_frame}\begin{proofEnvironment}[#1]}{\end{proofEnvironment}\end{proof_frame}}

\usepackage{tikz-3dplot}

\begin{document}
	
\newpage
\setlength{\columnsep}{0.6cm}
\twocolumn

\title{\bfseries Protein-environment-sensitive computational epitope accessibility analysis from antibody dose-response data}

\renewcommand{\Affilfont}{\footnotesize \normalfont}
\renewcommand{\Authfont}{\small \bfseries \raggedright}

\author[\faEnvelopeO,1]{Dominik Tschimmel}
\author[\faFlask,1]{Momina Saeed}
\author[\faFlask,1]{Maria Milani}
\author[\faUser,2]{Steffen Waldherr}
\author[\faUser, \faEnvelope,1]{Tim Hucho}

\affil[1]{Translational Pain Research, Department of Anesthesiology and Intensive Care Medicine, University Hospital Cologne, University of Cologne, 50931 Cologne, Germany}
\affil[2]{Molecular Systems Biology, Department of Functional and Evolutionary Ecology, University of Vienna, 1030 Vienna, Austria}
\affil[\faEnvelopeO]{orcid: 0000-0002-5238-9930}
\affil[\faEnvelope]{tim.hucho@uk-koeln.de, orcid: 0000-0002-4147-9308}
\affil[\faFlask]{Data acquisition during master's thesis}
\affil[\faUser]{Supervision of the research project and of the manuscript writing}

\maketitle

\begin{refsection}

\begingroup
\bfseries
\noindent
\section*{Abstract}

Antibodies are widely used in life-sciences and medical therapy. Yet, broadly applicable methods are missing to determine, in the biological system of choice, antibody specificity and its quantitative contribution to e.g. immunofluorescence stainings. Thereby, antibody-based data often needs to be seen with caution. Here, we present a simple-to-use approach to characterize and quantify antibody binding properties directly in the system of choice. We determine an epitope accessibility distribution in the system of interest based on a computational analysis of antibody-dilution immunofluorescence stainings. This allows the selection of specific antibodies, the choice of a dilution to maximize signal-specificity, and an improvement of signal quantification. It further expands the scope of antibody-based imaging to detect changes of the subcellular nano-environment and allows for antibody multiplexing.

\endgroup

\section*{Introduction}

Antibody-based applications range from basic research and diagnostics \cite{DeMatos_2010} \cite{Wen_2013} to therapeutic intervention \cite{Wirth_2023} \cite{Axelsen_2024} \cite{Guan_2024}. Accordingly, the purpose varies from e.g. target identification and quantification to the delivery of pharmaceuticals. Central for each of these applications are parameters defining the binding specificity. Nevertheless, such parameters are rarely accessible and in even fewer cases such information has been obtained in the exact system of choice.

Just like any other protein-protein interaction, antibodies bind to target proteins by biochemical, non-covalent interactions. Thereby, interaction strength as well as accessibility are dependent on the primary, secondary, tertiary, and quaternary structure of the proteins involved and on factors such as pH, ionic strength, temperature, and exposure time. This creates in part unique subcellular binding environments leading to heterogeneity of binding properties even for a single epitope in a fully purified system \cite{Svitel_2003}. 

Approaches such as immunocytochemical stainings are mostly negligent of these defining subcellular binding aspects and are based on wrong assumptions such as a fully established binding equilibrium. Ignoring these in part environment-specific antibody properties results in a misguided choice of antibody concentrations, undefined ratio of specific versus unspecific binding, misleading interpretation, inaccurate quantification, and\slash or unspecific targeting of therapeutic targets.

Broadly applicable technical approaches for characterization of biological system-specific binding properties of antibodies are currently missing. Gold standard approaches are e.g. the generation of knock-out\slash knock-down organisms, binding studies on expression libraries, and\slash or interaction studies by e.g. surface plasmon resonance approaches (REFS). But knock-out approaches alter the expression of many hundreds to thousands of transcripts involved in even more biological processes, thereby potentially altering the environment of the protein target of choice. Expression libraries are prone to similar artifacts, generated by strongly altered molecular environments, as well as to problems due to e.g. altered glycosylation patterns and\slash or underrepresentation of e.g. transmembrane protein expressions. Surface plasmone resonance systems are best suited for the analysis of at least in part purified protein systems, again therefore working in strongly altered binding environments. And problematic for any such approaches, a detailed analysis for avoiding false equilibrium-state assumptions is missing.

To identify the epitope accessibility properties in the system of choice, we established a computational approach that combines the concentration dependent signal intensity with binding properties, obtained from a binding model. The most often assumed Langmuir isotherm \cite{Alberti_2012} \cite{Latour_2014} cannot be applied as central assumptions such as a fully established binding equilibrium are violated by e.g. multiple finalizing washing steps. Hence, we define a model consisting of a superposition of pseudo-first-order reactions for each potential epitope class and\slash or nano-environment, with a fixed incubation time $\tau$ :
\[r_i = g_i \cdot (1-e^{-k_i \tau a})\ . \]
Here, $a$ denotes the antibody concentration\slash dilution and $g_i$ denotes the number of epitopes that have the binding rate constant $k_i$ (model derivation and validation see supplement \ref{sup-sec: accumulation model}). 

Further, we define the accessibility constant, $K_\tau \coloneqq \frac{1}{k \tau}$, which measures the accumulation rate of antibodies, in resemblance to the often used dissociation constant $K_d$. Since a priori knowledge about the number of epitope classes is impossible, the superposition is approximated as Fredholm integral equation:
\begin{equation}
	r_{\text{total}} = \int_{0}^\infty g(K_\tau) (1-e^{-\nicefrac{a}{K_\tau}})\  d K_\tau \ ,
	\label{main-eq: Fredholm equation}
\end{equation}
where $g(K_\tau)$ describes the distribution of epitope classes with respect to their binding rate\slash accessibility. Hence, we call plots\slash discrete approximations of $g(K_\tau)$ accessibility histogram.

Using Fredholm integral equations to describe heterogeneous binding is well established for binding isotherms \cite{Sips_1948} \cite{House_1978} \cite{Sposito_1980} and kinetic models \cite{Svitel_2003}\linebreak \cite{Forssen_2018} \cite{Malakhova_2020}. Because the inference of a density $g(K_\tau)$ form experimental data is often an ill-posed (curve-fitting) problem \cite{Provencher_1982}, we used the Bayesian framework to include prior information, as done by \cite{Svitel_2008}, and implemented\slash adapted a finite grid approximation akin to \cite{Zhang_2017} \cite{Zhang_2019}.

In this paper, we validate that immunocytochemical staining occurs in a non-equilibrium state. We obtained epitope accessibility distributions from dilution series immunostaining in HeLa cells with exemplary antibodies directed against NF200 and the ribosomal protein RPS11. This established optimal antibody dilutions which minimized the unspecific background. By analysis of the conformation-sensitive epitope of the regulatory subunit RII beta of PKA, we show that the accessibility distribution allows distinct quantification of the varying protein complex states \cite{Isensee_2018}. Finally, we outline and demonstrate how the accessibility histogram allows to perform a double-staining protocol which results in a binding-property-guided antibody multiplexing. 

\section*{Results}

The main focus of this paper are the applications of the accessibility analysis. For this reason, and because a thorough introduction to the involved model requires further explanation, we have deferred the details to the supplement (\ref{sup-sec: accumulation model} - \ref{sup-sec: accessibility histogram and units}). Yet, the accumulation hypothesis that arises from the lack of an equilibrium is not only essential for the equations of the model but also for the interpretation of the results that will be presented. Hence, we outline the principle and the validation of the accumulation hypothesis already here, and repeat the arguments in greater detail in the supplement.

\subsection*{Antibody accumulation}

As already mentioned in the introduction, multiple washing steps prior to the actual measurement disrupt any equilibrium that might have settled during the antibody incubation. The accumulation hypothesis assumes that antibodies which are measurable must have bound permanently. As the name suggests, these permanently bound antibodies accumulate during the incubation period. The rate of this accumulation is among others proportional to the antibody concentration. Yet, since the incubation time is the same for all antibody concentration conditions (in a dilution series experiment), higher antibody concentrations lead to more antibodies that will have bound permanently in the end. Figure \ref{fig: illustration dose-response behavoir} illustrates this principle.

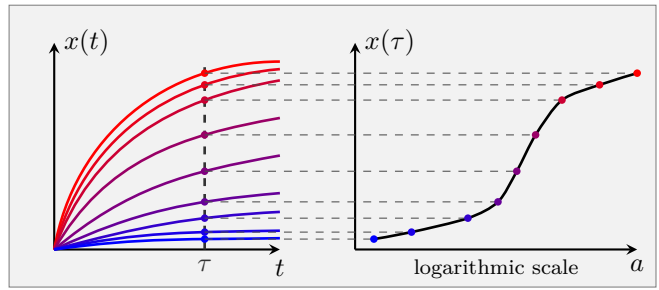
\begin{figure}[h]
	\centering
	\begin{tikzpicture}[xscale = 0.5,yscale = 0.5]
		\filldraw[fill = gray!10, draw = black!50] (-1.2,-1) rectangle (16,6.5);

		\draw[black,thick,-stealth] (0,0) -- (6,0)node[below]{$t$};
		\draw[black,thick,-stealth] (0,0) -- (0,5.5)node[right]{$x(t)$};

		\draw[black!80,dashed,line width = 1pt, name path = t] (4,0)node[below]{$\tau$} --(4,5);

		\draw[red,line width = 1pt, name path = p1] (0,0) to[out=80, in = 180] (6,5);
		\draw[red!90!blue,line width = 1pt, name path = p2] (0,0) to[out=75, in = 185] (6,4.8);
		\draw[red!80!blue, line width = 1pt, name path = p3] (0,0) to[out=70, in = 190] (6,4.5);
		\draw[red!60!blue, line width = 1pt, name path = p4] (0,0) to[out=60, in = 190] (6,3.5);
		\draw[red!50!blue, line width = 1pt, name path = p5] (0,0) to[out=40, in = 190] (6,2.5);
		\draw[red!40!blue, line width = 1pt, name path = p6] (0,0) to[out=25, in = 185] (6,1.5);
		\draw[red!20!blue, line width = 1pt, name path = p7] (0,0) to[out=15, in = 183] (6,1);
		\draw[red!10!blue, line width = 1pt, name path = p8] (0,0) to[out=12, in = 181] (6,0.5);
		\draw[blue, line width = 1pt, name path = p9] (0,0) to[out=8, in = 181] (6,0.3);

		\path [name intersections={of=p1 and t,by=t1}];
		\path [name intersections={of=p2 and t,by=t2}];
		\path [name intersections={of=p3 and t,by=t3}];
		\path [name intersections={of=p4 and t,by=t4}];
		\path [name intersections={of=p5 and t,by=t5}];
		\path [name intersections={of=p6 and t,by=t6}];
		\path [name intersections={of=p7 and t,by=t7}];
		\path [name intersections={of=p8 and t,by=t8}];
		\path [name intersections={of=p9 and t,by=t9}];

		\fill[red] (t1) circle (0.1);
		\fill[red!90!blue] (t2) circle (0.1);
		\fill[red!80!blue] (t3) circle (0.1);
		\fill[red!60!blue] (t4) circle (0.1);
		\fill[red!50!blue] (t5) circle (0.1);
		\fill[red!40!blue] (t6) circle (0.1);
		\fill[red!20!blue] (t7) circle (0.1);
		\fill[red!10!blue] (t8) circle (0.1);
		\fill[blue] (t9) circle (0.1);

		\draw[black,thick,-stealth] (8,0) -- (15.5,0)node[below]{$a$} node[below,pos = 0.5]{{\footnotesize logarithmic scale}};
		\draw[black,thick,-stealth] (8,0) -- (8,5.5)node[right]{$x(\tau)$};

		\coordinate (T1) at ($(t1) + (11.5,0)$);
		\coordinate (T2) at ($(t2) + (10.5,0)$);
		\coordinate (T3) at ($(t3) + (9.5,0)$);
		\coordinate (T4) at ($(t4) + (8.8,0)$);
		\coordinate (T5) at ($(t5) + (8.3,0)$);
		\coordinate (T6) at ($(t6) + (7.8,0)$);
		\coordinate (T7) at ($(t7) + (7,0)$);
		\coordinate (T8) at ($(t8) + (5.5,0)$);
		\coordinate (T9) at ($(t9) + (4.5,0)$);

		\draw[line width = 1pt,smooth] plot coordinates {(T1) (T2) (T3) (T4) (T5) (T6) (T7) (T8) (T9)};

		\draw[dashed, black!70] (t1) -- (T1);
		\fill[red] (T1) circle (0.1);

		\draw[dashed, black!70] (t2) -- (T2);
		\fill[red!90!blue] (T2) circle (0.1);

		\draw[dashed, black!70] (t3) -- (T3);
		\fill[red!80!blue] (T3) circle (0.1);

		\draw[dashed, black!70] (t4) -- (T4);
		\fill[red!60!blue] (T4) circle (0.1);

		\draw[dashed, black!70] (t5) -- (T5);
		\fill[red!50!blue] (T5) circle (0.1);

		\draw[dashed, black!70] (t6) -- (T6);
		\fill[red!40!blue] (T6) circle (0.1);

		\draw[dashed, black!70] (t7) -- (T7);
		\fill[red!20!blue] (T7) circle (0.1);

		\draw[dashed, black!70] (t8) -- (T8);
		\fill[red!10!blue] (T8) circle (0.1);

		\draw[dashed, black!70] (t9) -- (T9);
		\fill[blue] (T9) circle (0.1);

	\end{tikzpicture}
	\caption{The left-hand plot illustrates the antibody accumulation over time with different binding rates caused by different antibody concentrations (red = high concentration, blue = low concentration). The right-hand side plot illustrates how incubating all antibody concentration conditions for the same duration $\tau$ creates the dose-response behavior.}
	\label{fig: illustration dose-response behavoir}
\end{figure}

There is a simple validation experiment to check if the accumulation hypothesis applies: varying the incubation time. Decreasing the incubation time should decrease the response value for all concentrations. Predicting the expected behavior in detail requires further calculations (see supplement \ref{sup-sec: Validation of the accumulation hypothesis}), however. Yet, it is clear that the response decrease should not be uniform. Instead, the decrease should propagate from high concentrations (right) to low concentrations (left). 

This follows from the observation, that a certain threshold concentration is needed to cover all epitopes with antibodies during the incubation time. Any concentration that is higher will lead to the same response, as there cannot bind more antibodies then there are epitopes, which produces a saturation effect. Decreasing the incubation time then should lead to a decrease of the response for concentrations barely higher than the threshold concentration. In contrast, responses for concentrations much larger than the threshold concentration should still reach the saturation. In other word, the saturation point of the dose-response curve is reached at higher concentrations when the incubation time is decreased.

\begin{figure}[h]
    \centering
    \includegraphics[width = \columnwidth]{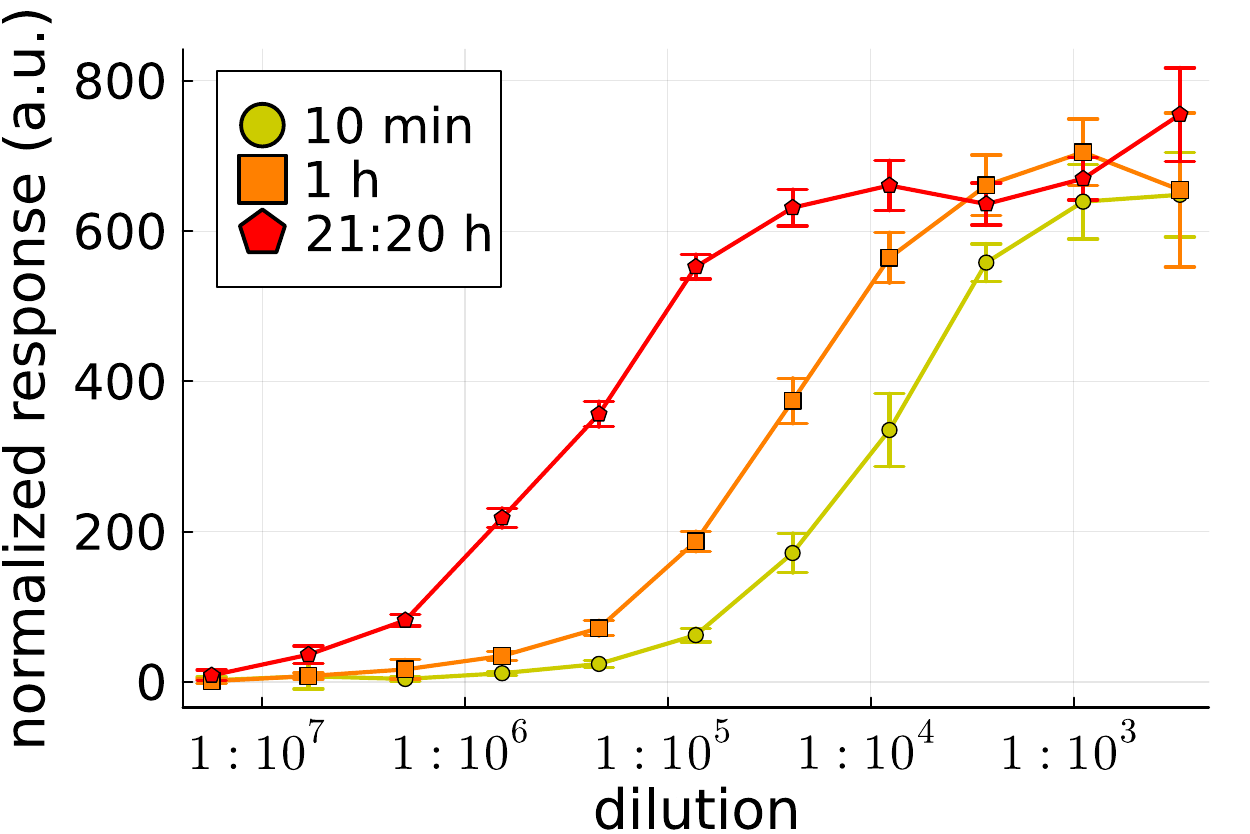}
    \caption{Dose-response curves (anti-NF200 antibody, HeLa cells) with different incubation times of the primary antibodies. Each data point is the mean of 16 replicates with the standard deviation of the replicates as error bars.}
    \label{fig: accumulation hypothesis validation}
\end{figure}

The outlined behavior is exactly what can observed in a simple experiment. Figure \ref{fig: accumulation hypothesis validation} shows 3 dose-response curves which were obtained with different incubation times: \unit[21:20]{h}, \unit[1]{h} and \unit[10]{mins}. In the margins of noise, the \unit[21:20]{h}-curve is always above the \unit[1]{h}-curve and the \unit[1]{h}-curve is always above the \unit[10]{min}-curve. Furthermore, the shift of the saturation point to higher concentrations, when shorter incubation times are used, appears as predicted. 

All in all, this validates the accumulation hypothesis. Considering the model equations in detail allows to check that further predictions are met. Furthermore, once it is validated that the accessibility histogram describes the actual binding properties of antibody binding process, one can apply the accessibility analysis to the dose-response curves of figure \ref{fig: accumulation hypothesis validation}. This will allow to speculate about the underlying antibody binding dynamics. For further information see supplement \ref{sup-sec: Validation of the accumulation hypothesis}.

\subsection*{The validation system}

\begin{figure*}[!ht]
	\centering
	\includegraphics[width = \textwidth]{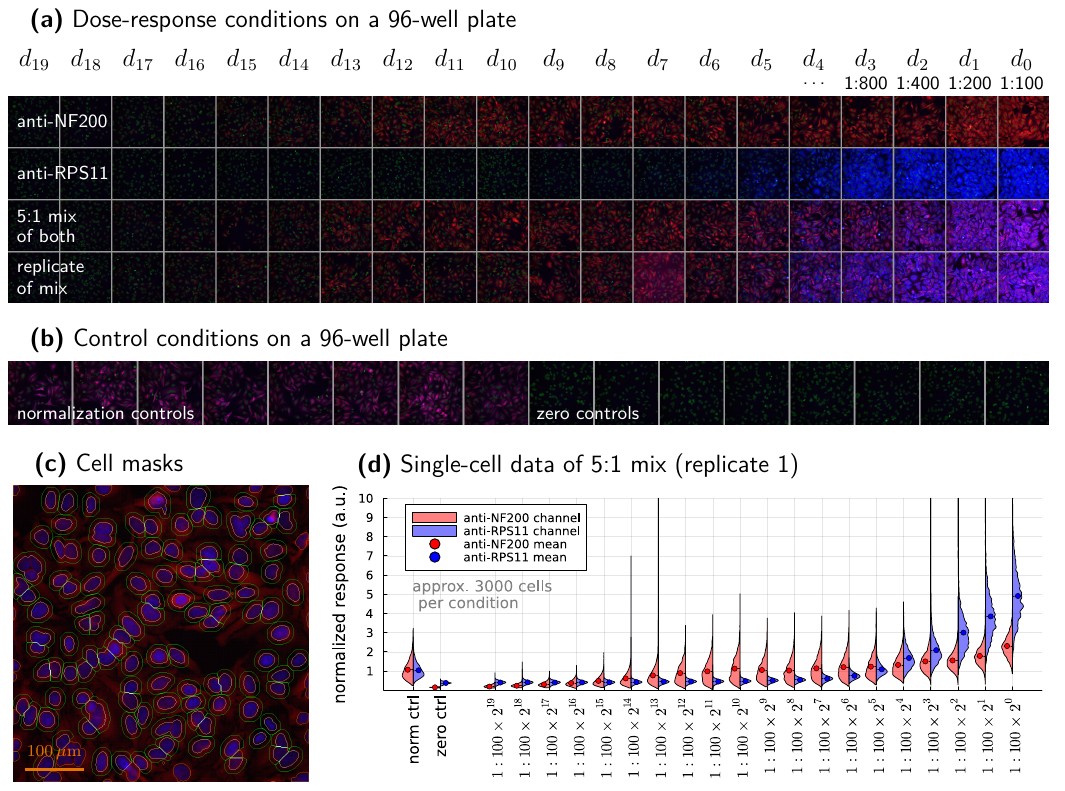}
	\caption{\textbf{Data acquisition from a 96-well plate for the validation system with HCS microscopy}\\[0.5em]
	\textbf{(In general)} The color channel signal intensities were normalized with respect to the normalization control. The total brightness of the images was uniformly increased for better visibility (does not apply to the data quantification). The colors do not represent the fluorescence emission frequencies but were chosen for better visibility in case of red-green blindness. Subfigure (c) was edited further for better visibility.\\[0.5em]
	\textbf{(a)} Example images of the 1:2 dilution series for the primary-antibody conditions: Only anti-NF200 (monoclonal mouse antibody) in the first row, only anti-RPS11 (monoclonal rabbit antibody) in the second row, 5:1 mix of anti-NF200 and anti-RPS11 in the last 2 rows (replicates). For the pre-dilutions, the anti-RPS11 antibody was diluted 1:100 and the anti-NF200 antibody was diluted 1:20 (leading to the 5:1 mixing-ratio). All pre-dilutions were labeled as 1:100 in all images\slash plots for consistent comparisons of the conditions (this essentially emulates a higher concentrated vendor dilution for the anti-NF200 antibody). In all cases both secondary antibodies (blue: anti-rabbit Alexa Fluor 555 and red: anti-mouse Alexa Fluor Plus 647) as well as Hoechst were used. \textbf{(b)} Normalization controls (first 8 wells): Only anti-NF200 diluted 1:1000 (from the actual vendor stock) but two anti-mouse secondary antibodies (blue: Alexa Fluor 555 and red: Alexa Fluor Plus 647). Zero controls (remaining 8 wells): no primary antibodies, only secondary antibodies and stains (anti-rabbit Alexa Fluor 555, anti-mouse Alexa Fluor Plus 647 and Hoechst). \textbf{(c)} Example image of the automatic quantification from the High Content Screening microscope software. Cell nuclei were identified with the Hoechst staining (here blue instead of green for better visibility in case of red-green blindness). The quantification regions (green circles) were defined by increasing the nuclei regions (white circles). Nuclei regions too close to the edges, regions that are deformed etc., were excluded (orange circles). \textbf{(d)} Distribution of average cell-responses for a single antibody-mix replicate. All responses were normalized to the average response of the normalization controls.}
	\label{main-fig: high content}
\end{figure*}

\begin{figure*}[!h]
	\centering
	\includegraphics[width = \textwidth]{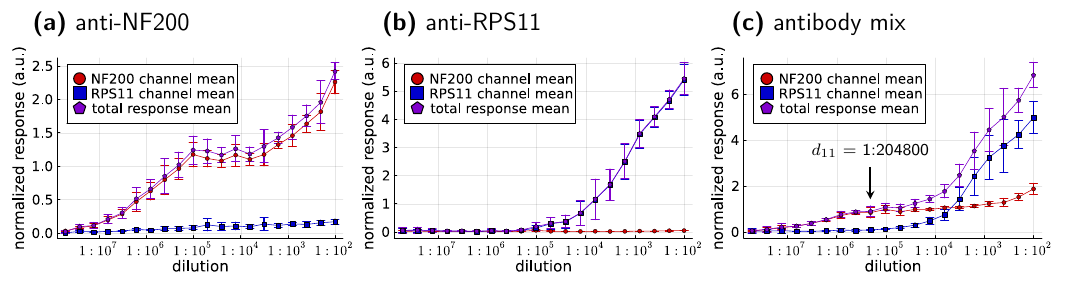}
	\caption{\textbf{Dose-response curves of the validation system (with individual scales)}\\[0.5em]
	Dose-response curves for the three conditions: \textbf{(a)} only anti-NF200 \textbf, \textbf{(b)} only anti-RPS11 and \textbf{(c)} the antibody mix. Both color channels, here red for anti-NF200 and blue for anti-RPS11, were measured in each case. The total intensity is the sum of both channels. Each individual dilution series on a plate defines a replicate, hence there are 4 replicates for the individual antibody conditions and 8 for the antibody mix condition. Two replicates were excluded completely, one for anti-NF200 and one for the antibody mix (see supplement \ref{sup-sec: removed replicates} for details). The individual does-response curves can be found in supplement \ref{sup-subsec: Individual replicates}. The optimal dilution quotient ($d_{11}$) that maximizes the anti-NF200 response while minimizing the anti-RPS11 response (splitting of the NF200 channel curve and the total-signal curve), is marked in the antibody-mix plot. Since each color channel was quantified, even if the respective antibody was not present, it is also possible to estimate the channel bleed-through: around \unit[10]{\%} from the NF200 channel to the RPS11 channel and almost negligible the other way around.}
	\label{main-fig: dose-response}

\end{figure*}

\begin{figure*}[!h]
	\centering
	\includegraphics[width = \textwidth]{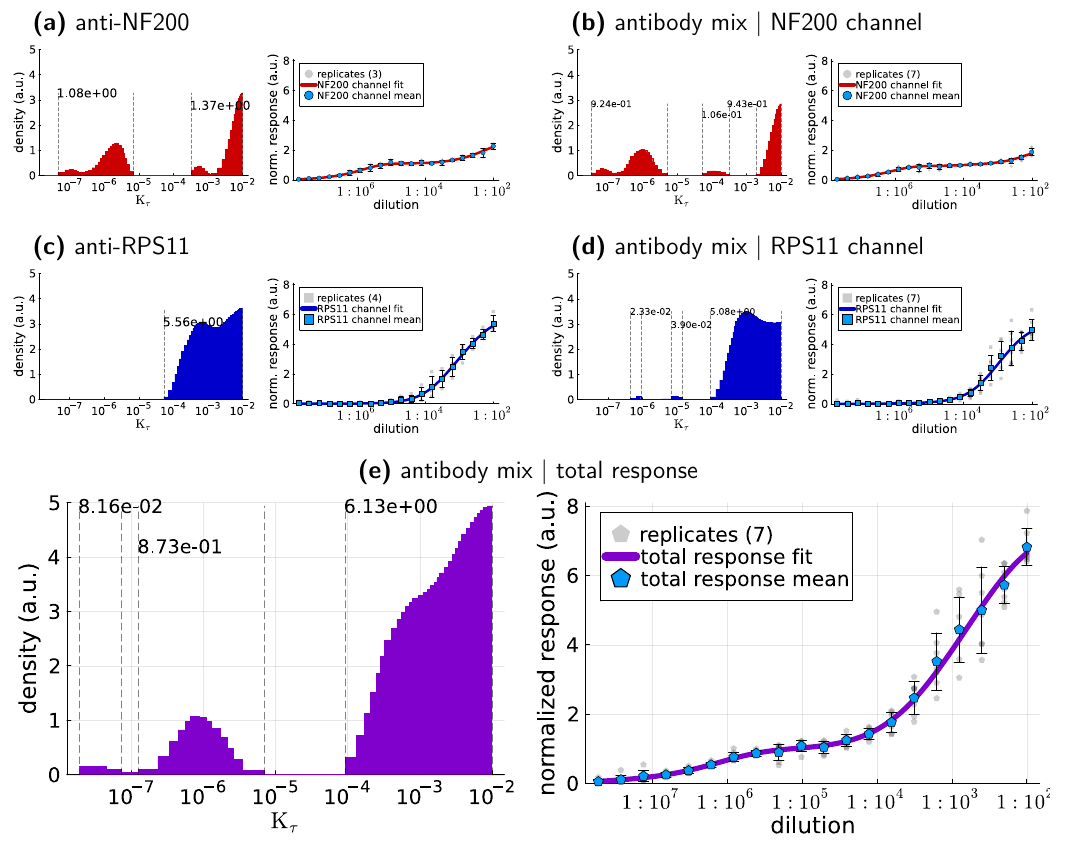}
	\caption{\textbf{Fitted curves and accessibility histograms}\\[0.5em]
	Accessibility analysis of the dose-response data (right-hand plots) from figure \ref{main-fig: dose-response} to obtain the corresponding accessibility histograms (left-hand plots). The data points are the mean values and the error bars are the standard deviations of the replicates. The black labels in the histogram plots show the amount of epitopes (in units of normalized response) of the individual peaks enclosed by the dashed lines. The differently colored plots belong to the following conditions and channels: \textbf{(a)} the NF200 channel (red) for the anti-NF200 antibody, \textbf{(b)} the NF200 channel (red) for the antibody mix, \textbf{(c)} the RPS11 channel (blue) for the anti-RPS11 antibody, \textbf{(d)} the RPS11 channel (blue) for the antibody mix and \textbf{(e)} the total response (sum of channels in violet) for the antibody mix.
	}
	\label{main-fig: histograms}
\end{figure*}

The general principle of the accessibility analysis should be applicable to all immunostaining experiments that allow for a direct or indirect numerical quantification of the amount of bound antibodies. In the same way, the accessibility analyses should work with most antibody-cell combinations. The only exception is the determination of optimal antibody concentrations\slash dilutions, which requires monoclonal antibodies as premise to draw the required conclusions from the accessibility histogram.

In contrast to the application of the accessibility analysis, the validation requires specific properties, such that the obtained results can be compared to already established methods. Thus, we chose to mix two monoclonal antibodies from different host species. Using distinct fluorescence labels then allows to measure the individual antibody contributions simultaneously by common multi-color microscopy of stained cells. The antibodies were mixed at a 5:1 ratio such that the fluorescence signal of one of the antibodies is dominant at lower concentrations. In this way, we created a system in which one of the antibodies behaves like a ``specific'' signal that appears at low concentrations, while the other antibody mimics the behavior of ``unspecific'' staining that only appears at higher concentrations. The accessibility analysis, which uses only the total signal, is validated by considering the individual contributions of the antibodies, which are obtained from the color channel measurements.

In detail, we used fixed HeLa cells and incubated them with  a monoclonal mouse anti-NF200 antibody and a monoclonal rabbit anti-RPS11 antibody. Note however, that the particular target epitopes are irrelevant for the validation. The antibodies were merely chosen because of their dose-response behavior and because they were in stock at our laboratory. To achieve the 5:1 mixing ratio, we pre-diluted the anti-NF200 antibody 1:20 and the anti-RPS11 antibody 1:100. These pre-dilutions basically function as stock solutions, where the 1:20 dilution emulates that the antibody vendor would provide a higher concentrated antibody solution. Consequently, we labeled both pre-dilutions as 1:100 and use dilutions instead of concentrations in this paper. We define the term ``dilution quotient`` to disambiguate what lower dilution means: $1:200$ is a lower ``dilution quotient'' than $1:100$, etc.. 

Working with dilution quotients instead of concentrations does not affect the accessibility analysis. First, dilution quotients are proportional to the actual concentration because of the order ($1:100 < 1:200$) that we chose. Thus, dilution quotients constitute just a change of units for the concentration. Second, the shape of the accessibility histogram is invariant under unit changes (see supplement \ref{sup-sec: accessibility histogram and units} for a proof and further details).

Since the accessibility analysis is applied to dose-response data, 4 dose-response curves were prepared per 96-well plate (4 replicate plates in total): only anti-NF200, only anti-RPS11 and the antibody mix (twice). The antibody-mix condition is used for the validation, while the other conditions are mainly used for additional controls and comparisons. For each condition, a 1:2 dilution series was performed with 20 dilution steps from the pre-dilutions (Fig. \ref{main-fig: high content}a). In addition, each plate contained 8 zero controls and 8 normalization controls (Fig. \ref{main-fig: high content}b). The zero controls did not contain the primary antibodies and were used to remove the signal offset created by e.g. autofluorescence. The normalization controls contained only anti-NF200 with a 1:1 mixture of anti-mouse Alexa Fluor 555 and anti-mouse Alexa Fluor Plus 647. The resulting signal intensities were used as reference values to normalize all other measured signal intensities to (see methods for details).

All antibody signal intensities were obtained with a High Content Screening (HCS) microscope, identifying and quantifying around 3000 cells per well (Fig. \ref{main-fig: high content}c). Figure \ref{main-fig: high content}d shows exemplarily the single-cell data (average responses of the cells) for a single dilution-series replicate of the antibody mix. Using an HCS microscope allowed for the quantification of thousands of cells. This large number ensures a good estimate of the population mean, as the standard error of the mean scales with $\frac{1}{\sqrt{n}}$. Consequently, the population mean was used for the dose-response curves. 

As can be seen in figure \ref{main-fig: high content}a, the anti-NF200 signal (red) begins to appear at lower dilution quotients, compared to the anti-RPS11 signal (blue). This is true for both the individual antibody conditions and the antibody mix condition, where the color turns from red to violet (i.e. red + blue). Inspecting the dose-response curves (Fig. \ref{main-fig: dose-response}) confirms this visual inspection. Thus, as intended we have constructed a validation system where the anti-RPS11 signal (blue) mimics the unwanted background signal. Hence, it is possible to determine the optimal antibody dilution quotient for our artificial system from the color channel responses, by picking the largest dilution quotient at which the anti-RPS11 signal (blue) does not contribute. Inspecting the dose-response curves of the antibody mix in figure \ref{main-fig: dose-response}c, this optimal dilution quotient is given by the splitting point of the anti-NF200 and the total response curve, which is around $d_{11} = 1:205800$.

\subsection*{Accessibility histograms describe binding properties}

\begin{figure*}[!h]
	\centering
	\includegraphics[width = \textwidth]{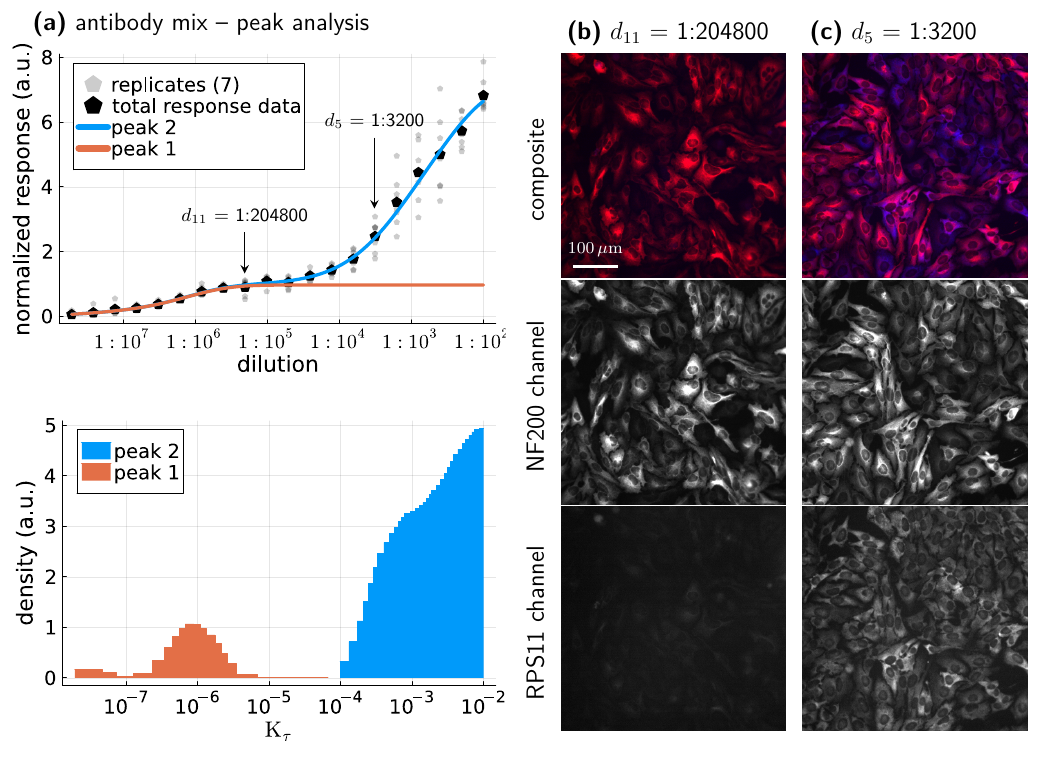}
	\caption{\textbf{Identification of the optimal dilution quotient with the accessibility analysis}\\[0.5em]
	\textbf{(a)} Individual contributions to the dose-response curve of the low-$K_\tau$ peaks (orange, $K_\tau \leq 10^{-5}$) and the high-$K_\tau$ peaks (azure, $K_\tau \geq 10^{-4}$). The dose-response curve of the high-$K_\tau$ peaks is added on top of the dose-response curve of the low-$K_\tau$ peaks. The splitting point between the individual contributions (colored curves) defines the optimal antibody dilution quotient. Both the optimal dilution quotient ($d_{11}$) and the dilutionq quotient at which both peaks contribute equally to the dose-response signal ($d_5$) are marked in the dose-response plot. \textbf{(b)} Example images for the $d_{11}$ dilution quotient. The brightness of all images was increased for better visibility. On top is the composite, in the middle is the monochrome image for the NF200 channel and at the bottom is the monochrome image for the RPS11 channel. The RPS11 channel image appears to be a faint copy of the NF200 channel image because of the 10\% channel bleed-through.  \textbf{(c)} Example images for the $d_5$ dilution quotient with the same brightness increase and the same order of channel-images.
	}
	\label{main-fig: optimal dilution}
\end{figure*}

\begin{figure*}[!h]
	\centering
	\includegraphics[width = \textwidth]{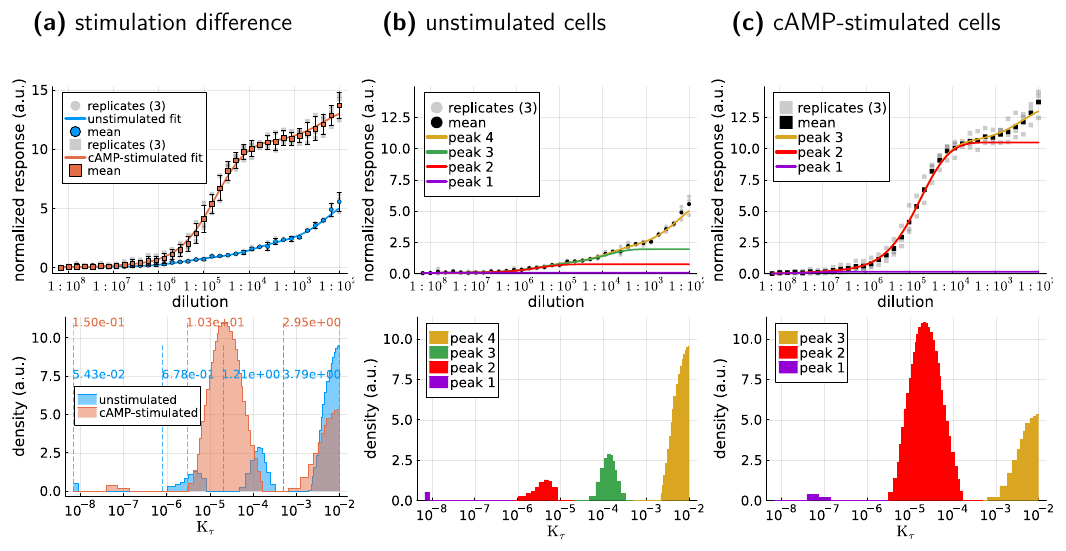}
	\caption{\textbf{Accessibility analysis of the pRII-antibody in  unstimulated and cAMP-stimulated HeLa cells}\\[0.5em]
	\textbf{(a)} Dose-response data and fitted curves for unstimulated HeLa cells (azure) and cAMP-stimulated HeLa cells (orange) as well as the corresponding accessibility histograms below. The color matched labels show the amount of eptiopes in units of normalized response for the individual peaks enclosed by the color matched dashed lines. \textbf{(b)} Peak analysis for the unstimulated cells, showing the contribution of the individual peaks to the dose-response curve. The color matched contributions to the dose-response curve are cumulative, i.e. added to the responses of peaks from left (low $K_\tau$) to right (high $K_\tau$). \textbf{(c)} Same as (b) for the cAMP-stimulated cells.
	}
	\label{main-fig: PKA}
\end{figure*}

While only the antibody mix condition is necessary for a validation, the other conditions provide good consistency checks for the accessibility analysis and the experimental design. Consequently, figure \ref{main-fig: histograms} shows the fitted dose-response curves and the corresponding accessibility histograms for all conditions. As already mentioned in the introduction, the histograms represent the accessibility distributions, i.e. the amount of epitopes (measured in units of the normalized response) for a given $K_\tau$. Thus, isolated\slash separate peaks indicate the presence of separated epitope classes.

Since the accessibility histograms have to cover several orders of magnitude, a logarithmic scale is used. The bin heights are divided by the visual widths of the bins such that the area of a peak in the plot corresponds to its effect on the dose-response curve (see methods and supplement \ref{sup-subsec: histogram visualization} for further details).

For the model fitting, we used a large regularization parameter to obtain smooth histograms (feature parsimony) whose corresponding model curve still describes the data well. That is, in all cases the fitted curves describe the data within the standard deviations of the replicates. The details about the fitting objective function and the chosen regularization\slash prior probability can be found in the methods section and the results for a smaller regularization parameter can be found in supplement \ref{sup-sec: histograms with weaker regularization}. Since the regularization parameter affects the peak widths, \cite{Forssen_2018} argues (in the context of kinetic models) that only the number of peaks and their positions should be considered. Nevertheless, it can be observed that similar dose-response curves produce similar histograms, which corroborates the consistency of the accessibility histogram. 

Comparing only the dose-response data in figures \ref{main-fig: histograms}(a-d) shows that the antibody mix produces indeed a superposition of two independent dose-response signals, as demanded for the validation system. 
By construction of the experiment, the total signal intensity (Fig. \ref{main-fig: histograms}e) is the sum of both color channels, i.e. the sum of both antibody signal intensities. What does not follow from the experiment design is that the same can be observed for the accessibility histograms: The histogram for the total signal intensity is approximately the sum of the histograms for the individual color channels. In other words, although the total signal intensity does not contain any information about the contributions of the individual antibodies, the histogram obtained from total-signal data still retrieves the individual features of both antibodies. This indicates that the accessibility histogram describes actual binding properties of antibodies.

In detail it can be observed that the peaks at $K_\tau \leq 10^{-5}$ of the anti-NF200 antibody are in the total-signal histogram as are the peak(s) at $K_\tau \geq 10^{-4}$ of the anti-RPS11 antibody (added together with the peaks of the anti-NF200 antibody at $K_\tau \geq 10^{-4}$). That the RPS11 histogram probably consists of 2 peaks can be seen in supplement \ref{sup-sec: histograms with weaker regularization}, where a weaker regularization is used. 

To show how the optimal antibody dilution quotient can be inferred from the accessibility histogram, we need to define what ``optimal'' means.  The optimal dilution quotient should minimize the unspecific background while maximizing the signal from the actual target epitopes. Since a slow binding rate (large $K_\tau$) corresponds to higher dilution quotients, where unspecific background is most prominent, it is reasonable to define the optimal dilution quotient as follows: the optimal dilution quotient maximizes the contributions from fast-binding peaks (small $K_\tau$) while minimizing the contributions from slow-binding peaks (large $K_\tau$). 

Strictly speaking, our definition of the optimal antibody dilution quotient requires monoclonal antibodies, as we assign high-$K_\tau$ peaks to unspecific binding (which may also include a cross-reaction to a different epitope). This assumption breaks down for polyclonal antibodies, which consist of a mixture of antibodies against different epitopes. High-$K_\tau$ peaks can belong to a highly specific subgroup of antibodies in the polyclonal mix that just happen to be highly dilute\slash rare. In fact, we have used this property to create our artificial validation system to behave as needed for the validation.

Because the accessibility model calculates the dose-res\-ponse curve from the histogram, i.e. the histogram is the parameter of the model, it is easy to obtain the relative contributions of different peaks. Figure \ref{main-fig: optimal dilution}a shows the individual dose-response curve contributions of the fast-binding peaks (orange) and the slow-binding peaks (azure) as cumulative plot (slow-binding curve added on top of the fast-binding curve). The splitting point of both curves thus marks the dilution quotient at which the slow-binding peak begins to contribute to the overall dose-response curve. Hence, the splitting point marks the optimal dilution quotient according to our definition. In this case the optimal dilution quotient is around $d_{11} = 1:204800$, coinciding with the optimal dilution quotient obtained from the color channel analysis in figure \ref{main-fig: dose-response}. Note, that only the total-signal data, devoid of any color channel information, was used to obtain the optimal dilution quotient with the accessibility analysis. This corroborates the use of the accessibility analysis for the identification of optimal antibody concentrations\slash dilution quotients in case of monoclonal antibodies.

We observe, that the image in figure \ref{main-fig: high content}a that corresponds to the optimal dilution quotient $d_{11}$ shows only a weak staining. This is inevitable, as we had to pick microscopy exposure times such that there is no signal clipping (for the images of the dose-response curve with the highest dilution quotients). In general, however, after having obtained the optimal dilution quotient from a dose-response experiment, one would adjust the exposure time for the experiment of interest. This signal amplification, be it due to longer exposure times or digital gain, is then justified, as the optimal dilution quotient minimizes the contribution from the unspecific background. To visually compare the optimal dilution quotient $d_{11}$ with the dilution quotient $d_5$ at which the anti-NF200 signal and the anti-RPS11 signal contribute equally (as determined by the accessibility analysis), we have increase the brightness of the images (Fig. \ref{main-fig: optimal dilution}b,c).

\subsection*{Conformational changes are visible in accessibility histograms}

 Both the fact that the sum of dose-response curves leads to the sum of the individual accessibility histograms and the consistency of histograms for similar dose-response curves indicate that the accessibility histogram corresponds to the binding properties of antibodies. If that is the case, a change of the target epitope accessibility should be observable in the accessibility histogram. For this purpose, we investigated dose-response curves of the pRII-antibody (see methods) against a phosphorylated epitope (pRII-epitope) of the regulatory subunit of PKA type II (PKA-II).

 At this point we only care about the gross structure of PKA-II, the location of the epitope and the effect of cAMP-stimulation: Inactive PKA-II consists of two catalytic and two regulatory (RII) subunits bound together. Upon cAMP stimulation, the catalytic subunits separate from the regulatory subunits. The results of \cite{Isensee_2018} suggest that
 \begin{enumerate*}[label=(\roman*)]
	\item the pRII-epitope is already phosphorylated in inactive PKA-II
	\item the pRII-epitope is blocked by the catalytic subunits in inactive PKA-II and becomes accessible upon cAMP-stimulation
	\item there is some conformational flexibility, making some of the pRII-epitope accessible even without a cAMP-stimulation.
 \end{enumerate*}

In summary, it can be expected that there is a pRII-antibody signal even for unstimulated HeLa cells. Be it because of the conformational flexibility or because of a basal activity of PKA during the cell fixation. For cAMP-stimulated cells, the signal intensity should then increase. Based on the results of \cite{Isensee_2018}, the accessibility histogram for unstimulated cells should have peaks that belong to accessible epitopes and the signal increase in cAMP-stimulated cells should correspond only to an increase of these accessible epitopes.

Figure \ref{main-fig: PKA}a shows the dose-response curves and the corresponding accessibility histograms for both unstimulated (azure) and cAMP-simulated (orange) HeLa cells. Figures \ref{main-fig: PKA}b,c show the corresponding peak analysis plots that were discussed in the previous section. It can be observed that the two middle peaks for unstimulated cells have become a single, larger and broader peak for cAMP-stimulated cells. In fact, the increase of epitopes is more than fivefold ($0.678 + 1.21 \rightarrow 10.3$ in units of normalized response). On the other hand, there is a minor decrease of epitopes ($3.79 \rightarrow 2.95$) and a slight broadening of the peak at the right-hand edge of the histogram. But overall, the effect of the rightmost peak on the dose-response curve is almost identical in both stimulation conditions, as can be seen in figures \ref{main-fig: PKA}b,c . Thus, the response increase for the cAMP-stimulated cells can be attributed completely to the increased middle peak.

 Based on the three major peaks one can define 3 accessibility classes of epitopes for the unstimulated cells. Since the least accessible epitopes (rightmost peak) do not seem to be affected by cAMP-stimulation and since they only account for the behavior of the data points with the largest antibody dilution quotients, it stands to reason that this is an unspecific background signal. This suggests $1:10^4$ as optimal dilution quotient for immunofluorescence staining, which is 20-times more dilute than the vendor suggestion ($1:500$). In any case, the other 2 peaks belong to more accessible epitopes, most likely the accessible pRII-epitopes. The massive increase of the amount of epitopes (middle peak) then meets the expectation that cAMP-stimulation exposes many pRII-epitopes. In this sense, the accessibility analysis agrees with the findings of \cite{Isensee_2018}. 

\subsection*{Computational multiplexing with accessibility histograms and double-staining}
\begin{figure*}[!h]
	\centering
	\includegraphics[width = \textwidth]{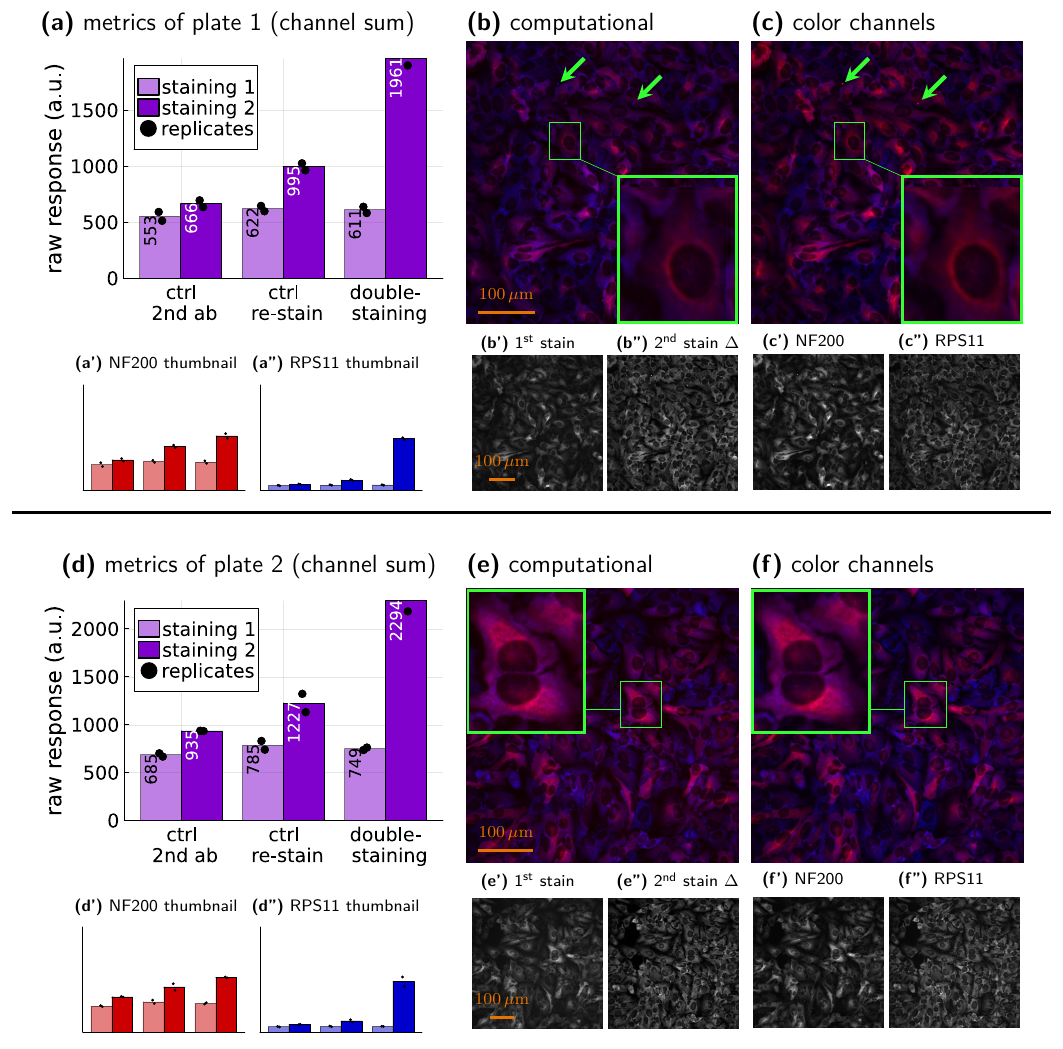}
	\caption{\textbf{Double-staining metrics and images}\\[0.5em]
	\textbf{(a)} Metrics of the different conditions, quantified with the HCS microscope: 2 wells for each condition with around 3000 cells quantified per well. The scatter points are the mean values over all quantified cells in the same well and the bars show the mean values of both wells. Light shades (left bars) are for the 1st staining and dark shades (right bars) are for the 2nd staining. The numbers in the bars are the rounded numeric values of the bars. \textbf{(a', a'')} Thumbnails for the bar plots of the individual channels (full plots in supplement figure \ref{sup-fig: optimal metrics}). \textbf{(b)} Computational composite of the double-staining protocol. The green arrows indicate dust\slash debris that got assigned to the wrong color. \textbf{(b')} 1st staining image (total signal) with correction for the double-application of secondary antibodies (brightness increase according to the fold change of ctrl 2nd ab in (a)). \textbf{(b'')} Difference image of the 2nd staining image (see supplement figure \ref{sup-fig: graphical illustration}) and the brightness-increased 1st staining image (b'). \textbf{(c)} Color composite from the channels (anti-NF200 in red, anti-RPS11 in blue) after the 2nd staining. The green arrows show again the dust\slash debris spots. \textbf{(c')} Monochrome NF200 channel image after the 2nd staining. \textbf{(c'')} Monochrome RPS11 channel image after the 2nd staining.
	\textbf{(d-f)} Independent repetition of the experiment.
	}
	\label{main-fig: double-staining}
\end{figure*}

For the optimal dilution quotient, as explained above, it is important to separate the response of low-$K_\tau$ peaks (left) from the response of high-$K_\tau$ peaks (right). But in figure \ref{main-fig: PKA}a, there are three peaks for unstimulated HeLa cells and one might want to isolate the contribution of only the middle peak. Unfortunately, there is no single dilution quotient that achieves this. As the cumulative peak analysis plots (Fig. \ref{main-fig: optimal dilution}a and \ref{main-fig: PKA}b,c) already suggested, the contributions of fast-binding\slash low-$K_\tau$ peaks appear first and are always present for dilution quotients at which slower-binding\slash high-$K_\tau$ peaks begin to contribute. But staining the same sample twice should allow to isolate the contribution of a high-$K_\tau$ peak.

In detail this means that the cells are stained first with the antibody dilution quotient that maximizes the contribution of the the low-$K_\tau$ peaks. Then, a microscopy image of the resulting staining is captured. After this, the cells are stained again, this time with a higher antibody dilution quotient (i.e. higher concentration). Then, a second microscopy image is created. By choice of the antibody dilution quotients, the 1st staining image would consist only of the the signal from the low-$K_\tau$ peaks. The 2nd staining image would then additionally contain the contributions from the high-$K_\tau$ peaks of interest. Thus, aligning the images and subtracting the pixel intensities of the 1st staining image from the 2nd staining image should isolate the contributions of the peaks of interest (see supplement Fig. \ref{sup-fig: graphical illustration} for the images and a graphical illustration of the idea).

Since we must be able to validate the results of the double-staining approach, we decided to use the validation system once again. The idea is that the low-$K_\tau$ peaks (orange in figure \ref{main-fig: optimal dilution}a) belong solely to the anti-NF200 antibody (red channel). And the high-$K_\tau$ peaks (azure in figure \ref{main-fig: optimal dilution}a) belong partly to the anti-RPS11 antibody (blue channel). Thus, the goal is to separate the individual antibody responses from only the gray-scale total intensity images by using the double-staining approach. Comparing the results with the individual antibody responses, using the color channels, then allows to assess the double-staining approach. 

However, the high-$K_\tau$ peaks consist of peaks from both antibodies, anti-RPS11 and anti-NF200 (see Fig. \ref{main-fig: optimal dilution}) -- a property we also observed for other antibody mixtures that we tested as potential validation systems. In other words, even if the attribution of the signals to the peaks works flawlessly, the signal attributed to the high-$K_\tau$ peaks will comprise both antibody signals. This complicates a direct comparison to the color channels and constitutes a deficit of the validation system for the validation of the double-staining approach. Fortunately, though, the high-$K_\tau$ peaks of the anti-NF200 antibody has a higher $K_\tau$ than parts of the anti-RPS11 peaks. Thus, avoiding too high 2nd staining concentrations alleviates the deficits of the validation system. Note, however, that our general recommendation (for 2 peaks) is to use high 2nd staining concentrations to obtain a meaningful signal increase after the 2nd staining.

Having separated the contributions of the individual peaks allows to create a composite image by using different colors for the individual contributions. For brevity, we refer to this image as computational composite in the following. This computational composite can easily be compared to the color channel composite of the 2nd staining, for a validation of the double-staining approach. 

Before we can continue with the results of the double-staining, there are some caveats with the double-staining approach that need to be addressed first. Any dust or debris that moves or is added in-between the stainings will be wrongly attributed (green arrows in Fig. \ref{main-fig: double-staining}b,c). Furthermore, the model for the accessibility analysis treats both the dynamics of the secondary antibodies and any depletion effects as parts of the binding properties for a given protocol (details in supplement \ref{sup-sec: accumulation model}). Hence, the results from the accessibility analysis are only valid for the very experimentation protocol and the very system they were obtained from. This is because using a different dilution of secondary antibodies, or as in this case applying the secondary antibodies twice, changes the obtained response for a given dilution quotient of the primary antibody. 

For the double-staining experiment (Fig. \ref{main-fig: double-staining}) we used the dilution quotient $d_{10} = 1:102400$ as 1st staining antibody-mix dilution. Although we have identified $d_{11} = 1:204800$ as optimal dilution quotient in the previous sections, the next higher dilution quotient was picked to a ensure saturation of signal contributions from the low-$K_\tau$ peaks (orange in Fig. \ref{main-fig: optimal dilution}). This is possible here, as the contributions of the high-$K_\tau$ peaks (azure) are still negligible for $d_{10}$. For the 2nd staining we chose the dilution quotient $d_6 = 1:6400$ to get a usable anti-RPS11 signal while avoiding to much signal increase from the anti-NF200 antibody. If not for the validation system deficits, we would have chosen $d_0 = 1:100$ here. To account for the changed protocol (secondary antibodies applied twice), we added a secondary antibody control (ctrl 2nd ab), where only the secondary antibodies were re-applied for the 2nd staining. Finally, to estimate the deficiency of the validation system, we added a re-stain control (ctrl re-stain), where the 1st staining was repeated for the 2nd staining. In all cases, secondary antibodies were applied with a 1:1000 dilution.

Figure \ref{main-fig: double-staining}a shows the metrics of the controls and the double-staining condition, quantified with the HCS microscope. As expected, even adding secondary antibodies twice leads to a signal increase. Re-staining with the same dilution quotient $d_{10}$ produces an even stronger increase, but not as strong as using the higher dilution quotient $d_6$. The fact that the anti-NF200 signal almost doubled in the double-staining condition shows again the limitations of the validation system. Note that the obtained metrics were not evaluated statistically, as their sole purpose is to provide rough estimates. For example, we use the secondary antibody control as conservative estimate for changes due to applying primary and secondary antibodies twice. That is, we increase the brightness of the 1st staining images by the fold change of the secondary antibody control.

Figures \ref{main-fig: double-staining}b',b'' show the images for the individual peak contributions and figure \ref{main-fig: double-staining}b shows the computational composite. Comparing these results to the color channels (Fig.\ref{main-fig: double-staining}c-c'') corroborates the double-staining approach, even under suboptimal conditions that are caused by the validation system deficits. The double-staining successfully retrieved the contributions of the individual antibodies, at least qualitatively, using only the total-signal-intensity images. Figures \ref{main-fig: double-staining}d-f'' show an independent repetition of the double-staining experiment to check replicability. In fact, we analyzed two different samples for each repetition, capturing multiple distinct images per sample, which corroborates the replicability further (see supplement figure \ref{sup-fig: double-staining view fields}).

At this point one might wonder if the accessibility analysis is really necessary. After all, staining twice and taking images in between is possible without the accessibility analysis. However, picking the staining dilutions without the information from the accessibility histograms, i.e. using arbitrary dilutions, leads to arbitrary signal separations that are meaningless. The results for using wrong dilutions can be found in supplement \ref{sup-sec: double-staining additional figures}.

\section*{Discussion}
In this paper, we have demonstrated applications of the computational accessibility analysis that we have implemented in open source Julia packages \cite{Tschimmel_2024_AntibodyPackages} to provide easy access. 

First, we have demonstrated that the accessibility analysis can determine optimal antibody concentrations\slash\hspace{0pt}dilution quotients for monoclonal antibodies that minimize the contribution of potentially unspecific background signals. The primary application may be the selection of antibody concentrations\slash dilution quotients for immunocytochemical and immunohistochemical staining experiments, which is often guided only by subjective visual interpretation. But the accessibility analysis could also be used to optimize immmunostainings seperately with simple dose-response experiments, before conducting further experiments like e.g. flow cytometry.

Second, we have provided evidence that the accessibility histogram describes binding properties of antibodies. Furthermore, we have shown that the accessibility analysis can be used to investigate accessibility changes of epitopes that my be caused by conformational changes or changes of the local epitope environment. Thus, the accessibility analysis could offer a cheap and easy first analysis to investigate e.g. conformational changes from simple dose-response experiments. In fact, the accessibility analysis should be applicable to any antibody-based experiment that allows a quantification of the amount of bound antibodies. The ability to vary the antibody concentration\slash dilution quotient is just a soft requirement, as varying the antibody concentration\slash dilution quotient could always be achieved by repeating the experiment from scratch. Thus, repeating e.g. a Western Blot experiment with different antibody concentrations\slash dilution quotients could allow to correlate the signal intensity to both the protein weight\slash size and the epitope binding rate\slash accessibility. 

Finally, we have illustrated how the binding information of the accessibility histograms can be used to investigate the epitope landscape further with a simple double-staining approach. Staining the same cell sample twice with concentrations\slash dilution quotients determined from the accessibility histogram allows to locate the contributions of accessibility histogram peaks within microscopy images. The basic idea of double-staining is akin to the label-coding described in \cite{Chen_2015} \cite{Eng_2019} \cite{Goltsev_2018}, where the label response code obtained from sequential imaging allows for a computational multiplexing. Yet, the working principle is fundamentally different. In case of the accessibility analysis, the multiplexing arises from the binding-property-heterogeneity of epitopes. Because of this, our binding-property-guided multiplexing could be well fit for the investigation and application of bi-specific antibodies. So far, bi-specific antibodies seem to be of interest only for therapeutics \cite{Chames_2009} \cite{Runcie_2018} \cite{Segaliny_2023}, as simultaneous staining of multiple targets is not desirable for immunostaining-antibodies. Hence, modeling approaches for bi-specific binders often have a therapeutic focus or investigate general binding properties \cite{Doldan-Martelli_2013} \cite{Harms_2014} \cite{Rhoden_2016} \cite{Steeg_2016}. The accessibility analysis could help in the research of bi-specific antibodies. Going one step beyond, our binding-property-guided multiplexing could even render bi-specific antibodies viable for immunostaining, allowing for a single-label multiplexing of different targets.

Despite providing replicable and consistent results for the applications that we have shown, the underlying pseudo-first-order model is only an empirical model. For example, neither antibody depletion effects nor the dynamics of the secondary antibodies are modeled. A major issue for extended models is, however, that any additional parameters to account for these effects are hardly measurable. And although depletion can be modeled, the superposition of epitope classes then leads to a coupled dynamical system without analytical solution. Furthermore, unit conversion factors do not cancel out anymore. Thus, precise measurements of the antibody concentration and the density of bound antibody-epitope complexes would become necessary. There is, however, a simple worst-case correction for depletion effects similar to the concentration correction described in \cite{Edwards_1998} that we discuss in supplement \ref{sup-sec: depletion correction}. 

Furthermore, it should be noted that information provided by the accessibility histogram depends on the dose-response curve. On the one hand, this means that the accessibility histogram can only be taken seriously if one can assume that the dose-response curve is characteristic for the system under investigation. Using only a single replicate that may be affected by strong noise\slash biological variability or data with strong outliers thus can produce wrong results. On the other hand, the information of the accessibility histogram is incomplete\slash inconclusive when the dose-response curve does not capture the whole binding behavior of the antibody. If the dose-response curve does not start from the baseline for the low concentrations\slash dilution quotients, low-$K_{\tau}$ peaks may be missed. However, this problem can easily be fixed by additional dilution steps. The opposite case, a dose-response curve that is flat and close to the baseline except for a few data points, occurs if the highest concentration\slash dilution quotient is not high enough. Often, this cannot be fixed easily, either because the antibody stock is not concentrated enough, or higher antibody concentrations\slash dilution quotients are too expensive. In this case, the histogram will show only a single peak, clipping to the right of the histogram, which does not provide any meaningful information. However, this scenario may be rephrased as warning sign, telling that among the tested concentrations\slash dilution quotients there is not enough information to distinguish epitope classes.

\begingroup
\raggedright
\printbibliography
\endgroup

\end{refsection}

\section*{Methods}

\subsection*{Cell culture}

HeLa cells were cultured in growth medium (DMEM with \unit[1]{\%} penicillin-streptomycin $\unit[10^4]{\nicefrac{U}{mL}}$ solution and \unit[10]{\%} fetal calf serum) at $\unit[37]{^\circ C}$ and \unit[5]{\%} CO\textsubscript{2} in the incubator. For passaging and seeding, the cells were washed with PBS and incubated with trypsin for $\unit[2]{mins}$ before adding fresh medium and subsequent centrifugation. After removing the supernatant, the cells were re-suspended with fresh medium. For the seeding, the cells were counted with a Neubauer-ruled hemocytometer and diluted with medium to obtain a cell density of approximately $\unit[10^8]{L^{-1}}$. Then $\unit[100]{\mu L}$ of the cell solution were added to each well of a $\mu$Clear 96-well plate, leading to approximately $10^4$ cells per well.

\subsection*{Stimulation and fixation}

After seeding the HeLa cells, the cell plates were incubated for \unit[24]{h} before cAMP-stimulation (only applicable for pRII-antibody experiments) and fixation. 

For the cAMP-stimulation $\unit[10]{\mu M}$ Sp-8-Br-cAMPS-AM were used per well. For this, Sp-8-Br-cAMPS-AM was diluted in DMSO and mixed with medium, taken from the cell plate. The mixing was done in a separate 96-well V-bottom plate on a $\unit[37]{^\circ C}$-pre-warmed thermal pad. Then the mixture was re-applied from the V-bottom plate back to the cell plate. Afterwards, the cells were incubated at $\unit[37]{^\circ C}$ and \unit[5]{\%} CO\textsubscript{2} for $\unit[15]{mins}$. Immediately after the incubation, the cells were fixed.

For the fixation, $\unit[100]{\mu L}$ of an \unit[8]{\%} Paraformaldehyde (PFA) solution were added to the $\unit[100]{\mu L}$ of medium (also containing stimulation compounds if applicable) that were already present in the well, leading to a \unit[4]{\%} PFA solution. After \unit[5]{mins} half of the well-content was removed and after another \unit[5]{mins} the cells were washed with PBS.

When the plates were not used immediately, they were sealed and stored in a fidge at $\unit[4]{^\circ C}$.

\subsection*{Immunocytochemistry}

The cells were permeabilized and blocked with \unit[2]{\%} normal goat serum, \unit[1]{\%} BSA, \unit[0.1]{\%} Triton X-100 and \unit[0.05]{\%} Tween 20, all dissolved in PBS, for \unit[1]{h} at room temperature. The primary antibodies were diluted in PBS with \unit[1]{\%} BSA according to the experiment plan in a separate 96-well V-bottom plate. After removing the blocking\slash permeabilization solution and applying the primary antibody solution, the cell plates were stored in the fridge at $\unit[4]{^\circ C}$ for \unit[24]{h}.

At the next day, the cells were washed 3 times with PBS. The plates were kept at room temperature for \unit[10]{mins} between the washings. Then the secondary antibodies and Hoechst, diluted 1:1000 in PBS, were applied. After this, the plates were kept in a dark environment for \unit[1]{h} at room temperature. Finally, the cells were washed 3 times with PBS, again with \unit[10]{mins} between the washings.\vspace{1em}

\noindent
\textbf{\small Validation system}\\
For the validation dose-response experiments and double-staining experiments the following stains, primary antibodies and secondary antibodies were used:

\begin{compactitem}
	\item \textbf{anti-NF200}: mouse monoclonal, anti-Neurofilament 200, Sigma-Aldrich, N0142-.2ML
	\item \textbf{anti-RPS11}: rabbit monoclonal, anti-Ribosomal protein S11\slash RPS11, Abcam, ab175213
	\item \textbf{Hoechst} 34580: Sigma-Aldrich, 63493
	\item \textbf{Alexa Fluor Plus 647}: goat polyclonal anti-mouse,  Invitrogen, \#A32728
	\item \textbf{Alexa Fluor 555}: goat polyclonal anti-rabbit,  Invitrogen, \#A-21429
	\item \textbf{Alexa Fluor 555}: goat polyclonal anti-mouse, Invitrogen, \#A-21424 (used for normalization controls)
\end{compactitem}\vspace{1em}

\noindent
\textbf{\small PKA experiments}\\
For the PKA experiments the following stains, primary antibodies and secondary antibodies were used:

\begin{compactitem}
	\item \textbf{anti-pRII}: rabbit monoclonal, anti-PKA R2\slash PKR2 (phospho S99), Abcam, ab32390
	\item \textbf{Hoechst} 34580: Sigma-Aldrich, 63493
	\item \textbf{Alexa Fluor Plus 647}: goat polyclonal anti-rabbit, Invitrogen, \#A32733
\end{compactitem}

\subsection*{Data quantification}
All plates were scanned with a laser-illuminated High Content Screening (HCS) microscope, (Cellinsight CX7 LZR, Thermo Scientific). The accompanying HCS Studio (v6.6.2) software was used to simultaneously capture the images and quantify the fluorescence signal intensities. For the quantification, cells were identified with the Hoechst staining and cell masks were defined by increasing the identified nucleus regions (example shown in figure \ref{main-fig: high content}c). Accordingly, the HCS Studio software calculated the measured, average fluorescence intensity for the individual cells (within the cell mask region) for each channel.  

For the validation dose-response experiments and double-staining experiments, images were captured until 3000 distinct cells were quantified per well with a hard limit of 50 images per well (preventing hour-long scanning of an accidentally empty well). Most wells reached 3000 distinct cells, with all but two wells above 2000 quantified cells (two exceptions with 1851 and 1993 cells respectively). For the PKA experiments, the cell target was set to 5000 (always reached). For all validation dose-response experiments, the exposure times were almost identical (for one plate \unit[0.022]{s} instead of \unit[0.025]{s} for the Alexa Fluor Plus 647 channel, to avoid signal clipping). For the PKA experiments, the exposure times were identical for each plate. The double-staining experiments used the same exposure times for both stainings. But the two plates were independent experiments using different exposure times.

The raw monochrome images (unedited images) for each channel were exported as PNG files for the dose-response experiments and as 16-Bit Tiff files for the double staining experiments. The single-cell data was exported as CSV file for each plate. The subsequent data analysis and image editing were done with Julia (v1.9.1).

\subsection*{Data normalization and dose-response\\ curves}

Since different fluorescence labels were used for the validation dose-response experiments (Alex Fluor 555 and Alexa Fluor Plus 647), and since one plate used a slightly different exposure time, a normalization condition was prepared to normalize the responses between the different plates and channels. This normalization condition consisted only of mouse anti-NF200 as primary antibody and both anti-mouse Alexa Fluor 555 and anti-mouse Alexa Fluor Plus 647 as secondary antibodies. Since the exposure times were identical for all PKA plates and only a single channel was used, the zero control was used for the normalization between different plates.

To obtain the normalization value, the mean population response (mean value of individual cell responses) was calculated for each normalization well individually, for each channel. The mean value of the mean population responses on a single plate was set as normalization value for the respective plate and channel. Then the response value of each cell was divided by the respective normalization value. (Instead of normalizing the single-cell responses it would be equally possible to just normalize the dose-response values at the end. Since the normalization is a linear transformation, the result would be identical up to floating point errors.)

After normalizing the single-cell responses, the mean population response was calculated for each well and each channel. Next, for each condition, the mean of the mean population responses was calculated as data point for the dose-response curve and the standard deviation of the mean population responses was used as measurement uncertainty.
That is, each well of a given condition was treated as a separate replicate. Finally, for each dose-response curve, the non-zero baseline signal was removed. For this, we used the baseline obtained from the zero controls. To prevent negative values, however, we used the lowest response of the dose-response curve, if subtracting the baseline would result in negative response values.

\subsection*{Image processing for visualization}

The raw monochrome PNG images for the individual channels were imported into Julia, essentially represented as matrices of \verb|Float64| values. To normalize the image brightness, the quotient of the normalization values for the respective plate was calculated (cf. ``Data normalization''):
\[ r = \frac{\text{norm. value}_{\text{Alexa Plus 647}}}{\text{norm. value}_{\text{Alexa 555}}} \ . \]
This quotient describes the factor by which the Alexa Fluor Plus 647 signal is stronger than the Alexa Fluor 555 signal, all else being equal (by assumption of the normalization control setup). Thus, each pixel intensity of the the Alexa Fluor 555 image matrices was multiplied with $r$. Finally, all pixel intensities were multiplied by a factor of 20 to increase the overall brightness for better visibility. The Hoechst channel images were not normalized but all pixel intensities were multiplied by a factor of 4 as to be visible but not distracting in the composite images.

For the composite images (Fig. \ref{main-fig: high content}a,b), the images were added pixel-wise, assigning the following colors (i.e. multiplying the pixel intensity with the respective color vector):
\begin{center}
	\definecolor{channel_red}{RGB}{255, 0, 0}
	\definecolor{channel_blue}{RGB}{0, 0, 255}
	\definecolor{channel_green}{RGB}{0,255, 0}
	\renewcommand{\arraystretch}{1.4}
\begin{tabular}{|l|c|}\hline
	\textbf{channel} & \textbf{color}\\ \hline
	Hoechst & $\text{RGB}(0,1,0)$ \textcolor{channel_green}{$\blacksquare$}\\ \hline
	Alexa Fluor Plus 647 & $\text{RGB}(1,0,0)$ \textcolor{channel_red}{$\blacksquare$}\\ \hline
	Alexa Fluor 555 & $\text{RGB}\left(0,0,1\right)$ \textcolor{channel_blue}{$\blacksquare$}\\ \hline
\end{tabular}	
\end{center}
Finally the image matrices were exported as PNG files with a $2\times 2$ pixel-binning and clipping of too high intensities.

For the images in figure \ref{main-fig: optimal dilution}b,c, the brightness of the individual-channel images was increased by an additional factor of 1.5 before creating the composite. The individual-channel images were exported as monochrome images and the composite was obtained as before, but this time without the Hoechst-channel images.

\subsection*{Model fitting}

We used the Julia packages \cite{Tschimmel_2024_AntibodyPackages} for all accessibility analyses. Here, we mostly give a high-level description not subject to the actual implementation in the packages. For instructions how to implement the methods described in this section and for a complete list of the available options, please refer to the extensive documentation of the packages.

\vspace{1em}

\noindent
\textbf{\small Fitting objective}\\
Let $\{(d_i,r_i, \Delta r_i)\}_{i=1}^n$ denote the dilution quotients, the response values (mean of replicates) and the measurement uncertainties (standard deviation of replicates) of a dose-response curve, then the objective function that was used for the model fitting reads:
\[\text{obj}(\lambda) = -\ell_0(\lambda) -\sum_{i_1}^n \frac{(r_i - f(d_i,\lambda))^2}{(2 \Delta r_i)^2}\ .\]
This is a logarithmic posterior distribution, assuming the measured responses to be independent and normally distributed. Here, $\lambda$ denotes the parameters for the model function $f$ and $\ell_0$ denotes the logarithmic prior.

Both the model function and the logarithmic prior depend on the grid discretization of the $K_\tau$-domain that is used for the numerical approximation of the integral in equation \eqref{main-eq: Fredholm equation}. Thus, let $\{p_1,\ldots, p_m\}$ denote the discretization points of the $K_\tau$-domain, which define the intervals:
\[{[p_1,p_2], [p_2,p_3], \ldots, [p_{m-1},p_m]}\ .\]
The parameters in $\lambda = (\lambda_1, \ldots, \lambda_{m-1})$ are the amounts of epitopes with $K_\tau$ in the respective interval.

For a given grid, we used the following model function to approximate the integral model \eqref{main-eq: Fredholm equation}:
\[f(d,\lambda) = \sum_{j=1}^{m-1} \lambda_j \left(1-\exp\left(-\frac{d}{(p_{j+1}-p_j)/2}\right)\right)\ .\]
Note that $g(p)$ in \eqref{main-eq: Fredholm equation} is in fact the epitope density, not the amount of epitopes. Thus, its discretization should be $\frac{\lambda_i}{p_{i+1}-p_i}$, but the factor $\frac{1}{p_{i+1}-p_i}$ cancels out with the factor $p_{i+1}-p_i$ from the Riemann sum approximation.

For the logarithmic prior, we essentially used a Tikhonov regularization for the differences $\lambda_{i+1}-\lambda_i$ of the amounts of epitopes with $K_\tau$ in neighboring intervals. However, the differences were rescaled with respect to the visual log-scale of the histogram plots:
\begin{gather*}
\begin{aligned}
	\ell_0(\lambda) &= -\frac{\alpha}{m^2} \sum_{j=1}^{m-2} \left( \frac{\lambda_{j+1}}{\log(p_{j+2}) - \log(p_{j+1})} - \frac{\lambda_{j}}{\log(p_{j+1}) - \log(p_j)}  \right)^2 ,
\end{aligned}
\end{gather*}
where $\log$ is the logarithm with base 10 and $\alpha$ is the regularization scale. For all analyses in the main part of the paper we used $\alpha = 500$.

\vspace{1em}

\noindent
\textbf{\small Adaptive grids and objective function minimization}\\
We used adaptively refining grids for the minimization of the objective function. First, a coarse logarithmically spaced 2-interval grid (3 discretization points) was defined:
\[\begin{aligned}
	&p_1 = \min\{d_i\}, \quad p_2 = \log\left(\frac{\max\{d_i\} - \min\{d_i\}}{2}\right),\\
	&p_3 = \max\{d_i\} \ .
\end{aligned}\]
Then, the corresponding objective function was constructed, as explained above, and minimized with the gradient-free Nelder-Mead algorithm (using the Optim.jl package) using 2000 iterations. After obtaining the parameters $\lambda$ from the minimization, the contribution value of each interval was defined (cf. supplement \ref{sup-subsec: histogram visualization}) by 
\[\frac{\lambda_j}{\log(p_{j+1})-\log(p_j)}\ .\]
To refine the grid adaptively, the interval with the largest contribution value difference, compared to its neighbors, was split in two. Each of the intervals got $\frac{\lambda_j}{2}$ as new parameter (equal distribution of the amount of epitopes). After splitting the interval, the objective function was constructed for the now refined grid. Using the correspondingly refined parameters as new starting point, the objective was minimized again.

This optimization-refinement cycle was repeated 50 times to obtain a sufficiently fine grid. After the last refinement, the refined parameters were optimized with the LBFGS algorithm (2000 iterations), using numerically calculated gradients with an absolute tolerance of 1e-12.

\subsection*{Plotting histograms}

To plot the fitted parameters $\{\lambda_1,\ldots,\lambda{m-1}\}$ that correspond to a grid $\{p_1,\ldots,p_{m}\}$ as histogram, the bars are defined according to the intervals $[p_1,p_2],[p_2,p_3]\ \ldots$. Then, for each bar, its width $w_i$ in a logarithmic scale is calculated by 
\[w_i = \log(p_{i+1})-\log(p_i)\ .\]
The height of a bar in the histogram is then the parameter divided by this logarithmic width $\nicefrac{\lambda_i}{w_i}$. Finally, the bars are plotted in a logarithmic scale.

\subsection*{Double-staining}

The 1st staining of the double-staining experiments followed the protocol described in the ``Immunocytochemistry'' section. Both the actual double-staining condition and the controls were stained with the primary antibody mix (5:1 mix of anti-NF200 and anti-RPS11) using a dilution quotient of 1:102400 in total. Here, ``in total'' means, that the 5:1 mix, which is given the initial dilution quotient 1:100, was diluted by an additional factor 1:1024. In all cases, a 1:1000 dilution of the secondary antibodies\slash stains (anti-mouse Alexa Fluor Plus 647, anti-rabbid Alexa Fluor 555 and Hoechst) was used. After the immunocytochemistry, the prepared wells were scanned with the HCS microscope.

Immediately after scanning the plate, the 2nd staining was commenced, using the same protocol as before, but not repeating the blocking. For the double-staining condition, we used a dilution quotient of 1:6400 (in total) for the primary antibody mix. For the re-stain control, we used the same dilution quotient as used for the 1st staining and for the secondary antibody control, the wells were left untouched. After the \unit[24]{h} primary antibody incubation, the secondary antibodies but not the Hoechst stain were applied at the usual 1:1000 dilution to all conditions. After this, the wells were scanned again, using exactly the same settings as for the first scanning.

The raw images were exported as 16-bit Tiff images to retain all data. To align the images of the two different scans, we compared the respective Hoechst channel images. With the SubpixelRegistration.jl package we determined the pixel shifts that are necessary to align the Hoechst channel images. Then, we applied the same pixel shifts to the respective channel images and cropped all images to the overlapping region. For the rest of the image processing, the Hoechst images were not used. Furthermore, the brightness (pixel intensities) of all images was increase by a factor of 6 just before exporting the images (after the edits described below), to increase visibility. Note that no channel normalization was applied here, since all wells were prepared on the same plate (for one independent repetition of the experiment), and since the relative brightness of the fluorescence labels does not affect the result (same effect on computational multiplexing and the color validation).

We exported the color channel images of the 2nd staining without any additional editing as monochrome images (Fig. \ref{main-fig: double-staining} c',c'',f',f''). The corresponding composite images (Fig. \ref{main-fig: double-staining}c,f) were calculated as explained in the ``Image processing for visualization'' section, without the Hoechst channel, though.

For the computational multiplexing, we removed the color channel information by adding together the pixel intensities of the monochrome color channel images. This resulted in a single, total-signal-intensity, monochrome image per staining and removed the color information that would not be present in generic double-staining experiments. From the HCS microscope data, we determined the fold change of the secondary antibody control and increased the brightness of the 1st staining image by this fold change. Next, we calculated the pixel-wise difference between the 2nd staining image and the brightness-increased 1st staining image. All images were exported as monochrome images (figure \ref{main-fig: double-staining}b',b'',e',e'' and supplement figure \ref{sup-fig: graphical illustration}). Finally, we created a composite image (Fig. \ref{main-fig: double-staining}b,e) as done before, using the colors $\text{RGB}(1,0,0)$ for the brightness-increased 1st staining image and $\text{RGB}(0,0,1)$ for the difference image.

\section*{Acknowledgements}
We would like to thank  Maike Siobal and Lars Schmidt for providing technical assistance in the conduction of the experiments.

\clearpage

\begin{refsection}

\appendix
\onecolumn
\KOMAoptions{fontsize=12pt}
\newgeometry{left=2.8cm, right = 2.8cm, top=2cm, bottom = 2.5cm}
\fancyfootoffset{0pt}

\renewcommand\thefigure{\thesection.\arabic{figure}}
\setcounter{figure}{0}

\begin{center}
	\Huge \bfseries Supplement
\end{center}
\vspace{0.5cm}

\setlength{\baselineskip}{15pt}

\tableofcontents

\noindent
\rule{\columnwidth}{1pt}

\section{Antibody accumulation model for the anaylsis of immunostaining processes}
\label{sup-sec: accumulation model}

The Langmuir isotherm \cite{Langmuir_1918} is a popular model to describe the commonly observed dose-response behavior of antibody binding \cite{Latour_2014}. However, the Langmuir model is inappropriate for most immunostaining protocols. Thus, a new model needs to be provided to explain dose-response curves for these protocols. Using an empirical model for this, allows to implicitly incorporate effects that are practically immeasurable.

\subsection{The Langmuir model is inappropriate}

In \cite{Latour_2014} the author provides 3 major issues why the Langmuir model does not apply to antibody binding: 
\begin{itemize}
\item the binding sites are not identical,
\item multilayer adsorption is possible,
\item the response is not a dynamical equilibrium.
\end{itemize}
The first problem, non-identical binding sites, can be addressed by considering a superposition model, which we will address in subsection \ref{sup-subsec: Heterogeneous binding landscapes}. The second problem, multilayer adsorption, can be circumvented by interpreting the Langmuir model as empirical model. ``Empirical model'' means that one assumes a top-down approach, describing only the average amount of the involved quantities over time (compartment model) but not the involved physical principles or the microscopic behavior: Let $x(t)$ denote the average surface density of bound antibody-epitope complexes, $a(t)$ denote the average free antibody concentration and $g(t)$ denote the average surface density of free epitopes. Then we may describe the change of bound antibody-epitope complexes with the Langmuir kinetics
\begin{equation}
	\label{sup-eq: empirical Langmuir equation}
	\frac{d}{dt}x(t) = k_a a(t) g(t) - k_d x(t)\ ,
\end{equation}
where $k_a$ and $k_d$ are rate constants (including unit conversions). This approach does not require require monolayer adsorption.

\begin{remark}[\bfseries Interpretation of empirical models]
	Since the empirical model \eqref{sup-eq: empirical Langmuir equation} only describes the average quantities, the rate constants describe only the raw observed rates of binding\slash dissociation. They no longer describe any chemical\slash physical process, but encompass multiple effects compressed into a single average rate. 
\end{remark}
\begin{remark}[\bfseries Benefits of empirical models]
	\label{sup-rem: advantage of empirical model}
	The lack of chemical\slash physical interpretation of parameters in an empiricla model may reduce what the model can explain. On the other hand, one gains more leeway in the numerical description of different processes, as complex effects can be absorbed into average parameters.
	
	Since the antibody accumulation model is intended for a broad application, not a precise account of reality, we will use the absorption of complex effects into average parameters quite extensively. Among others, this allows to account for practically immeasurable quantities and allows for convenient simplifications.
\end{remark}

While non-identical binding sites can be described by superposition models and multilayer adsorption can be absorbed into average parameters for an empirical model, the problem of irreversible binding is more fundamental than even addressed by \cite{Latour_2014}. Often, the Langmuir isotherm, i.e. the equilibrium state of the Langmuir rate equation, is used to interpret binding data obtained from protocols that use washing steps. But ``washing'' means that after the antibody incubation phase any equilibrium that might have settled will be washed away, often multiple times before the actual scanning. Yet, there is still a response signal even for protocols with extensive washing. Thus, we must conclude that the measured antibodies have bound permanently (for the time frame of the experiment). But then, we need a new explanation for the commonly observed dose-response behavior of immunostaining.

\begin{remark}[\bfseries Reversible binding]
	That the measurable antibodies must have bound permanently does not mean that no antibodies can bind reversibly. But this temporary antibody binding is not measured. However, since we intend to use an empirical model, these effects can be understood as binding obstruction that are part of the average binding rate.
\end{remark}

\subsection{From the Langmuir model to the accumulation model}

As mentioned in remark \ref{sup-rem: advantage of empirical model}, the accumulation model should be an empirical model with broad applicability. Thus, we may incorporate the permanent binding of antibodies starting form equation \eqref{sup-eq: empirical Langmuir equation}:
\[\frac{d}{dt}x(t) = k_a a(t) g(t) - k_d x(t)\ .\]
For this, we set $k_d = 0$, meaning that there is no dissociation. Furthermore, for simplicity, we assume $a(t)$ to be constant. Note, that this assumption is violated, especially for low antibody concentrations, because of antibody depletion. But using the correct term $a-\beta x(t)$ with an appropriate unit conversion factor $\beta$ will prove impractical for most applications (see remark \ref{sup-rem: antibody depletion and superposition} and subsection \ref{sup-subsec: unit invariance fails}). In addition, a simple depletion correction method to gauge the worst-case deviation of the simplified model will be provided in section \ref{sup-sec: depletion correction}.

With $k_d = 0$ and $a(t) = a$, the rate equation simplifies to 
\begin{equation}
	\label{sup-eq: simplified rate equation}
	\frac{d}{dt} x(t) = k_a a (g - x(t))\ ,
\end{equation}
where $g$ denotes the total surface density of epitopes. This differential equation has a simple analytical solution:
\[x(t) = g - (g-x_0) e^{-k_a a (t)}\ .\]
In immunostaining experiments there are usually no antibodies in the system prior to the antibody incubation $x(0) = x_0 = 0$ and the antibody is incubated for a fixed amount of time $\tau$. Thus, the surface density of bound antibody after the antibody incubation is
\[x = g(1-e^{-k_a a \tau})\ . \]
Next, we can define the constant
\begin{equation}
	\label{sup-eq: accumulation model single}
	K_\tau = \frac{1}{k_a \tau} \quad  \Rightarrow  \quad x = g(1-e^{-\nicefrac{a}{K_\tau}})\ ,
\end{equation}
s.t. the system-describing constant $K_\tau$ has the same dimension as the antibody concentration $a$. In that way, the system constant $K_\tau$ behaves similar to the well known $K_d$ of the Langmuir isotherm. It marks the antibody concentration $a = K_\tau$ at which approximately $63.2$\% of the epitopes are occupied by antibodies. As reverse of the binding rate $K_\tau \sim \frac{1}{k_a}$, we may also interpret $K_\tau$ as binding obstruction constant, i.e. as accessibility measure.

Observe, as we have already explained verbally in the main part of the paper, that the binding rate $\frac{d}{dt}x(t)$ is not only dependent on the accessibility constant $K_\tau$ but also to the antibody concentration $a$. Since the incubation time usually is the same for all conditions (i.e. all concentrations), the dose-response behavior can be explained as follows: Lower concentrations take longer to reach the same epitope coverage levels than higher concentrations. Figure \ref{fig: illustration dose-response behavoir} in the main part of the paper illustrates this principle, which we refer to as accumulation hypothesis.

\subsection{Heterogeneous binding landscapes}
\label{sup-subsec: Heterogeneous binding landscapes}

So far, the mathematical description covers only a single interaction, where all antibodies and all epitopes are identical. Assuming epitopes not to interfere with each other, the presence of multiple epitope classes in a system can be described by a superposition:
\begin{equation}
	\label{sup-eq: accumulation model sum}
	x = \sum_{i} g_i \left(1-\exp\left( -\frac{a}{K_{\tau,i}}\right)\right)\ .
\end{equation}
Here, $K_{\tau,i}$ is the accessibility constant of the $i$-th epitope class and $g_i$ the amount of epitopes in the $i$-th epitope class. Since there is only a finite number of epitopes in a real system, the superposition is in fact a finite sum. Yet, for the application to inverse problems, i.e. determining the parameters $\{(g_i,K_{\tau,i})\}$ from dose-response data, this model is not well suited, as the number of epitope classes is unknown a priori. However, we can approximate the sum by an integral, promoting the discrete numbers $g_i$ to an accessibility density $g(K_\tau)$: 
\begin{equation}
	\label{sup-eq: accumulation model integral}
	x \approx \int_0^\infty g(K) (1-e^{-\frac{a}{K}})\ dK\ .
\end{equation}
In this way, we obtain a Fredholm integral equation which is commonly used in the literature to describe heterogeneous binding \cite{Sips_1948} \cite{House_1978} \cite{Sposito_1980} \cite{Svitel_2003} \cite{Forssen_2018} \cite{Malakhova_2020}. Observe that fitting the Fredholm model no longer requires a priori knowledge about the number of epitope-classes, fixing the issues of the discrete model \eqref{sup-eq: accumulation model sum}. An approach already applied to binding curves (kinetic data) by \cite{Svitel_2003}.

\begin{definition}
	In the following we refer to equation \eqref{sup-eq: accumulation model single} and equation \eqref{sup-eq: accumulation model sum} as \textbf{discrete accumulation model} and refer to equation \eqref{sup-eq: accumulation model integral} as \textbf{Fredholm accumulation model}.
\end{definition}

At this point, one might be tempted to model additional aspects like the secondary antibody dynamics. But the binding behavior of the secondary antibodies depends on the density of bound primary antibodies, which is hard to measure. Secondary antibodies are used to make primary antibodies measurable in the first place. Thus, the argument about empirical models is employed here again. The dynamics of the secondary antibodies (and additional real-world effects) are included in the density $g(K_\tau)$. In other words, the density $g(K_\tau)$ and the resulting histogram become a top-level characterization of the whole immunostaining protocol in a given cell system. In that regard the sometimes invalid assumption $a(t) = a$ is less of a problem for the application of the model to investigate antibody behavior in a cell system. In fact, the consistency of the applications presented in the paper corroborates the merit of this top-level description approach.

\begin{remark}[\bfseries Antibody depletion and superposition]
	\label{sup-rem: antibody depletion and superposition}
	Since the antibody concentration was assumed to be constant, and each epitope can belong only to one class, the binding processes of the different epitope classes are independent. Thus, the overall response curve is just the sum of the individual response curves. If the antibody can deplete, however, the epitope classes compete for the available antibodies. Thus, the heterogenous binding must be modelled at the level of the differential equation (with $a$ denoting the initial antibody concentration and $\beta$ converting the units):
	\[\frac{d}{dt}x_i(t) = k_{a_i} \left(a - \beta \sum_{j}x_j(t)\right)(g_i-x_i(t))\ .\]
	Observe that this is a coupled system of differential equations, which in this case does not appear to have a simple analytical solution.
\end{remark}

\section{The accessibility histogram and its visualization}
\label{sup-sec: accessibility histogram}

Approximating the sum of the discrete accumulation model \eqref{sup-eq: accumulation model sum} by an integral removes the need to determine the unknown number of epitope classes. Yet, a generic continuous function is difficult to represent in a computer and is thus hard to use for inverse problems (parameter estimation form data). Hence, an approximation of the density by a finite set of parameters is required.

\subsection{Histograms as approximation of the accessibility density}
Recall that the accessibility density $g(K_\tau)$ is in fact just an approximation of a finite set of epitope classes $\{(g_i,K_{\tau,i})\}$. Since physically meaningful binding rates should be finite, there are finite values $p_{\min}$ and $p_{\max}$ s.t.
\[0 < p_{\min} \leq K_{\tau,i}\leq p_{\max} < \infty \qquad \forall\ i\ .\]
In other words, for any density $g(K_\tau)$ that describes a real system, there is a smallest positive accessibility constant $0<p_{\min}$ and a largest finite accessibility constant $p_{\max} < \infty$ for which $g(K_\tau)\neq 0$.

As outlined in the methods section, the accessibility density $g(K_\tau)$ is approximated by constant functions over a finite grid. For this, $m-2$ additional discretization points between $p_1 = p_{\min}$ and $p_m = p_{\max}$ are defined $\{p_1,\ldots,p_m\}$, s.t.
\[p_i \leq p_j \qquad \forall\  i < j\ .\]
These discretization points define the following intervals:
\[I_1 = [p_1,p_2], \quad I_2 = [p_2,p_3], \ldots \ ,\quad  I_{m-1} = [p_{m-1},p_m]\ . \]
An approximation of $g(K_\tau)$ by constant functions over the intervals $\{I_j\}_{j_1}^{m-1}$ then reads
\[g(K_\tau)\approx \sum_{j=1}^{m-1} g(\tfrac{p_{j+1}-p_j}{2}) \chi_{I_j}(K_\tau) \qquad \text{where}\qquad \chi_{I_j}(K_\tau) = \left\{ \begin{array}{ll}
1 & \quad ,\ K_\tau \in I_j \\  
0 & \quad ,\  K_\tau \not\in I_j
\end{array} \right. \ .\]
In case of an unknown accessibility density that is to be inferred form data, the approximation consists of $m-1$ parameters $\lambda_1,\ldots, \lambda_{m-1}$ that are to be estimated form the data:
\[g(K_\tau)\coloneqq \sum_{j=1}^{m-1}\tfrac{\lambda_j}{p_{j+1}-p_j} \chi_{I_j}(K_\tau)\ . \]
We divided the parameters by the respective interval lengths in this definition s.t. the parameters $\lambda_j$ correspond to the number of epitopes in the intervals $I_j$. This can be seen as follows.

When we apply the approximation of the density in the Fredholm accumulation model, we obtain
\[\begin{aligned}
	x &\approx \int_0^\infty g(K)(1-e^{-\frac{a}{K}})\ dK \approx \int_0^\infty \sum_{j=1}^{m-1}\tfrac{\lambda_j}{p_{j+1}-p_j} \chi_{I_j}(K)(1-e^{-\frac{a}{K}})\ dK\\ 
	&\quad = \sum_{j=1}^{m-1} \tfrac{\lambda_j}{p_{j+1}-p_j} \int_{I_j} 1-e^{-\frac{a}{K}} \ dK\ .
\end{aligned}\]
Next, assuming the $I_j$ to be sufficiently small, we can approximate the integrals as products. Let $\langle p_j \rangle\coloneqq \frac{p_{j+1}-p_j}{2}$ denote the middle point of the interval $I_j = [p_j,p_{j+1}]$, then
\[\int_{I_j} 1-e^{-\frac{a}{K}}\ dK \approx \left(1-e^{-\frac{a}{\langle p_j \rangle}}\right) (p_{j+1}-p_j)\ . \]
Thus, the Fredholm accumulation model becomes
\[\begin{aligned}
	x &\approx \sum_{j=1}^{m-1} \tfrac{\lambda_j}{p_{j+1}-p_j} \int_{I_j} 1-e^{-\frac{a}{K}} \ dK \approx \sum_{j=1}^{m-1} \tfrac{\lambda_j}{p_{j+1}-p_j} \left(1-e^{-\frac{a}{\langle p_j\rangle}}\right)(p_{j+1}-p_j) \\
	&\quad = \sum_{j=1}^{m-1} \tfrac{\lambda_j}{p_{j+1}-p_j} (p_{j+1}-p_j) \left(1-e^{-\frac{a}{\langle p_j\rangle}}\right) = \sum_{j=1}^{m-1} \lambda_j \left(1-e^{-\frac{a}{\langle p_j\rangle}}\right)\ .
\end{aligned}\]
Comparing the last term with the discrete accumulation model, we can observe that the $(\lambda_j, \langle p_j \rangle)$ define the epitope classes for the discrete accumulation model. In other words, $\lambda_j$ corresponds to the number of epitopes with $K_\tau = \langle p_j \rangle$. That is, by construction, the number of epitopes with $K_\tau \in I_j$.

\subsection{Visualization and dose-response contribution}
\label{sup-subsec: histogram visualization}

As stated in the main part of the paper and in the methods section, the bin heights are rescaled w.r.t. the bin widths as they appear in the logarithmically scaled plot. This is a mere visualization aid which ensures that the visual area of peaks in the logarithmically scaled histograms corresponds 1:1 to the effect on the dose-response curve. This can easily be illustrated by specifying a density with two peaks, as the corresponding dose-response curves and peak contributions can readily be obtained from the accumulation model.

\begin{wrapfigure}{r}{0.5\textwidth}
		\centering
		\includegraphics[width = 0.45\textwidth]{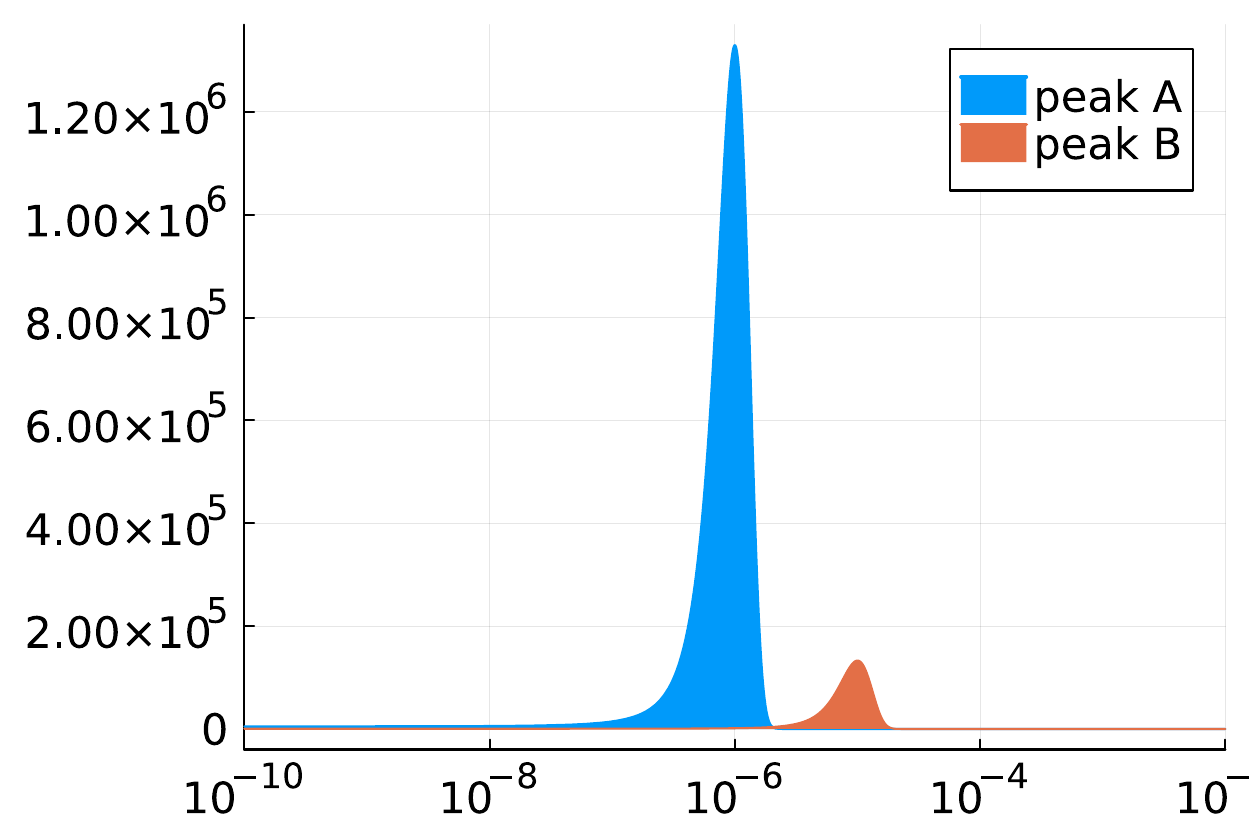}
		\caption{Plot of $g(K)$.}
		\label{sup-fig: density plot}
		\vspace{-1cm}
\end{wrapfigure}
For the illustration, consider the sum of two normal distributions $\mathcal{N}(10^{-6},3\cdot 10^{-7})$ and $\mathcal{N}(10^{-5},3\cdot 10^{-6})$:
\[ \begin{aligned}
	 g(K) &= \tfrac{1}{3\cdot 10^{-7}\cdot\sqrt{2\pi }}\ e^{-\frac{1}{2}\left(\frac{K-10^{-6}}{3\cdot 10^{-7}}\right)^2}\\[1em] 
	 &\ +\  \tfrac{1}{3\cdot 10^{-6}\cdot\sqrt{2\pi }}\  e^{-\frac{1}{2}\left(\frac{K-10^{-5}}{3\cdot 10^{-6}}\right)^2}\ .
\end{aligned}
	 \]
The corresponding density is plotted in a logarithmic scale in figure \ref{sup-fig: density plot}. Now, this density can approximated with a grid. Since the grids are refined adaptively during the fitting process, leading to unequal interval sizes, we should do the same for this illustration. Here, we use 25 logarithmically spaced discretization points from $10^{-10}$ to $10^{-6}$, 50 logarithmically spaced discretization points from $10^{-6}$ to $10^{-5}$ and finally 100 logarithmically spaced discretization points from $10^{-5}$ to $10^{-2}$.

\begin{figure}[ht]
	\centering
	\begin{subfigure}[c]{0.49\textwidth}
		\caption{Gird without normalization}
		\includegraphics[width = \textwidth]{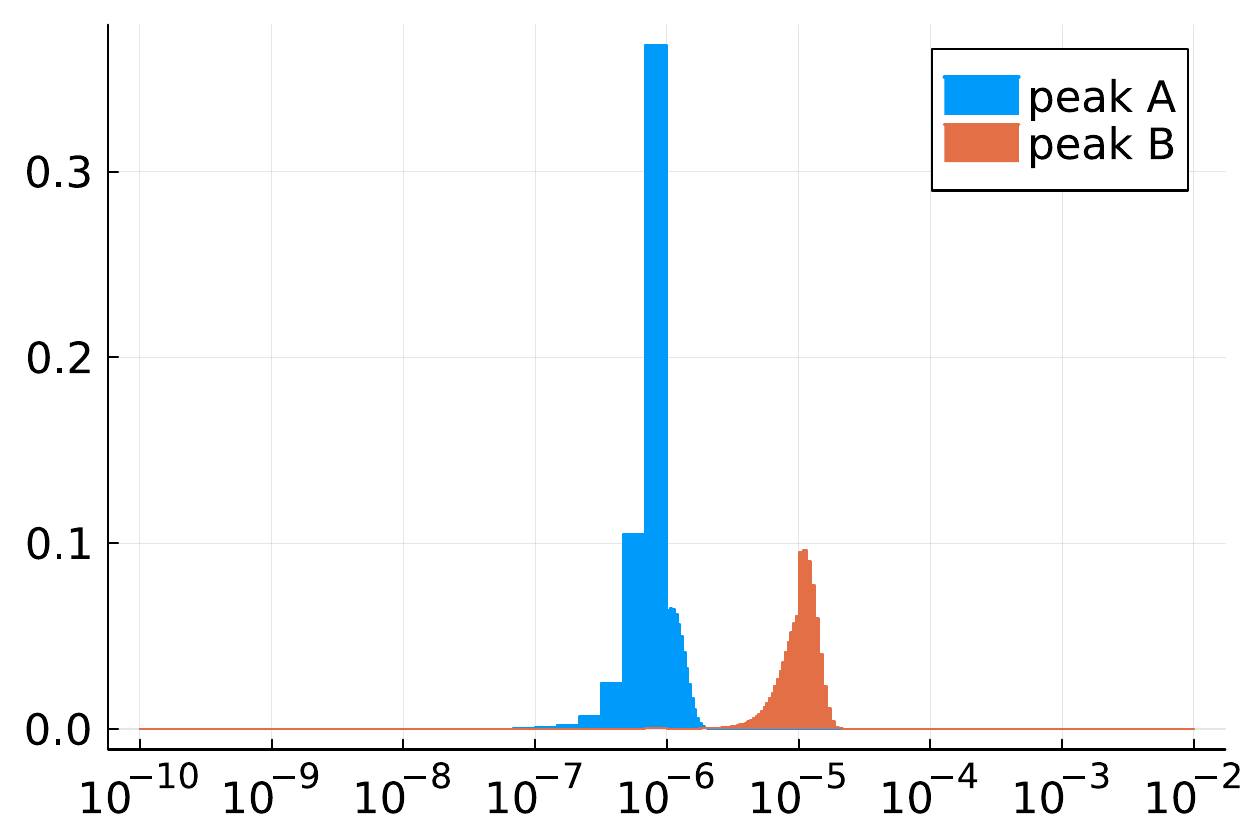}
	\end{subfigure}
	\begin{subfigure}[c]{0.49\textwidth}
		\caption{Grid with volume normalization}
		\includegraphics[width = \textwidth]{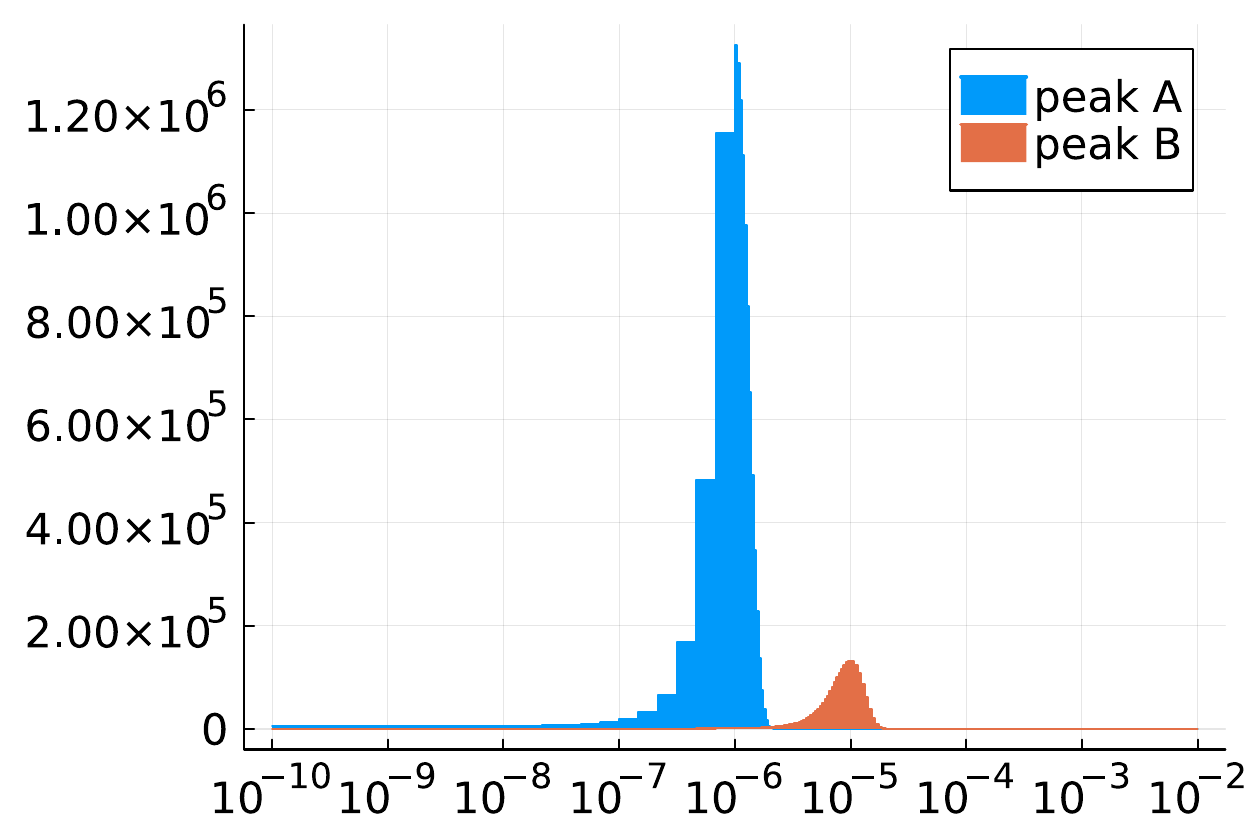}
	\end{subfigure} \vspace{1em}

	\begin{subfigure}[c]{0.49\textwidth}
		\caption{Gird with visual log-volume normalization}
		\includegraphics[width = \textwidth]{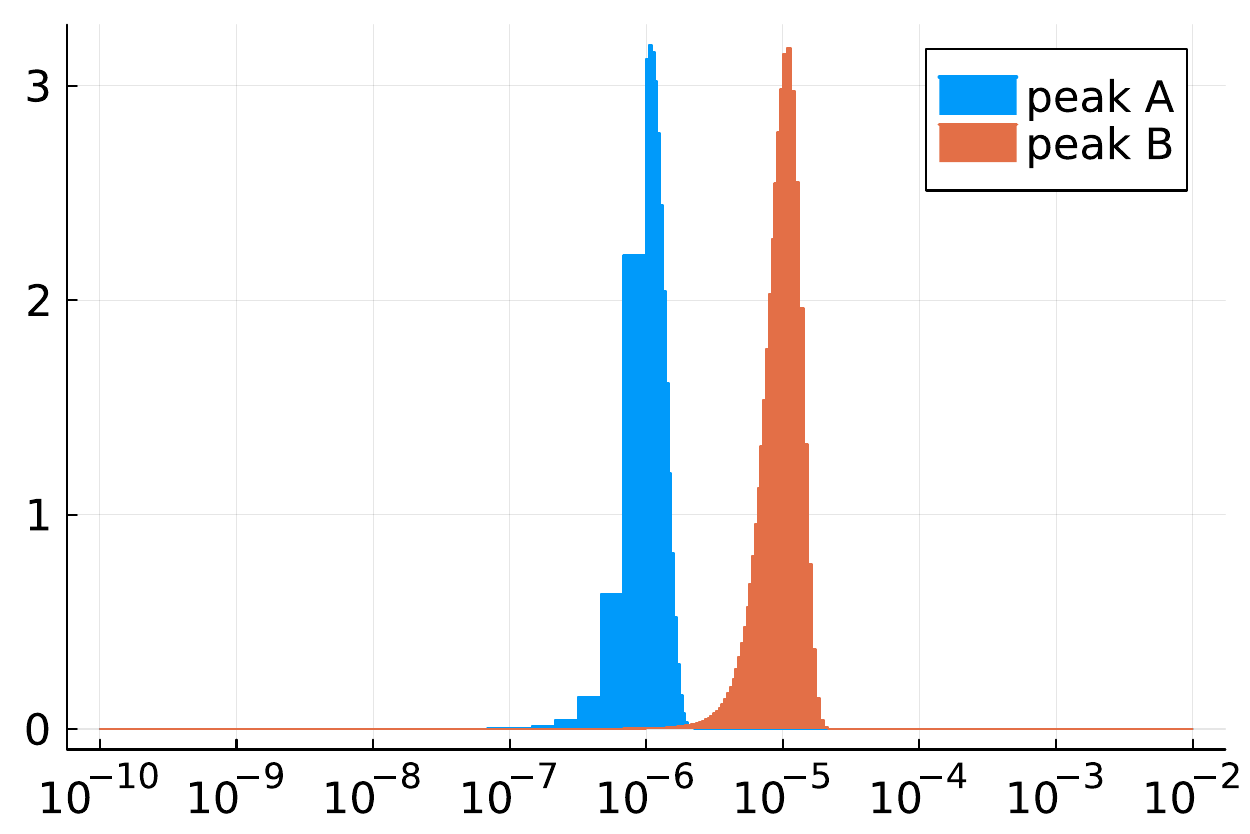}
	\end{subfigure} 
	\begin{subfigure}[c]{0.49\textwidth}
		\caption{Corresponding dose-response plot}
		\includegraphics[width = \textwidth]{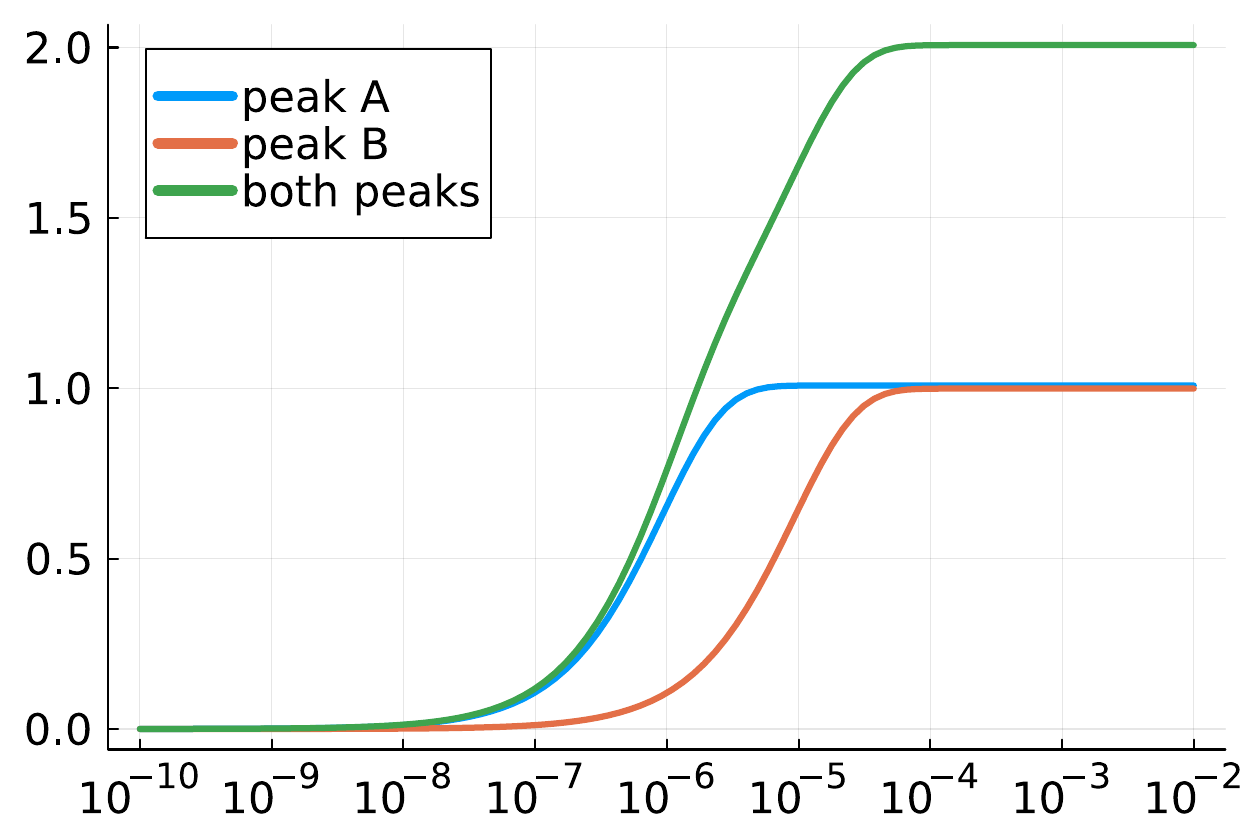}
	\end{subfigure}
	\caption{Comparison of bar height normalizations for the accessibility histogram.}
	\label{sup-fig:Grid visualization}
\end{figure}

Figure \ref{sup-fig:Grid visualization} shows 3 normalizations for the bin heights: the raw parameters without normalization $\lambda_i$ in fig. \ref{sup-fig:Grid visualization}a, the parameters divided by the true interval lengths $\nicefrac{\lambda_i}{V_i}$ in fig. \ref{sup-fig:Grid visualization}b and the parameters divided by the visual interval lengths in the logarithmically scaled plot $\nicefrac{\lambda_i}{w_i}$ in fig. \ref{sup-fig:Grid visualization}c. Furthermore, the dose-response curve contributions of the two peaks are plotted in figure \ref{sup-fig:Grid visualization}d.

For the raw parameters, recall from the methods section that the parameters correspond to the number of epitopes with $K_\tau$ in the respective intervals. Thus, without normalization, the bin heights would correspond to the contribution of each bin to the dose-response curve. However, since the bins have different widths, it is difficult to see how many bins comprise a peak. Instead, one is drawn to look at the total surface area of a peak. But this does not correspond to the effect of the peak on the dose-response curve (for the unnormalized plot). Hence, not normalizing the bin heights produces a misleading histogram, prone to misinterpretation.

Observe that when the bin heights are divided by the true lengths of the intervals, the true shape of the density $g(K)$ is recovered. Yet, as will be argued in subsection \ref{sup-subsec: histogram and units}, the density $g(K)$ is of little interest. The density is meaningless (unit dependent) outside of the integral. Only the histogram, arising from bin-wise evaluation of the integral, corresponds to the dose-response contribution. In this case, observe that peak B is much smaller than peak A. But the dose-response plot (fig. \ref{sup-fig:Grid visualization}d) shows that both peaks contribute equally strong to the overall dose-response curve. Thus, using the volume normalization for the histogram is also misleading\slash unintuitive.

As mentioned earlier, one is inclined to interpret the visual surface area of a peak as measure for the contribution of a peak to the dose-response curve. Dividing the parameters by the visual bin widths leads to histogram bars whose area is equal to the respective parameters. This is, because the visual surface of these rescaled bars is given by $w_i \cdot \nicefrac{\lambda_i}{w_i} = \lambda_i$. That the visual area of peaks for this normalization corresponds 1:1 to the effect on the dose-response curve can be observed in figure \ref{sup-fig:Grid visualization}c. Both peaks have the same surface area, in accordance with the fact that both peaks contribute equally strong to the overall dose-response curve (fig. \ref{sup-fig:Grid visualization}d).

\FloatBarrier

\section{The accessibility histogram does not depend on the choice of units}
\label{sup-sec: accessibility histogram and units}

A common problem in biological research is that most quantities cannot be measured directly. In these cases, proportional quantities need to be used for the measurement, where the corresponding proportionality constants remain unknown for the most part. This requires additional caution when binding models are used, as a lack of knowledge of these proportionality factors could lead to misinterpretations. Fortunately, the accumulation model and the shape of the accessibility histogram do not depend on the choice of units under certain circumstances.

\subsection{Measurement limitations and unknown proportionality factors}
\label{sup-subsec: measurement limitations}

Both the antibody concentration $a$ and the density of bound antibody-epitope complexes $x$ cannot be set up\slash measured directly in most biological laboratories. Instead of the antibody concentration the antibody dilution quotient $\mathfrak{a}$ is used and instead of the bound antibody-epitope complex density the fluorescence response $\mathfrak{x}$ is measured. These quantities are related by proportionality factors:
\[\mathfrak{a} = \eta a \qquad \text{and}\qquad \mathfrak{x} = \xi x \ .\]

\begin{remark}[\bfseries Proportionality and Units]
	Measuring a quantity indirectly by measuring a proportional quantity is equivalent to choosing a different measurement unit. For that reason, we will use proportionality factor and unit conversion factor synonymously in the following.
\end{remark}

The conversion between dilution quotient and concentration is quite simple if the concentration of the base dilution is known. And often antibody vendors provide the concentration, albeit in terms of protein weight per volume. However, the ratio of antibodies that have denatured while storing the antibody (freeze\slash thaw damage) is not known without proper measurements. But these measurements are either tedious or require special instruments. Hence, they are not performed regularly for simple immunostaining experiments such that $\eta$ is usually not known exactly.

While the factor $\eta$ can be obtained, it is nearly impossible to determine $\xi$ precisely with fluorescence-microscope setups. This is, because the proportionality factor $\xi$ comprises multiple unknown proportionalities. 
\begin{itemize}
	\setlength{\itemsep}{-0.2em}
	\item The digital signal is proportional to the analogue signal of the image sensor.
	\item The amount of photons detected by the image sensor is proportional to the amount of emitted photons.
	\item The amount of emitted photons (fluorescence response) is proportional to the density of fluorescence labels.
	\item The density of fluorescence labels is proportional to the density of bound secondary antibodies.
	\item The density of bound secondary antibodies is proportional to the density of bound primary antibodies.
\end{itemize}

\subsection{Unit conversion factors cancel out in the discrete accumulation model}
\label{sup-subsec: unit factors discrete model}

Note that we have chosen some properties of the units already in the derivation of the accumulation model (section \ref{sup-sec: accumulation model}). Observe that both the density of bound antibody-epitope complexes $x$ and the 
the density of epitopes $g$ describe the same thing: epitope densities, where $x$ describes the subset of epitopes that has antibodies bound to them. Since we added no conversion factor to the term $g-x(t)$ in \eqref{sup-eq: simplified rate equation}, we chose that $g$ and $x(t)$ must have the same unit. Furthermore, we used $k_a$ to cancel out the unit of the antibody concentration $a$ and to add the missing unit of $\frac{1}{\text{time}}$ in equation \eqref{sup-eq: simplified rate equation}. Since we defined $K_\tau = \frac{1}{k_a \tau}$ and since $\tau$ is a time point, it follows that the accessibility constant $K_\tau$ has the same unit as the antibody concentration $a$.

By restricting the choice of units in this way, we ensured that unit conversion factors cancel out in the accumulation model.

\begin{corollary}
Let $\mathfrak{a} = \eta a$, $\mathfrak{K}_{\tau,i}= \eta K_{\tau,i}$, $\mathfrak{x} = \xi x$ and $\mathfrak{g}_i = \xi g_i$ be proportional quantities. Furthermore, let 
\[x = \sum_{i=1}^n g_i\left(1-\exp\left(-\frac{a}{K_{\tau,i}}\right)\right)\ .\]
Then it holds for all proportionality constants $\eta,\xi > 0$ and all numbers of epitope classes $n\in \mathbb{N}$ that
\[\mathfrak{x} = \sum_{i=1}^n \mathfrak{g}_i\left(1-\exp\left(-\frac{\mathfrak{a}}{\mathfrak{K}_{\tau,i}}\right)\right)\ .\]
\end{corollary}
\needspace{5cm}
\begin{proof}
\[\begin{aligned}
\mathfrak{x} &= \xi x = \xi \sum_i g_i\left(1-\exp\left(-\frac{a}{K_{\tau,i}}\right)\right) = \sum_{i=1}^n \xi g_i\left(1-\exp\left(-\frac{\eta a}{\eta K_{\tau,i}}\right)\right)\\[1em]
&= \sum_{i=1}^n \mathfrak{g}_i\left(1-\exp\left(-\frac{\mathfrak{a}}{\mathfrak{K}_{\tau,i}}\right)\right)
\end{aligned} \]	
\qed
\end{proof}

\noindent
In particular, this corollary implies that the unit conversion factors $\eta$ and $\xi$ need not be known for the discrete accumulation model. At least if it suffices to know\slash use $\mathfrak{K}_{\tau,i}$ and $\mathfrak{g}_i$ in the same units as $\mathfrak{a}$ and $\mathfrak{x}$ respectively. Note that this is the case for the applications of the accessibility histogram that we presented in the paper.

\subsection{Conversion factors do not cancel out for the antibody depletion model}
\label{sup-subsec: unit invariance fails}

As alluded to before, using $x-\beta a$ to model antibody depletion prevents the cancellation of proportionality factors. This can already be seen in the case where only a single epitope class is considered. But although the antibody depletion model can be solved analytically for a single epitope class (cf. remark \ref{sup-rem: antibody depletion and superposition}), it is easier to consider the differential equation instead.

Let us first consider again the depletion-free accumulation model:
\[\frac{d}{dt}x(t) = k_a a (g-x(t)) = \frac{\tau}{K_\tau}\  a (g-x(t))\ .\]
Using the proportional quantities, we see that the unit conversion factors cancel out as before:

\begin{corollary}
Let $\mathfrak{a} = \eta a$, $\mathfrak{K}_{\tau} = \eta K_\tau$, $\mathfrak{x}(t) = \xi x(t)$ and $\mathfrak{g} = \xi g$ be proportional quantities. Furthermore, let 
\[\frac{d}{dt} x(t) = \frac{\tau}{K_\tau} a(g-x(t))\ .\]
Then it holds for all proportionality constants $\eta,\xi > 0$  that
\[\frac{d}{dt} \mathfrak{x}(t) = \frac{\tau}{\mathfrak{K}_\tau} \mathfrak{a}(\mathfrak{g}-\mathfrak{x}(t))\ .\]
\end{corollary}
\needspace{5cm}
\begin{proof}
\[
	\begin{aligned}
		\frac{d}{dt} \mathfrak{x}(t) &= \frac{d}{dt} \xi x(t) = \xi \frac{\tau}{K_{\tau}}a (g-x(t)) =  \frac{\tau a}{K_{\tau}} \xi(g-x(t)) = \frac{\tau  \eta a}{\eta K_\tau} \left(\xi g-\xi x(t)\right) \\
		&= \frac{\tau}{\mathfrak{K}_\tau} \mathfrak{a}(\mathfrak{g}-\mathfrak{x}(t))
	\end{aligned}
\]
\qed
\end{proof}

\noindent
For the antibody depletion model, assume that $x$ is measured as particle surface density (number of bound antibodies per unit area) and $a$ is measured as particle volume density. Then the required conversion factor is given by
\[\beta = \frac{A}{V}\ ,\]
where $V$ is the well volume and $A$ is the surface area of the bottom of the well ($A$ converts $x$ to the number of particles, which is converted  to a volume density with $\frac{1}{V}$). Then the differential equation is given by
\[\frac{d}{dt}x(t) = k_a \left(a- \frac{A}{V} x(t)\right)(g-x(t)) = \frac{\tau}{K_\tau} \left(a- \frac{A}{V} x(t)\right)(g-x(t))\ .\]
If we now consider the proportional quantities, it follows that

\begin{corollary}
\label{sup-cor: depletion model no cancellation}
Let $\mathfrak{a} = \eta a$, $\mathfrak{K}_{\tau} = \eta K_\tau$, $\mathfrak{x}(t) = \xi x(t)$ and $\mathfrak{g} = \xi g$ be proportional quantities. Furthermore, let 
\[\frac{d}{dt} x(t) = \frac{\tau}{K_\tau} \left(a-\frac{A}{V}x(t)\right)(g-x(t))\ .\]
Then it holds for all proportionality constants $\eta,\xi \in >0$  that
\[\frac{d}{dt} \mathfrak{x}(t) = \frac{\tau}{\mathfrak{K}_\tau} \left(\mathfrak{a} - \frac{A}{V} \frac{\eta}{\xi}\right)(\mathfrak{g}-\mathfrak{x}(t))\ .\]
\end{corollary}
\begin{proof}
\[\begin{aligned}
	\frac{d}{dt} \mathfrak{x}(t) &= \frac{d}{dt}\xi x(t) = \xi \frac{\tau}{K_\tau} \left(a-\frac{A}{V} x(t)\right)(g-x(t)) = \frac{\tau (a-\frac{A}{V}x(t))}{K_\tau} \xi (g-x(t))\\[1em]
	&=  \frac{\tau \eta (a-\frac{A}{V} \frac{\xi}{\xi}x(t))}{\eta K_\tau}(\xi g-\xi x(t)) = \frac{\tau  (\eta a-\frac{A}{V} \eta \frac{\xi}{\xi}x(t))}{\eta K_\tau}(\xi g-\xi x(t))\\[1em]
	&  = \frac{\tau}{\mathfrak{K}_\tau} \left(\mathfrak{a} - \frac{A}{V} \frac{\eta}{\xi}\right)(\mathfrak{g}-\mathfrak{x}(t))
\end{aligned}\]
\qed
\end{proof}

\noindent
As corollary \ref{sup-cor: depletion model no cancellation} shows, the unit conversion factors appear explicitly in the equation when the units are changed.  For the antibody depletion model this necessitates precise knowledge of the unit conversion factors $\eta$ and $\xi$.

\subsection{Choice of units and the accessibility histogram}
\label{sup-subsec: histogram and units}

So far, we have only considered the discrete accumulation model. Since the Fredholm accumulation model is obtained as approximation of the discrete accumulation model, one might expect that the unit conversion factors should still cancel out. The intuition is that until the approximation as integral model is applied the discrete accumulation model does not depend on the chosen units. Thus, the approximation as integral should also not depend on the units. Unfortunately, integration by substitution leads to the following result:

\begin{corollary}
	\label{sup-cor: density transformation}
Let $\mathfrak{a} = \eta a$, $\mathfrak{K} = \eta K$ and $\mathfrak{x} = \xi x$  be  proportional quantities. Furthermore, let 
\[x = \int_0^\infty g(K)\left(1-\exp\left(-\frac{a}{K}\right)\right)\  dK\ .\]
Then it holds for all proportionality constants $\eta,\xi > 0$  that
\[\mathfrak{r} = \int_0^\infty \mathfrak{g}(\mathfrak{K})\left(1-\exp\left(-\frac{\mathfrak{a}}{\mathfrak{K}}\right)\right)d\mathfrak{K} \qquad \text{where}\qquad \mathfrak{g}(\bullet) =  \tfrac{1}{\eta} \xi g(\tfrac{1}{\eta} \bullet)\ .\]
\end{corollary}
\begin{proof}
	\label{sup-prf: density transformation}
\[
	\begin{aligned}
		\mathfrak{x} &= \xi x = \xi \int_0^\infty g(K)\left(1-\exp\left(-\frac{a}{K}\right)\right)dK = \int_0^\infty  \xi g(K)\left(1-\exp\left(-\frac{a}{K}\right)\right)dK\\[1em] 
		&= \int_{0 \eta}^{\infty \eta} \xi g(\tfrac{1}{\eta} \mathfrak{K})\left(1-\exp\left(-\frac{a}{\tfrac{1}{\eta} \mathfrak{K}}\right)\right)d(\tfrac{1}{\eta} \mathfrak{K}) = \int_0^\infty \xi \tfrac{1}{\eta} g(\tfrac{1}{\eta} \mathfrak{K})\left(1-\exp\left(-\frac{\mathfrak{a}}{\mathfrak{K}}\right)\right)d\mathfrak{K}\\[1em]
		&\eqqcolon \int_0^\infty \mathfrak{g}(\mathfrak{K})\left(1-\exp\left(-\frac{\mathfrak{a}}{\mathfrak{K}}\right)\right)d\mathfrak{K}
	\end{aligned}
\]
\qed
\end{proof}

\noindent
As was the case for the discrete accumulation model, a coefficient $\xi$ converts the values of $g(K)$ to the same unit that $\mathfrak{x}$ is measured in. Yet, the transformation $\mathfrak{g}(\bullet) = \tfrac{1}{\eta}\xi g(\tfrac{1}{\eta} \bullet)$ contains the unit conversion factor $\eta$ as well. Thus, the density function $g(\bullet)$ depends on the chosen unit for the antibody concentration.  This might seem unsettling at first, as it could mean that the accessibility histogram, which is derived from the density $g(\bullet)$, depends on the chosen antibody concentration unit. Furthermore, the additional coefficient $\tfrac{1}{\eta}$ seems unnatural as it converts the values of the density function $g(\bullet)$ w.r.t. the antibody concentration conversion factor (which should have no effect on the measurement of the epitopes).

Before we address the implications for the accessibility histogram, let us recall how the integral equation came about. It was a mere approximation of a discrete sum, to ease the applicability of the model. In fact, to retrieve the discrete model, the density function $g(\bullet)$ must be a sum of delta-functions, which only makes sense when integrated. In that regard, the density $g(\bullet)$ may be understood in the same way as a probability density. It is just a density function whose meaning arises only after integration. For probability densities, integration over a subset yields the probability that the parameter is in the subset. Here, integration over a subset leads to the response contribution of epitopes with $K_\tau$ in the subset. Thus, any odd behavior of the density function $g(\bullet)$ upon unit changes can be neglected, as long as the integral remains unchanged. Fortunately, this is the case, as can be seen in proof \ref{sup-prf: density transformation}. 

The fact that the integral does not depend on the choice of units also addresses the concerns about the accessibility histogram. After all, the accessibility histogram is obtained by evaluating the integral piece-wiese over the bins (i.e. intervals) of the histogram. For this, the density $g(\bullet)$ is assumed to be piece-wise constant over the bins (cf. methods: ``adaptive grids'' and section \ref{sup-sec: accessibility histogram}):
\[g(\bullet) = \sum_{i=j}^m \tfrac{f_j}{\text{Vol}(I_j)} \chi_{I_j}(\bullet) \qquad \text{where}\qquad \chi_{I_j}(\bullet) = \left\{ \begin{array}{ll}
		1 & \quad , \bullet\in I_j\\
		0  & \quad , \text{else}
	\end{array} \right.\]
\[\Rightarrow \quad x = \int_0^\infty g(K)(1-\exp(-\tfrac{a}{K})) \ dK = \sum_{j=1}^m \tfrac{f_j}{\text{Vol}(I_j)} \int_{I_j} (1-\exp(-\tfrac{a}{K}))\ dK\ .\]
Thus, to investigate the effect of the choice of units, we should focus on decompositions of the density function $g(\bullet)$:

\begin{theorem}
	\label{sup-thm: histogram and unit change}
	Let $I\subseteq \mathbb{R}_{\geq 0}$ be an interval and let $I_j = [p_j,q_j]$ be intervals with $q_j > p_j$ and $p_{j+1} = q_j$ s.t. $\bigcup_{j=1}^m I_j = I$. Furthermore, consider the following decomposition of the density function:
	\[g(K) = \sum_{j=1}^m f_j \tfrac{1}{q_j-p_j} \chi_{[p_j,q_j]}(K)\ .\]
	In addition, define the following integral and its approximation
	\[H_j \coloneqq \int_{p_j}^{q_j} \tfrac{1}{q_j - p_j}(1-\exp(- \tfrac{a}{K}))\ dK\qquad \text{and}\qquad h_j \coloneqq 1-\exp\left(- \frac{a}{q_j-p_j}\right) .\]
	Then the Fredholm accumulation model (over $I$) and its approximation read:
	\[x = \sum_{j=1}^m f_j H_j \qquad \text{and} \qquad x_{\text{approx.}} \coloneqq \sum_{j=1}^{m}f_j h_j\ .\]
	Let $\mathfrak{x} = \xi x$, $\mathfrak{f}_j = \xi f_j$ and $\mathfrak{a}=\eta a$,$\mathfrak{K} = \eta K$, $\mathfrak{p}_j = \eta p_j$, $\mathfrak{q}_j = \eta q_j$ be proportional quantities, then it holds for all $\xi,\eta > 0$ that
	\[\mathfrak{x} = \sum_{j=1}^m \mathfrak{f}_j H_j\ , \qquad \mathfrak{x}_{\text{approx.}} = \sum_{j=1}^m \mathfrak{f}_j h_j\]
	\[H_j = \int_{\mathfrak{p}_j}^{\mathfrak{q}_j} \frac{1}{\mathfrak{q}_j - \mathfrak{p}_j} \left(1-\exp\left(-\frac{\mathfrak{a}}{\mathfrak{K}}\right)\right) \ d\mathfrak{K}\ , \qquad h_j = 1-\exp\left(- \frac{\mathfrak{a}}{\mathfrak{q}_j - \mathfrak{p}_j}\right)\ .\]
\end{theorem}
\begin{proof}
	The term $H_j$ is defined such that the expression $x = \sum_{j=1}^m f_j H_j$ is true:
	\[\begin{aligned}
		x &= \int_0^\infty g(K)(1-\exp(-\tfrac{a}{K}))\ dk = \int_0^\infty \sum_{j=1}^m f_j \tfrac{1}{q_j-p_j}\chi_{[p_j,q_j]}(K)(1-\exp(-\tfrac{a}{K}))\ dK\\[1em]
		&= \sum_{j=1}^m f_j \int_0^\infty \chi_{[p_j,q_j]}(K)\tfrac{1}{q_j-p_j}(1-\exp(-\tfrac{a}{K}))\ dK\\[1em] 
		&= \sum_{j=1}^m f_j \int_{p_j}^{q_j} \tfrac{1}{q_j-p_j}(1-\exp(-\tfrac{a}{K}))\ dK  = \sum_{j=1}^m f_j H_j\ .
	\end{aligned} \]	
	Furthermore $x_{\text{approx.}} = \sum_{j=1}^m f_j h_j$ is just a definition. Thus it remains to show the properties for the proportional quantities.\\
	
	\noindent
	First, observe that the proportionality factor simply cancels out in $h_j$:
	\[1-\exp\left(-\frac{\mathfrak{a}}{\mathfrak{q}_j -\mathfrak{p}_j}\right) = 1-\exp\left(-\frac{\eta a}{\eta q_j - \eta p_j}\right) = 1 - \exp\left(-\frac{a}{q_j-p_j}\right) = h_j\ .\]
	For $H_j$ we use integration by substitution:
	\[\begin{aligned}
		H_j &= \int_{p_j}^{q_j} \frac{1}{q_j -p_j}\left(1-\exp\left(-\frac{a}{K}\right)\right) \ dK = \int_{\eta p_j}^{\eta q_j} \frac{1}{q_j -p_j}\left(1-\exp\left(-\frac{a}{\frac{1}{\eta}\mathfrak{K}}\right)\right)\ d(\tfrac{1}{\eta}\mathfrak{K}) \\[1em]
		&= \int_{\eta p_j}^{\eta q_j} \frac{1}{\eta}\frac{1}{q_j-p_j}\left(1-\exp\left(-\frac{\eta a}{\mathfrak{K}}\right)\right)\ d\mathfrak{K} = \int_{\eta p_j}^{\eta q_j} \frac{1}{\eta q_j-\eta p_j}\left(1-\exp\left(-\frac{\eta a}{\mathfrak{K}}\right)\right)\ d\mathfrak{K}\\[1em]
		&= \int_{\mathfrak{p}_j}^{\mathfrak{q}_j} \frac{1}{\mathfrak{q}_j - \mathfrak{p}_j} \left(1-\exp\left(-\frac{\mathfrak{a}}{\mathfrak{K}}\right)\right) \ d\mathfrak{K}\ .
	\end{aligned}\]
	Thus, $H_j$ and $h_j$ do not depend on the proportionality factor $\eta$. Furthermore, they do not depend on $\xi$, as no quantity that is used to calculate them depends on $\xi$. Thus, calculating $\mathfrak{x}$ and $\mathfrak{x}_{\text{approx.}}$ becomes a trivial insertion of the definitions $\mathfrak{x} = \xi x$ and $\mathfrak{f}_j = \xi f_j$:
	\[\mathfrak{x} = \xi x = \xi \sum_{j=1}^m f_j H_j = \sum_{j=1}^m \xi f_j H_j = \sum_{j=1}^m \mathfrak{f}_j H_j\ ,\]
	\[\mathfrak{x}_{\text{approx.}} = \xi x_{\text{approx.}} = \xi \sum_{j=1}^m f_j h_j = \sum_{j=1}^m \xi f_j h_j = \sum_{j=1}^m \mathfrak{f}_j h_j\ .\]
	\qed
\end{proof}

\begin{remark}[\bfseries Interpretation of theorem \ref{sup-thm: histogram and unit change}]
Essentially, theorem \ref{sup-thm: histogram and unit change} has the following meaning. When the antibody concentration and the epitope density are measured in other units than number of particles per volume\slash surface area, only the numbers written to the axes of the accessibility histogram change. The shape of the histogram remains unaffected. That is, the relation between the $\{\mathfrak{f}_j\}$ (measured in the new units) is the same as the relation between the $\{f_j\}$. Furthermore, the intervals $[p_j,q_j]$ to which the coefficients $f_j$ belong are transformed according to the antibody concentration unit change into $[\mathfrak{p}_j,\mathfrak{q}_j]$, s.t. the relation between the bin-widths of the histogram is preserved.
\end{remark}

\section{Validation of the accumulation hypothesis}
\label{sup-sec: Validation of the accumulation hypothesis}

The antibody accumulation hypothesis originated from the discrepancy between the Langmuir-model assumptions and the experimental protocols. Yet, despite being well motivated, a validation of the accumulation hypothesis is still due. In the results section we already briefly discussed validation data. Having derived the accumulation model and the accessibility histogram in great detail, we can derive further predictions and compare them to the experimental data. Furthermore, since we have validated in the main part of the paper that the accessibility histogram describes the binding properties of antibodies, we can apply the accessibility analysis. This will suggest properties of the antibody binding dynamics.

\subsection{Validation}

The idea is quite simple. Recall the accumulation hypothesis which states that the dose-response behavior originates from a premature interruption of the antibody accumulation (and the different binding speeds that depend on the antibody concentration $\frac{d}{dt} x(t)\sim a$). If that is the case, increasing the incubation time should increase the amount of bound antibodies, and thus the resulting fluorescence response for a given dose. Furthermore, this effect should apply to all antibody concentrations. For a dilution series (dose-response curve), this means that increasing the incubation time should appear as if the stock solution of primary antibody was higher concentrated. Thus, dose-response curves with longer incubation times appear to be shifted to the left in a logarithmically scaled plot.

\begin{remark}[{\bfseries Left-shift of dose-response curves -- explanation}]
	\label{sup-rem: incubation time left-shift}
	To understand, why dose-response curves with longer incubation times appear left-shifted in a logarithmically scaled plot, we assume for the moment that doubling the incubation time has the same effect as doubling the concentration. I.e. the response function for the experiments with twice as long incubation times $r_{2\tau}(a)$ is given by
	\[r_{2\tau}(a) = r_\tau(2a)\ , \]
	where $r_\tau(a)$ denotes the response function of experiments with the original incubation time. This means that a given response value $r_{2\tau}(\frac{a}{2}) = r_\tau(a)$ of the $r_{2\tau}$-curve appears already at the concentration $\frac{a}{2}$, while the same response value for the $r_\tau$-curve appears at a higher concentration $a$. In a logarithmic scale, this transformation $a \mapsto \frac{a}{2}$ becomes $\log(a) \mapsto \log(\frac{a}{2}) = \log(a)-\log(2)$. This transformation describes the aforementioned left-shift.
	\vspace{1em}

	This proportionality assumption, doubling the incubation time is the same as doubling the antibody concentration, is not unmotivated. Observe that the antibody concentration and the incubation time always appear as product  $a\cdot \tau$ in the accumulation model:
	\[1-\exp\left(-\frac{a}{K_\tau}\right) = 1- \exp\left(-\frac{a}{\frac{1}{k_a \cdot \tau}}\right) =1-\exp(-k_a \cdot a \cdot \tau) \ .\]
	However, the accumulation model is just an empirical model, neglecting effects like antibody depletion. Thus, a perfect proportionality is not to be expected. The effect of increasing the incubation time should diminish for longer incubation times.
\end{remark}

\begin{figure}[h]
	\centering

	\begin{subfigure}[c]{0.49\textwidth}
		\centering
		{\large \textbf{(a)} raw data}
		
		\includegraphics[width = \textwidth]{images/supplement_figure_accumulation/common_data.pdf}
	\end{subfigure}
	\begin{subfigure}[c]{0.49\textwidth}
		\centering
		{\large \textbf{(b)} shifted curves}
		
		\includegraphics[width = \textwidth]{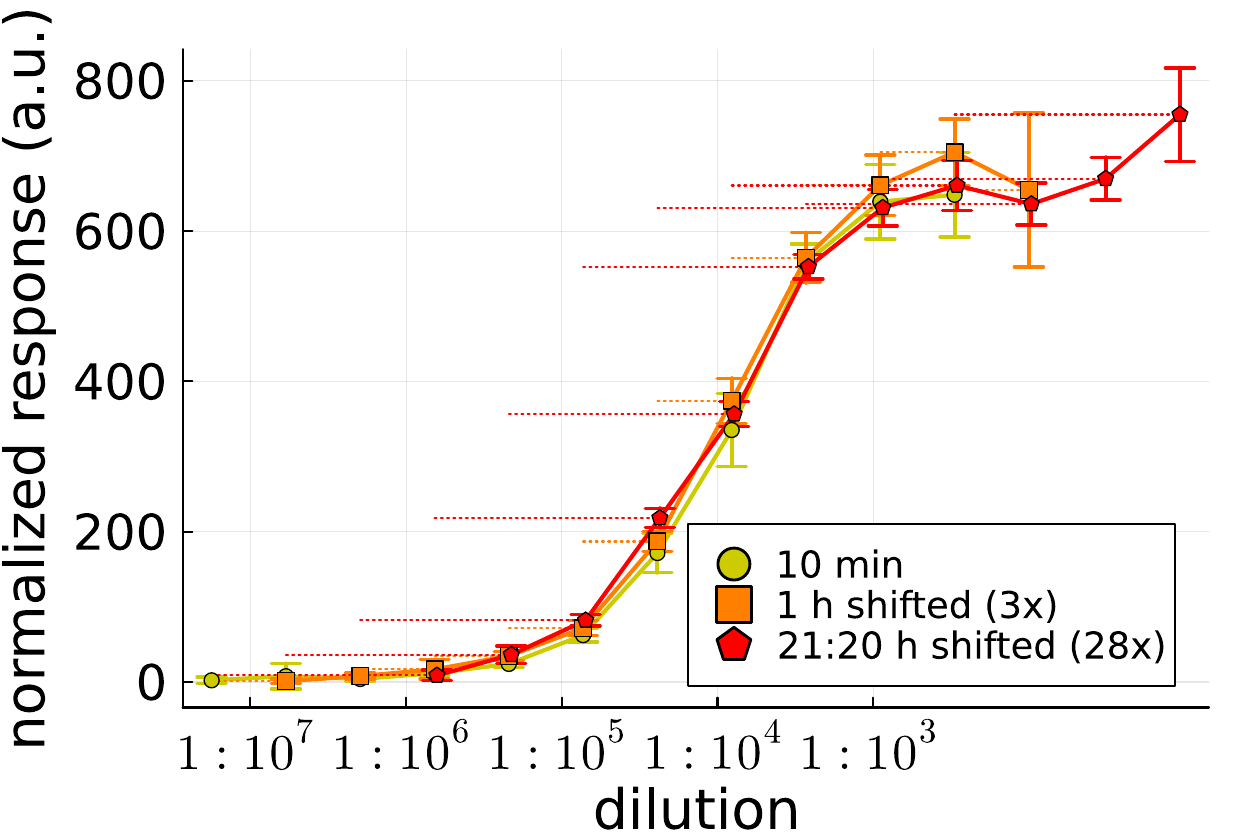}
	\end{subfigure}
	\caption{Dose-response curves (anti-NF200 antibody, HeLa cells) with different incubation times of the primary antibodies. Each data point is the mean of 16 replicates with the standard deviation of the replicates as error bars. \textbf{(a)} shows the raw data plots. \textbf{(b)} shows the same response curves but plotted with increased concentrations (3 times higher concentrations for the \unit[1]{h}-incubation curve and 28 times higher concentrations for the \unit[21:20]{h}-incubation curve). The dashed lines indicate the shift of the curves.}
	\label{sup-fig: incubation time data}
\end{figure}

Figure \ref{sup-fig: incubation time data}a shows three dose-response curves, each created with a different primary antibody incubation time (\unit[10]{mins}, \unit[1]{h} and \unit[21:20]{h}). As predicted by the accumulation hypothesis, longer incubation times for the primary antibody lead to stronger responses. Furthermore, also as predicted, the dose-response curves that were created with longer incubation times are shifted to the left. Figure \ref{sup-fig: incubation time data}b demonstrates that the dose-response curves have the same shape, i.e. are indeed just shifted. This plot was created by uniformly increasing the dilution quotients by a factor of 3 for the \unit[1]{h}-incubation curve and by a factor of 28 for the \unit[21:20]{h}-incubation curve. All in all, the results shown in figure \ref{sup-fig: incubation time data} validate the accumulation hypothesis.

Note, that as explained in remark \ref{sup-rem: incubation time left-shift}, the shift is not proportional to the increase of the incubation time. The \unit[1]{h}-incubation is 6-times longer than the \unit[10]{min}-incubation but the dose-response curves are just shifted apart by a factor of 3. The \unit[21:20]{h}-incubation is $\frac{64}{3}$-times longer than the \unit[1]{h}-incubation but the dose-response curves are just shifted apart by a factor of $\frac{28}{3}$. As speculated in remark \ref{sup-rem: incubation time left-shift}, the effect of increased incubation times decreases the longer the reference incubation time already is in the first place.

\begin{figure}[t]
	\centering

	\begin{subfigure}[c]{0.9\textwidth}
		\centering 
		{\large \textbf{(a)} 10-minute primary antibody incubation}
		\vspace{0.3cm}

		\includegraphics[width = 0.49\textwidth]{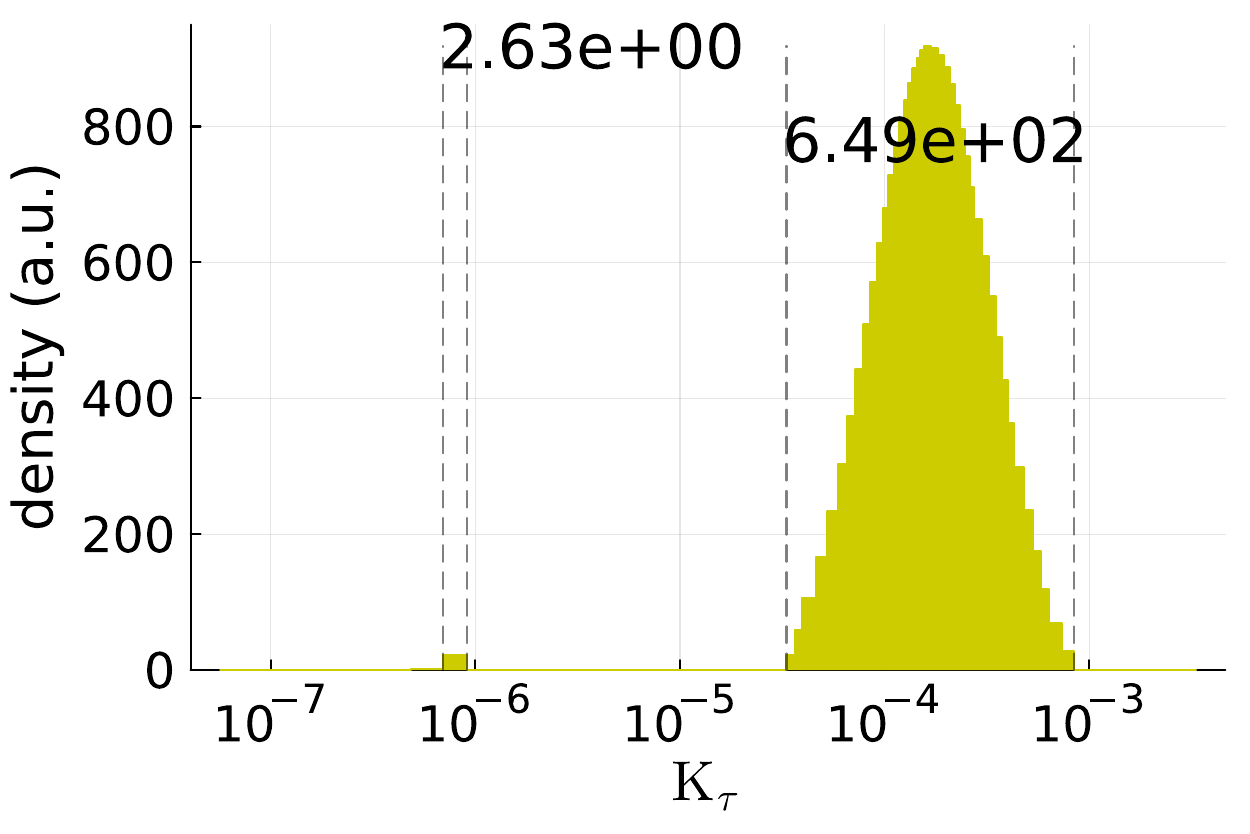}
		\includegraphics[width = 0.49\textwidth]{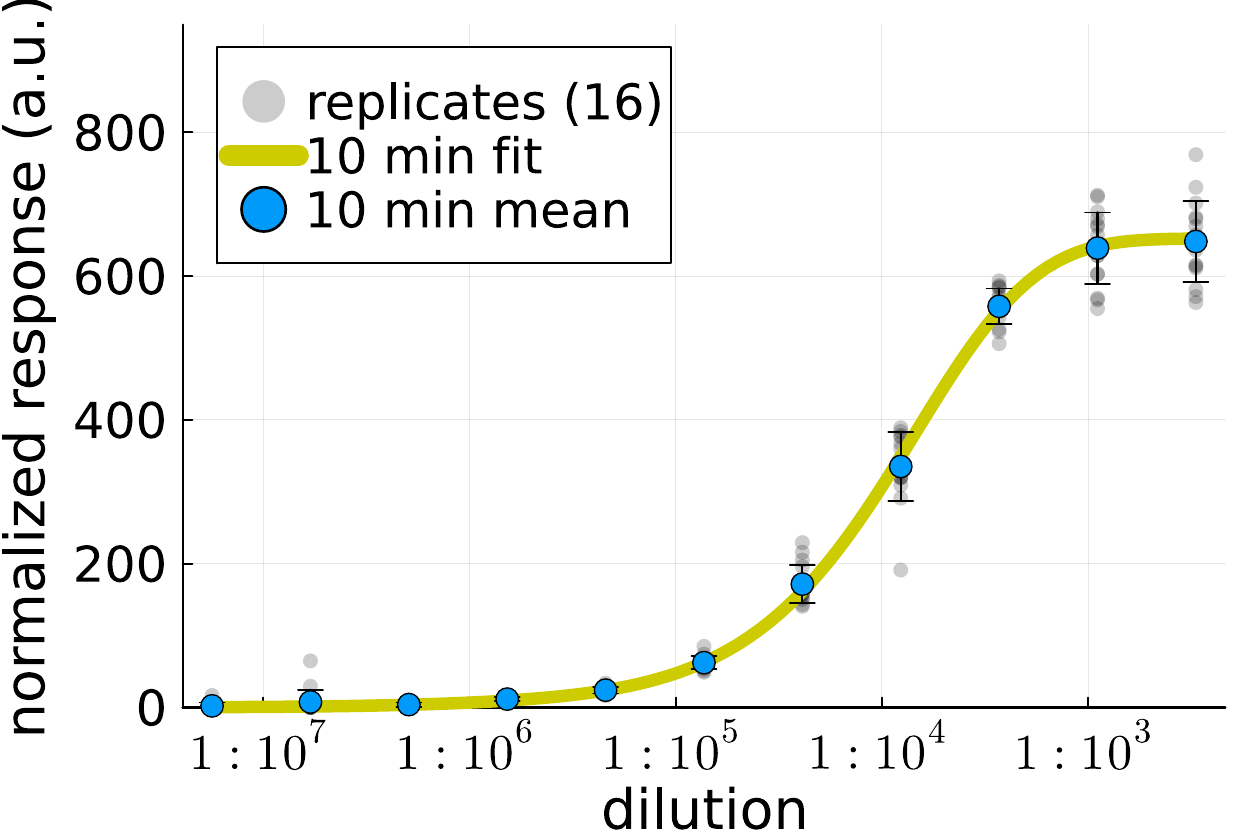}
	\end{subfigure}
	\vspace{0.6cm}

	\begin{subfigure}[c]{0.9\textwidth}
		\centering 
		{\large \textbf{(b)} 1-hour primary antibody incubation}
		\vspace{0.3cm}

		\includegraphics[width = 0.49\textwidth]{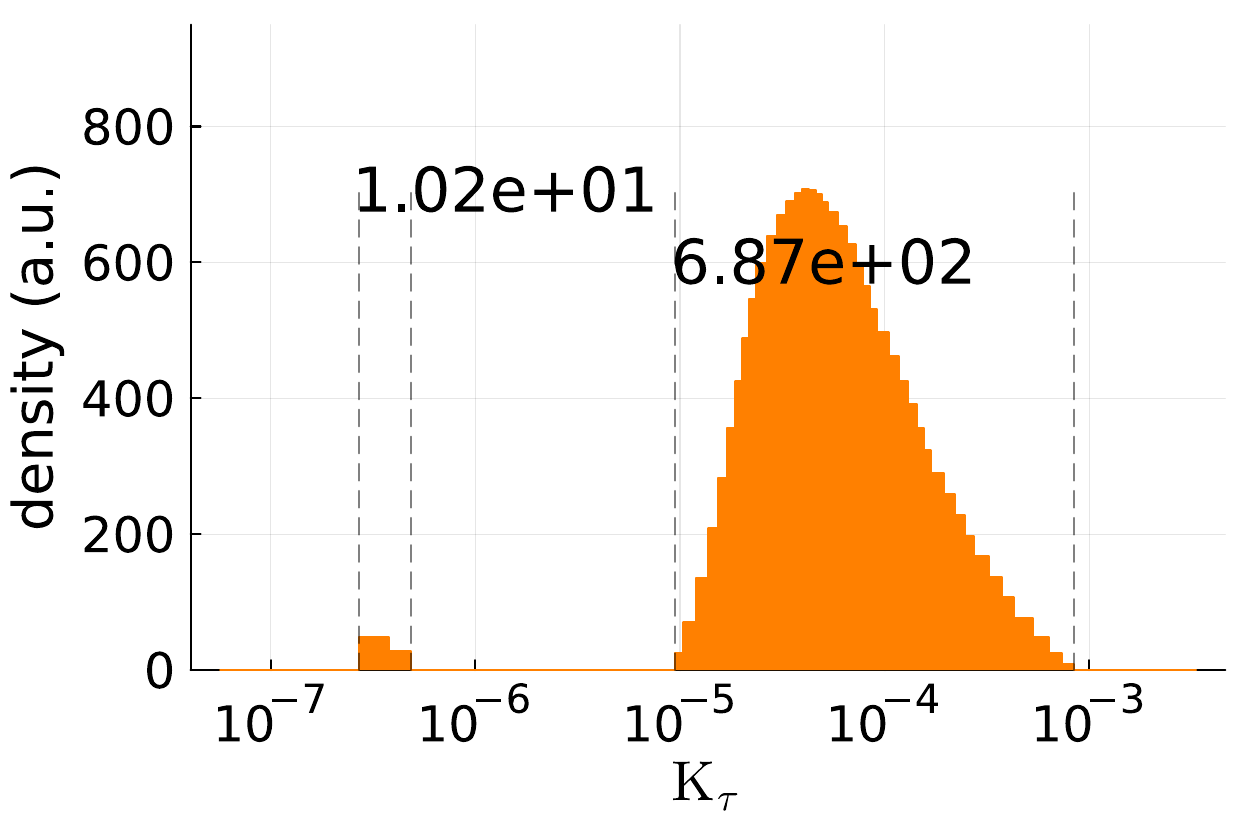}
		\includegraphics[width = 0.49\textwidth]{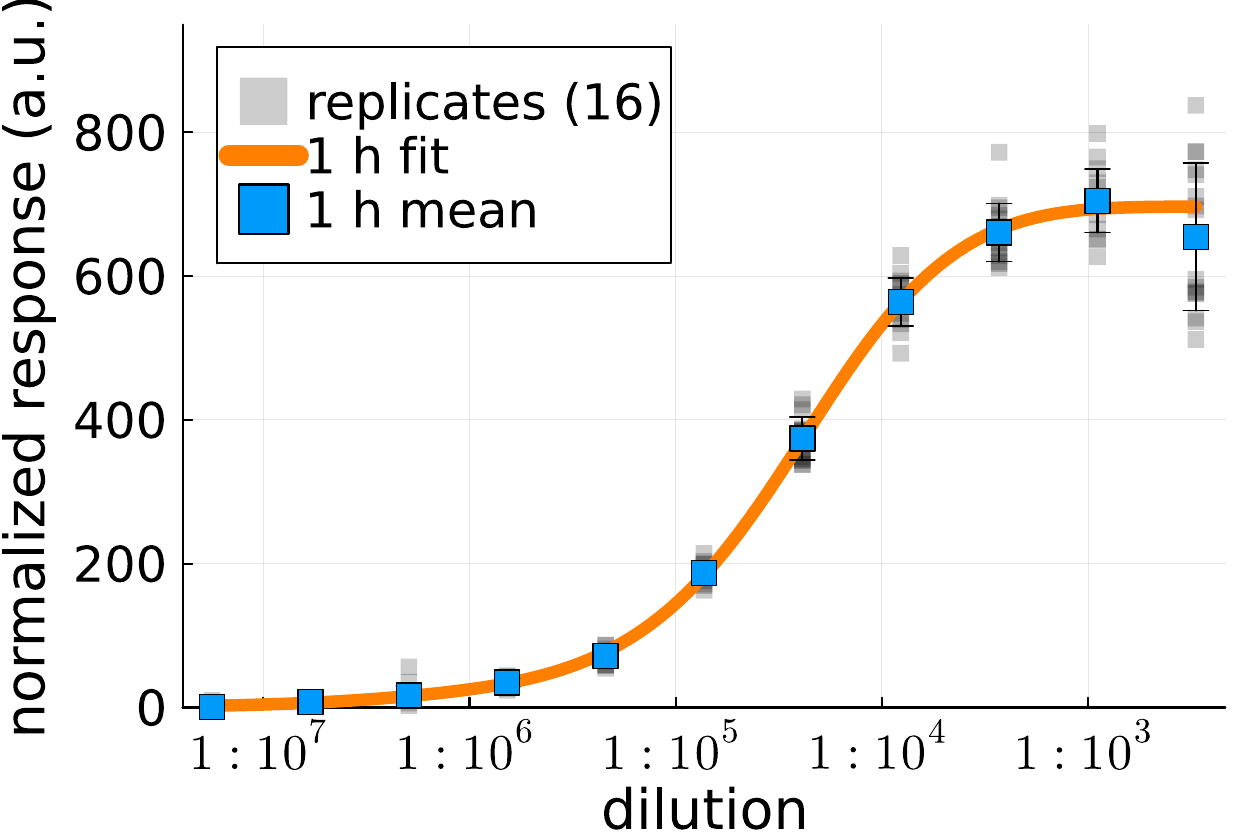}
	\end{subfigure}
	\vspace{0.6cm}

	\begin{subfigure}[c]{0.9\textwidth}
		\centering 
		{\large \textbf{(c)} 21:20-hour primary antibody incubation}
		\vspace{0.3cm}

		\includegraphics[width = 0.49\textwidth]{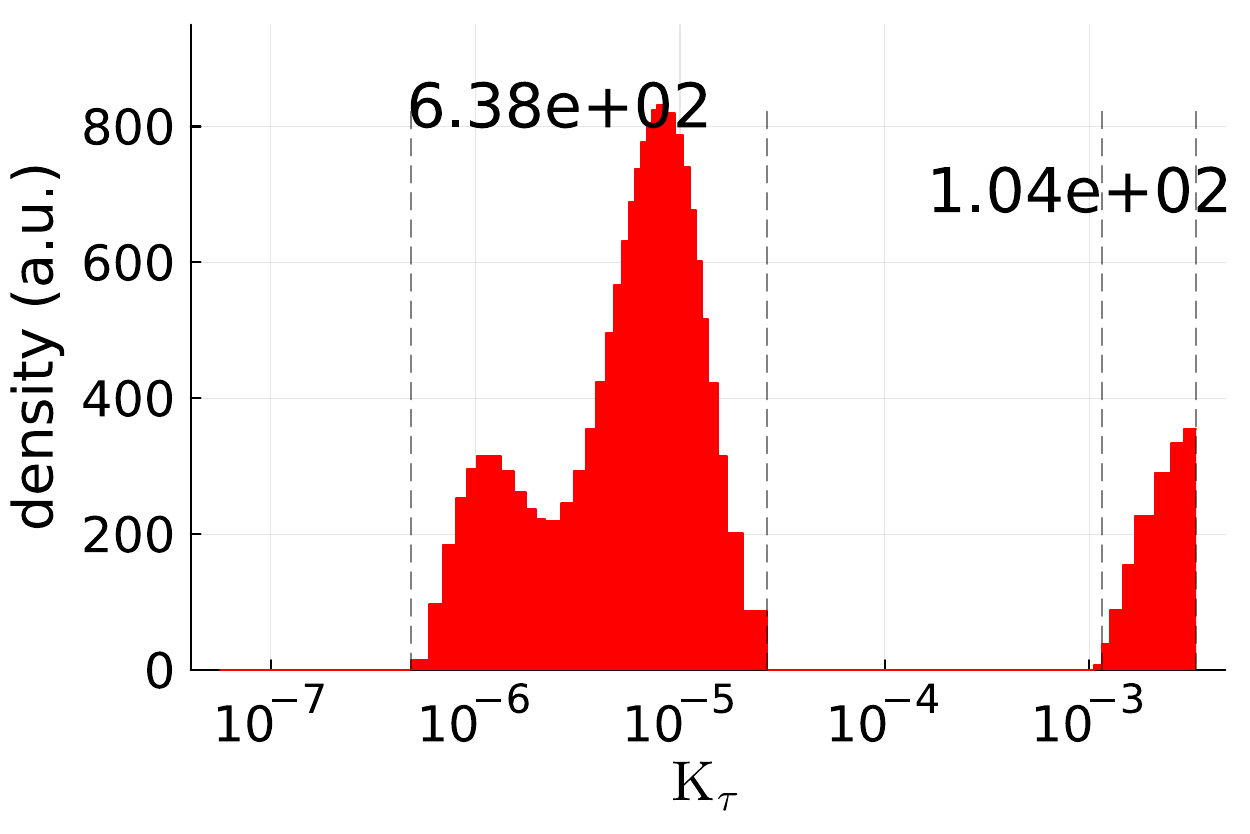}
		\includegraphics[width = 0.49\textwidth]{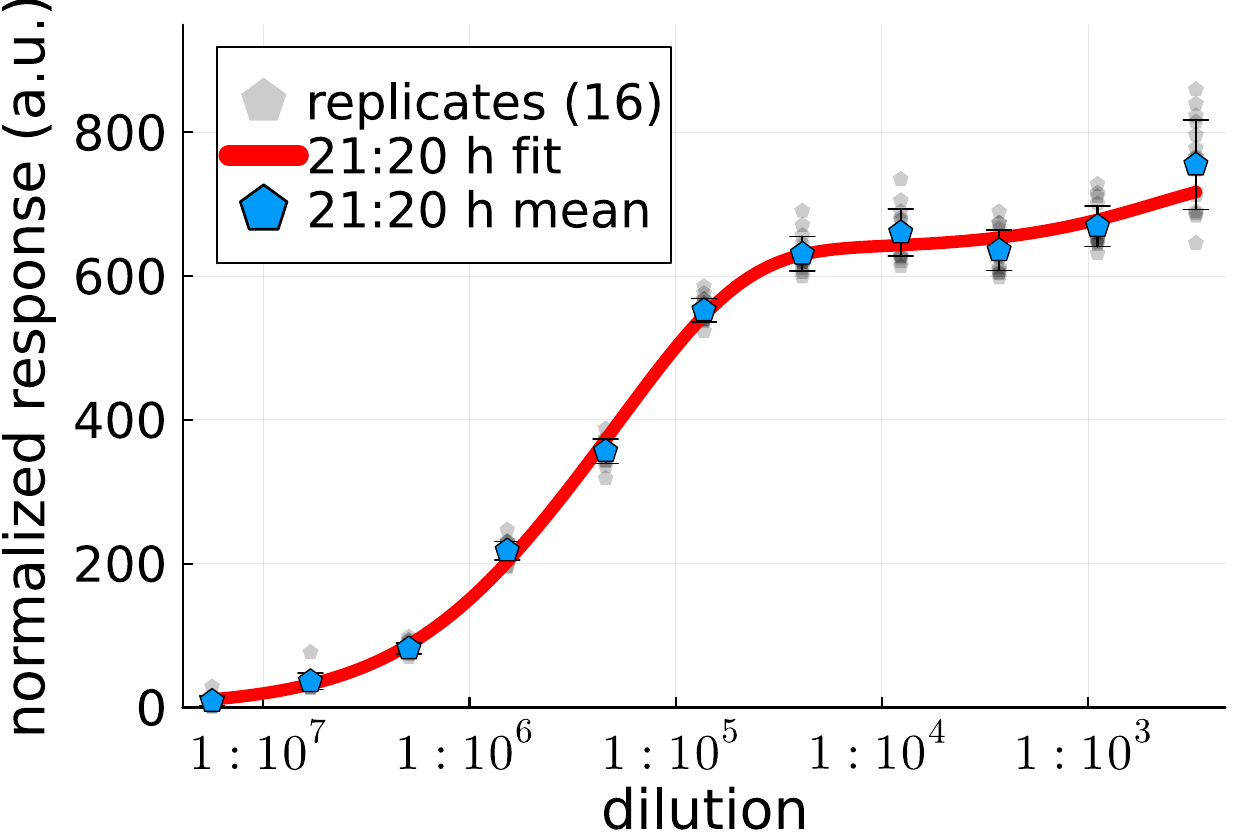}
	\end{subfigure}

	\caption{Accessibility analysis of the dose-response curves for the different primary antibody incubation times form figure \ref{sup-fig: incubation time data}(a).}
	\label{sup-fig: incubation time accessibility analysis}
\end{figure}

\subsection{Accessibility analysis}

Although the raw data plots are enough to validate the accumulation hypothesis, applying the accessibility analysis is worthwhile at this point. Especially since we have already validated the accessibility analysis in the main part of the paper. Accordingly, figure \ref{sup-fig: incubation time accessibility analysis} shows the results of the accessibility analysis for for the 3 dose-response curves from figure \ref{sup-fig: incubation time data}a.  

First, we observe that the longer the incubation of the primary antibodies, the further left are the peaks in the accessibility histogram. This is hardly surprising, as the dose-response curves are shifted to the left, as explained above. Since the $K_\tau$ values of peaks in the histogram correspond to the dilution quotients of features in the dose-resposne curve, this left-shift must appear in the histogram. But recall the validation system from the main part. We increased the concentration of the anti-NF200 stock solution (by labelling the 1:20 dilution quotient as 1:100) to shift the corresponding histogram peaks to the left. In this regard, the left-shift upon longer incubation times corroborates the previous statement: Increased incubation times appear as if the concentration of the primary antibody stock solution was increased.

A more interesting result is that the single peak in figure \ref{sup-fig: incubation time accessibility analysis}a disperses as it shifts to the left in figures \ref{sup-fig: incubation time accessibility analysis}b and \ref{sup-fig: incubation time accessibility analysis}c. At first sight, this seems surprising, as one would expect just a shift of the peak. Yet, the dispersion makes sense in light of the binding process and the accessibility interpretation.

\begin{remark}[{\bfseries Reasonable incubation times}]
	The illustrated antibody binding process assume, of course, that the incubation times are reasonable. That is, the incubation times are long enough s.t. antibodies could at least reach epitopes (in a straight line) with their diffusion velocity. All explanations break down if infinitely short or infinitely long incubation times are considered.
\end{remark}

\noindent
For antibodies to bind to an epitope, they have to be in close vicinity to the epitope. However, antibody movement is in general not a directed process, but governed by Brownian motion. That is, the path of an antibody is essentially a random walk. Accessible and less accessible epitopes can differ in two aspects. First, less accessible epitopes can be harder to reach, requiring longer\slash less probable random walks. Second, less accessible epitopes can be partially blocked or have a lower affinity such that the antibody requires multiple attempts to bind (the antibody must approach the epitope from the right angle). This again requires to a less probable random walk. 

Because of the stochastic nature of antibody binding, there are two ways to achieve a certain amount of bound antibodies, i.e. a certain response signal. One can either increase the incubation time or one can use more antibodies (higher concentration), or a combination of both.

In case of long incubation times, fewer antibodies (lower concentration) are required to obtain a certain response signal, as the long incubation time gives each individual antibody a higher chance to bind eventually. However, the accessible epitopes act as obstruction for less accessible epitopes in this case. That is, because of the location of accessible epitopes, a random walk is likely to pass by at such an epitope. Furthermore, accessible epitopes are often easier to bind to. Hence, accessible epitopes essentially catch antibodies before they have the chance to reach inaccessible epitopes. Since low amounts of antibodies already suffice to obtain a certain response signal in case of long incubation times, there are not enough antibodies left to bind to less accessible epitopes. In other words, large enough amounts of antibodies to cover most accessible epitopes are required before a substantial amount of antibodies can bind to less accessible epitopes. Thus, the difference between accessible epitopes and less accessible epitopes should be most pronounced in case of long incubation times. 

In case of short incubation times, most antibody random walks do not lead to binding events, such that more antibodies (higher concentrations) are required to obtain a certain response signal. Because of the excess of antibodies, a large proportion of accessible epitopes can be bound to by antibodies, which blocks their capability to obstruct antibodies from reaching less accessible epitopes. In addition, because of the excess of antibodies, a sufficiently large amount of free antibodies remains to have a chance to reach and bind to less accessible epitopes. In other words, at concentrations where the accessible epitopes get bound to, there are enough antibodies left that have a chance, albeit a lower one, to bind to less accessible epitopes. Thus, the distinction of accessible and less accessible epitopes is less pronounced in case of short incubation times.

\begin{remark}[{\bfseries Antibody depletion does not cause peak dispersion}]
	One might interject at this point, that longer incubation times could lead to stronger depletion effects which could equally well explain the odd behavior of dispersing peaks. The idea is, that longer incubation times allow more antibodies to bind to epitopes, exacerbating the depletion.
	\vspace{14pt}

	But the depletion effect should have the opposite effect, it should compress accessible peaks back into inaccessible peaks (shift to the right). The reason is, as discussed in section \ref{sup-sec: depletion correction}, that the concentration becomes lower because of antibody depletion. This slows down the antibody binding, s.t. the binding rate $k_a$ is underestimated, i.e. $K_\tau = \frac{1}{k_a \tau}$ is overestimated (larger than it should be). By the way, the same reasoning explains why antibody depletion leads to an underestimation of the dissociation constant $K_d$ in the Langmuir model. See e.g. \cite{Edwards_1998} and \cite{Jarmoskaite_2020} for further details.
	\vspace{14pt}

	In summary, depletion effects can be ruled out as reason for the peak dispersion observed when longer incubation times are used.
\end{remark}

\FloatBarrier

\section{Depletion correction}
\label{sup-sec: depletion correction}

In subsection \ref{sup-subsec: unit invariance fails} we have seen that the conversion factor $\beta$ for the antibody depletion term $a-\beta x(t)$ explicitly depends on the choice of units (see corollary \ref{sup-cor: depletion model no cancellation}). Since the unit conversion between response signal (used for the measurement) and bound antibody-epitope complex density is unknown in most cases (as argued in subsection \ref{sup-subsec: measurement limitations}), $\beta$ often remains unknown in practice. For that reason, we had to neglect antibody depletion in the accumulation model. Nevertheless, the maximal depletion can be estimated from the data, allowing to estimate the worst-case error of neglecting antibody depletion.

Let $\{(\mathfrak{a}_i, \mathfrak{x}_i)\}_{i=1}^n$ be dose-response data in experimental units, e.g. dilution quotient and fluorescence response and recall the derivation of the accumulation model: The antibody binds permanently to the epitopes, accumulating over time. Thus, the corresponding response $\mathfrak{x}(t)$ increases over time (assuming one would stop the incubation and measure $\mathfrak{x}(t)$ at the time $t$ after washing). The antibody incubation is eventually stopped after a period of time $\tau$, s.t. the response that is measured is given by $\mathfrak{x} \coloneqq \mathfrak{x}(\tau)$. During the antibody incubation, the initial antibody dilution quotient $\mathfrak{a}$ has decreased by the amount of antibodies that have bound to epitopes. Hence, the dilution quotient of unbound antibodies $\mathfrak{b}$, after the incubation phase, is given by (using the still unknown conversion factor $\boldsymbol{\beta}$ in experimental units):
\[\mathfrak{b} = \mathfrak{a}-\boldsymbol{\beta} \mathfrak{x}\ .\]
Note that $\boldsymbol{\beta}$ is a fixed constant for a given choice of units. Hence, it holds for all data points that the remaining free antibody dilution quotient is
\[\mathfrak{b}_i = \mathfrak{a}_i-\boldsymbol{\beta} \mathfrak{x}_i\ .\]

At this point, we still do not know the conversion factor $\boldsymbol{\beta}$. But since the remaining antibody dilution quotient must not be negative, we can estimate the largest possible conversion factor 
\[\widehat{\boldsymbol{\beta}}\coloneqq \max\{\gamma \geq 0 \mid \mathfrak{a}_i -\gamma \mathfrak{x}_i \geq 0\quad \forall\ i\}\ .\]
Then, following the idea of \cite{Edwards_1998}, we may correct the initial antibody dilution quotients using the largest possible conversion factor. 
\[\widetilde{\mathfrak{b}}_i \coloneqq \mathfrak{a}_i -\widehat{\boldsymbol{\beta}}\mathfrak{x}_i \qquad \leadsto \qquad \text{corrected data:}\quad \{(\widetilde{\mathfrak{b}}_i,\mathfrak{x}_i)\}_{i=1}^n\ .\]
This constitutes a worst-case scenario for two reasons. First, by definition $\widehat{\boldsymbol{\beta}}$ is larger than the true conversion factor $\boldsymbol{\beta}$, which leads to smaller values for the $\widetilde{\mathfrak{b}}_i$. Second, using the values $\widetilde{\mathfrak{b}}_i$ assumes that the antibodies were depleted for the whole antibody incubation phase, when in reality the antibody dilution quotient started from the initial dilution quotient in the beginning, gradually decreasing over time.

\begin{figure}[h!]
	\centering 
	\centerline{\includegraphics[width = 18cm]{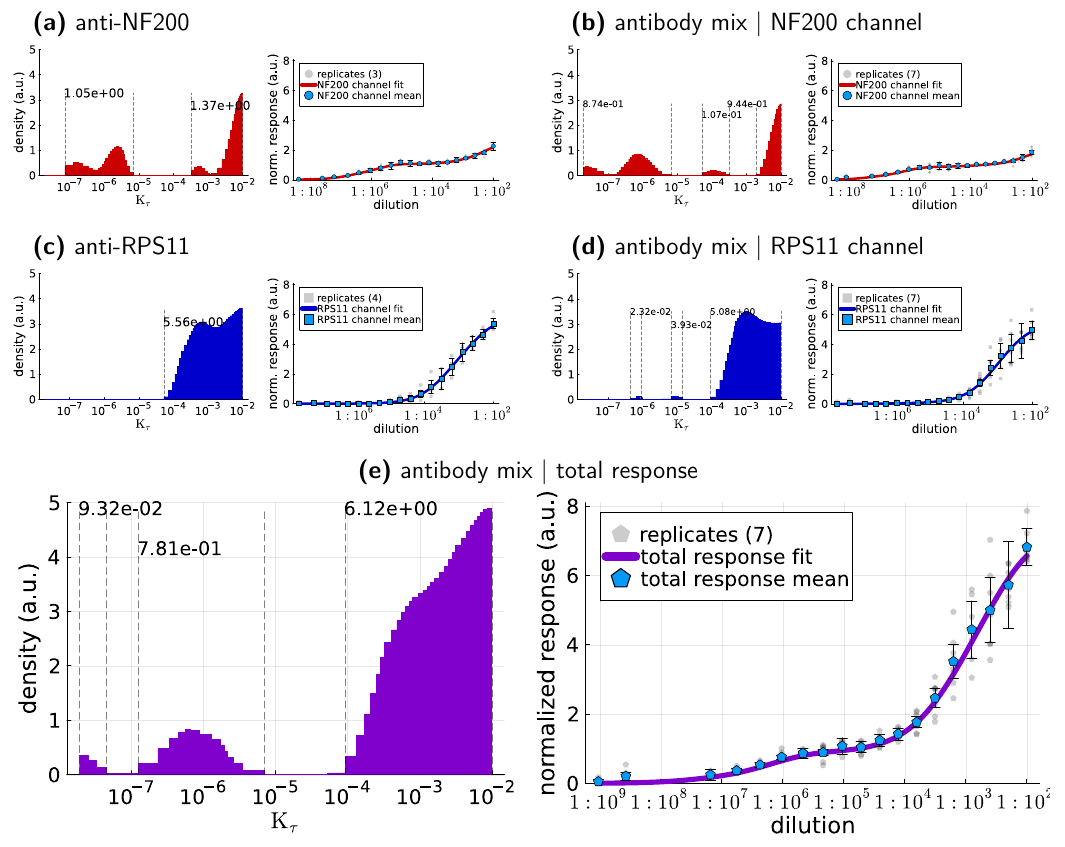}}
	\caption{Depletion-corrected dose-response curves of the validation system (cf. figure \ref{main-fig: histograms}) and the corresponding accessibility histograms.}
	\label{sup-fig: depletion}
\end{figure}

\begin{figure}[h!]
	\centering 
	\begin{subfigure}[c]{0.49\textwidth}
		\includegraphics[width = \textwidth]{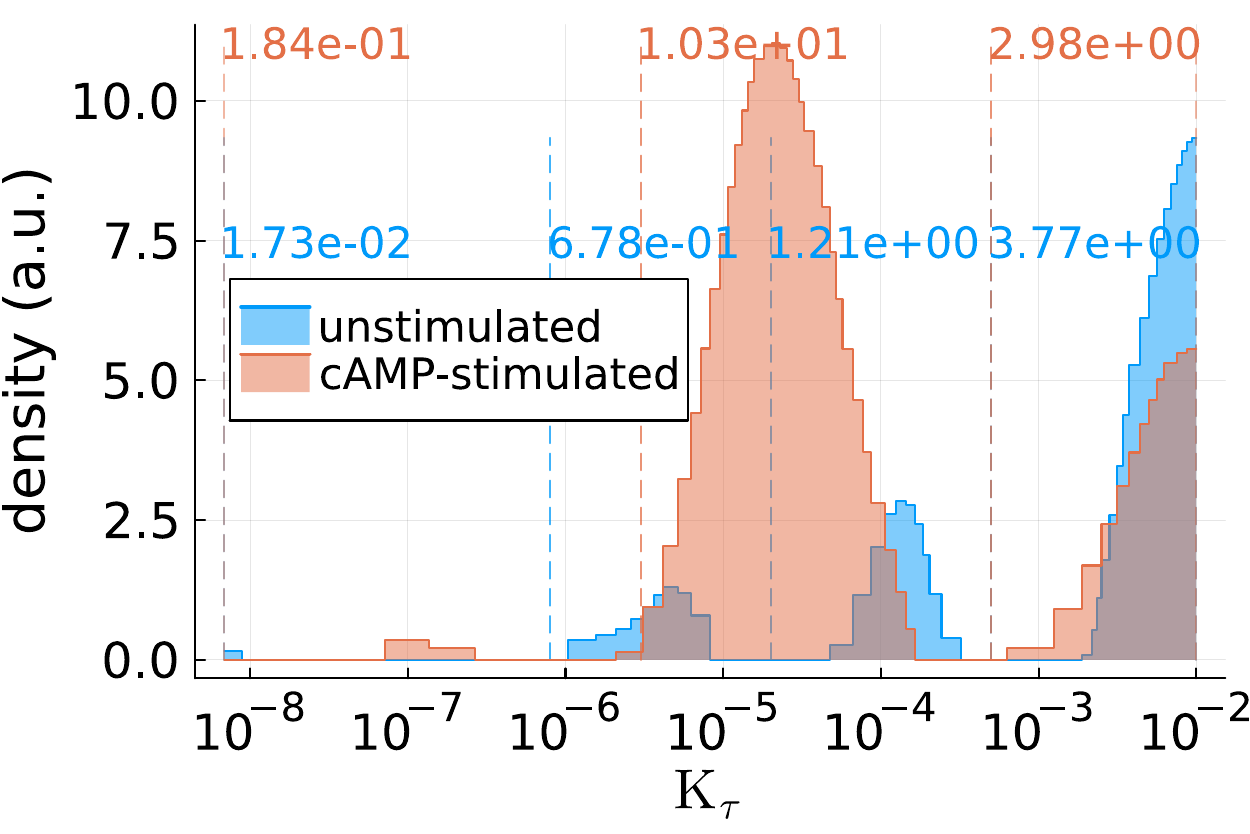}
	\end{subfigure}
	\begin{subfigure}[c]{0.49\textwidth}
		\includegraphics[width = \textwidth]{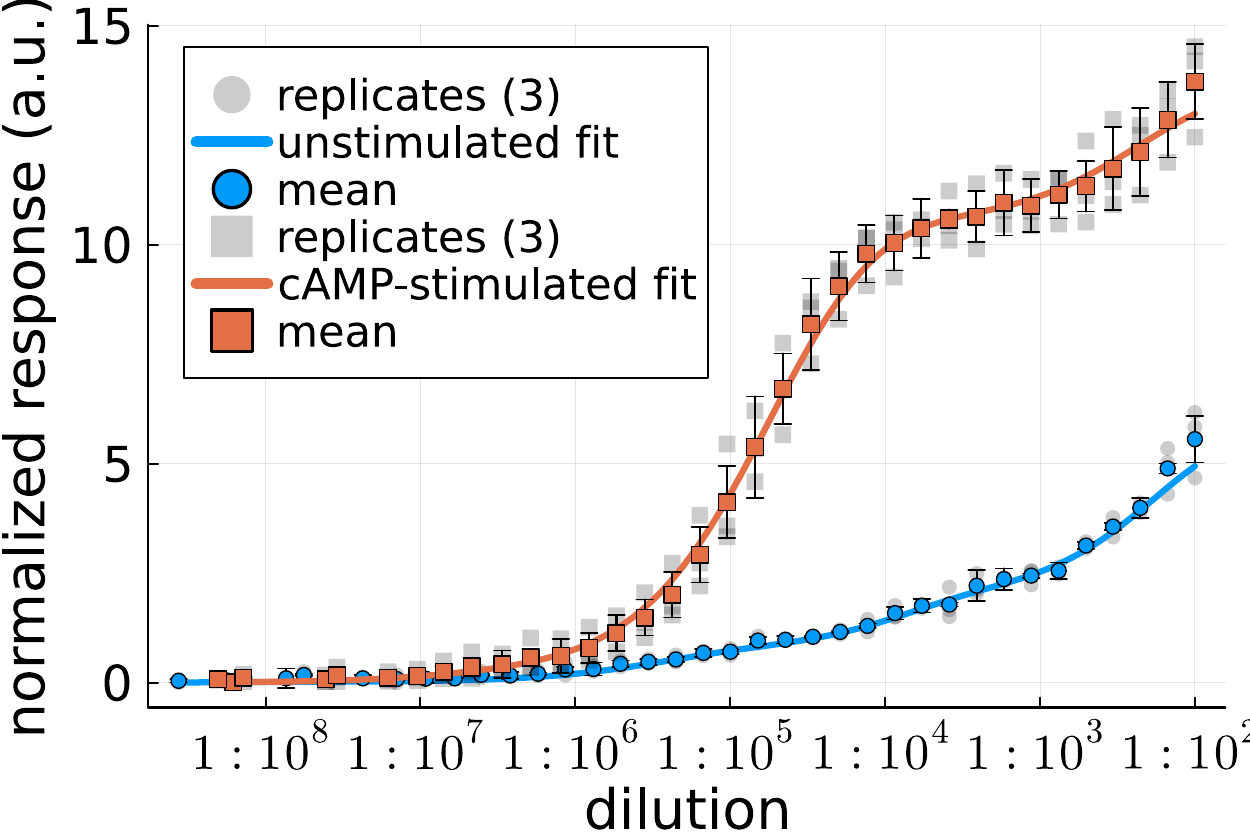}
	\end{subfigure}
	\begin{subfigure}[c]{0.49\textwidth}
	\end{subfigure}
	\caption{Depletion-corrected dose-response curves of the pRII-antibody (cf. figure \ref{main-fig: PKA}) and the corresponding accessibility histograms.}
	\label{sup-fig: PKA depletion}
\end{figure}
 
Figure \ref{sup-fig: depletion} shows depletion-corrected dose-response curves (for the data from figure \ref{main-fig: histograms}) and the accessibility histograms obtained from them. The depletion correction was applied before removing the baseline of the dose-response curves. For the plots, but not for the curve fitting, any data points with dilution quotient equal to zero were removed, as they would conflict with the logarithmic scale.

Observe that the depletion correction only affects the lower dilution-quotient data points, as expected. Nevertheless, the obtained accessibility histograms remain almost the same as for the uncorrected dose-response curves in figure \ref{main-fig: histograms}. Only a slight shift to the left for the leftmost peaks can be observed in fig. \ref{sup-fig: depletion}b and \ref{sup-fig: depletion}e. This indicates that depletion did not affect the results in the main part. The same is true for the PKA-data (see figure \ref{sup-fig: PKA depletion}).

\FloatBarrier

\section{Histograms with weaker regularization}
\label{sup-sec: histograms with weaker regularization}

In the main part of the paper, we used a large regularization parameter ($\alpha = 500$) to reduce peak features that could be caused by noise. However, a strong regularization may also mask true features that correspond to the actual binding behavior. For this reason, we investigate the validation system histograms (cf figure \ref{main-fig: histograms})  for a weaker regularization ($\alpha = 40$).

\subsection{Overview}

All in all, it can be observed that reducing the regularization strength does not lead to many additional peaks. The only exception is that the broad peak of the anti-RPS11 antibody becomes two distinct peaks. The presence of two peaks could already be assumed from figure \ref{main-fig: histograms}, but becomes more apparent when weaker regularizations are used.

\begin{center}
	\centerline{\includegraphics[width = 18cm]{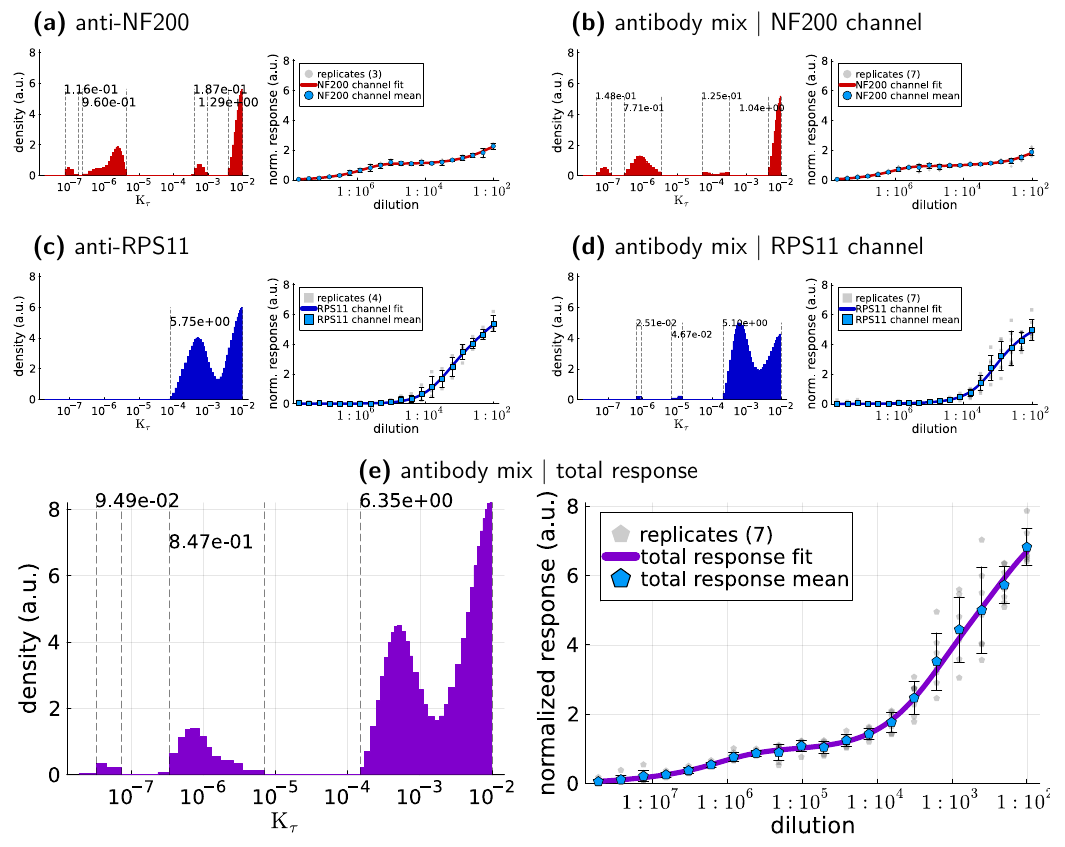}}
	\captionof{figure}{Re-analysis of the validation system, using a weaker regularization: $\alpha = 40$ (cf. fig. \ref{main-fig: histograms}, where $\alpha = 500$).}
	\label{sup-fig: weaker regularization}
\end{center}

\FloatBarrier

\subsection{Individual replicates}
\label{sup-subsec: Individual replicates}

When a weak regularization parameter is used to extract more features, it is a good idea to check if the obtained features are not noise artifacts. For this purpose, one may compare the accessibility histograms of the individual replicates. Figures \ref{sup-fig: rep 1}-\ref{sup-fig: rep 7} show the dose-response curves and the corresponding accessibility histograms of the individual replicates for the antibody-mix condition. 

In general, the major features of the histograms for the mean value dose-response curves (fig. \ref{sup-fig: weaker regularization}b,d,e) are present in all replicate histograms, confirming the results in figure \ref{sup-fig: weaker regularization}. In particular, this corroborates that there are 2 separate peaks in the histogram of the anti-RPS11 antibody. However, the histograms are not identical between the different replicates, which is due to the usual replicate to replicate variation. This is especially noticeable for replicate 1 (fig. \ref{sup-fig: rep 1}) and replicate 7 (fig. \ref{sup-fig: rep 7}), which may be regarded as outliers.

Note that because of the general replicate to replicate variability one should (if possible) use the mean value of replicates for the dose-response curve. Individual replicates can then be used to corroborate findings, but one should avoid drawing conclusions from a single replicate, as replicates 1 and 7 demonstrate.

\vspace{0.5cm}
\noindent
\begin{minipage}{\textwidth}
	\centering

	{\large \bfseries Replicate 1}\vspace{0.5em}

		\includegraphics[width = 0.32\textwidth]{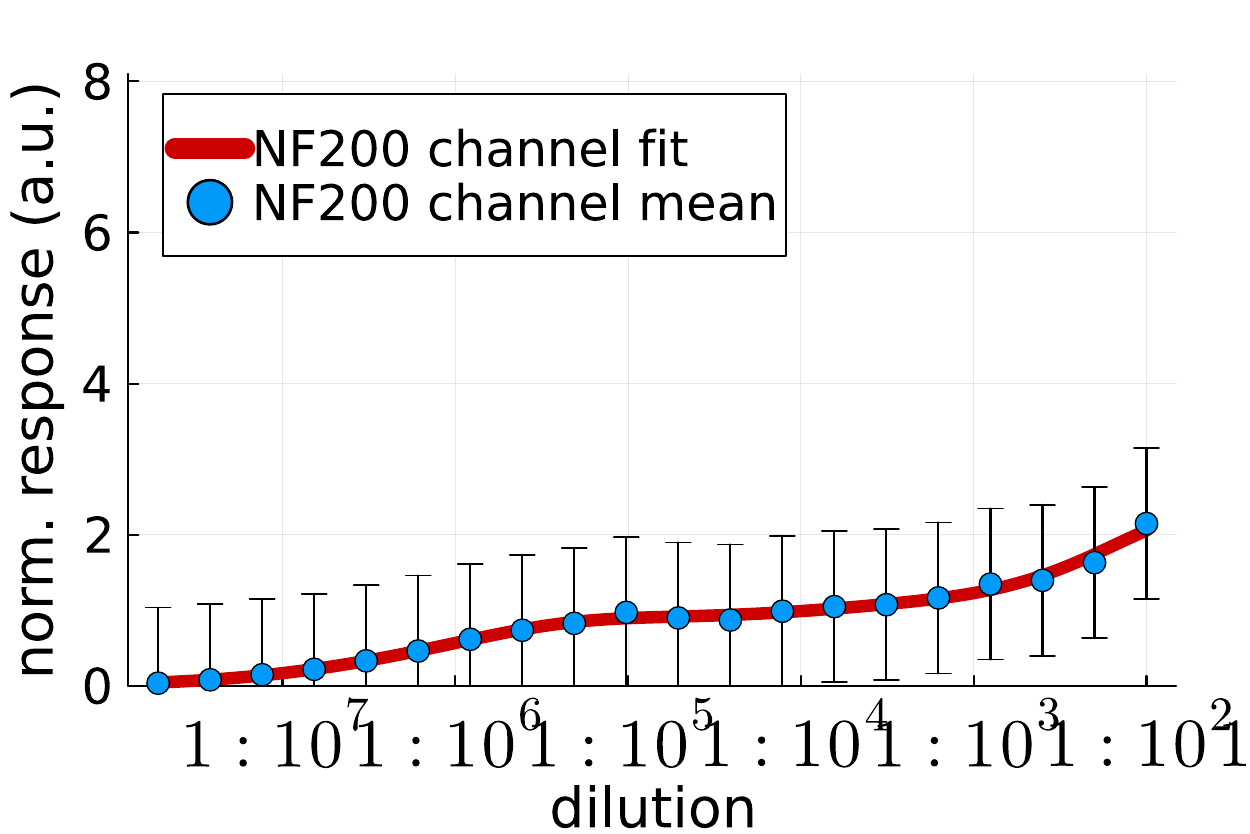}
		\includegraphics[width = 0.32\textwidth]{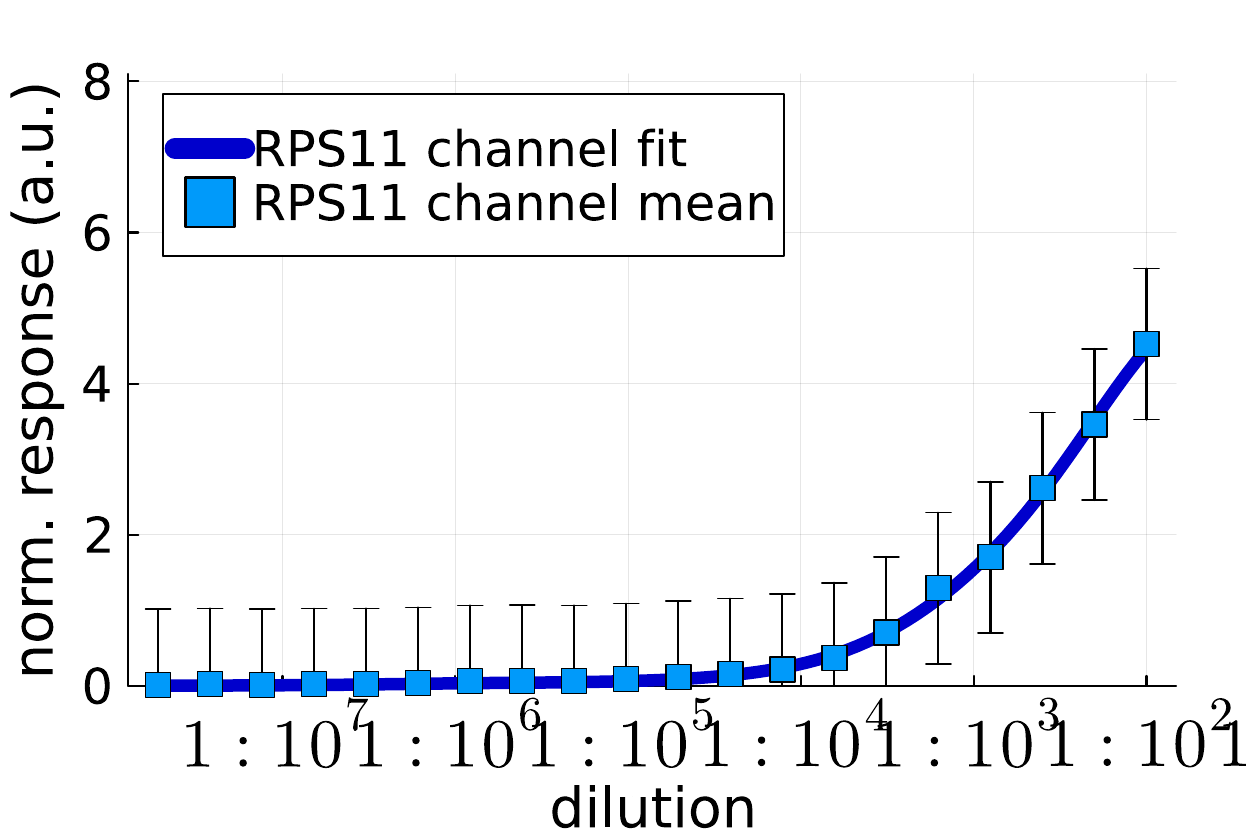}
		\includegraphics[width = 0.32\textwidth]{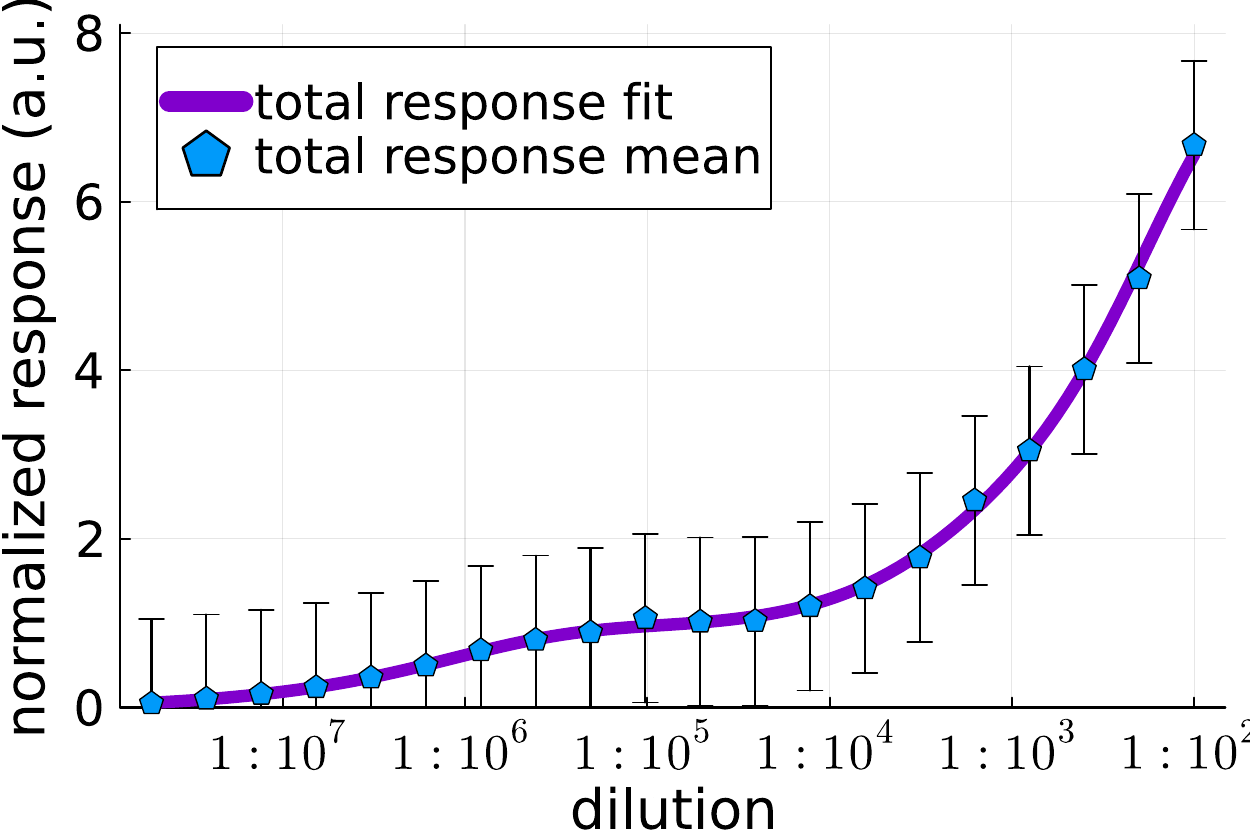}

		\includegraphics[width = 0.32\textwidth]{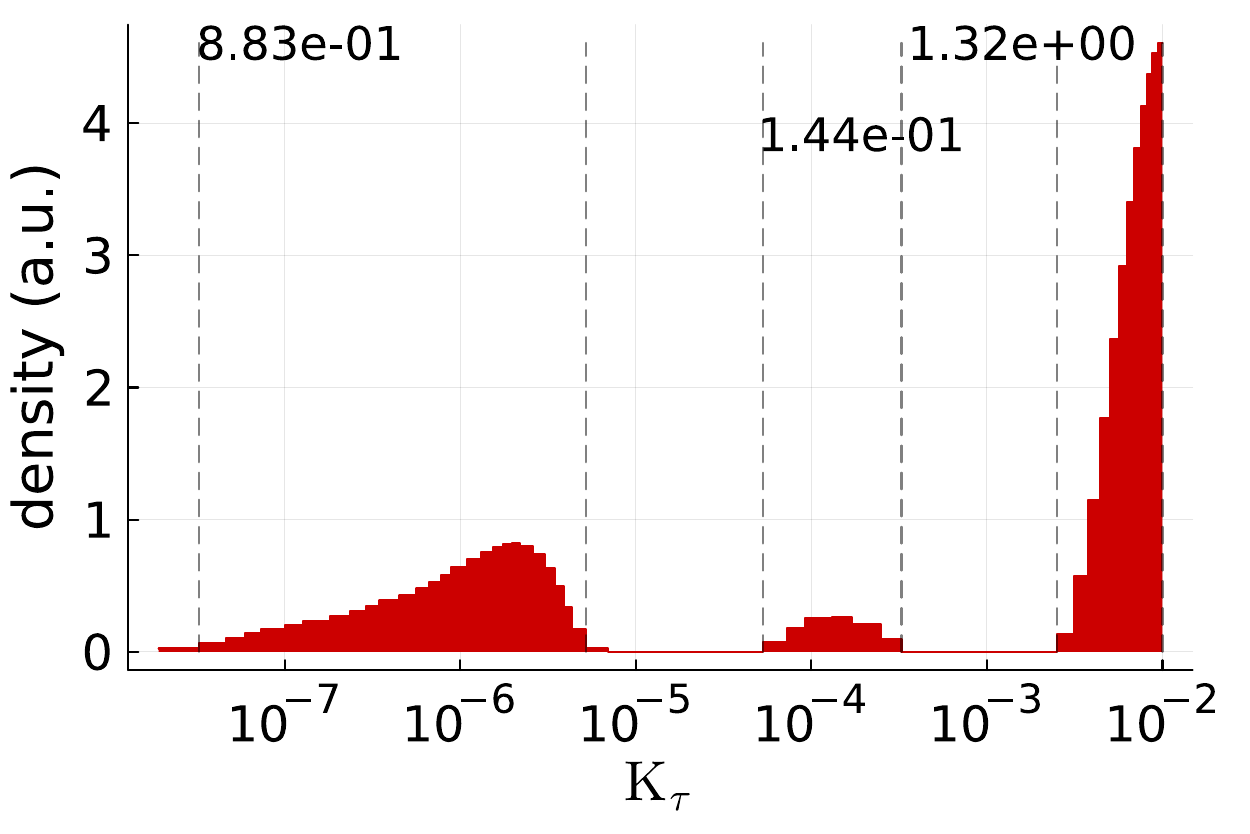}
		\includegraphics[width = 0.32\textwidth]{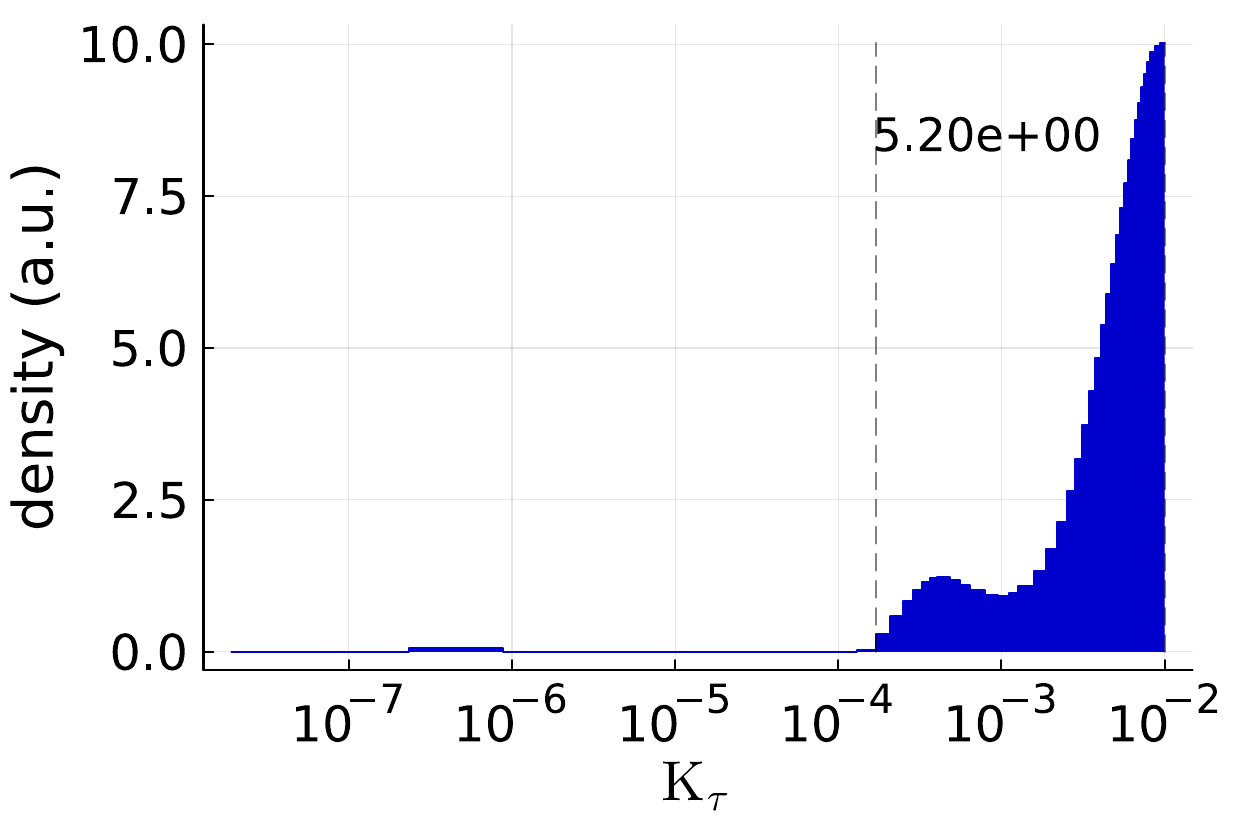}
		\includegraphics[width = 0.32\textwidth]{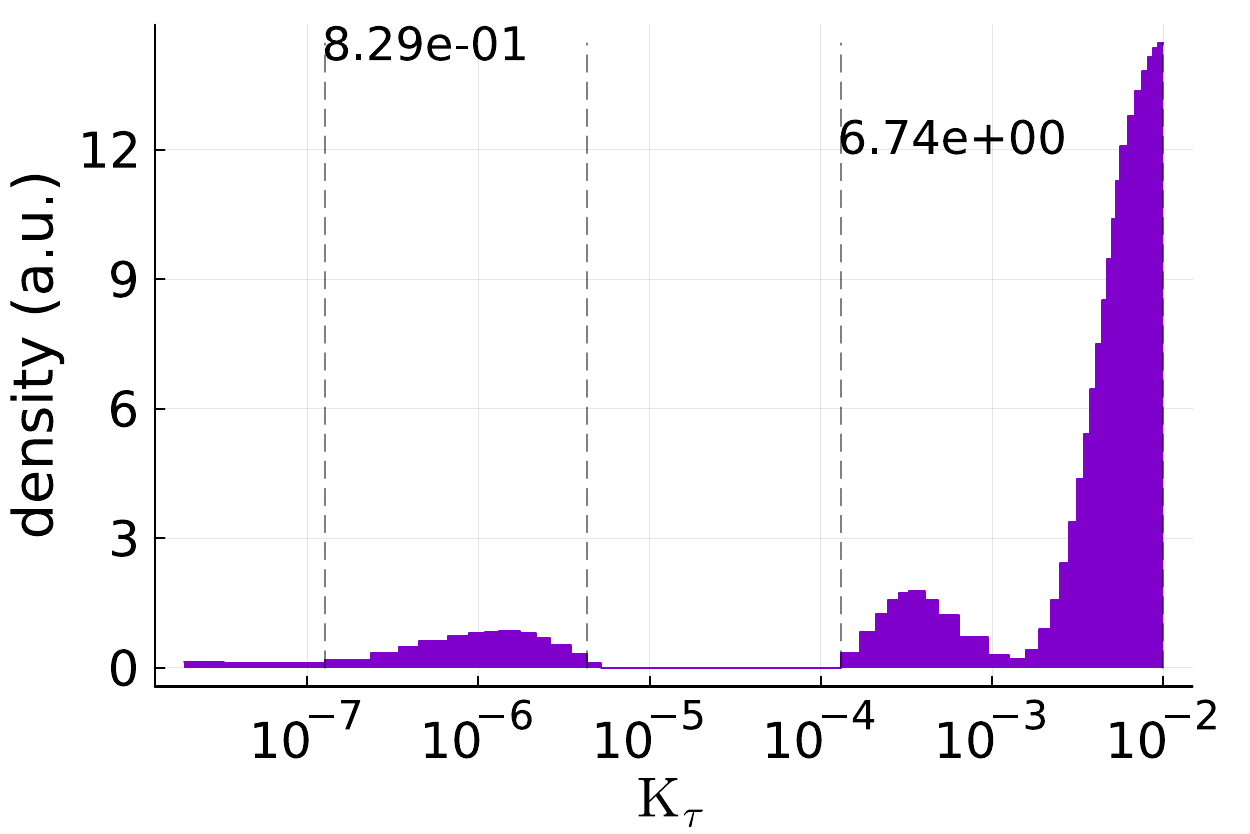}

	\captionof{figure}{Dose-response curves and corresponding accessibility histograms of replicate 1 (for the antibody mix condition).}
	\label{sup-fig: rep 1}
\end{minipage}

\vspace{0.5cm}
\noindent
\begin{minipage}{\textwidth}
	\centering
	{\large \bfseries Replicate 2}\vspace{0.5em}

		\includegraphics[width = 0.32\textwidth]{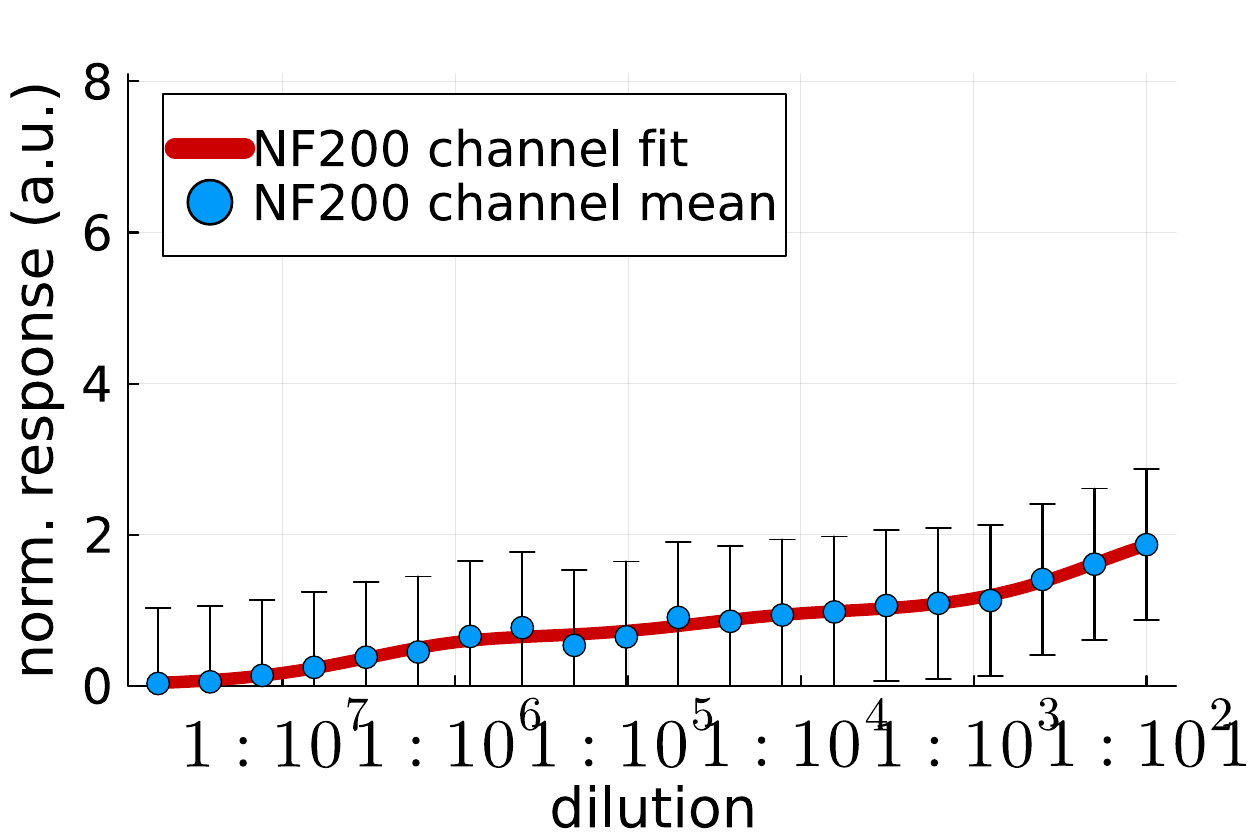}
		\includegraphics[width = 0.32\textwidth]{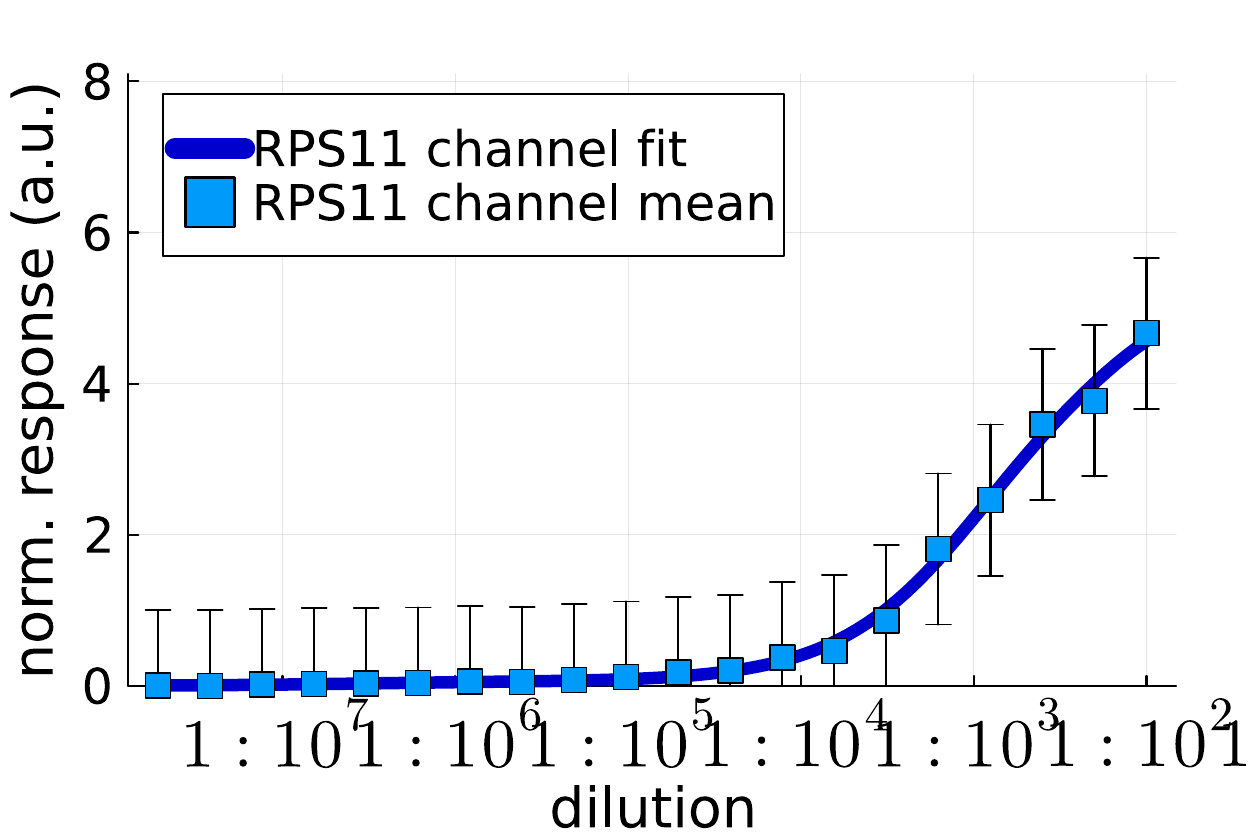}
		\includegraphics[width = 0.32\textwidth]{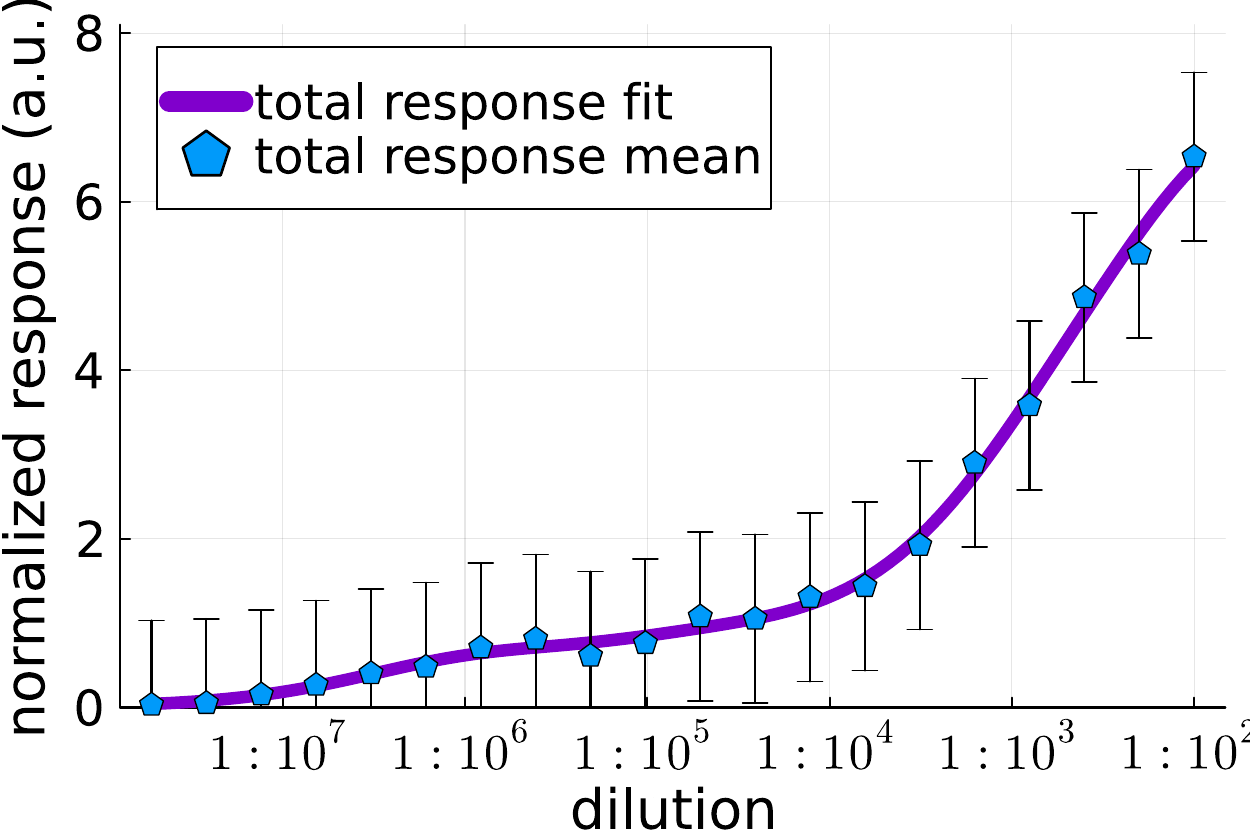}

		\includegraphics[width = 0.32\textwidth]{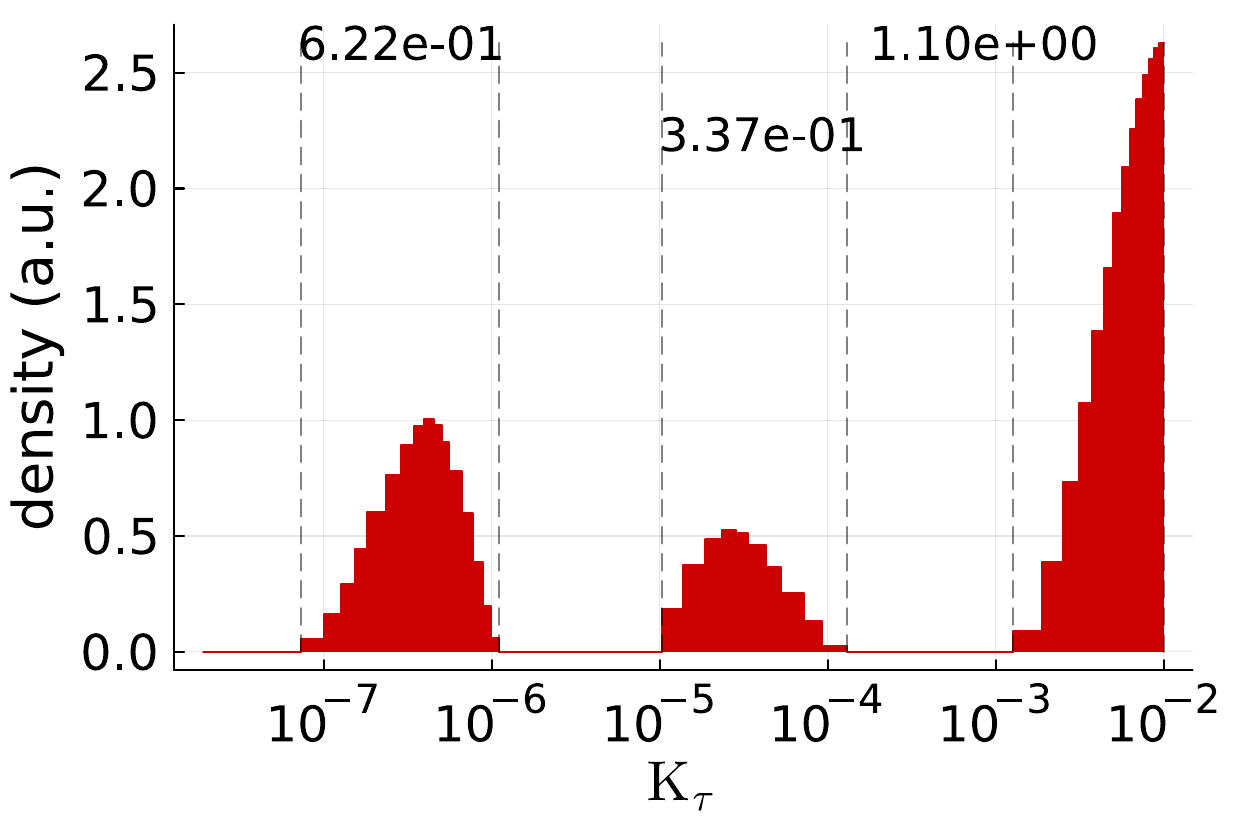}
		\includegraphics[width = 0.32\textwidth]{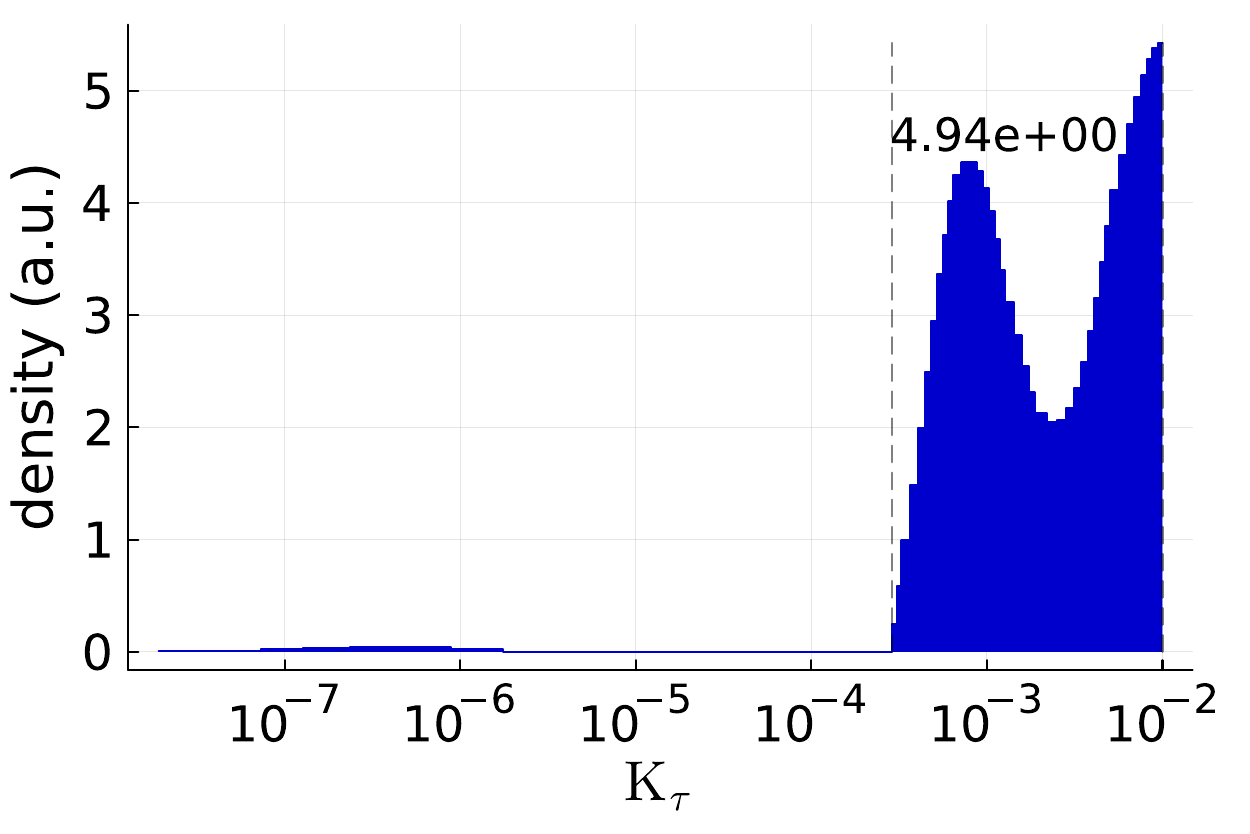}
		\includegraphics[width = 0.32\textwidth]{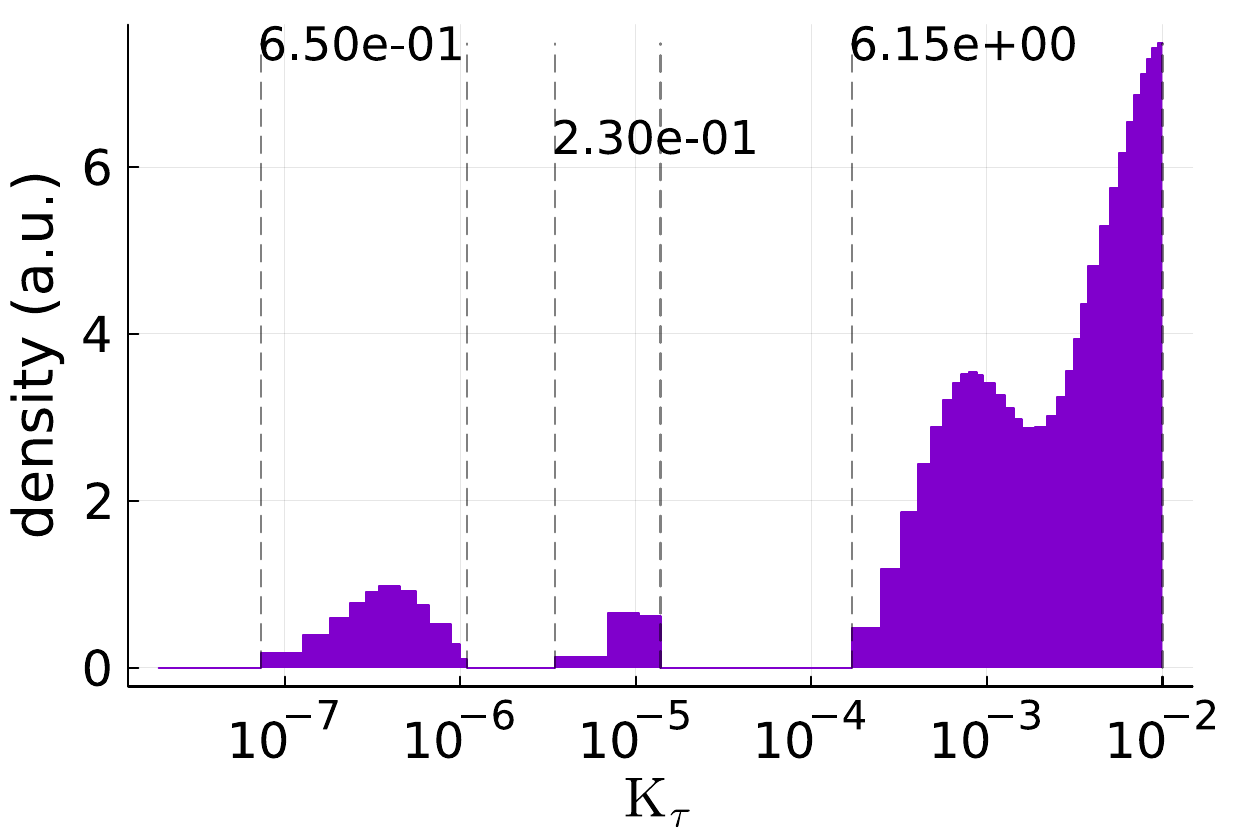}

	\captionof{figure}{Dose-response curves and corresponding accessibility histograms of replicate 2 (for the antibody mix condition).}
	\label{sup-fig: rep 2}
\end{minipage}

\vspace{0.5cm}
\noindent
\begin{minipage}{\textwidth}
	\centering
	{\large \bfseries Replicate 3}\vspace{0.5em}

		\includegraphics[width = 0.32\textwidth]{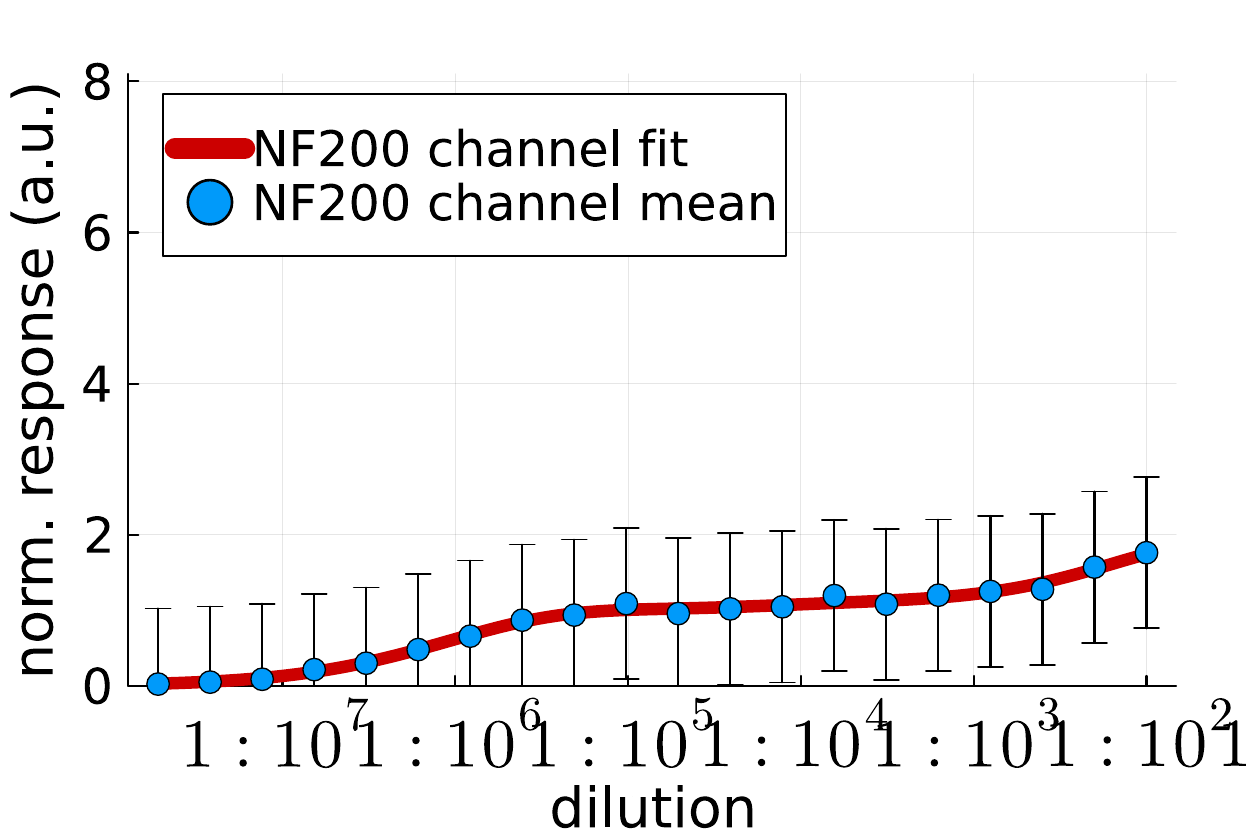}
		\includegraphics[width = 0.32\textwidth]{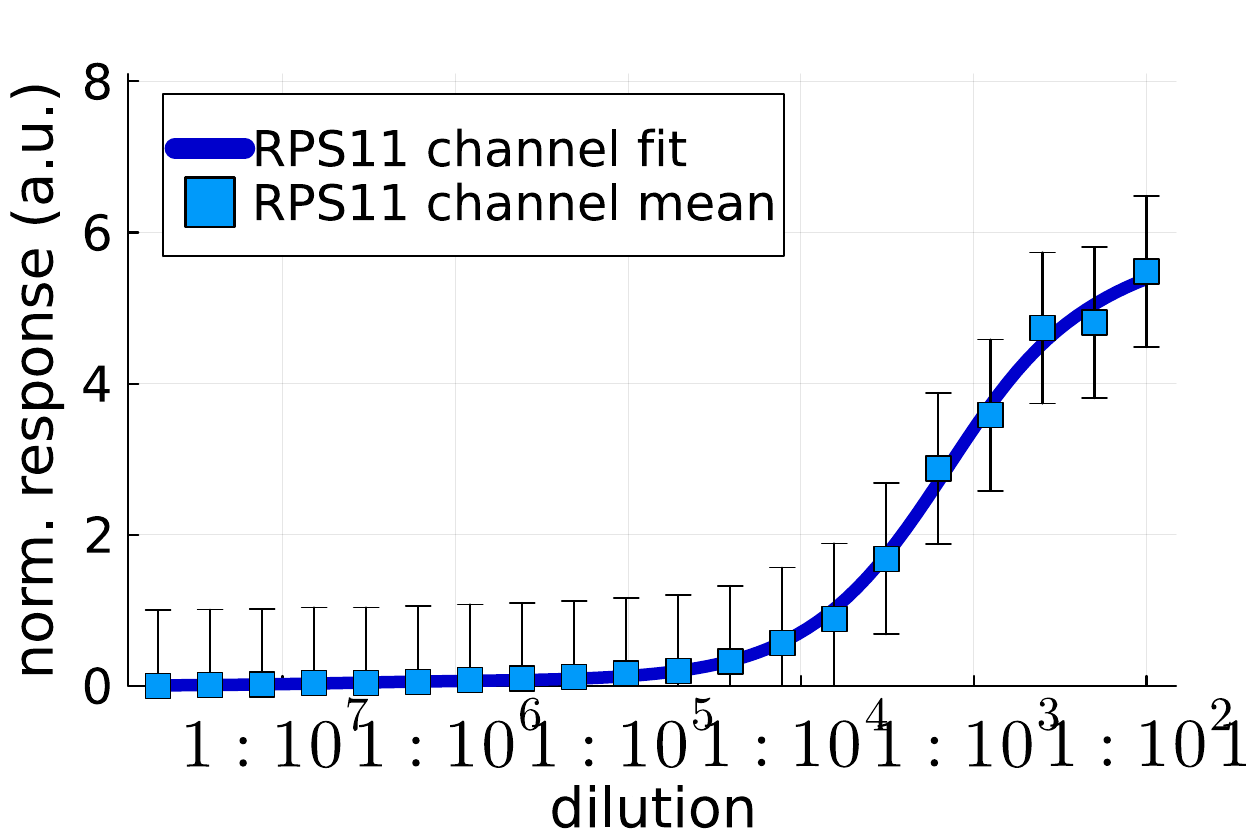}
		\includegraphics[width = 0.32\textwidth]{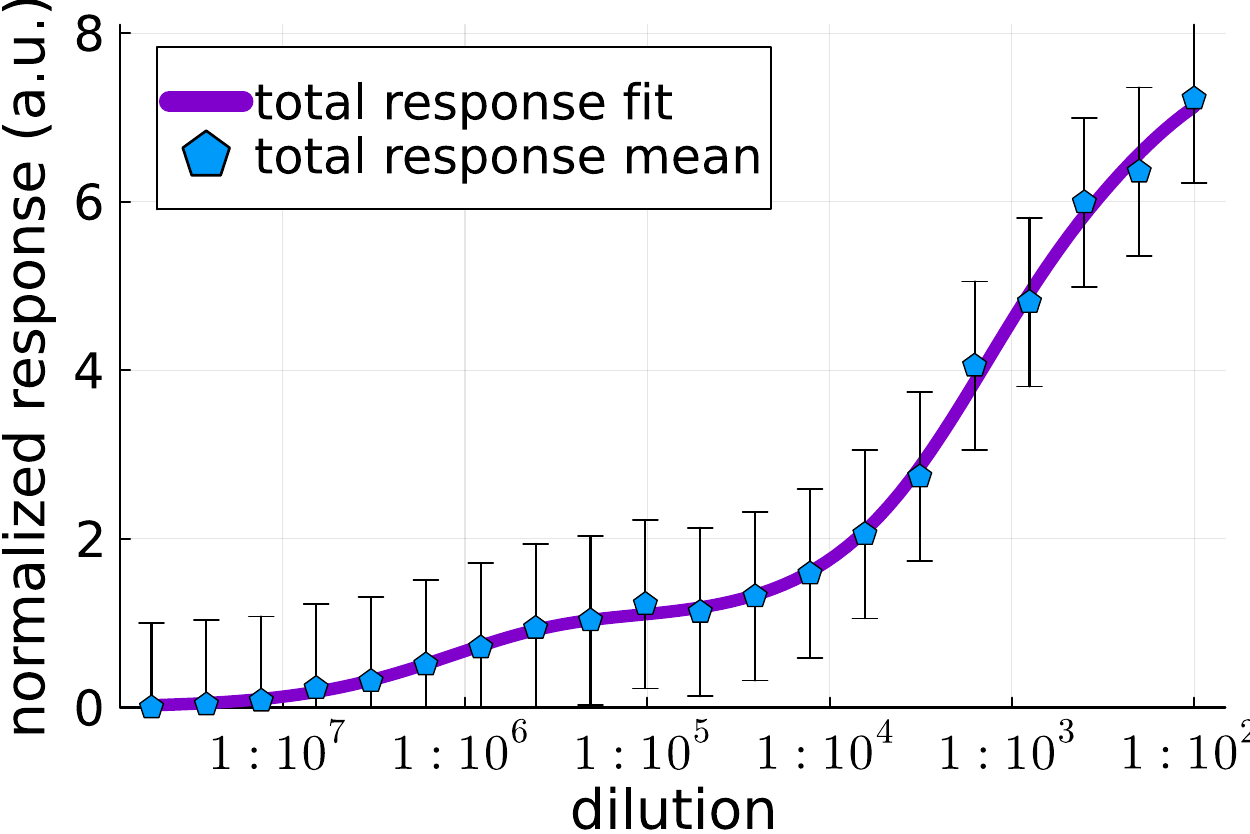}

		\includegraphics[width = 0.32\textwidth]{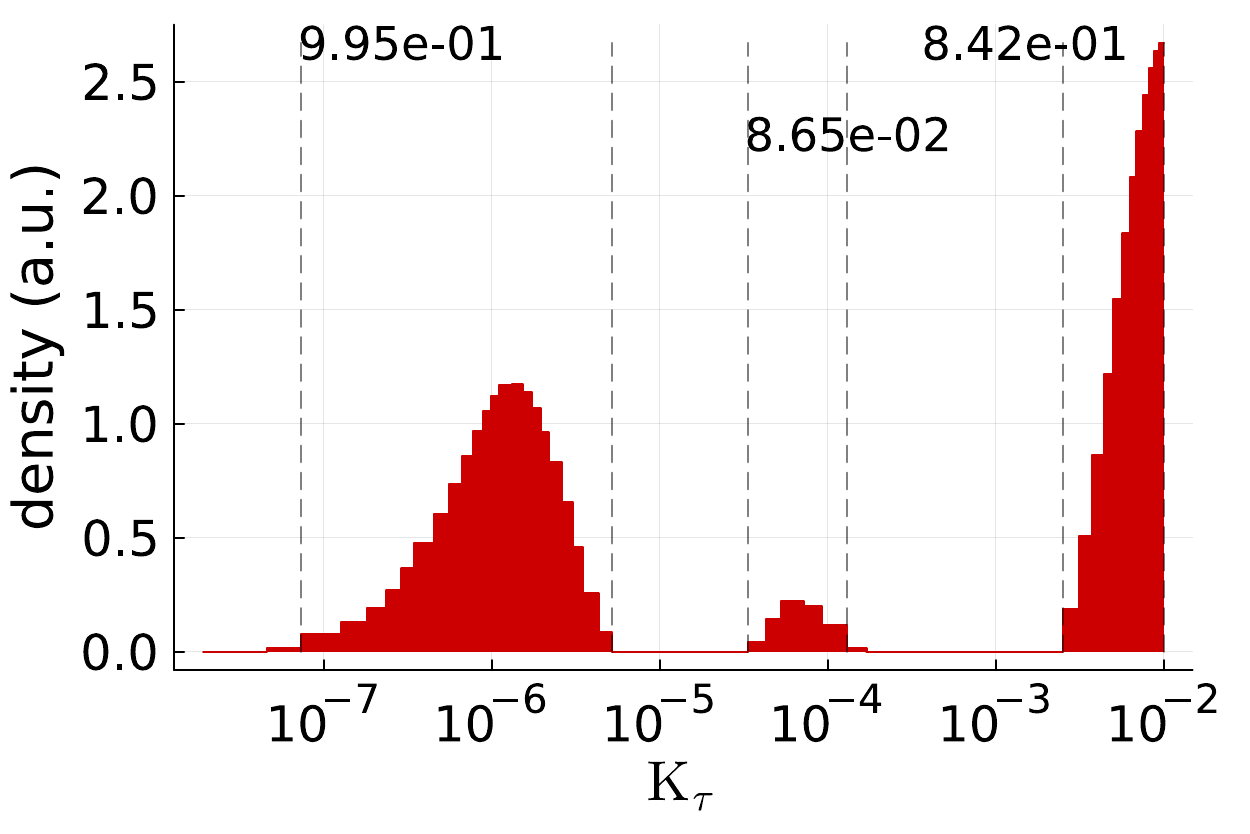}
		\includegraphics[width = 0.32\textwidth]{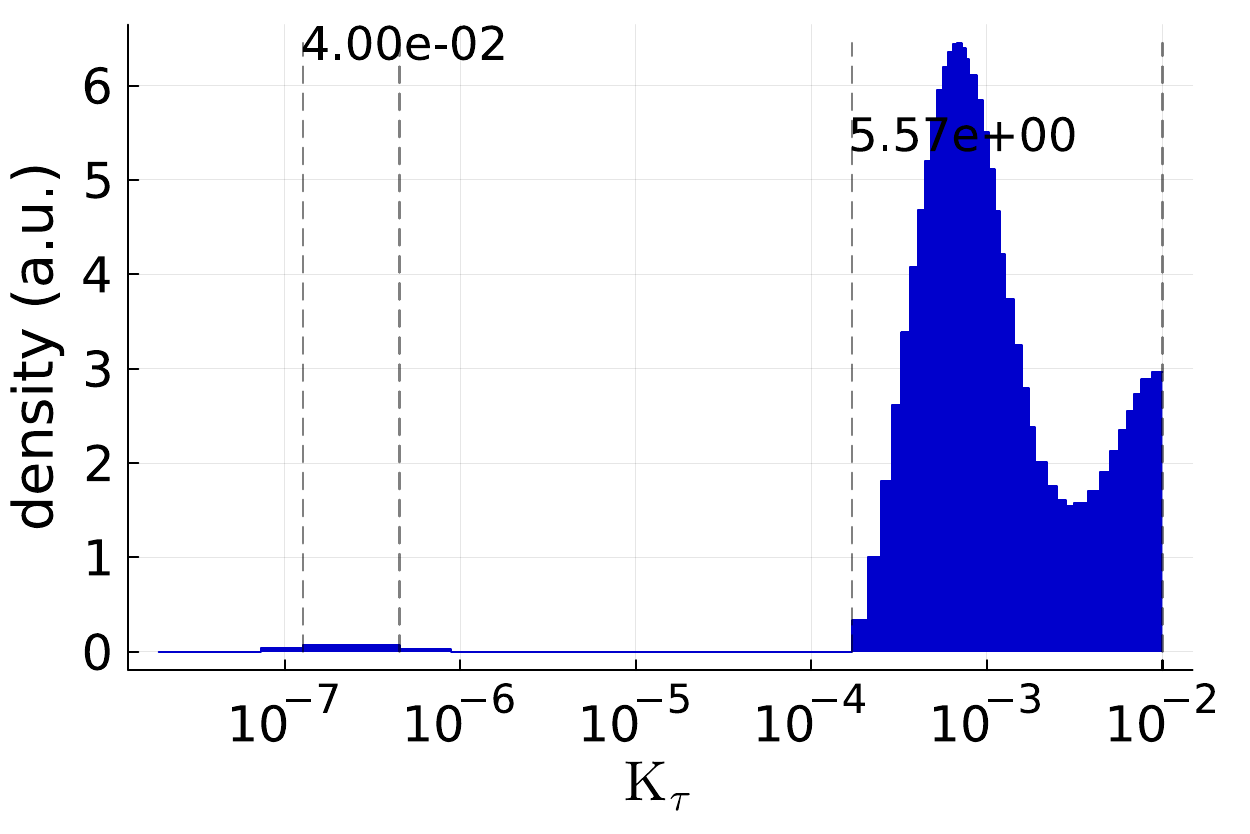}
		\includegraphics[width = 0.32\textwidth]{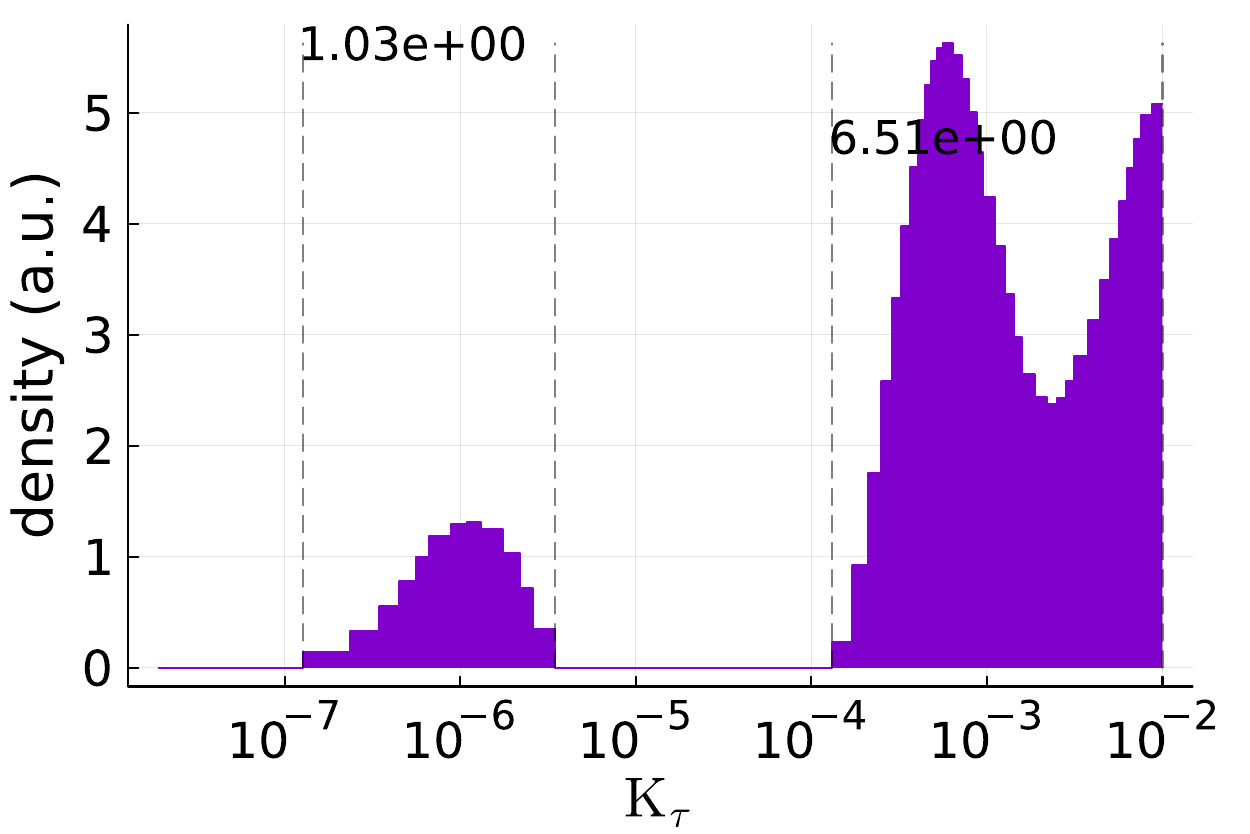}

	\captionof{figure}{Dose-response curves and corresponding accessibility histograms of replicate 3 (for the antibody mix condition).}

	\label{sup-fig: rep 3}
\end{minipage}

\vspace{0.5cm}
\noindent
\begin{minipage}{\textwidth}
	\centering
	{\large \bfseries Replicate 4}\vspace{0.5em}

		\includegraphics[width = 0.32\textwidth]{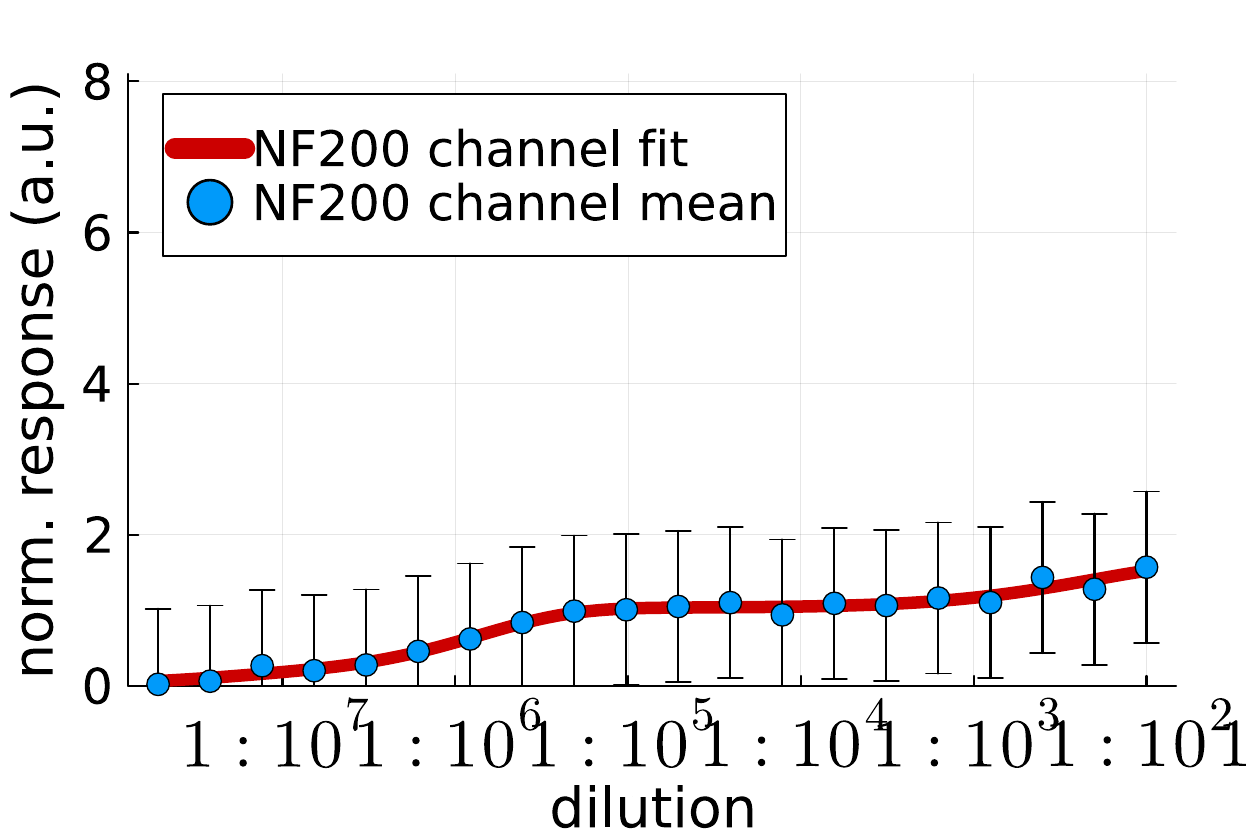}
		\includegraphics[width = 0.32\textwidth]{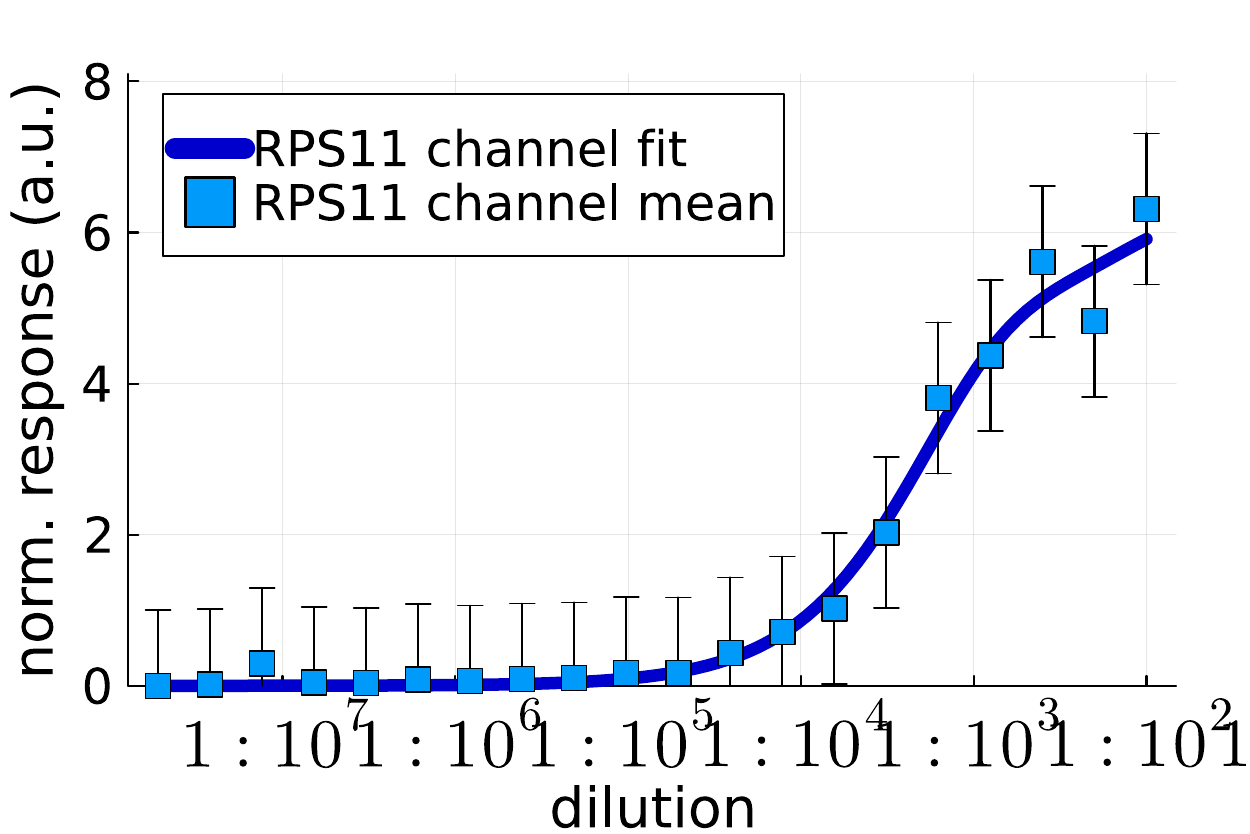}
		\includegraphics[width = 0.32\textwidth]{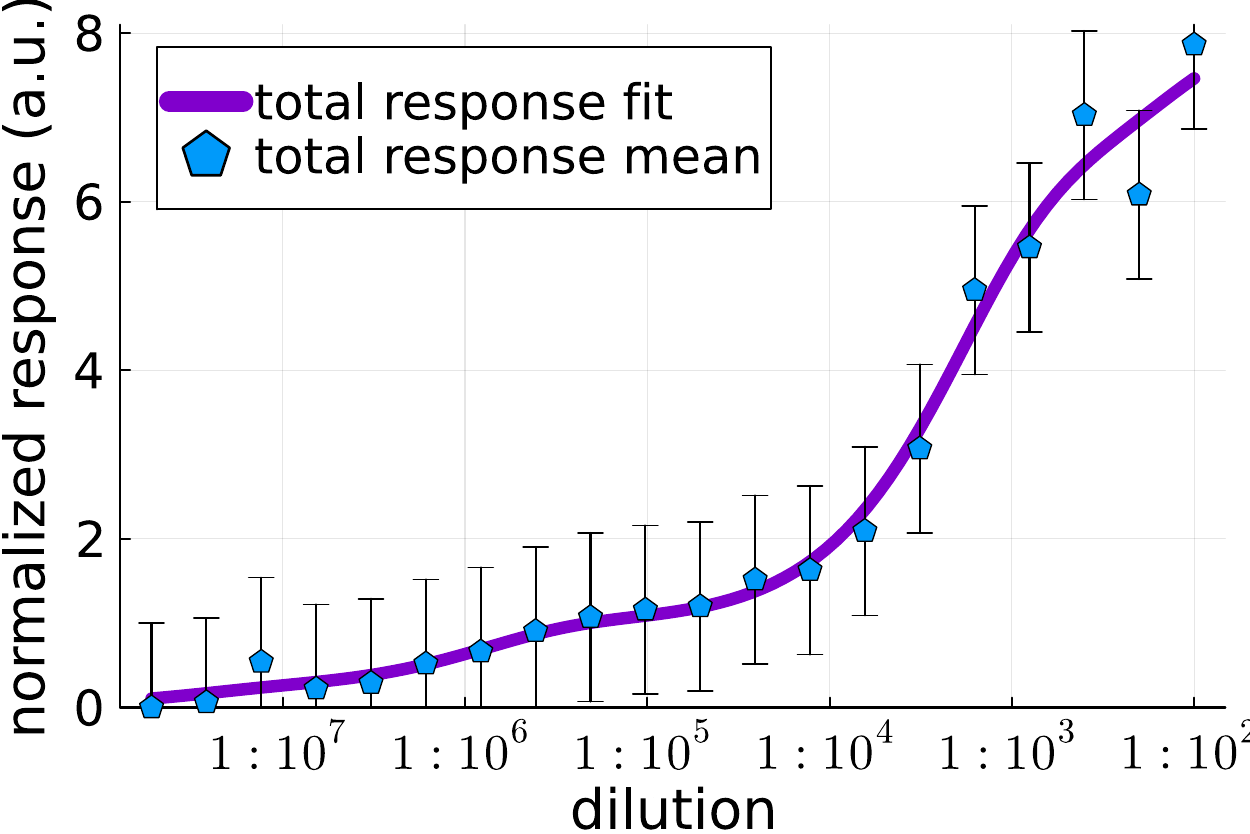}

		\includegraphics[width = 0.32\textwidth]{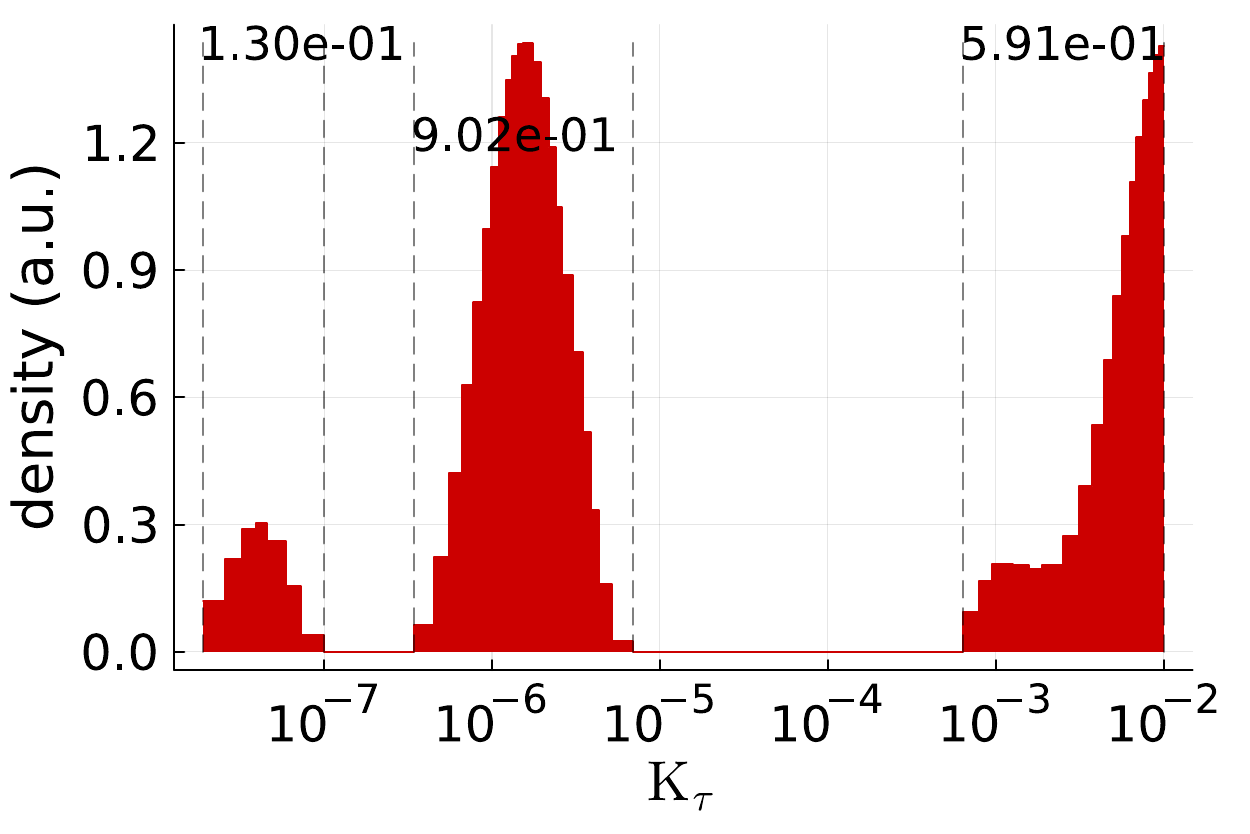}
		\includegraphics[width = 0.32\textwidth]{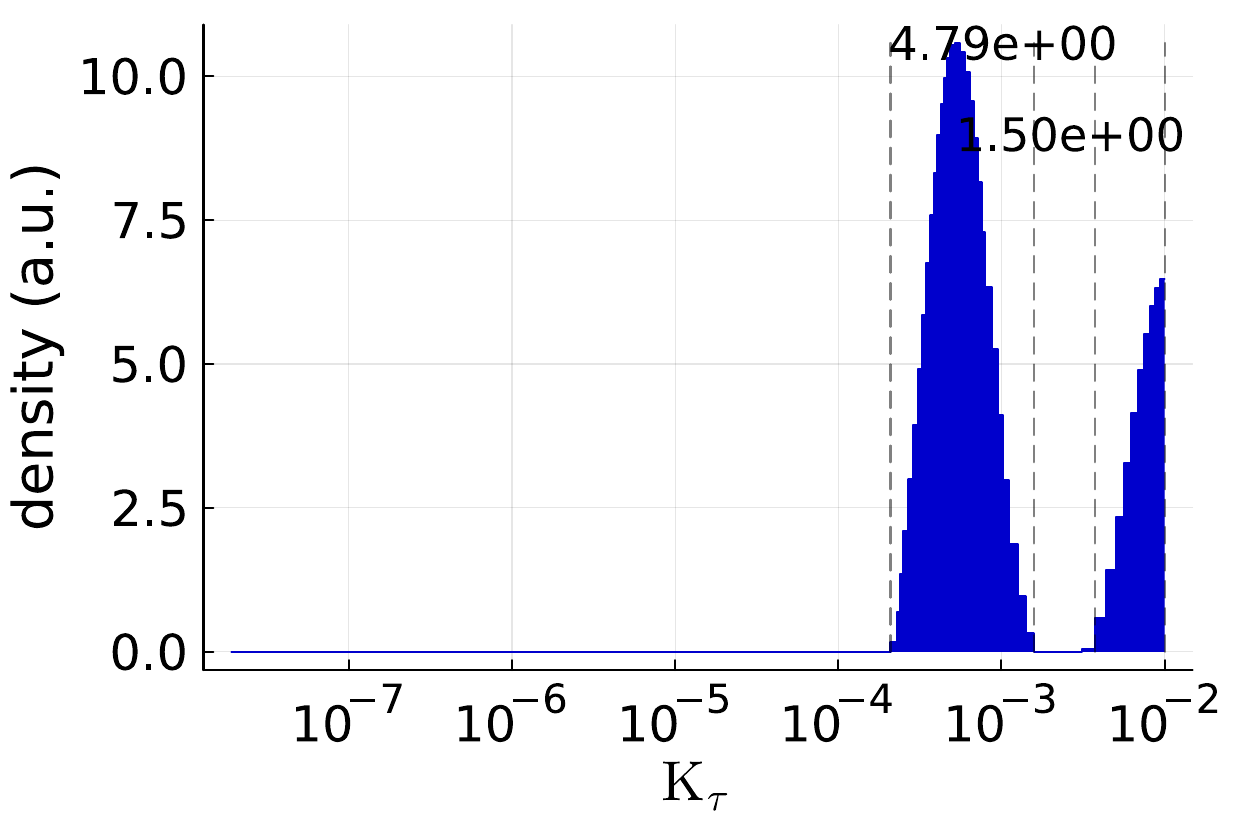}
		\includegraphics[width = 0.32\textwidth]{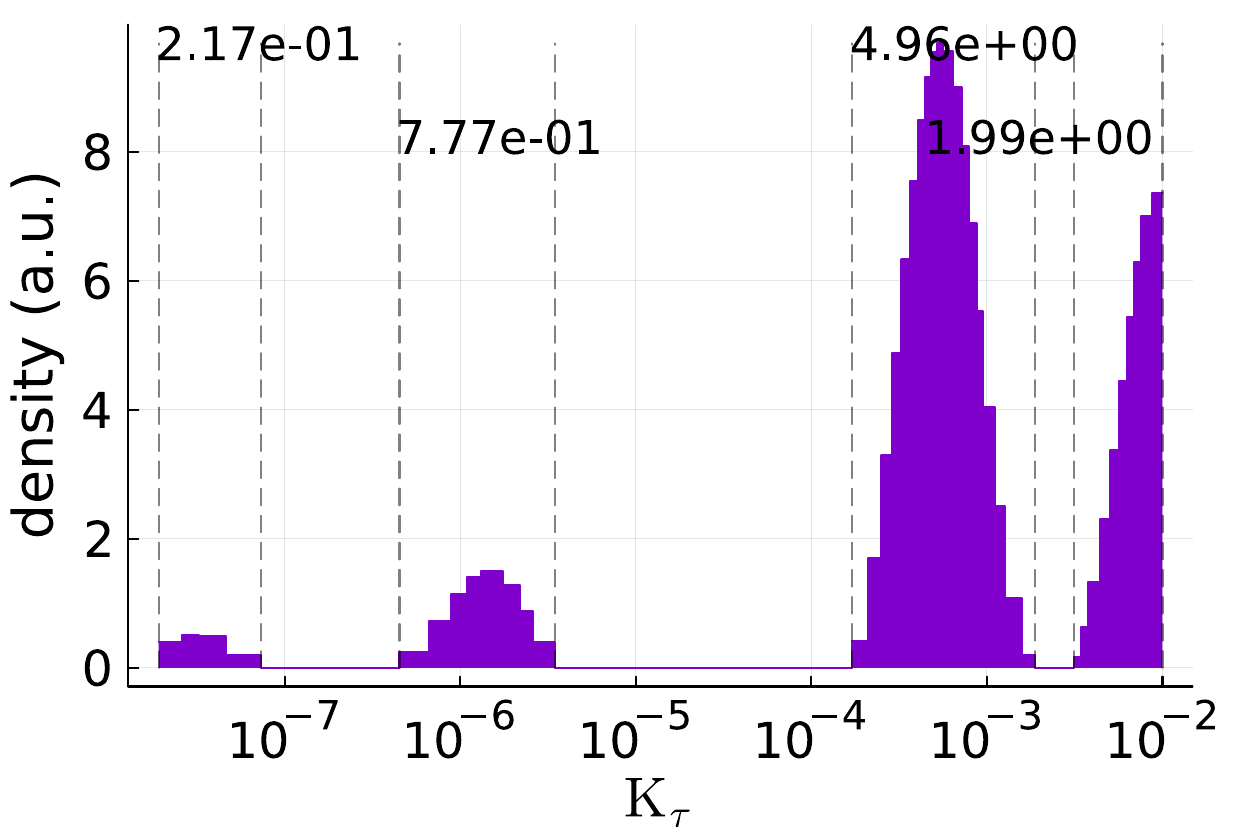}

	\captionof{figure}{Dose-response curves and corresponding accessibility histograms of replicate 4 (for the antibody mix condition).}
	\label{sup-fig: rep 4}
\end{minipage}

\vspace{0.5cm}
\noindent
\begin{minipage}{\textwidth}
	\centering
	{\large \bfseries Replicate 5}\vspace{0.5em}

		\includegraphics[width = 0.32\textwidth]{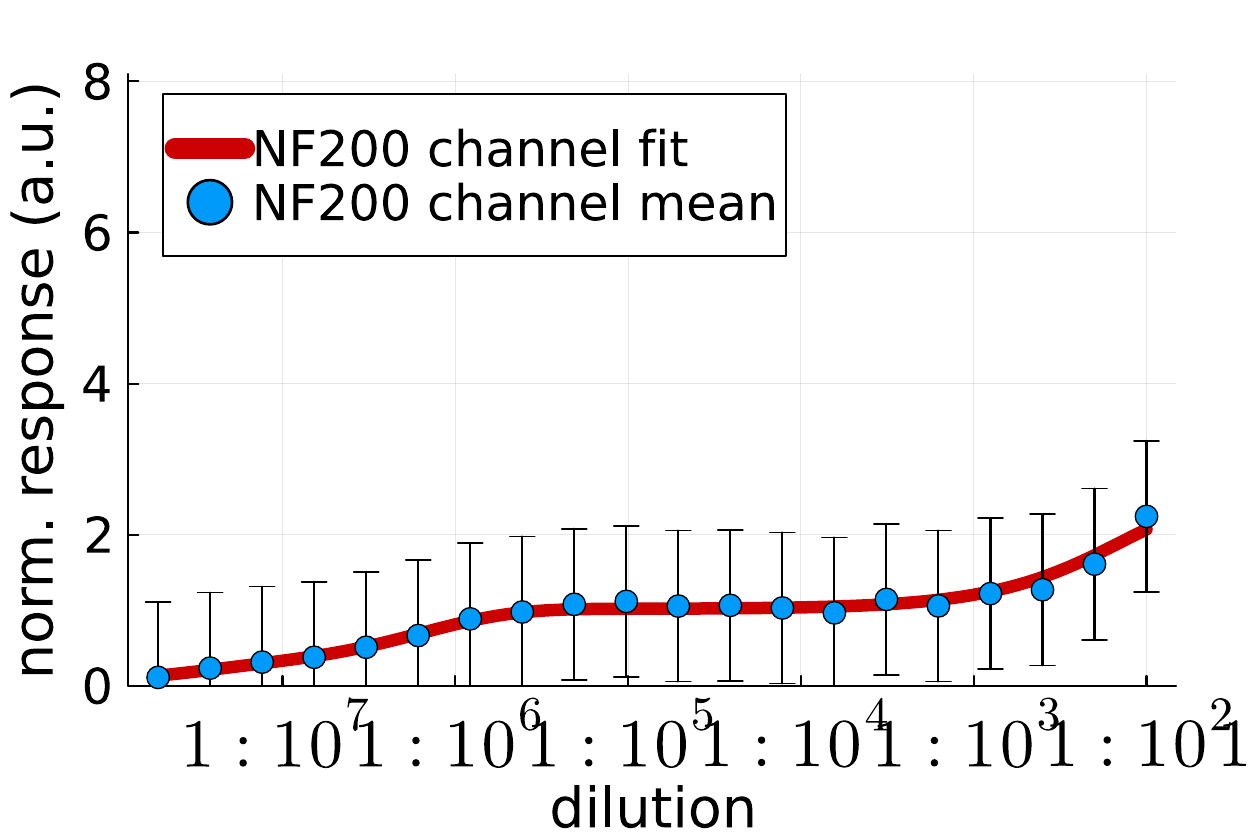}
		\includegraphics[width = 0.32\textwidth]{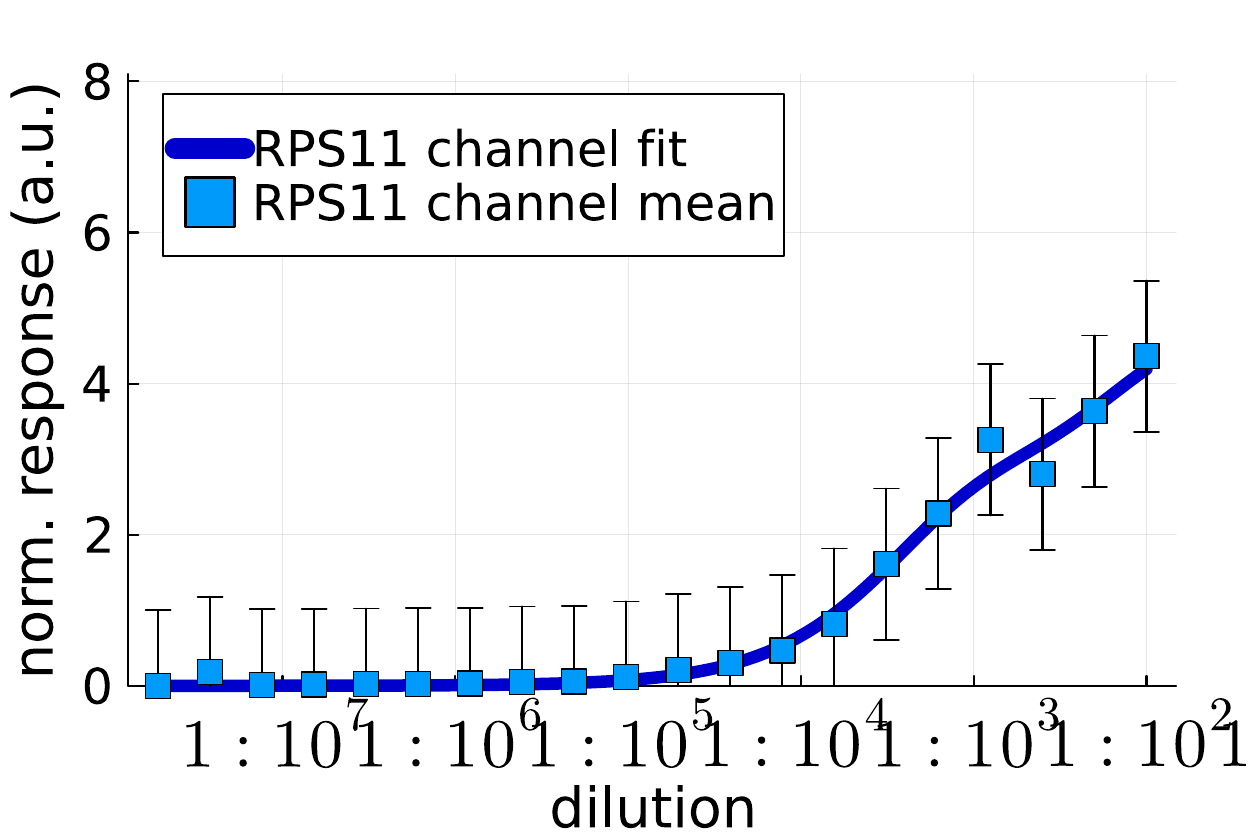}
		\includegraphics[width = 0.32\textwidth]{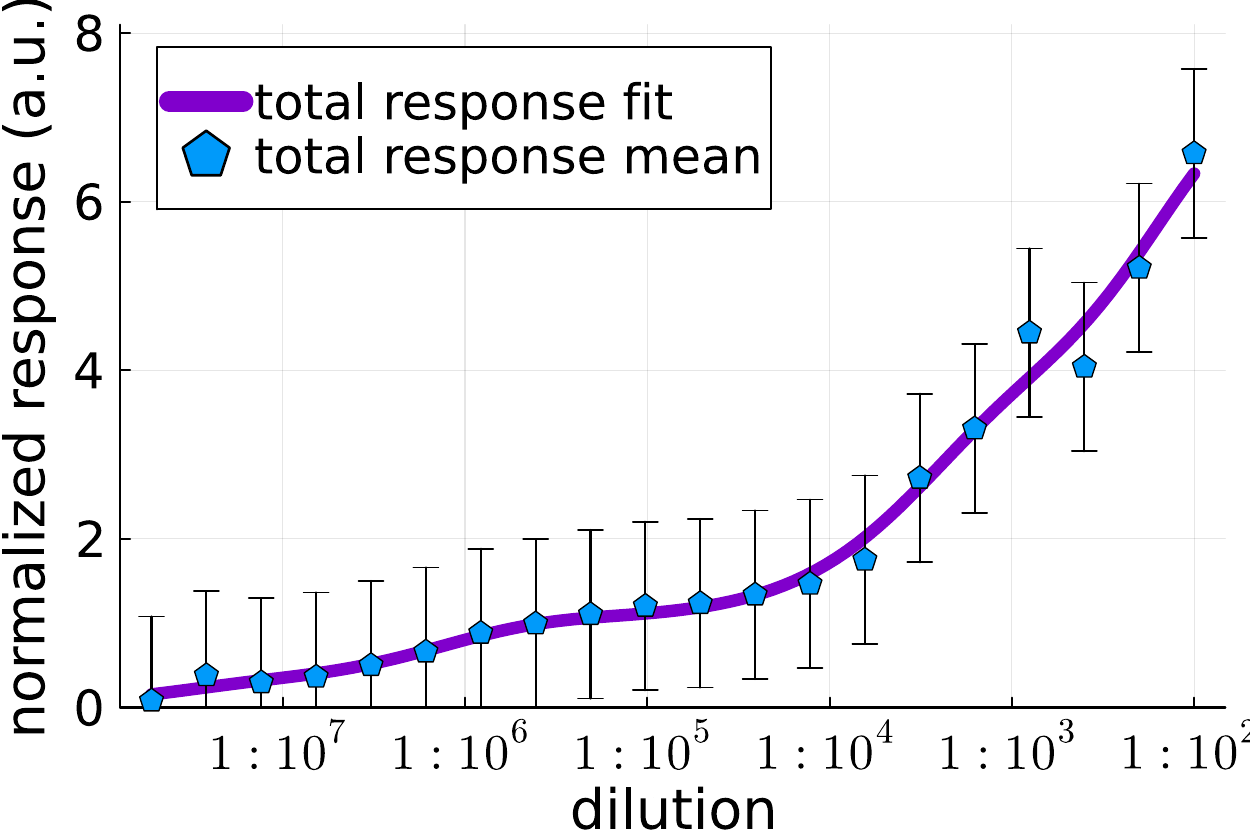}

		\includegraphics[width = 0.32\textwidth]{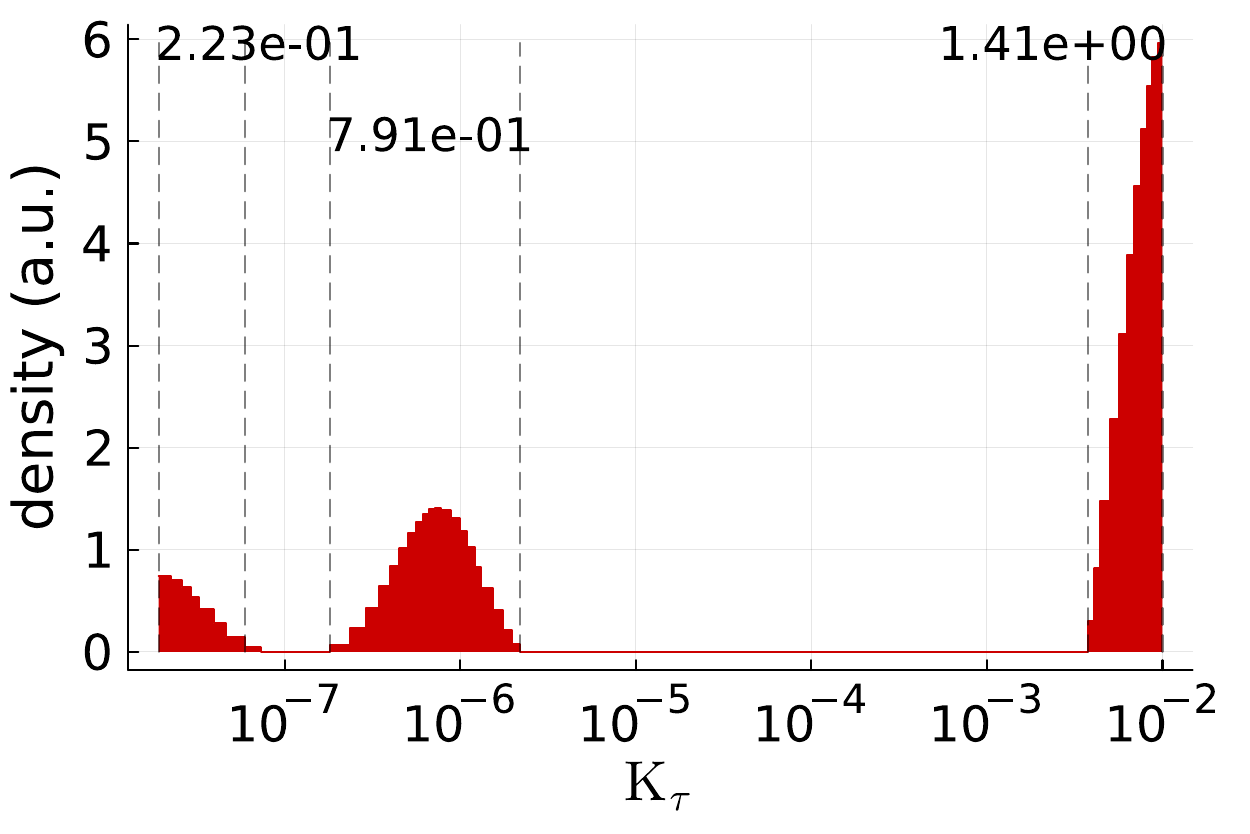}
		\includegraphics[width = 0.32\textwidth]{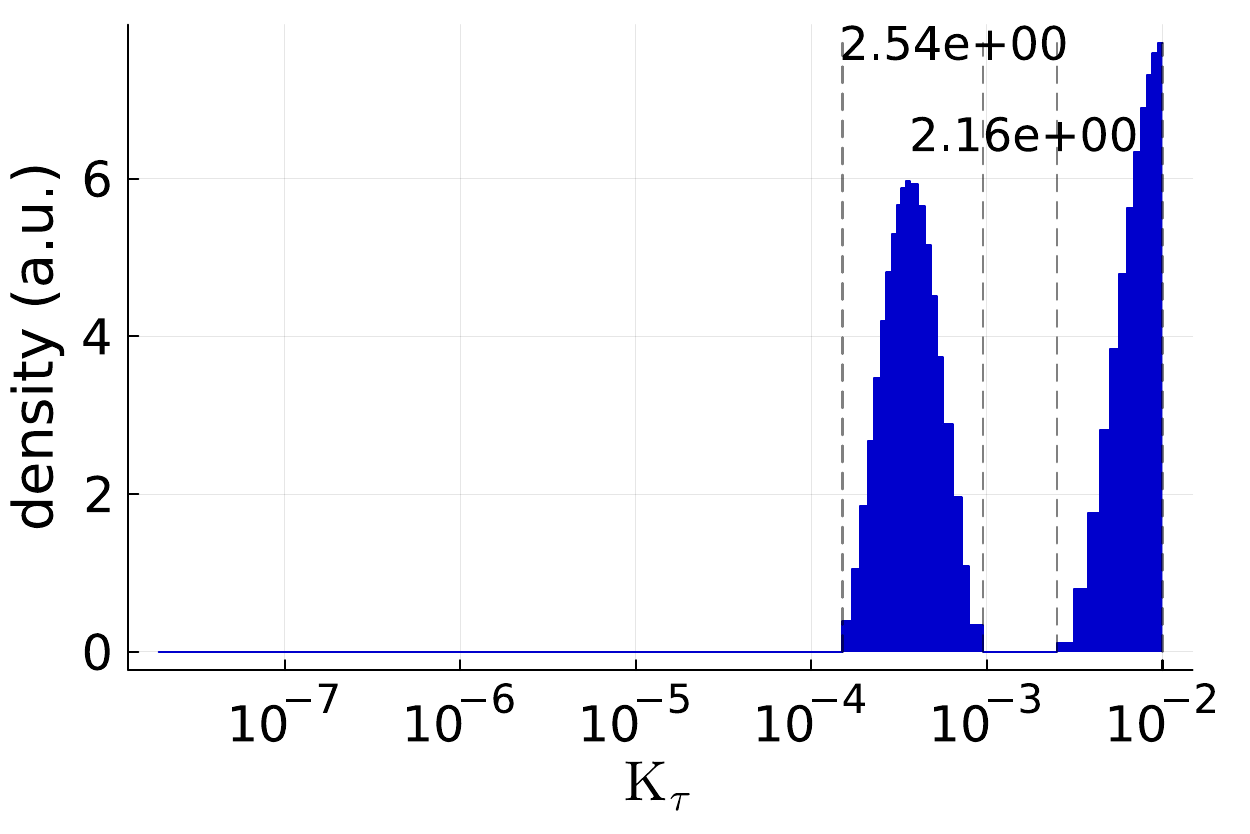}
		\includegraphics[width = 0.32\textwidth]{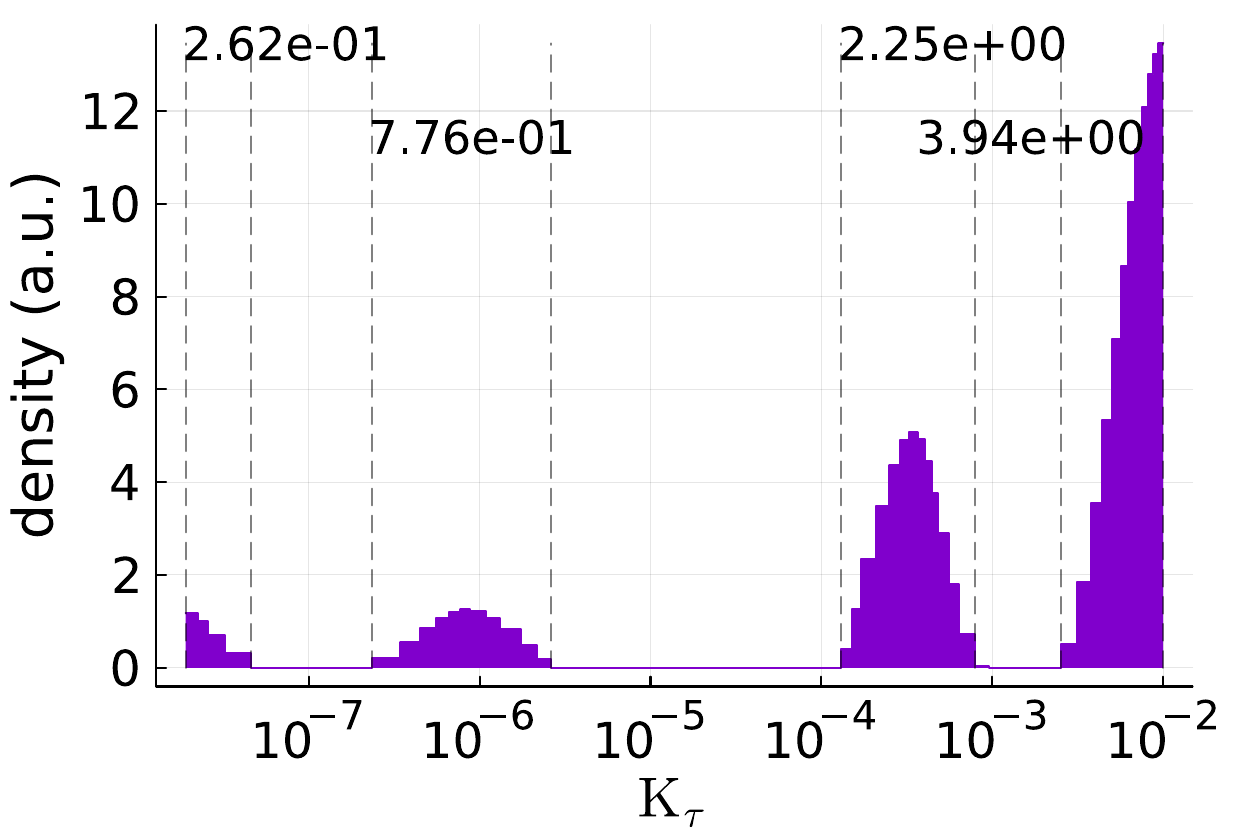}

	\captionof{figure}{Dose-response curves and corresponding accessibility histograms of replicate 5 (for the antibody mix condition).}
	\label{sup-fig: rep 5}
\end{minipage}

\vspace{0.5cm}
\noindent
\begin{minipage}{\textwidth}
	\centering
	{\large \bfseries Replicate 6}\vspace{0.5em}

		\includegraphics[width = 0.32\textwidth]{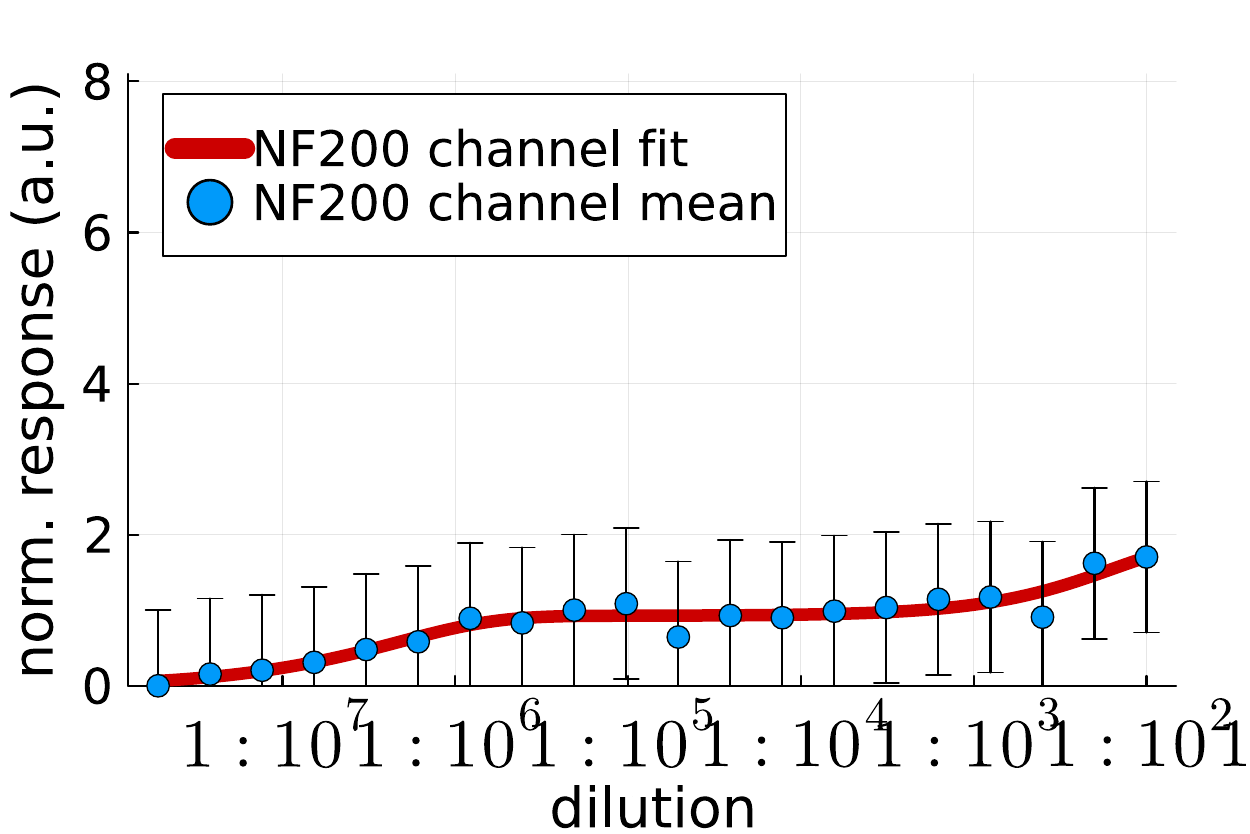}
		\includegraphics[width = 0.32\textwidth]{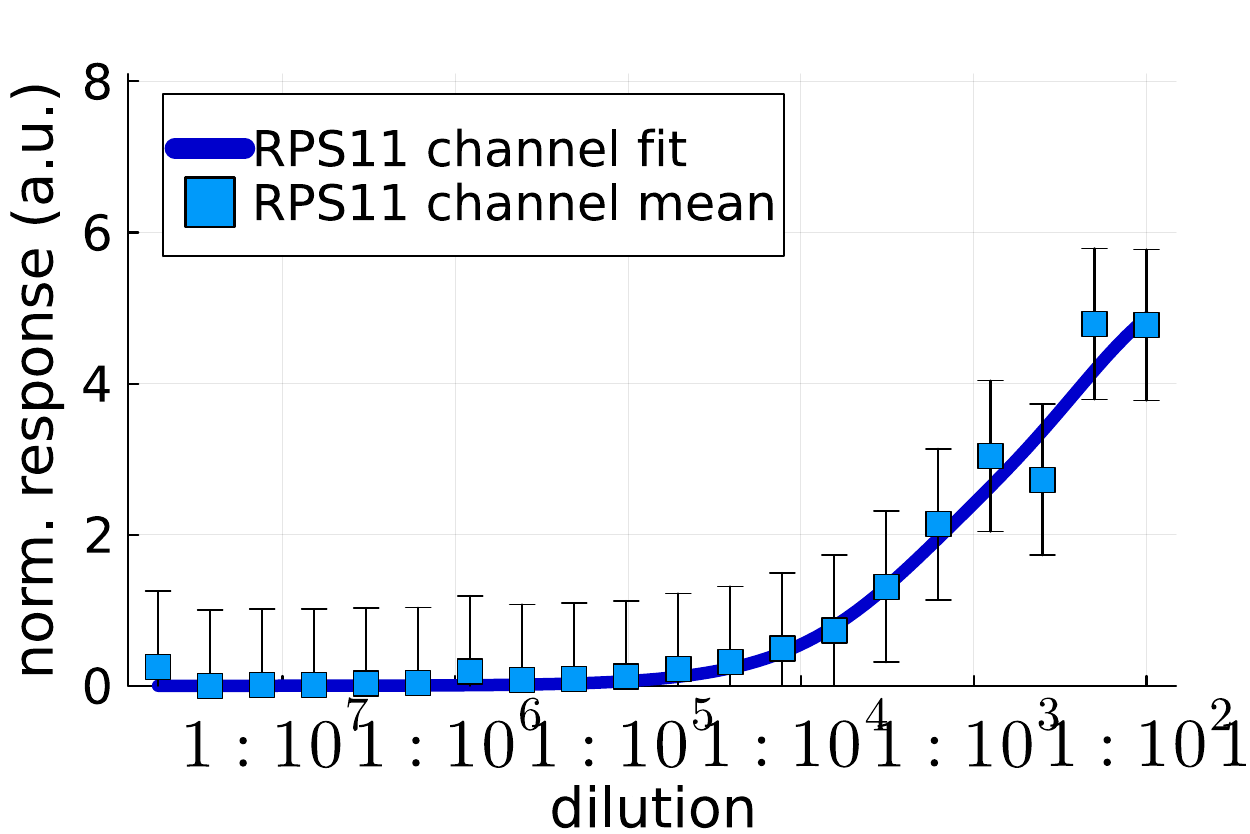}
		\includegraphics[width = 0.32\textwidth]{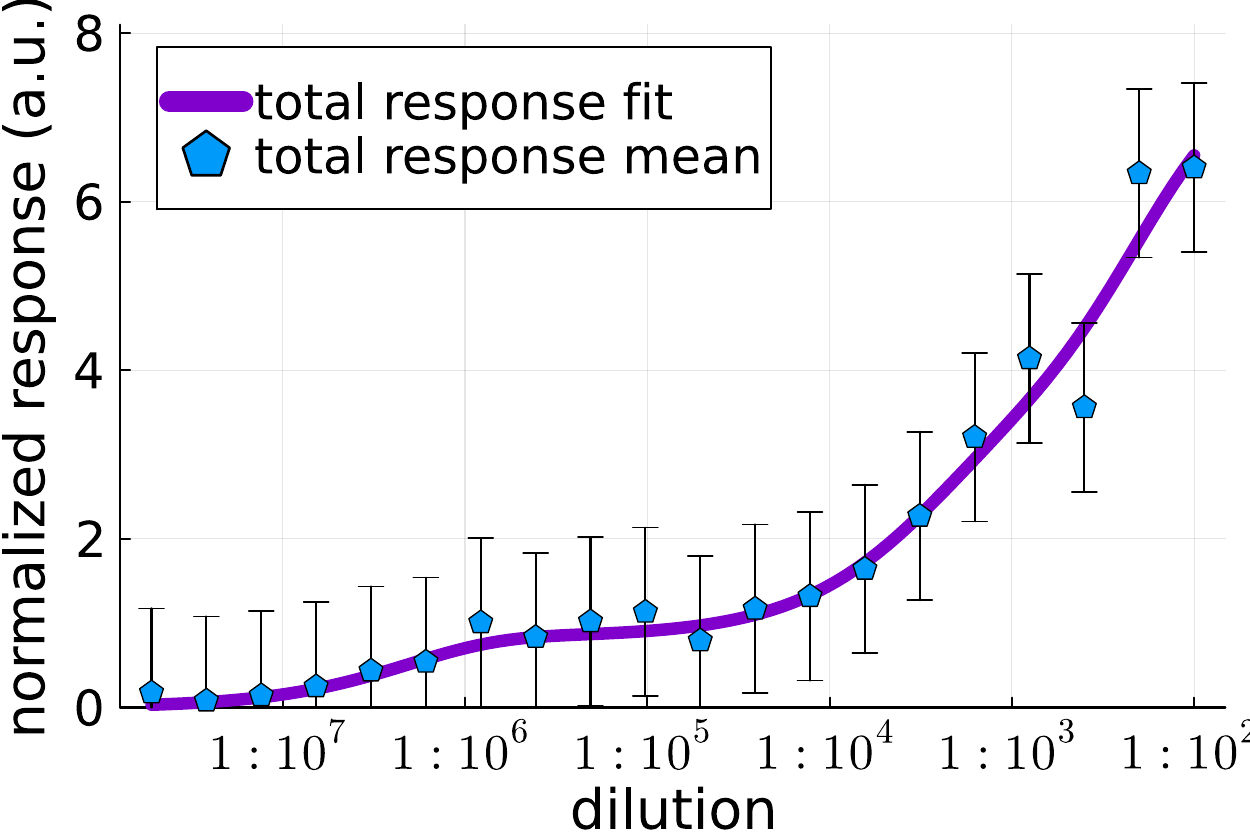}

		\includegraphics[width = 0.32\textwidth]{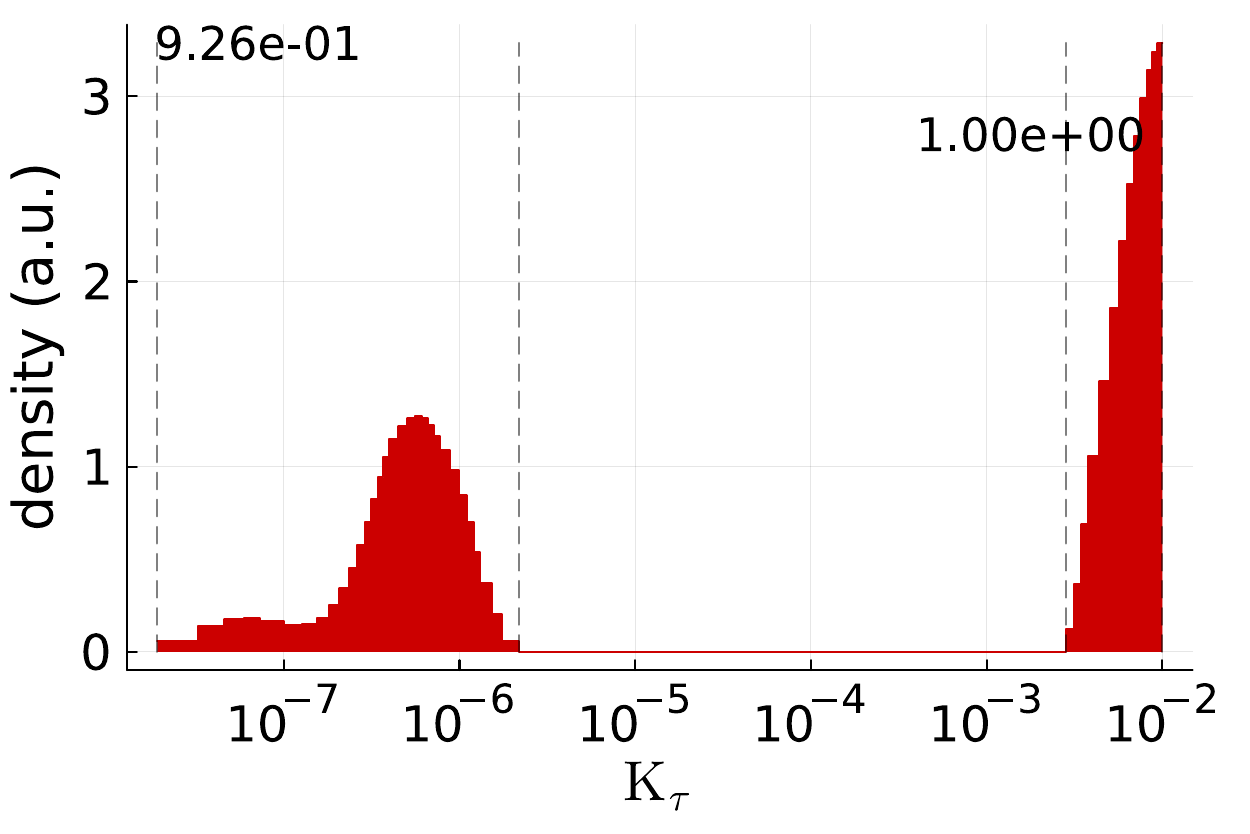}
		\includegraphics[width = 0.32\textwidth]{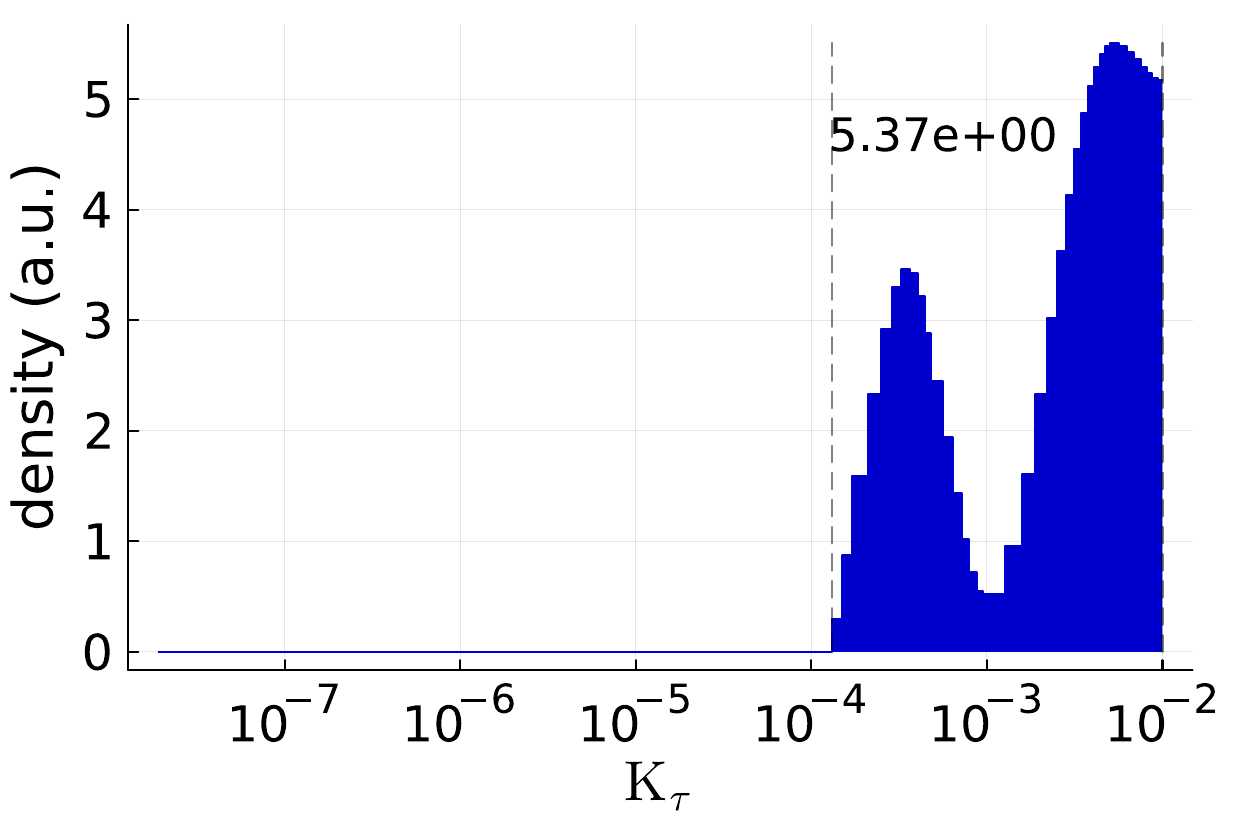}
		\includegraphics[width = 0.32\textwidth]{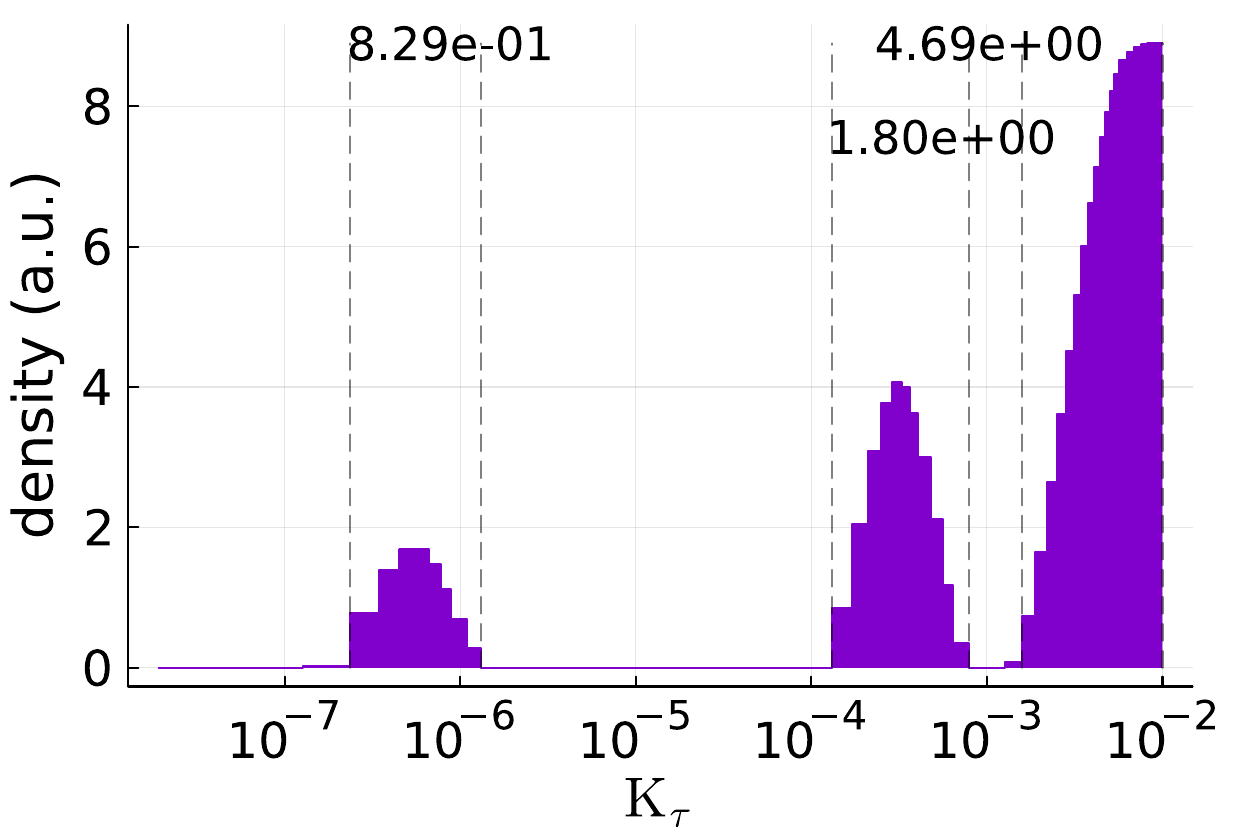}

	\captionof{figure}{Dose-response curves and corresponding accessibility histograms of replicate 6 (for the antibody mix condition).}
	\label{sup-fig: rep 6}
\end{minipage}

\vspace{0.5cm}
\noindent
\begin{minipage}{\textwidth}
	\centering
	{\large \bfseries Replicate 7}\vspace{0.5em}

		\includegraphics[width = 0.32\textwidth]{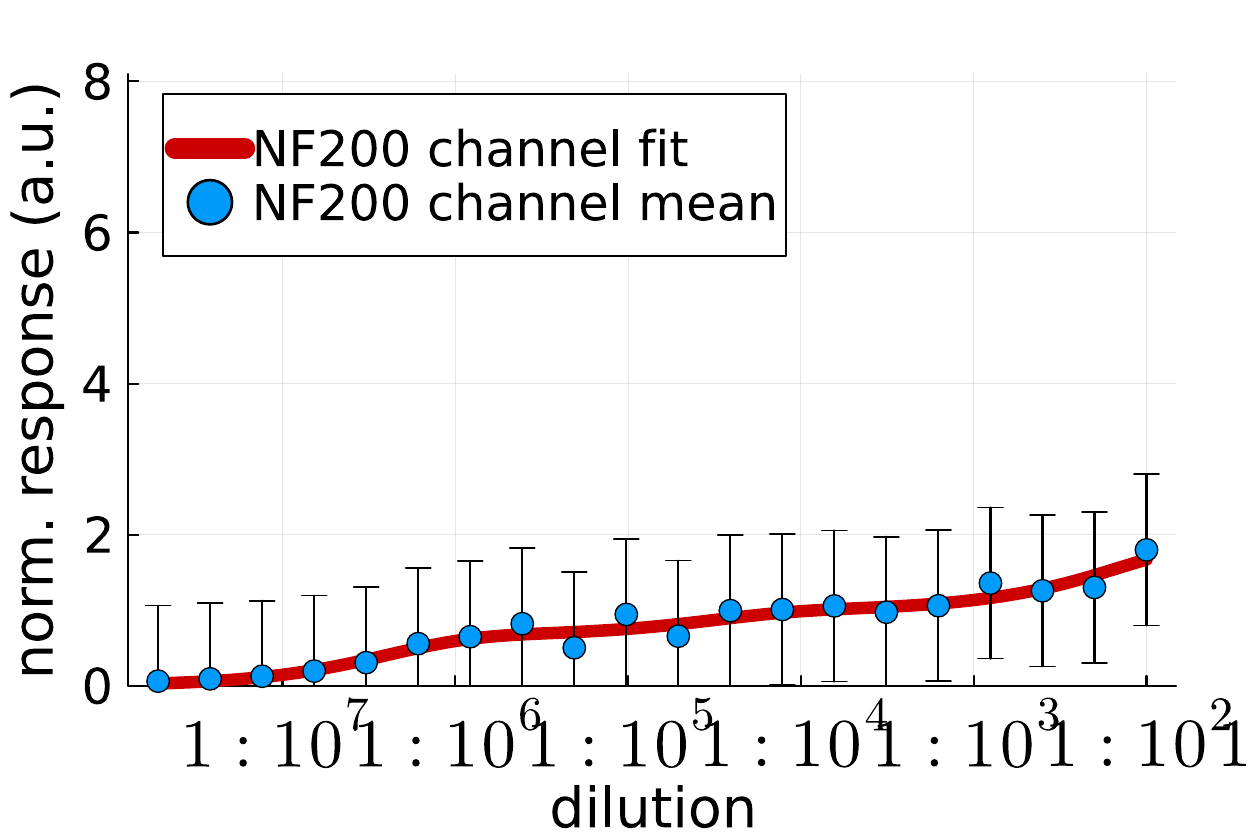}
		\includegraphics[width = 0.32\textwidth]{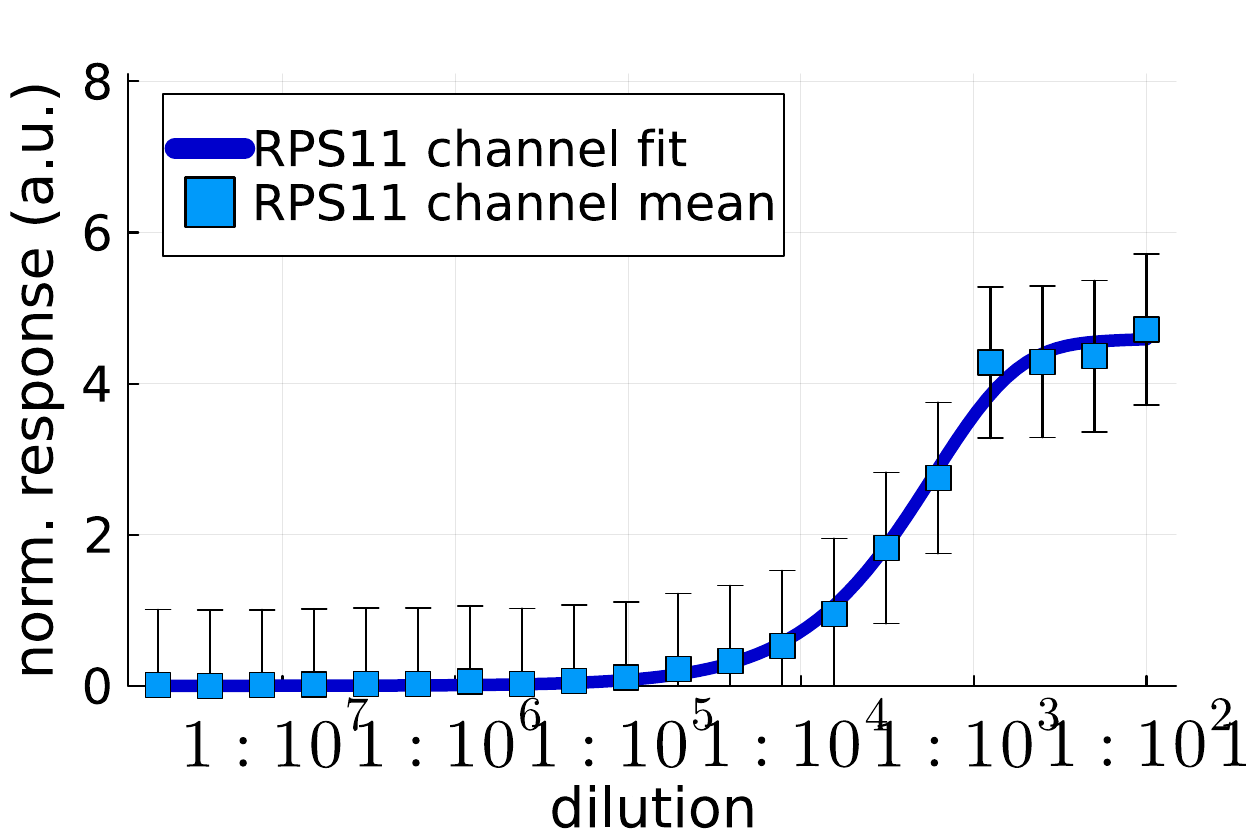}
		\includegraphics[width = 0.32\textwidth]{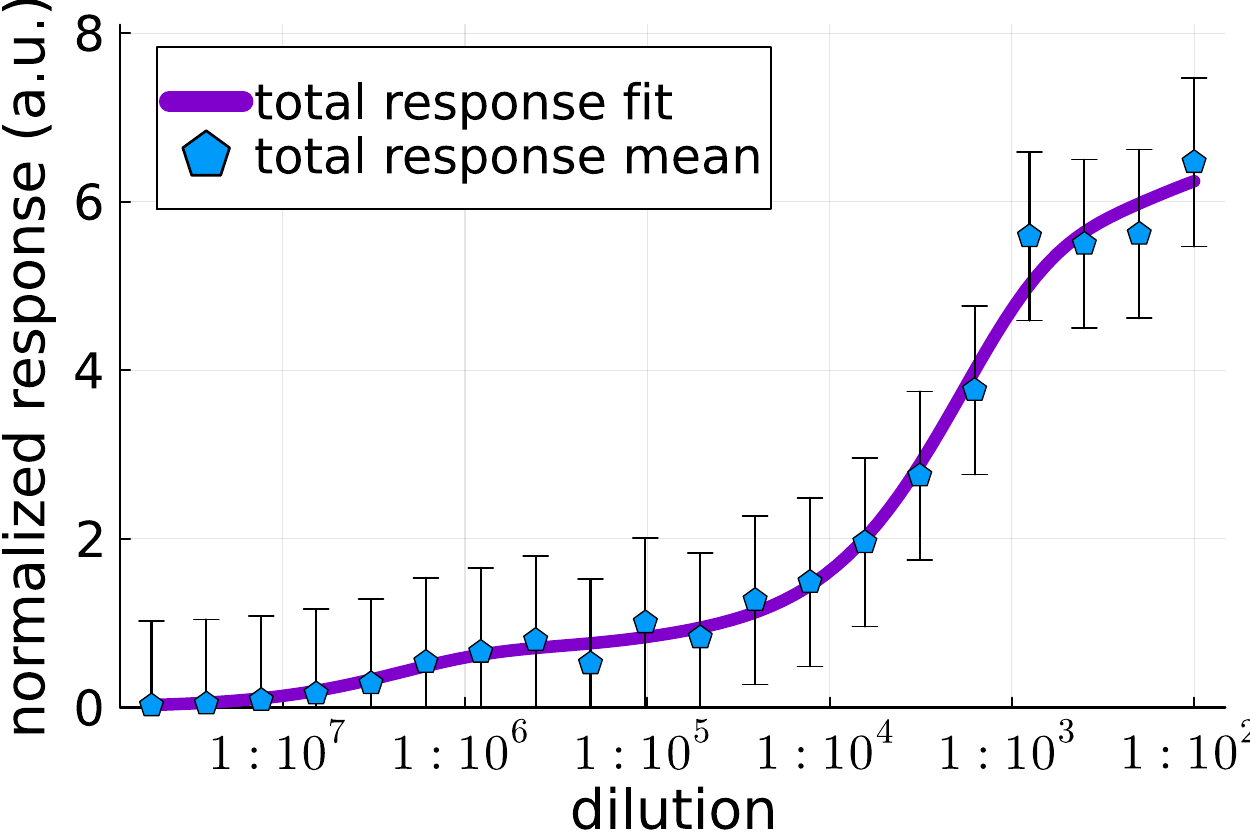}

		\includegraphics[width = 0.32\textwidth]{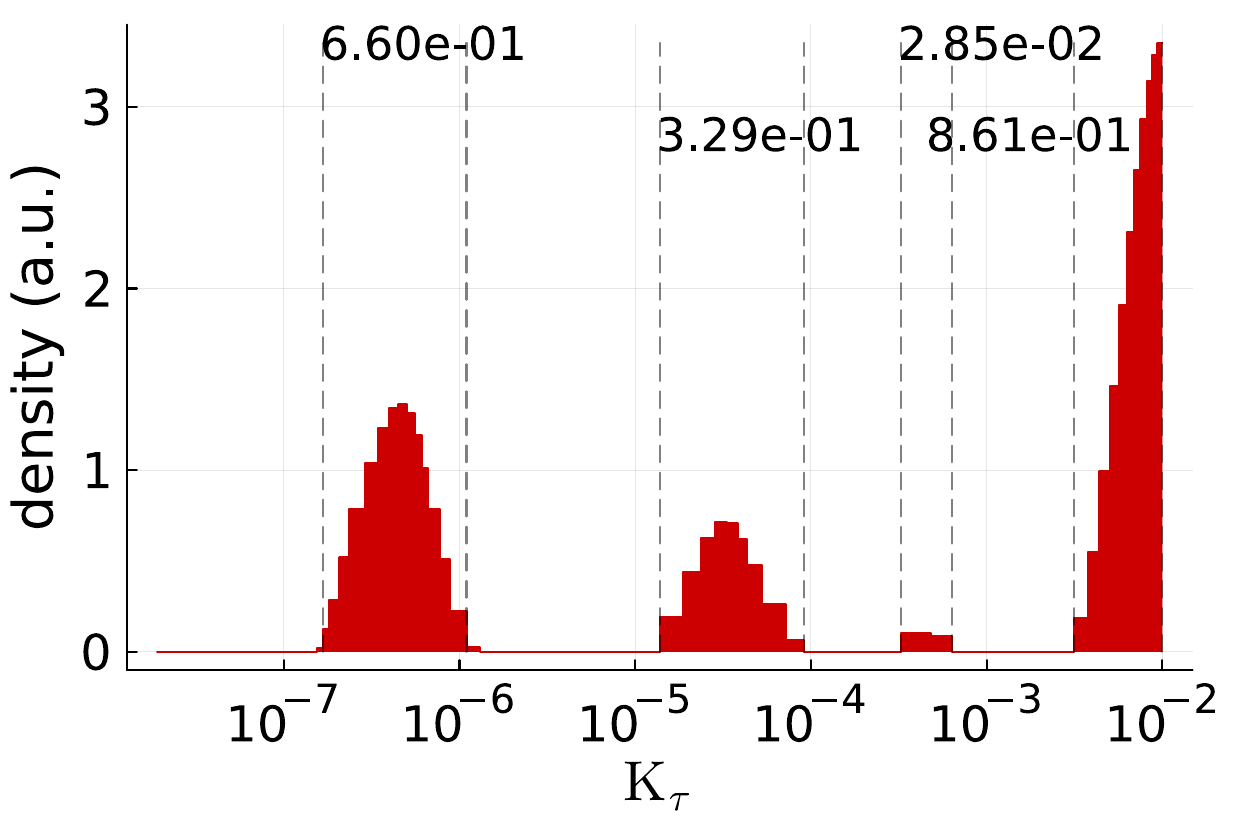}
		\includegraphics[width = 0.32\textwidth]{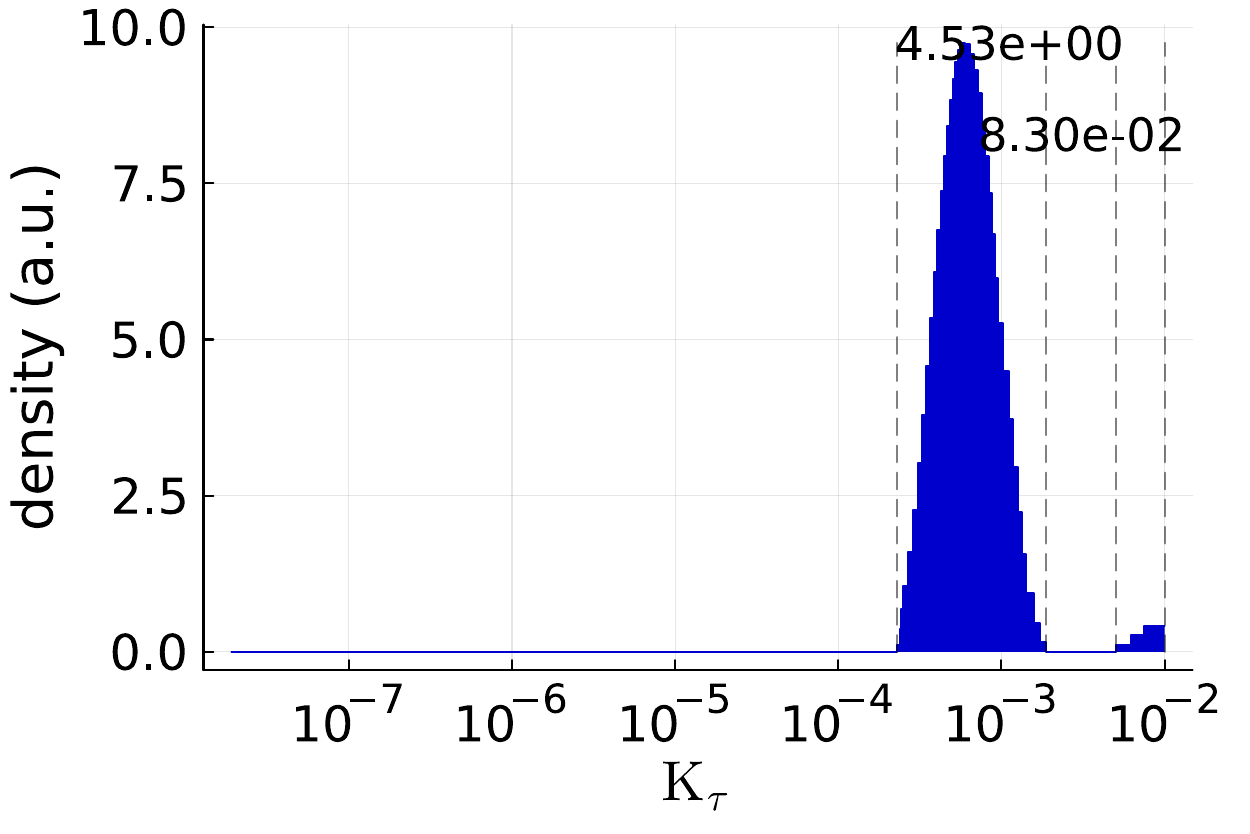}
		\includegraphics[width = 0.32\textwidth]{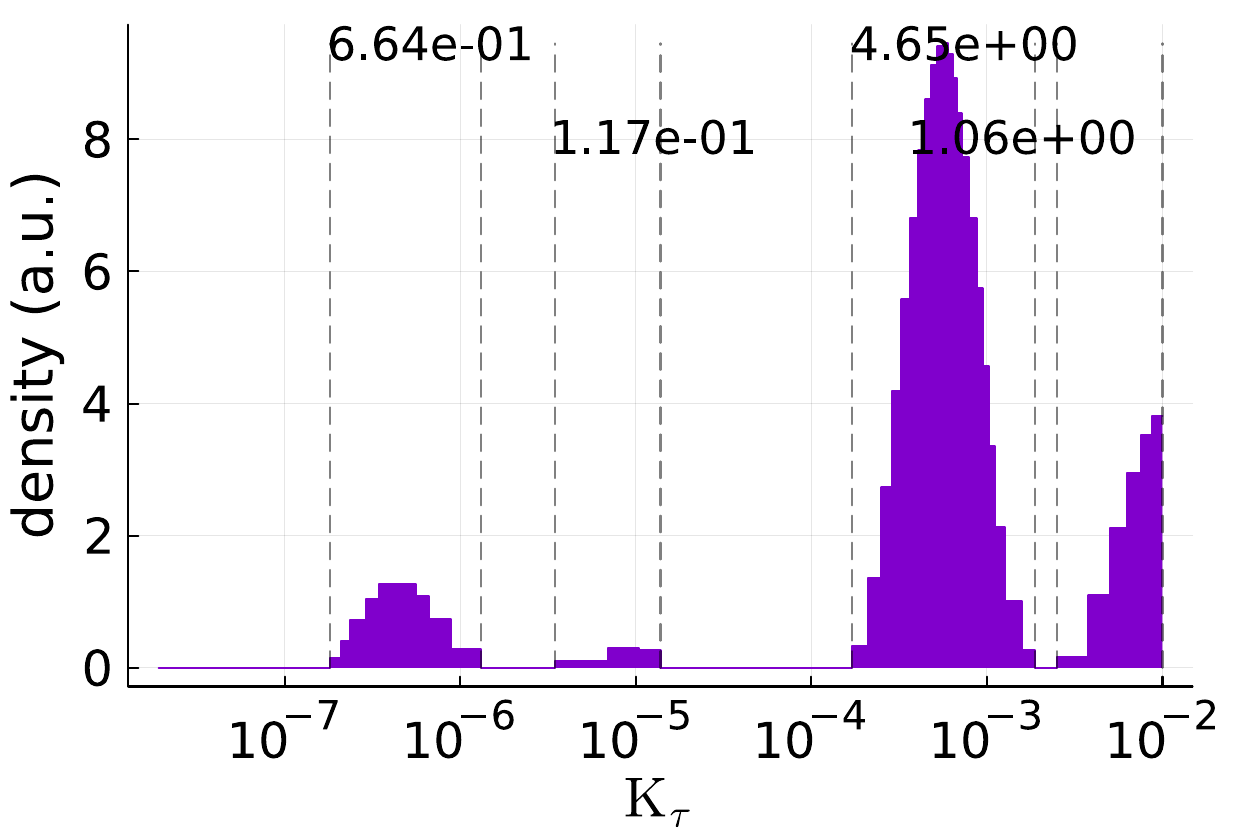}

	\captionof{figure}{Dose-response curves and corresponding accessibility histograms of replicate 7 (for the antibody mix condition).}
	\label{sup-fig: rep 7}
\end{minipage}

\subsection{Uncertainty analysis}

To investigate the results obtained for the weaker regularization (fig. \ref{sup-fig: weaker regularization}) further, we can quantify the uncertainty of the different peaks in the histogram. For this, we may pick a single peak, vary its size and observe the resulting fitting-objective-function values. Since our objective function is a logarithmic posterior distribution (see methods), the values can be converted into posterior probabilities. Thus, it is possible to obtain how much less likely a given change of the peak size is, compared to the best-fitting peak size.

Figures \ref{sup-fig: uncertainty 1}-\ref{sup-fig: uncertainty 4} show the uncertainty estimations for the 4 peaks of the total response histogram of the antibody mix. In all cases, the red line is the best-fitting histogram. The color ribbons (red $\rightarrow$ white $\rightarrow$ blue) show progressively less probable changes of the peak size. The corresponding effect of the changed peak size on the dose-response curve is shown by color-matched ribbons in the dose-response plots. The probability factors of the different colors (figure legends) are relative to the probability of the best-fitting histogram (which thus has the factor 1).

For peaks further left in the histogram (smaller $K_\tau$), changes of the overall peak size are highly unlikely (even 1000 times less probable peak sizes are tightly confined to the best-fitting peak size). Furthermore, only peak changes that have probability factors larger than $\nicefrac{1}{100}$ produce dose-response curves within the noise range of the dose-response data. In addition, observe that the leftmost peak can be removed completely (but not increased), without affecting the dose-response curve, suggesting that this is a noise artifact.

The size of peaks further right in the histogram (larger $K_\tau$) has much more leeway. This is because high-$K_\tau$ peaks affect mostly just the high-dilution-quotient points of the dose-response curve, thus leading only to deviations from a few data points. In other words, since there are no more data points to guide what peak-effects are allowed, there is much greater uncertainty for the high-$K_\tau$ peaks. Nevertheless, as before only peak size changes with probability factors larger than $\nicefrac{1}{100}$ produce dose-response curves within the noise range of the data. And as before, these changes are tightly confided to the best-fitting peak size if considered as relative deviation from the best-fitting peak size. That is, the deviations from the best-fitting peak sizes further left in the histogram are smaller, but the size of the peaks is much smaller as well.

The fact that the general peak size cannot be changed too much without producing dose-response curves that no longer match the data points (within the range of noise) corroborates the findings obtained with weaker regularization (fig. \ref{sup-fig: weaker regularization}).

\vspace{0.5cm}
\noindent
\begin{minipage}{\textwidth}
	\centering
		\includegraphics[width = 0.49\textwidth]{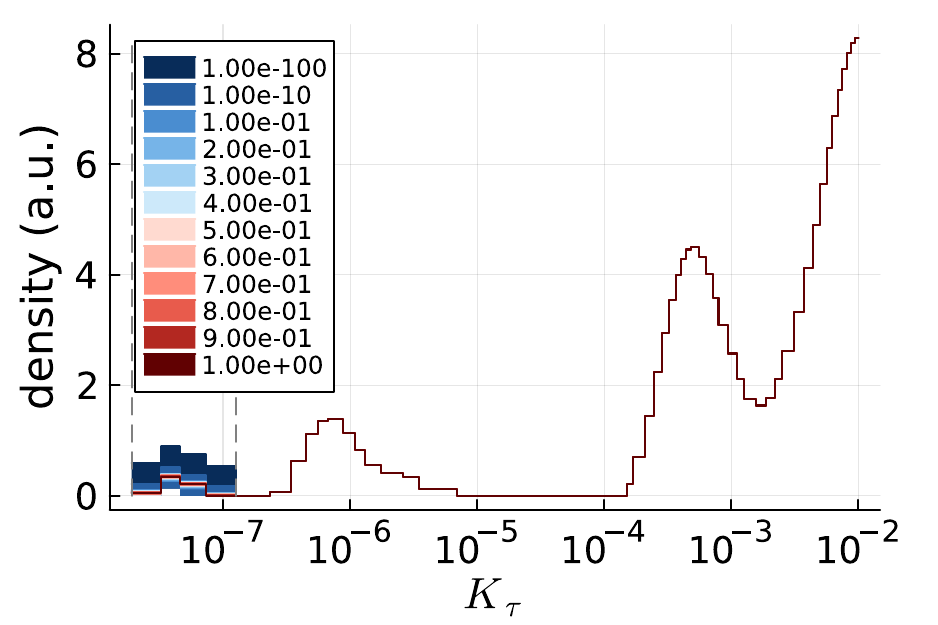}
		\includegraphics[width = 0.49\textwidth]{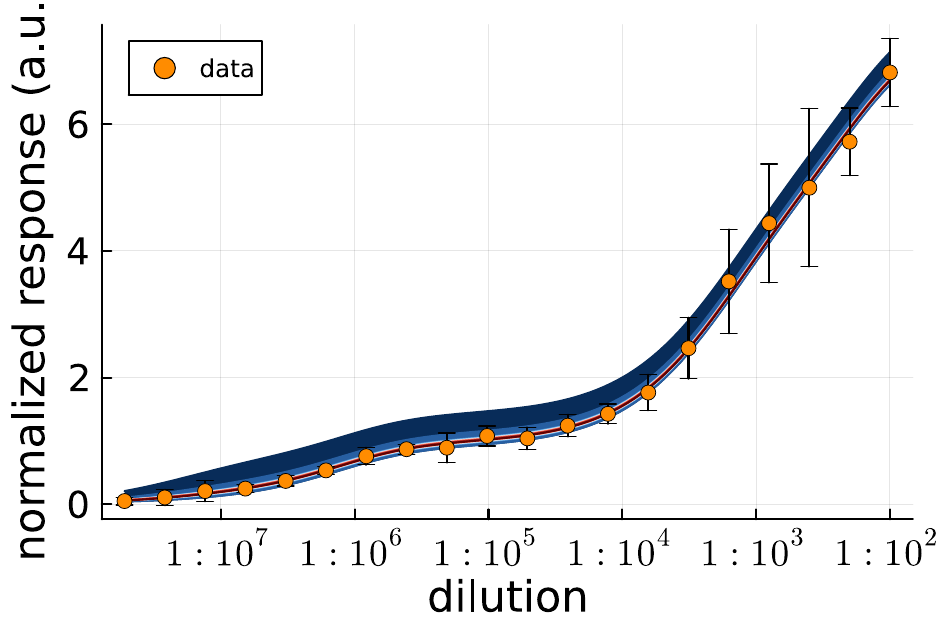}
		\captionof{figure}{Uncertainty analysis (1st peak) for the total response (channel sum) of the antibody mix condition.}
		\label{sup-fig: uncertainty 1}
\end{minipage}

\vspace{0.5cm}
\noindent
\begin{minipage}{\textwidth}
	\centering
		\includegraphics[width = 0.49\textwidth]{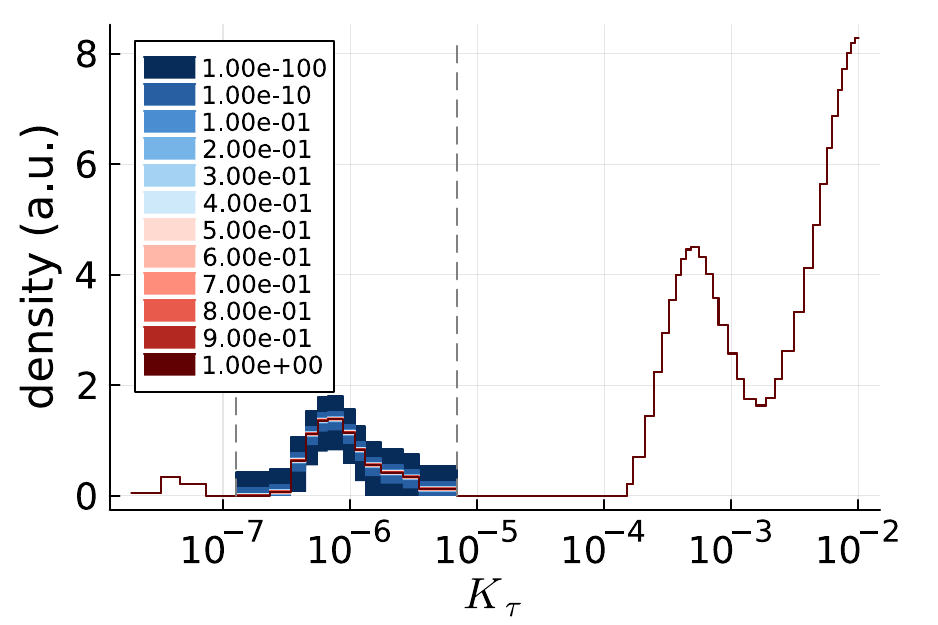}
		\includegraphics[width = 0.49\textwidth]{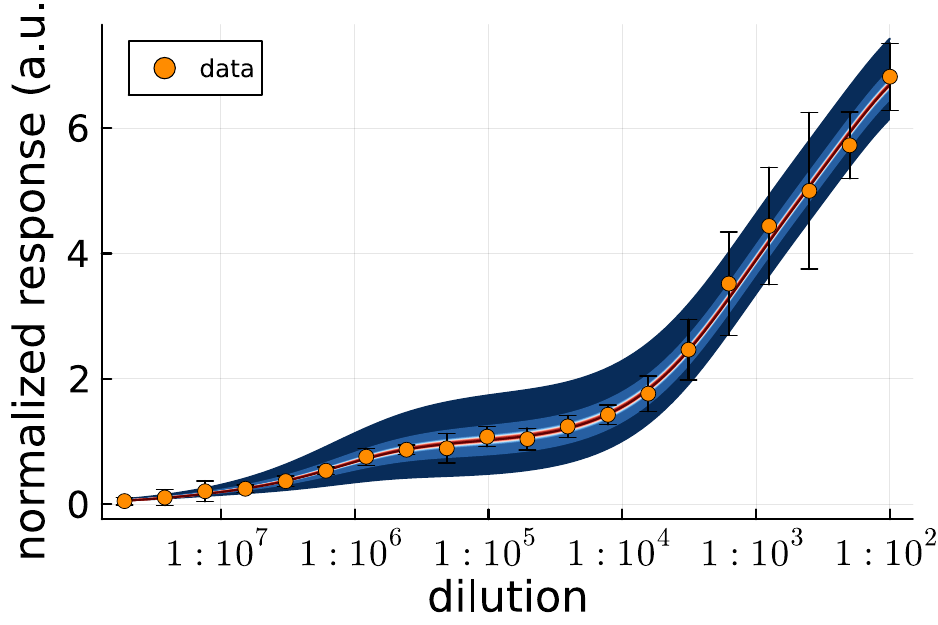}
		\captionof{figure}{Uncertainty analysis (2nd peak) for the total response (channel sum) of the antibody mix condition.}
		\label{sup-fig: uncertainty 2}
\end{minipage}

\vspace{0.5cm}
\noindent
\begin{minipage}{\textwidth}
	\centering
		\includegraphics[width = 0.49\textwidth]{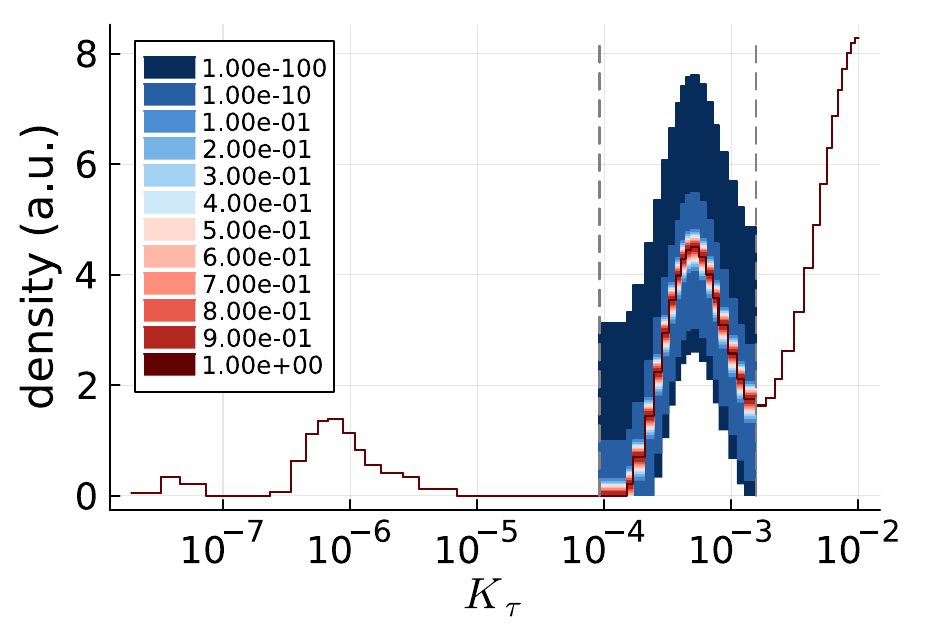}
		\includegraphics[width = 0.49\textwidth]{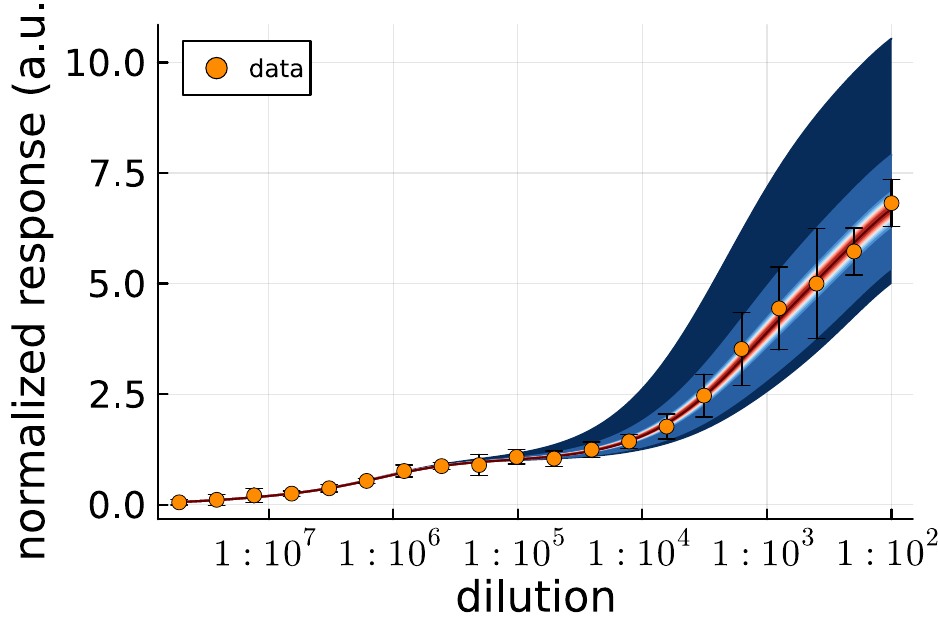}
		\captionof{figure}{Uncertainty analysis (3rd peak) for the total response (channel sum) of the antibody mix condition.}
		\label{sup-fig: uncertainty 3}
\end{minipage}

\vspace{0.5cm}
\noindent
\begin{minipage}{\textwidth}
	\centering
		\includegraphics[width = 0.49\textwidth]{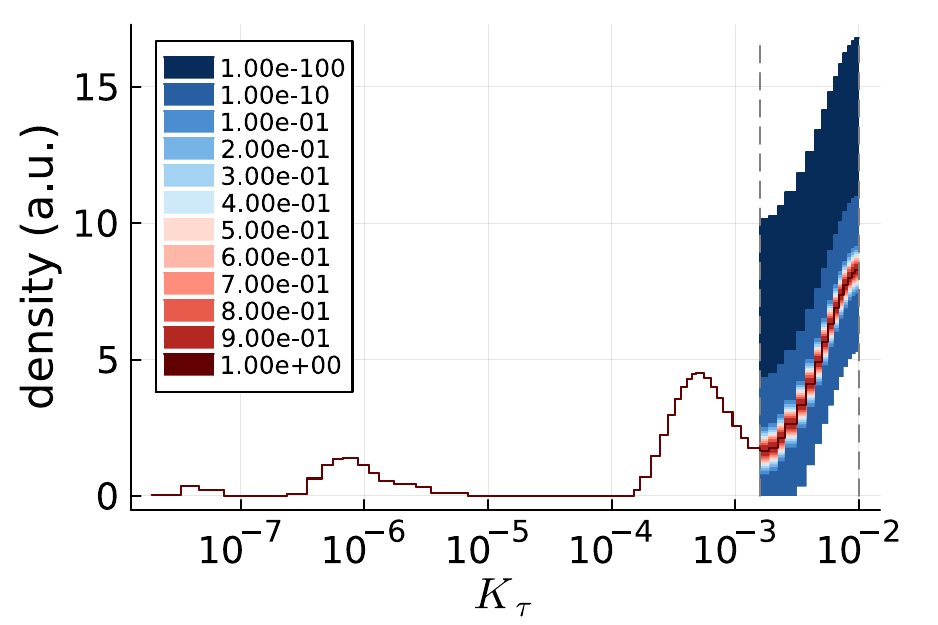}
		\includegraphics[width = 0.49\textwidth]{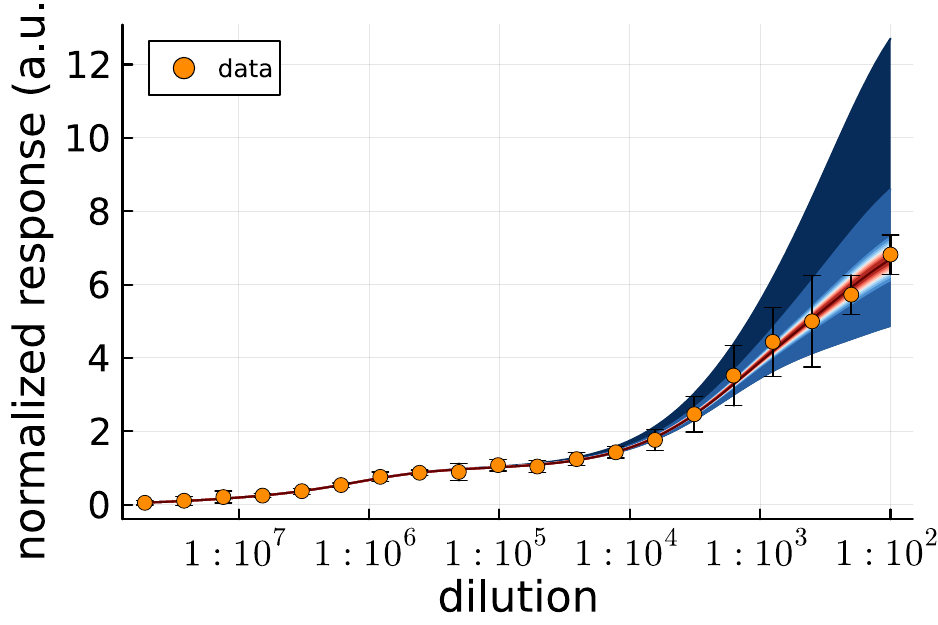}
		\captionof{figure}{Uncertainty analysis (4th peak) for the total response (channel sum) of the antibody mix condition.}
		\label{sup-fig: uncertainty 4}
\end{minipage}

\section{Removed replicates}
\label{sup-sec: removed replicates}

Two replicates were removed from further analysis, one replicate from the anti-NF200-antibody condition and one replicate from the antibody-mix condition. Figure \ref{sup-fig: outlier}a shows the dose-response data of the excluded anti-NF200 replicate (only NF200 channel) and figure \ref{sup-fig: outlier}b shows the excluded antibody-mix replicate (anti-NF200 and anti-RPS11 channel). In both cases it can be observed that the dose-response condition (higher dose means higher response) is violated multiple times beyond the margin of noise for the NF200 channel. Furthermore, the RPS11 channel in the antibody-mix condition contains an obvious outlier at the high-dilution-quotient end of the dose-response curve.

The violation of the dose-response behavior indicates that something must have gone wrong, as all other replicates follow the dose-response behavior within the margins of noise. What exactly went wrong is not known. But since both faulty replicates are from the same plate, it may be speculated that some of the cells have dried out during immunocytochemistry. Thus, we excluded the faulty replicates completely from all analyses.

\vspace{0.5cm}
\noindent
\begin{minipage}{\textwidth}
	\centering

	\begin{minipage}{0.49\textwidth}
		\centering
		\textbf{(a)} anti-NF200 condition

		\includegraphics[width = \textwidth]{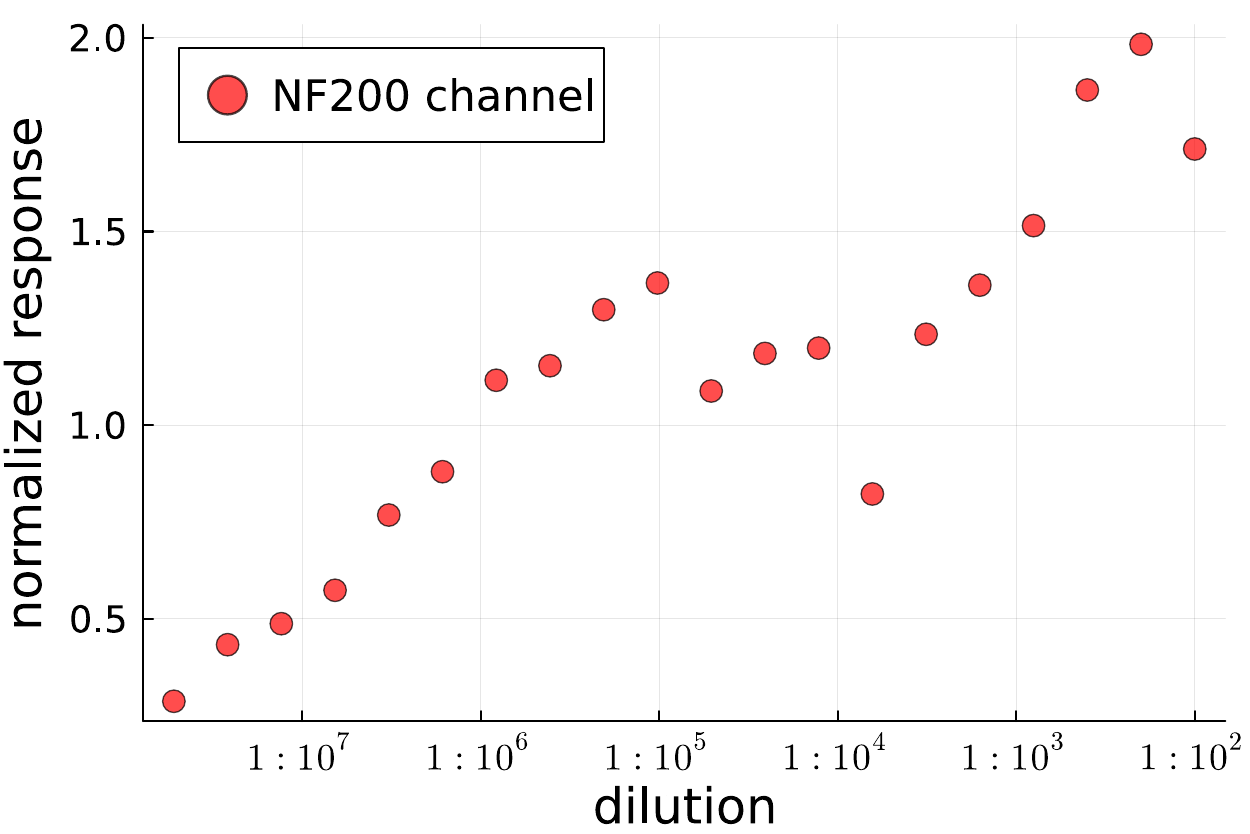}
	\end{minipage}
	\begin{minipage}{0.49\textwidth}
		\centering
		\textbf{(b)} antibody-mix condition

		\includegraphics[width = \textwidth]{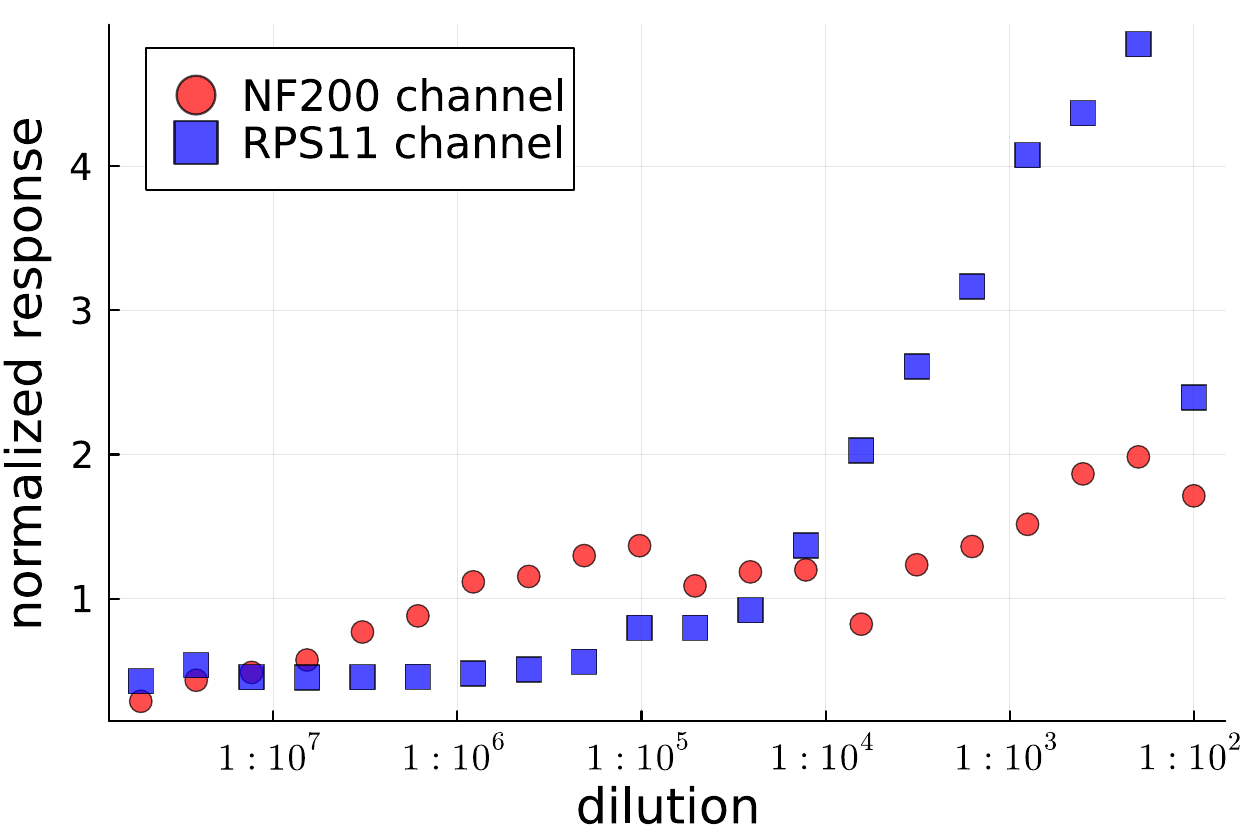}
	\end{minipage}

	\captionof{figure}{Excluded replicates.}
	\label{sup-fig: outlier}
\end{minipage}

\section{Double-staining: Additional figures}
\label{sup-sec: double-staining additional figures}

\subsubsection*{Graphical illustration of the double-staining approach}

For an illustration of the double-staining idea, as well as to provide missing images, figure \ref{sup-fig: graphical illustration} summarizes the image subtraction process graphically. As a short reminder: Two antibody concentrations\slash dilution quotients, corresponding to peaks in the accessibility histogram, are determined. Then, the cells are stained with the first  antibody concentration\slash dilution quotient and a microscopy image is captured (1st staining image). After this, the same cells are stained again, this time with with the second (higher) antibody concentration\slash dilution quotient, leading to the 2nd staining image. The 1st staining image contains only contributions from the low-$K_\tau$ peaks. To account for the double-application of secondary antibodies, the 1st-staining-image brightness is increased (according to the 2nd ab control, Fig. \ref{sup-fig: optimal metrics}a,d), producing the images in figure \ref{sup-fig: graphical illustration}b,e. These brightness-increased 1st staining images are subtracted from the 2nd staining images (Fig. \ref{sup-fig: graphical illustration}a,d). The resulting difference images (Fig. \ref{sup-fig: graphical illustration}c,f) then contain mostly the contributions of the high-$K_\tau$ peaks. 

\begin{figure}[!ht]
	\centering

	{\large\bfseries Plate 1}\vspace{0.5cm}

	\begin{minipage}{0.23\textwidth}
		\centering
		\textbf{(a)} 2nd staining
		\includegraphics[width = \textwidth]{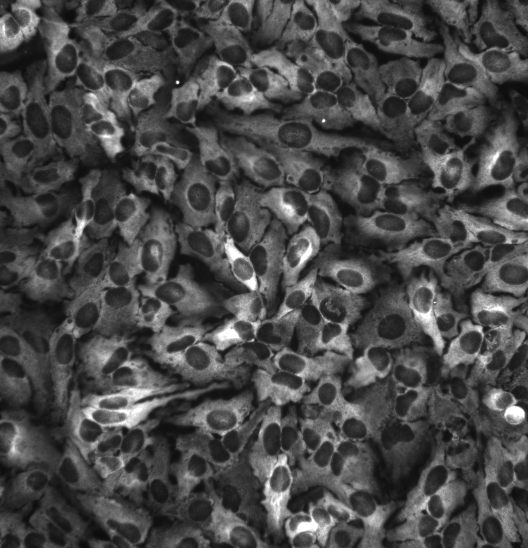}
		\vphantom{(a)}
	\end{minipage} 
	\textbf{\Huge $\ - \ $}
	\begin{minipage}{0.23\textwidth}
		\centering
		\textbf{(b)} 1st staining*
		\includegraphics[width = \textwidth]{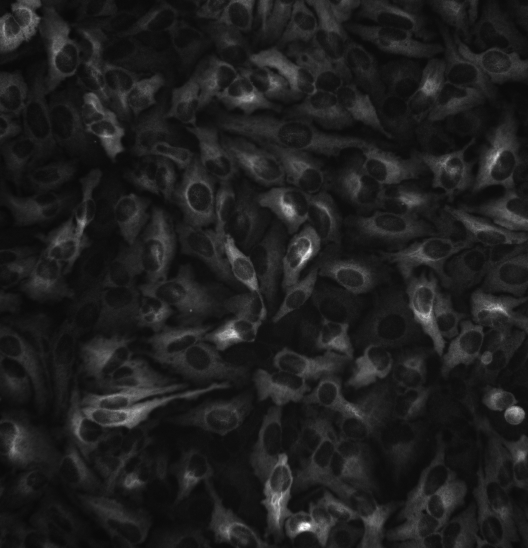}
		\vphantom{(a)}
	\end{minipage} 
	\textbf{\Huge $\ = \ $}
	\begin{minipage}{0.23\textwidth}
		\centering
		\textbf{(c)} Difference
		\includegraphics[width = \textwidth]{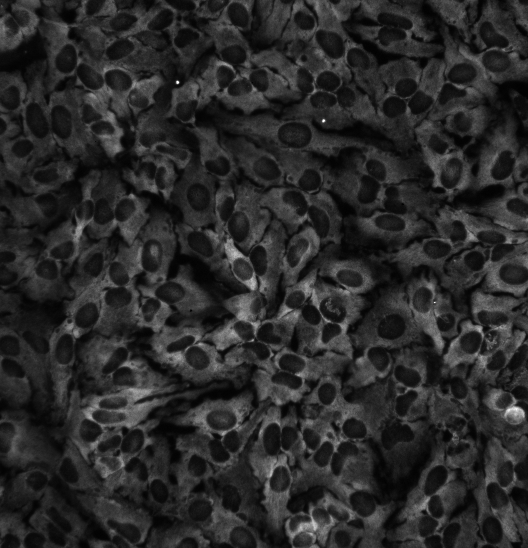}
		\vphantom{(a)}
	\end{minipage} 

	\vspace{0.2cm}

	{\large\bfseries Plate 2}\vspace{0.5cm}

	\begin{minipage}{0.23\textwidth}
		\centering
		\textbf{(d)} 2nd staining
		\includegraphics[width = \textwidth]{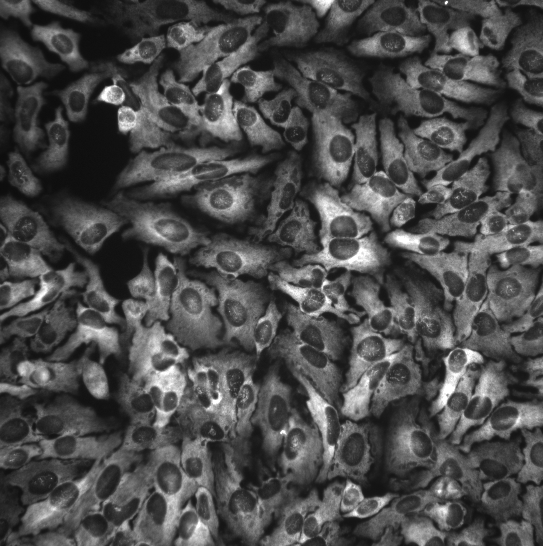}
		\vphantom{(a)}
	\end{minipage} 
	\textbf{\Huge $\ - \ $}
	\begin{minipage}{0.23\textwidth}
		\centering
		\textbf{(e)} 1st staining*
		\includegraphics[width = \textwidth]{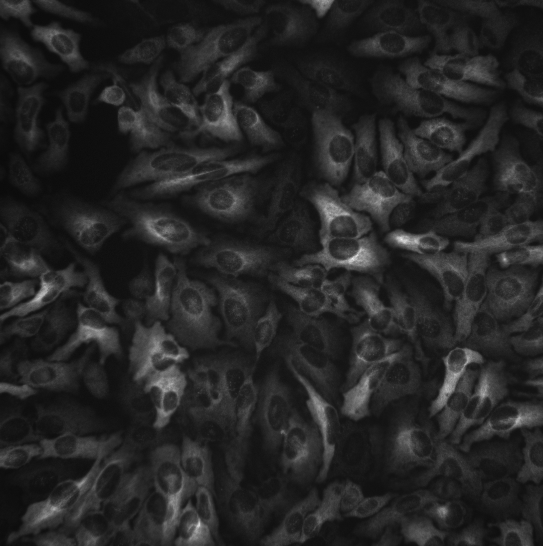}
		\vphantom{(a)}
	\end{minipage} 
	\textbf{\Huge $\ = \ $}
	\begin{minipage}{0.23\textwidth}
		\centering
		\textbf{(f)} Difference
		\includegraphics[width = \textwidth]{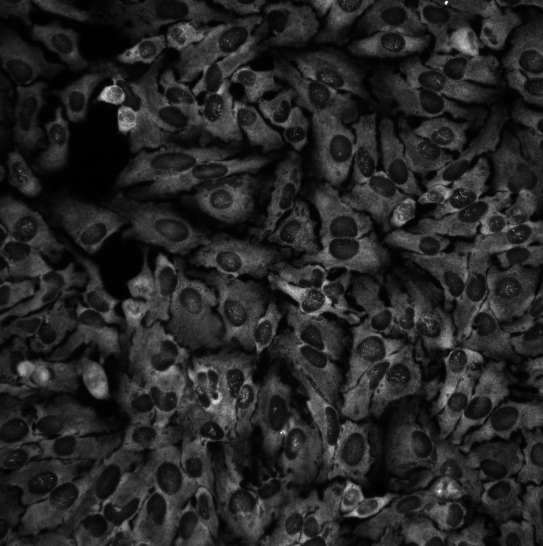}
		\vphantom{(a)}
	\end{minipage} 

	\caption{Illustration of the double-staining approach. \textbf{*)} The 1st staining images already contain the correction for the double-application of secondary antibodies (brightness increase according to ctrl 2nd ab in figures \ref{sup-fig: optimal metrics}a,d respectively).}
	\label{sup-fig: graphical illustration}
\end{figure}

\subsubsection*{Additional view fields}

To illustrate the consistency of the double-staining approach, figure \ref{sup-fig: double-staining view fields} contains additional double-staining composite images (next to the color channel composites). On each plate, 2 wells were used for the double-staining. For each well, distinct view fields were captured.

\subsubsection*{Full-size metrics plots}
In figure \ref{main-fig: double-staining} we used only thumbnails for the metrics of the anti-NF200 and anti-RPS11 channels, because of figure size limitations. Yet, we drew conclusions from these channels, i.e. that the NF200 response doubled, which indicated the deficits of the validation system for the verification of the double-staining approach. Thus, figure \ref{sup-fig: optimal metrics} contains full-size plots for the metrics of each channel.

Having all metrics next to each other allows to see that the majority of the signal increase after the 2nd staining does indeed come from the anti-RPS11 antibody (blue). But recall, as argued in the main part of the paper, that the high-$K_\tau$ peaks (in Fig. \ref{main-fig: optimal dilution}a) actually comprise peaks from both the anti-NF200 antibody (red) and the anti-RPS11 antibody (blue). Hence the aforementioned increase of the NF200 signal.

\begin{figure}[!ht]
	\centering
	{\large\bfseries Plate 1 metrics: optimal concentration for the 1st staining}\vspace{0.5cm}

	\begin{minipage}[c]{0.32\textwidth}
		\centering
		\textbf{(a)} Channel sum

		\includegraphics[width = \textwidth]{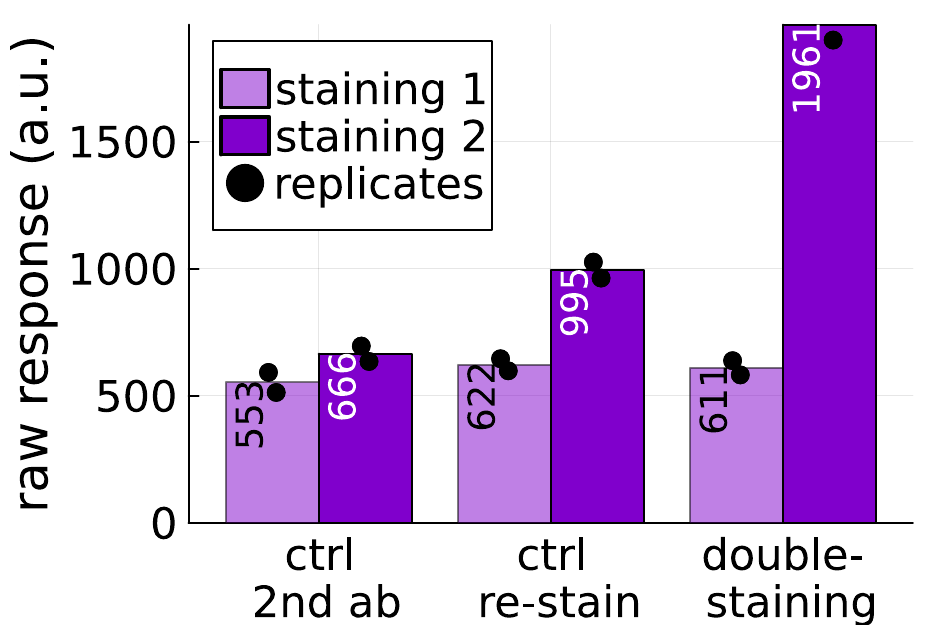}
	\end{minipage}
	\begin{minipage}[c]{0.32\textwidth}
		\centering
		\textbf{(b)} NF200 channel

		\includegraphics[width = \textwidth]{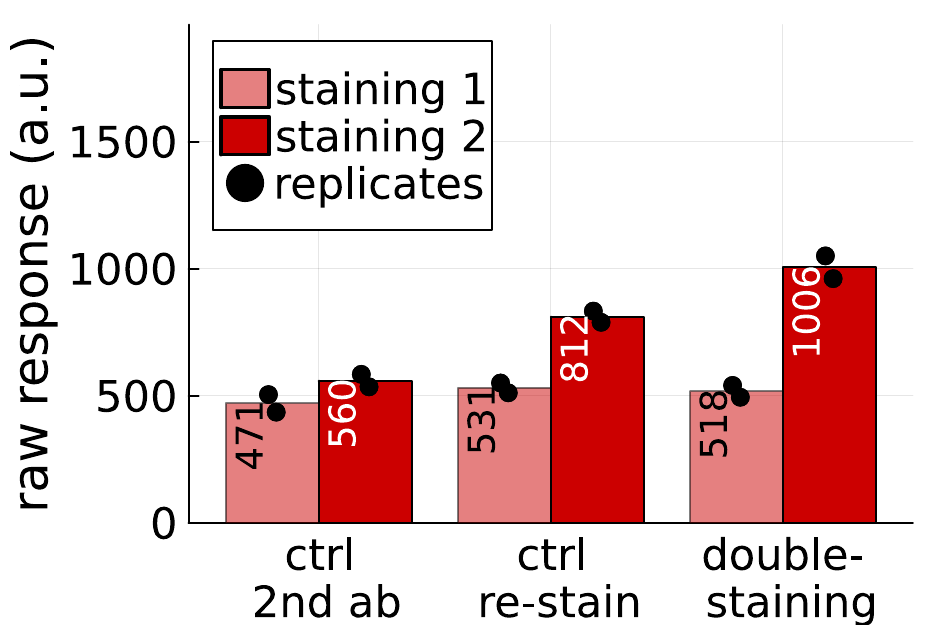}
	\end{minipage}
	\begin{minipage}[c]{0.32\textwidth}
		\centering
		\textbf{(c)} RPS11 channel

		\includegraphics[width = \textwidth]{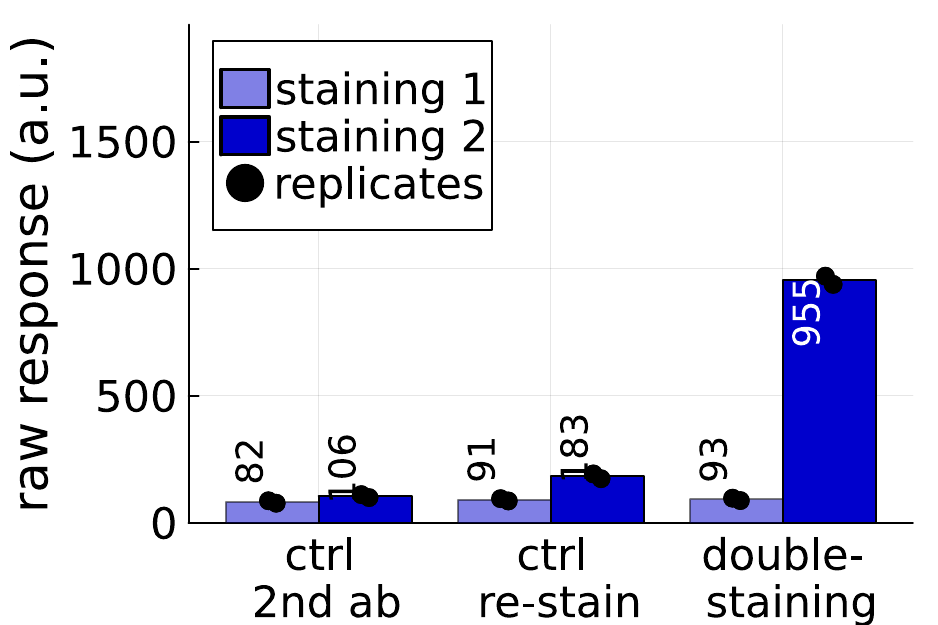}
	\end{minipage}

	\vspace{0.5cm}
	\rule{\textwidth}{2pt}
	\vspace{0.5cm}

	{\large\bfseries Plate 2 metrics: optimal concentration for the 1st staining}\vspace{0.5cm}
	\begin{minipage}[c]{0.32\textwidth}
		\centering
		\textbf{(d)} Channel sum

		\includegraphics[width = \textwidth]{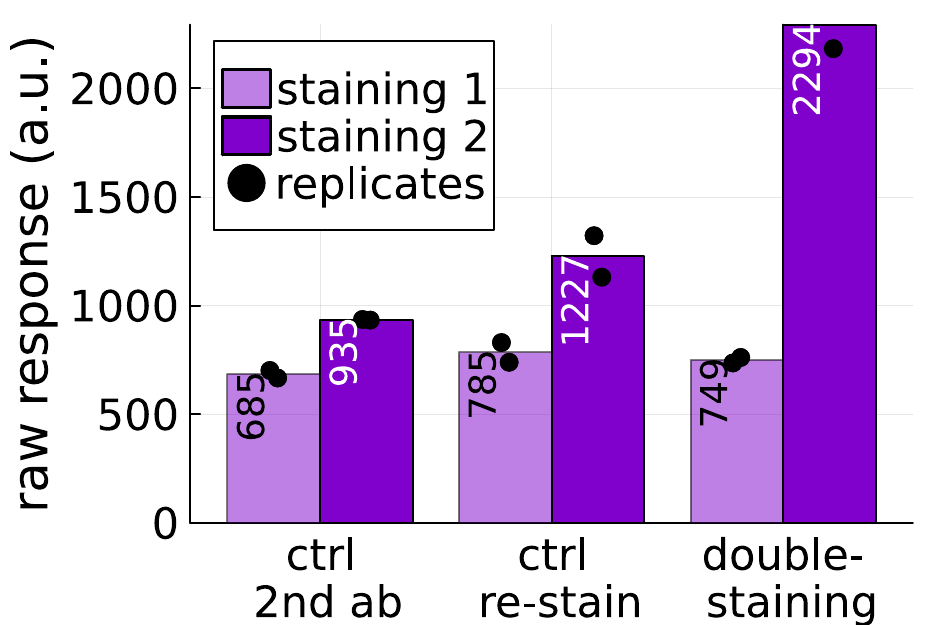}
	\end{minipage}
	\begin{minipage}[c]{0.32\textwidth}
		\centering
		\textbf{(e)} NF200 channel

		\includegraphics[width = \textwidth]{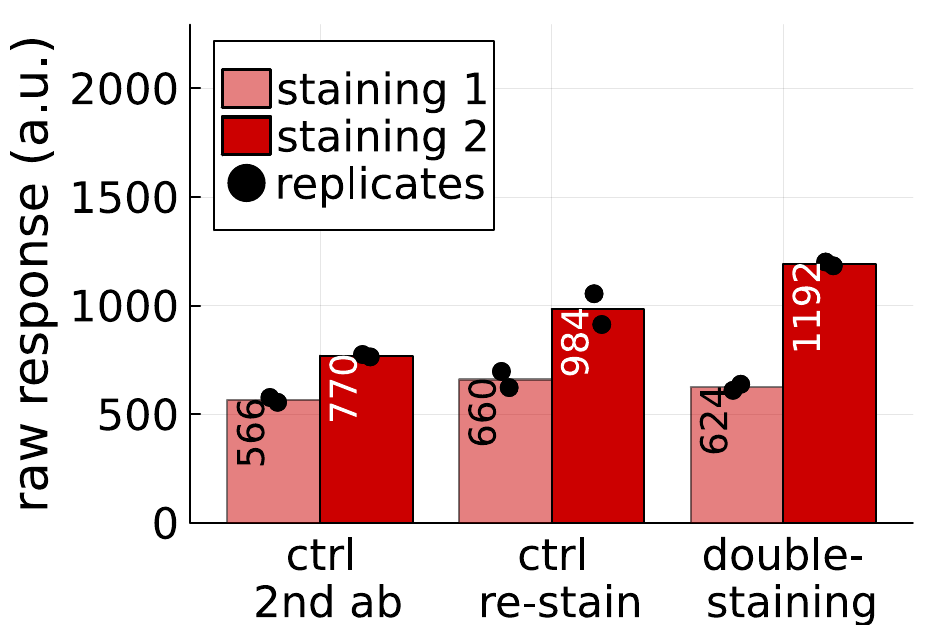}
	\end{minipage}
	\begin{minipage}[c]{0.32\textwidth}
		\centering
		\textbf{(f)} RPS11 channel

		\includegraphics[width = \textwidth]{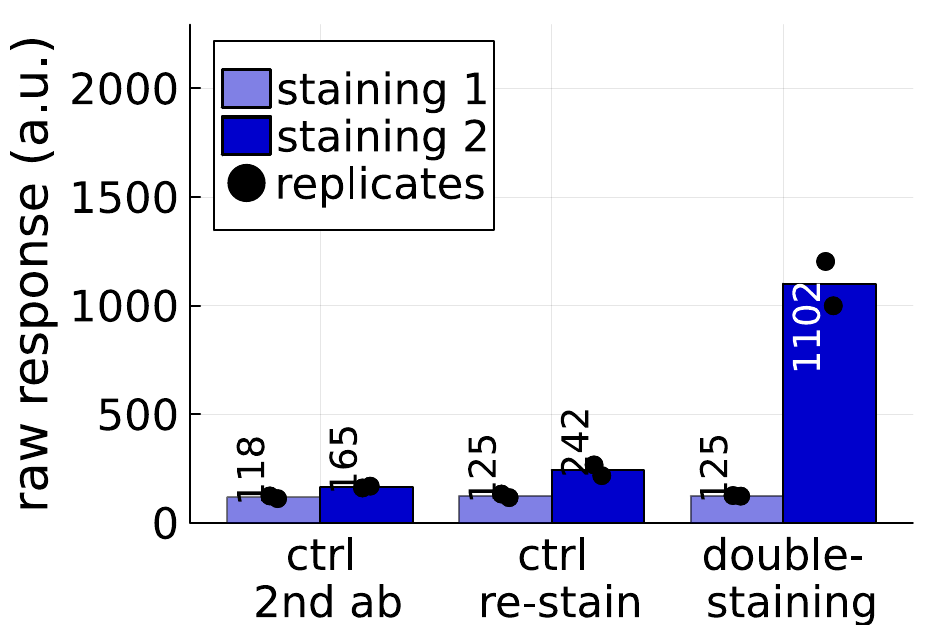}
	\end{minipage}
	\caption{Full-size figures for the thumbnails in figure \ref{main-fig: double-staining}.}
	\label{sup-fig: optimal metrics}
\end{figure}

\subsubsection*{Using improper antibody dilution quotients for the double-staining}

At the end of the double-staining section in the main part of the paper, we argued that using arbitrary antibody concentrations\slash dilution quotients for the double-staining would lead to improper results. For this reason, we repeated the double-staining with improper dilution quotients. Figure \ref{sup-fig: too low} contains the results for too low dilution quotients for the 1st staining. Accordingly, figure \ref{sup-fig: too high} contains the results for too high dilution quotients for the 1st staining. Note that the metrics always show the actual signal increase for the respective color channel, unlike the computational composites that attribute the signals to the wrong color channel because of the improper antibody dilution quotients for the 1st staining.

When too low dilution quotients are used for the 1st staining (here $d_{14} = 1:1638400$ instead of $d_{10} = 1:102400$), the 2nd staining leads to a substantial increase of the signal that belongs to the low-$K_\tau$ peaks. Recall that in the double-staining approach signal components are attributed to the peaks by signal differences between the 1st and the 2nd staining. Since the response of the low-$K_\tau$ peaks was not saturated, the corresponding increase gets wrongly attributed to the high-$K_\tau$ peaks. For the computational composite this means that signal contributions from the low-$K_\tau$ peaks (that should be assigned to the red composite color) get wrongly assigned to the blue composite color (used for the high-$K_\tau$ peaks). This can be seen in figures \ref{sup-fig: too low}d,i where the computational composites contain too much blue signal in comparison to the correct color channel composite images (figs. \ref{sup-fig: too low}e,j).

When too high dilution quotients are used for the 1st staining, the effect is reversed. Contributions of the high-$K_\tau$ peaks are already contained in the 1st staining image and thus get wrongly attributed to the low-$K_\tau$ peaks. This leads to an underestimation of the high-$K_\tau$-peak contributions. In figures \ref{sup-fig: too high}d,i it can be seen that the computational composite images contain slightly too much red signal in comparison to the correct channel composite images (figs. \ref{sup-fig: too high}e,j). Because of the deficits of the validation system, we could not use higher dilution quotients than $d_6 = 1:6400$ for the 2nd staining. Thus, for the illustration of too high 1st staining dilution quotients we could only use $d_{8} = 1:25600$ instead of $d_{10} = 1:102400$ for the 1st staining. Hence, the illustrated effect, too much red signal, is rather small here.

\FloatBarrier

\noindent
\begin{minipage}{\textwidth}
	\renewcommand{\arraystretch}{1.4}

	\begin{tabular}{>{\centering\arraybackslash}m{0.22\textwidth}|>{\centering\arraybackslash}m{0.22\textwidth}||>{\centering\arraybackslash}m{0.22\textwidth}|>{\centering\arraybackslash}m{0.22\textwidth}}

		\multicolumn{4}{c}{\textbf{\Large Plate 1}} \\ \hline

		\multicolumn{2}{c||}{\textbf{\large Well A}} & \multicolumn{2}{c}{\textbf{\large Well B}}\\ \hline

		\textbf{Computational} & \textbf{Color channel} & \textbf{Computational} & \textbf{Color channel} \\ \hline

		\vspace{0.5em}
		\includegraphics[width = 0.21\textwidth]{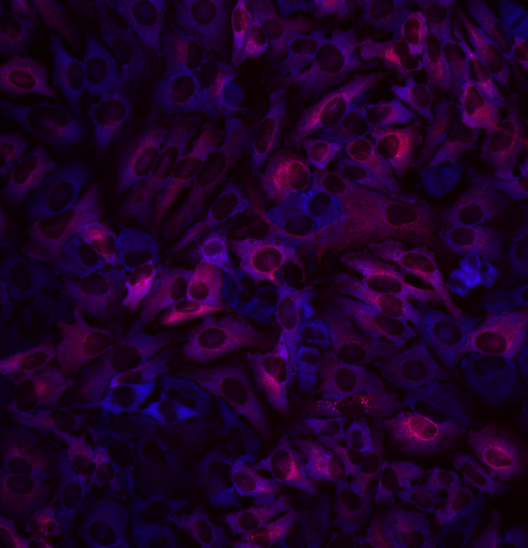} &

		\vspace{0.5em}
		\includegraphics[width = 0.21\textwidth]{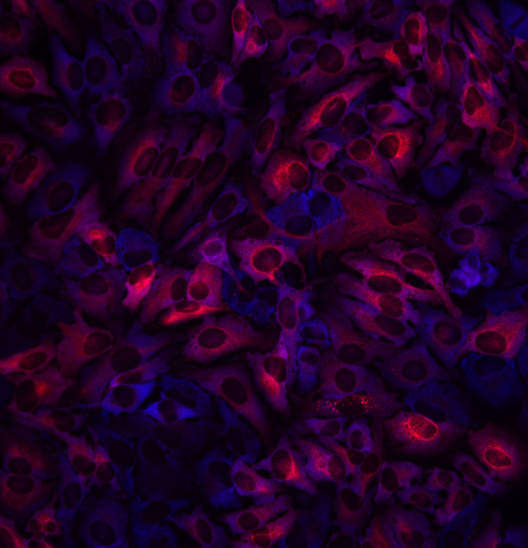}&

		\vspace{0.5em}
		\includegraphics[width = 0.21\textwidth]{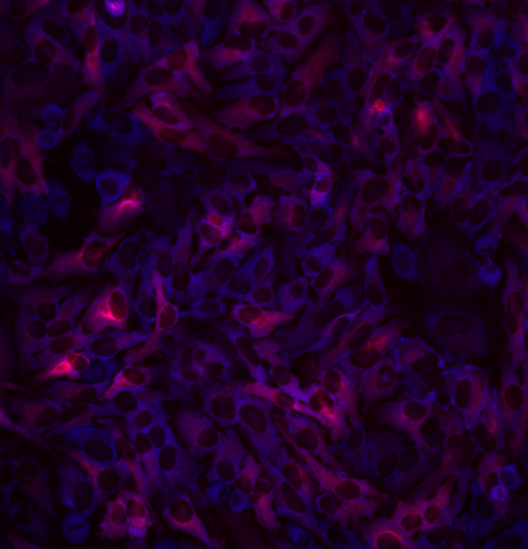} &

		\vspace{0.5em}
		\includegraphics[width = 0.21\textwidth]{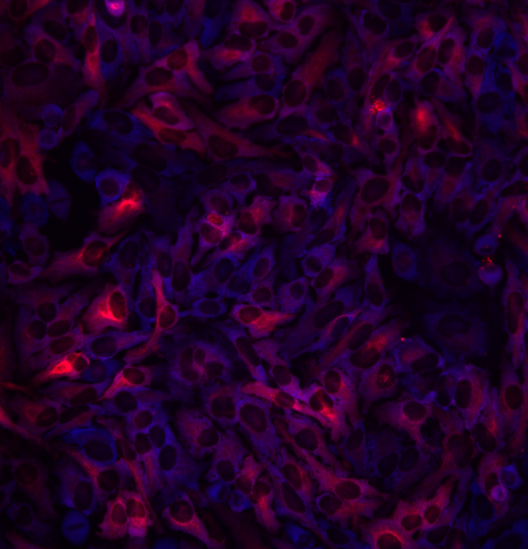}\\ \hline 

		\vspace{0.5em}
		\includegraphics[width = 0.21\textwidth]{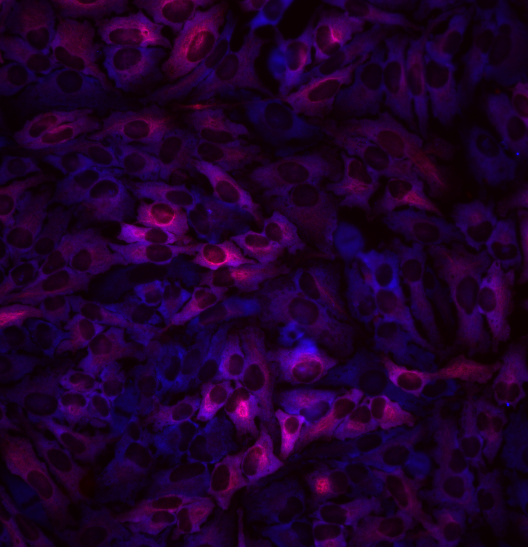} &

		\vspace{0.5em}
		\includegraphics[width = 0.21\textwidth]{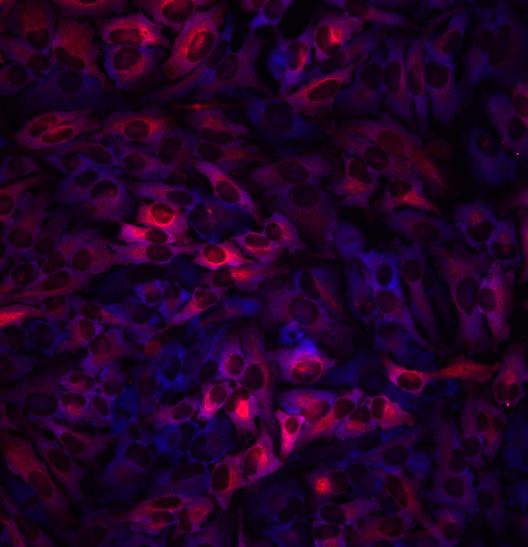}&

		\vspace{0.5em}
		\includegraphics[width = 0.21\textwidth]{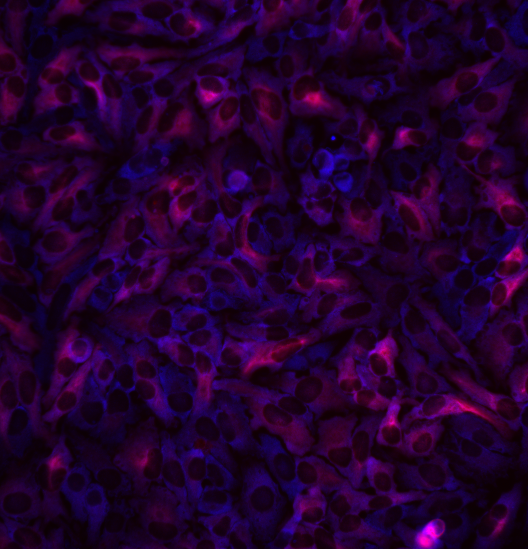} &

		\vspace{0.5em}
		\includegraphics[width = 0.21\textwidth]{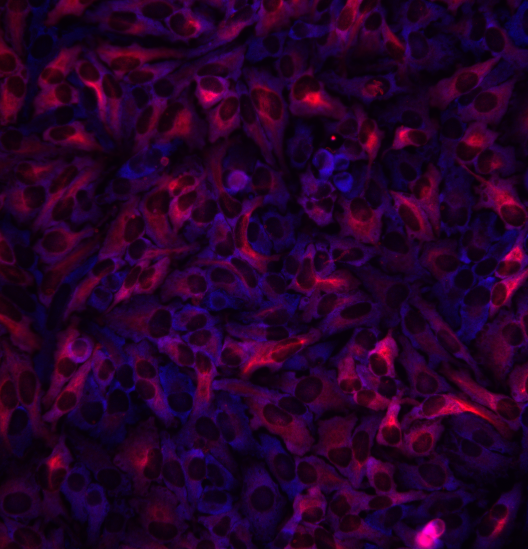}\\ \hline

	\end{tabular}

	\vspace{1em}

	\begin{tabular}{>{\centering\arraybackslash}m{0.22\textwidth}|>{\centering\arraybackslash}m{0.22\textwidth}||>{\centering\arraybackslash}m{0.22\textwidth}|>{\centering\arraybackslash}m{0.22\textwidth}}

		\multicolumn{4}{c}{\textbf{\Large Plate 2}} \\ \hline

		\multicolumn{2}{c||}{\textbf{\large Well A}} & \multicolumn{2}{c}{\textbf{\large Well B}}\\ \hline

		\textbf{Computational} & \textbf{Color channel} & \textbf{Computational} & \textbf{Color channel} \\ \hline

		\vspace{0.5em}
		\includegraphics[width = 0.21\textwidth]{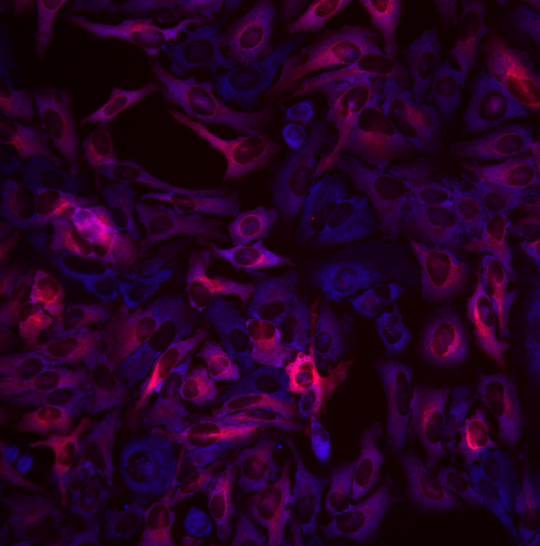} &

		\vspace{0.5em}
		\includegraphics[width = 0.21\textwidth]{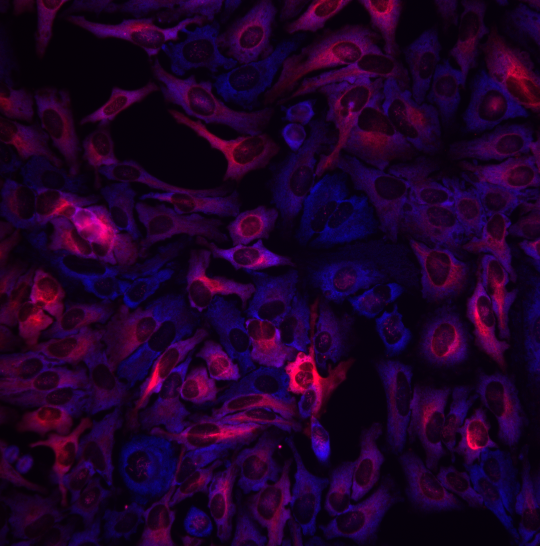}&

		\vspace{0.5em}
		\includegraphics[width = 0.21\textwidth]{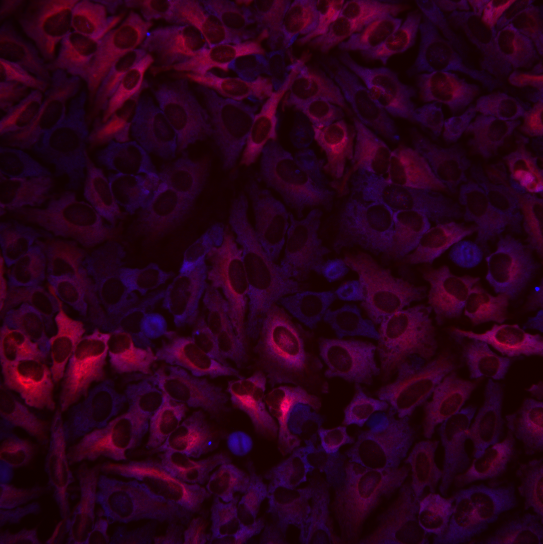} &

		\vspace{0.5em}
		\includegraphics[width = 0.21\textwidth]{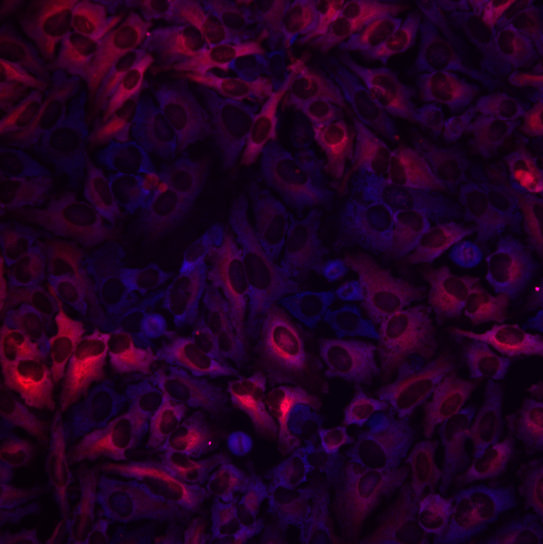}\\ \hline 

		\vspace{0.5em}
		\includegraphics[width = 0.21\textwidth]{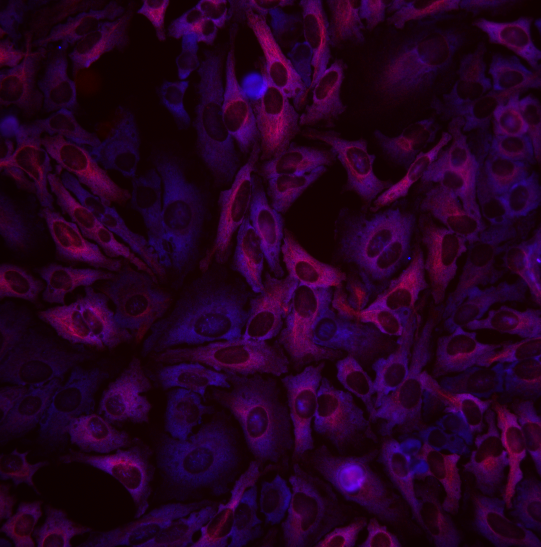} &

		\vspace{0.5em}
		\includegraphics[width = 0.21\textwidth]{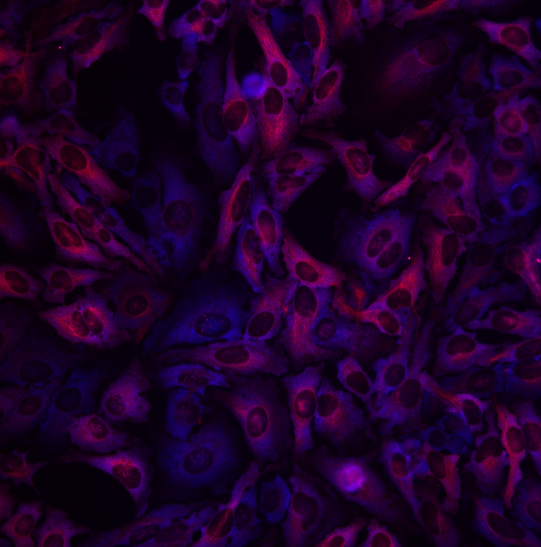}&

		\vspace{0.5em}
		\includegraphics[width = 0.21\textwidth]{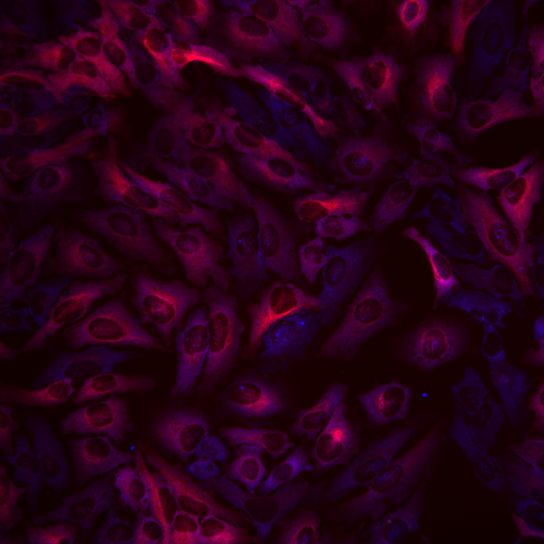} &

		\vspace{0.5em}
		\includegraphics[width = 0.21\textwidth]{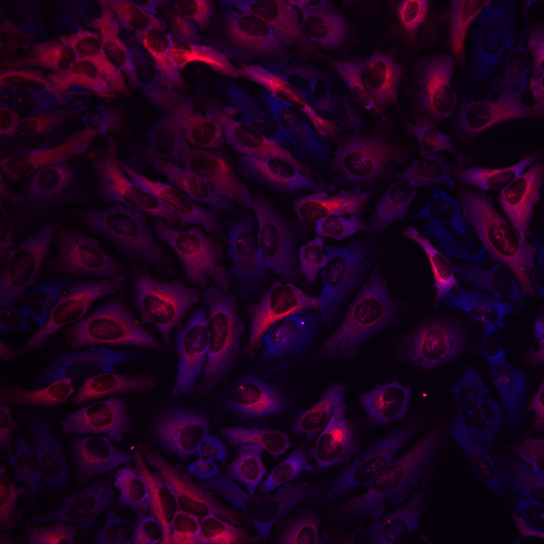}\\ \hline

	\end{tabular}

	\captionof{figure}{Additional double-staining composites and color channel composites from both plates. On each plate, 2 separate wells were used for the double-staining. For each well, multiple distinct view fields were captured.}
	\label{sup-fig: double-staining view fields}
\end{minipage}

\noindent
\begin{minipage}{\textwidth}
	\centering
	{\large\bfseries Plate 1: too low concentration for the 1st staining}\vspace{0.5cm}

	\begin{minipage}[c]{0.32\textwidth}
		\centering
		\textbf{(a)} Channel sum

		\includegraphics[width = \textwidth]{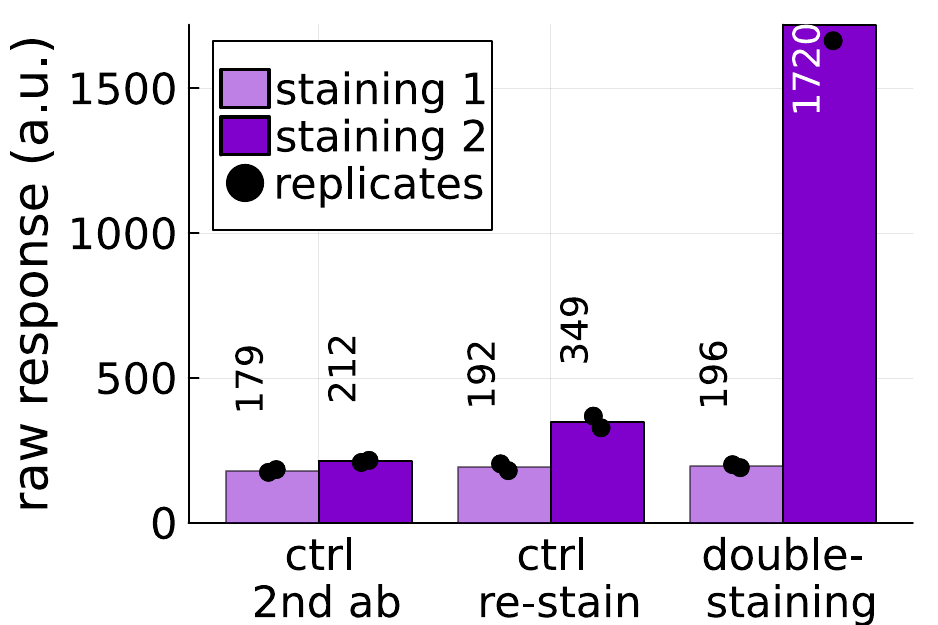}
	\end{minipage}
	\begin{minipage}[c]{0.32\textwidth}
		\centering
		\textbf{(b)} NF200 channel

		\includegraphics[width = \textwidth]{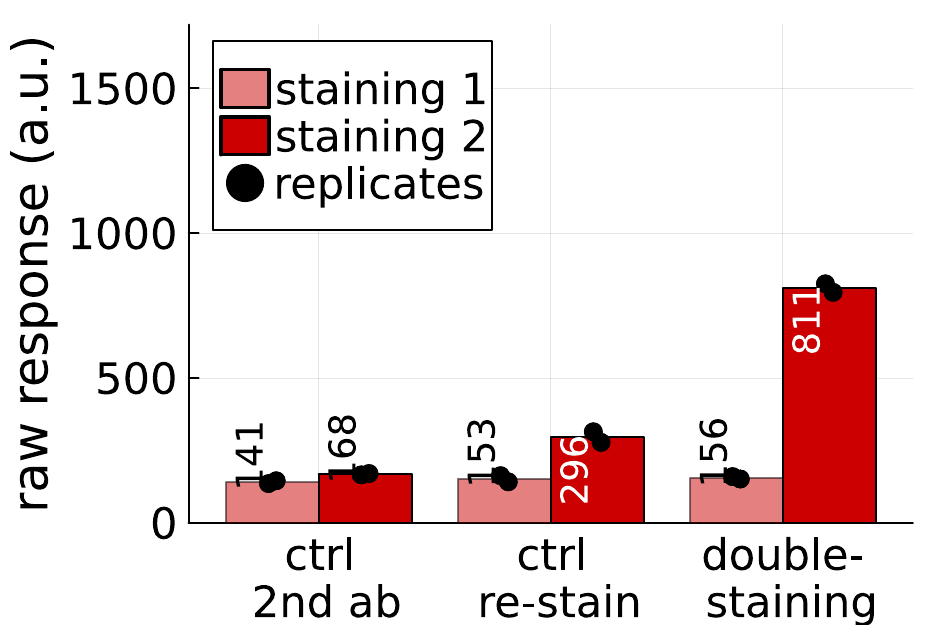}
	\end{minipage}
	\begin{minipage}[c]{0.32\textwidth}
		\centering
		\textbf{(c)} RPS11 channel

		\includegraphics[width = \textwidth]{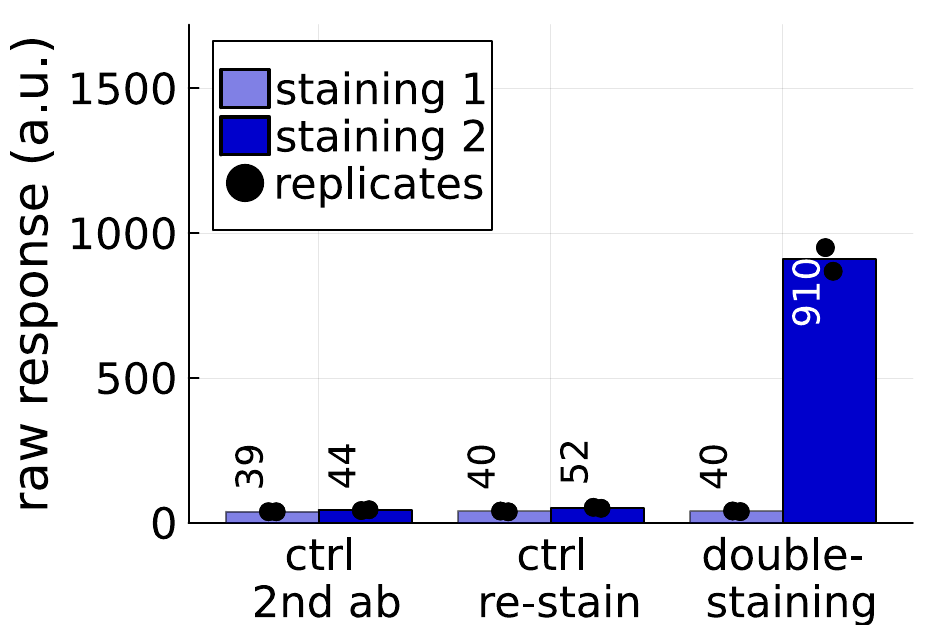}
	\end{minipage}\vspace{0.5cm}

	\begin{minipage}[c]{0.35\textwidth}
		\centering
		\textbf{(d)} Computational composite

		\includegraphics[width = 0.85\textwidth]{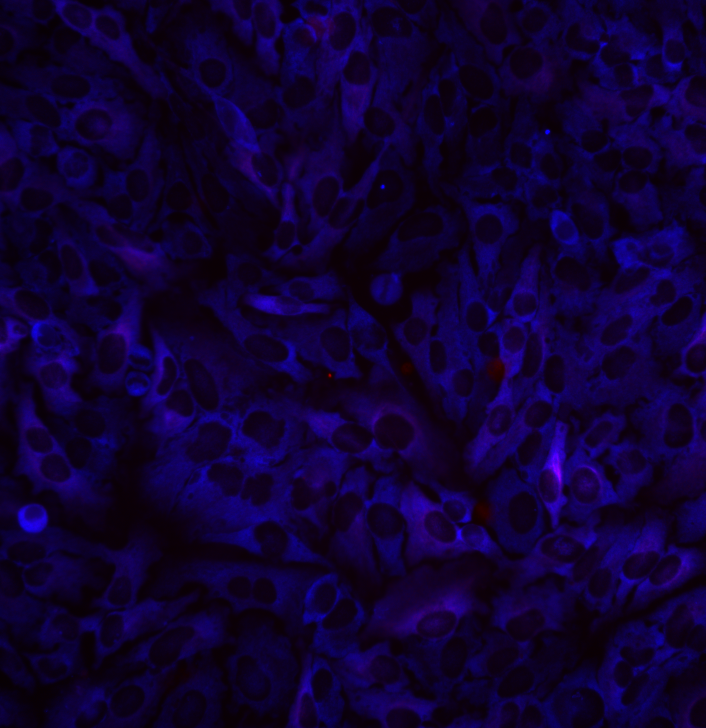}
	\end{minipage}
	\begin{minipage}[c]{0.35\textwidth}
		\centering
		\textbf{(e)} Channel composite

		\includegraphics[width = 0.85\textwidth]{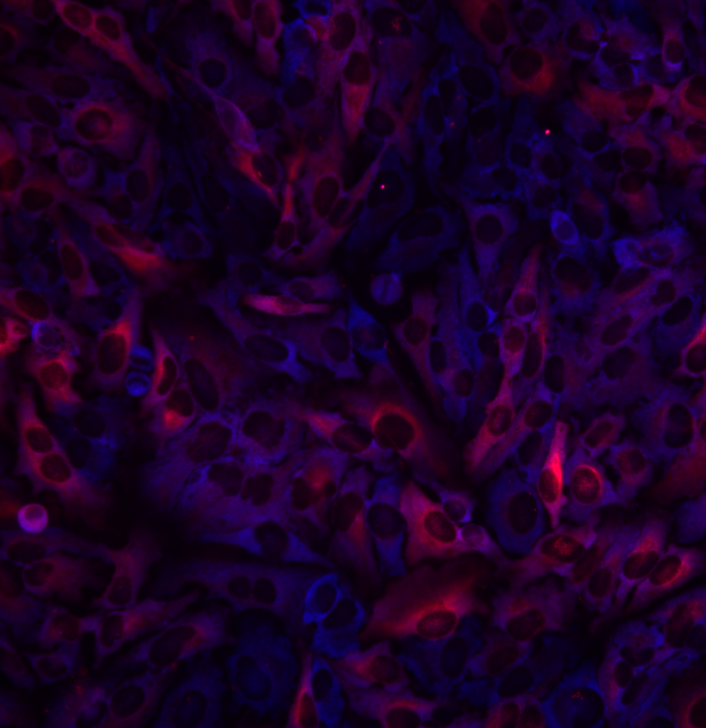}
	\end{minipage}

	\vspace{0.5cm}
	\rule{\textwidth}{2pt}
	\vspace{0.2cm}

	{\large\bfseries Plate 2: too low concentration for the 1st staining}\vspace{0.5cm}

	\begin{minipage}[c]{0.32\textwidth}
		\centering
		\textbf{(f)} Channel sum

		\includegraphics[width = \textwidth]{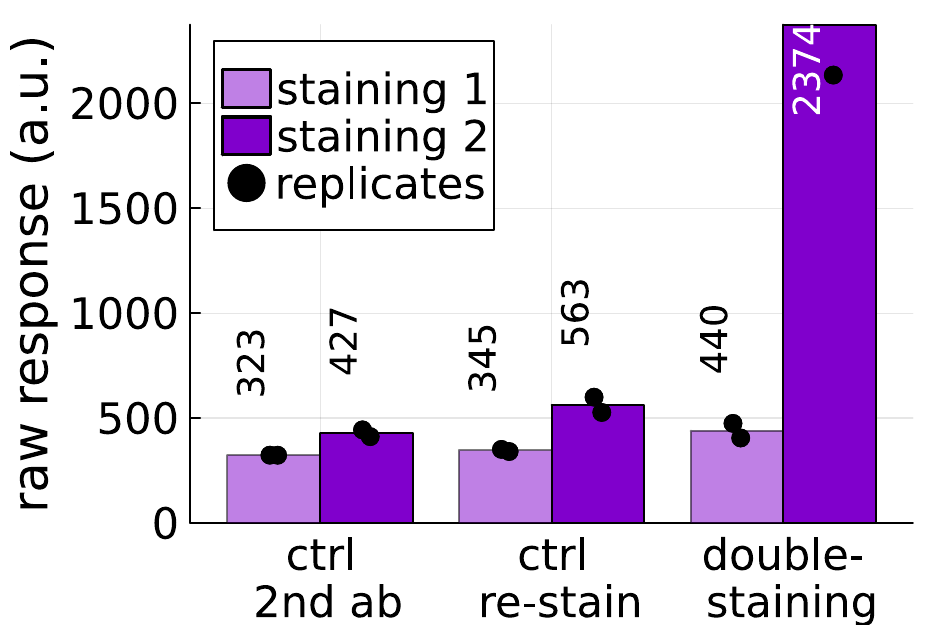}
	\end{minipage}
	\begin{minipage}[c]{0.32\textwidth}
		\centering
		\textbf{(g)} NF200 channel

		\includegraphics[width = \textwidth]{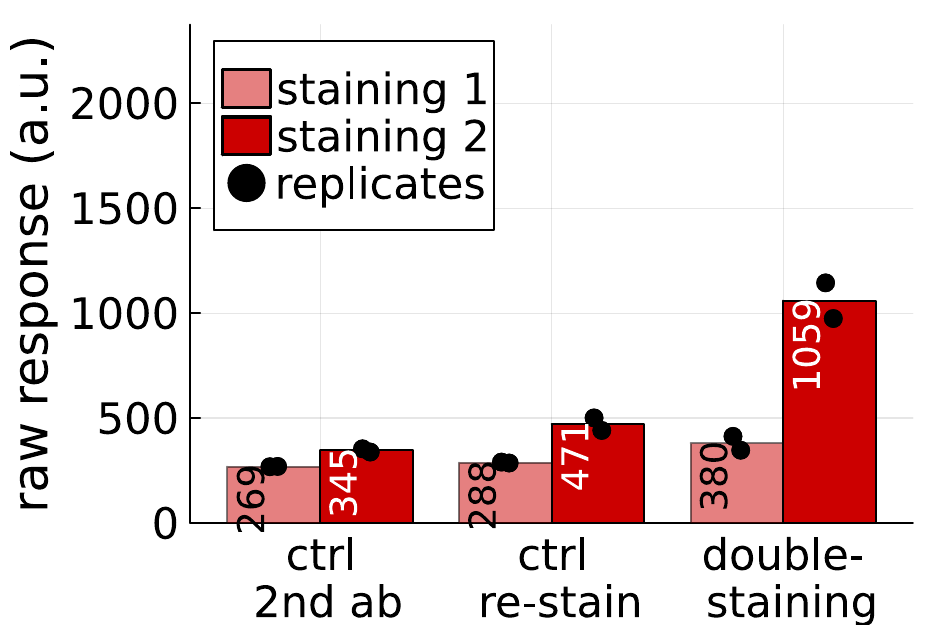}
	\end{minipage}
	\begin{minipage}[c]{0.32\textwidth}
		\centering
		\textbf{(h)} RPS11 channel

		\includegraphics[width = \textwidth]{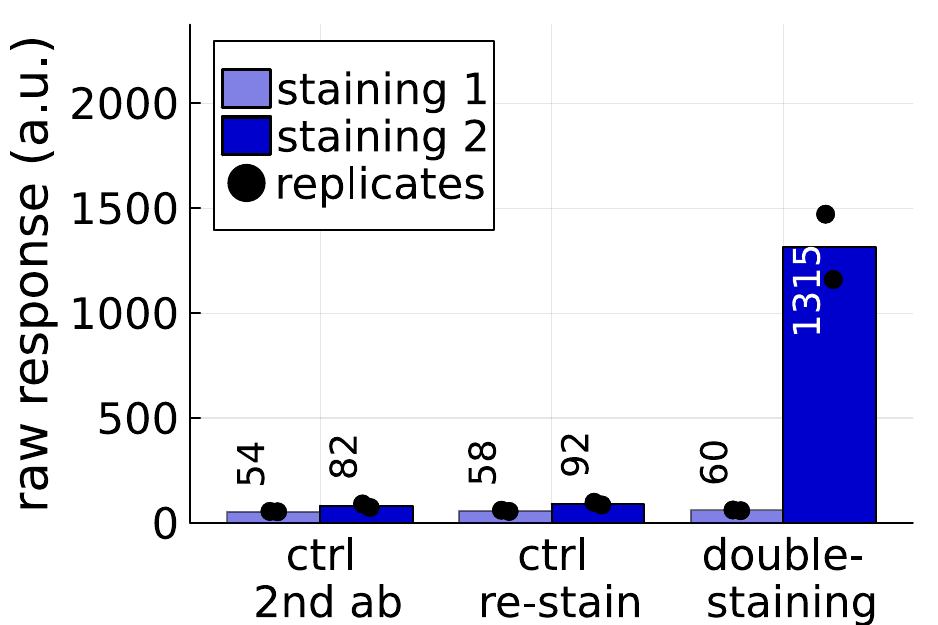}
	\end{minipage}\vspace{0.5cm}

	\begin{minipage}[c]{0.35\textwidth}
		\centering
		\textbf{(i)} Computational composite

		\includegraphics[width = 0.85\textwidth]{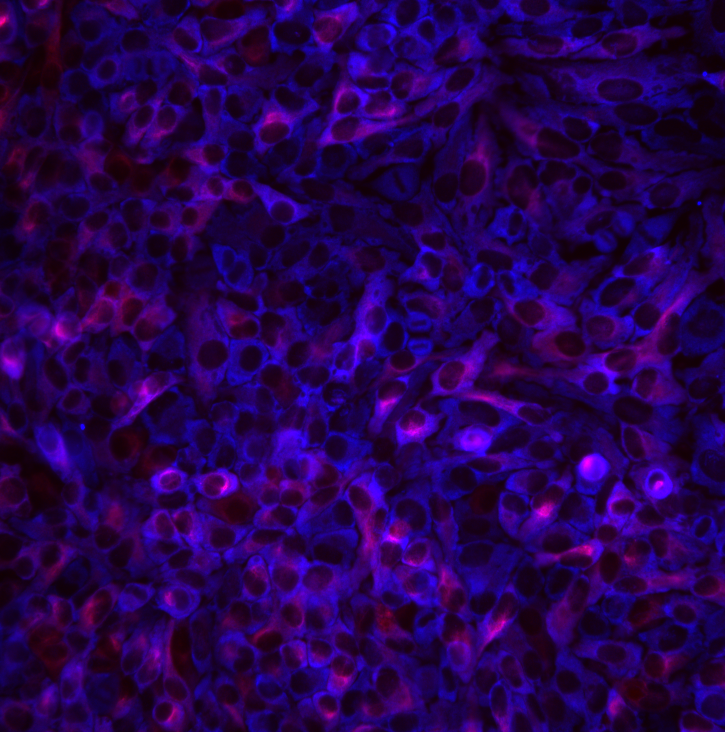}
	\end{minipage}
	\begin{minipage}[c]{0.35\textwidth}
		\centering
		\textbf{(j)} Channel composite

		\includegraphics[width = 0.85\textwidth]{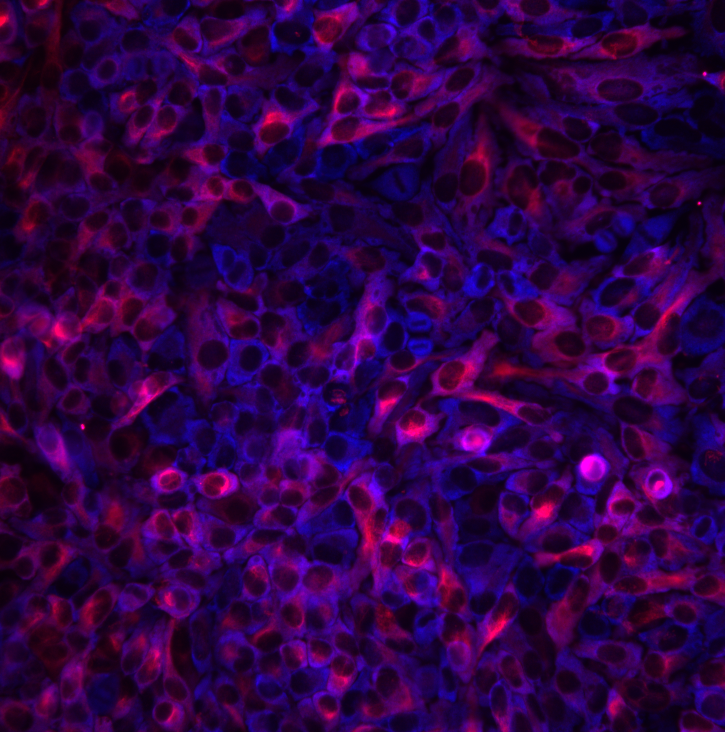}
	\end{minipage}
	\captionof{figure}{Double-staining experiment, using a too low dilution quotient ($d_{14} = 1:1638400$ instead of $d_{10} = 1:102400$) for the 1st staining.}
	\label{sup-fig: too low}
\end{minipage}

\noindent
\begin{minipage}{\textwidth}
	\centering
	{\large\bfseries Plate 1: too high concentration for the 1st staining}\vspace{0.5cm}

	\begin{minipage}[c]{0.32\textwidth}
		\centering
		\textbf{(a)} Channel sum

		\includegraphics[width = \textwidth]{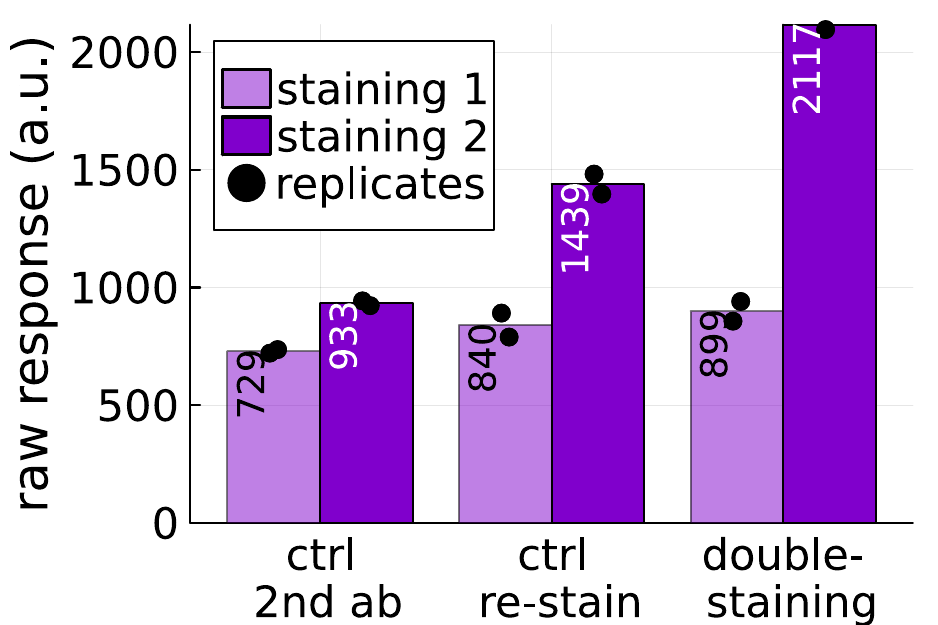}
	\end{minipage}
	\begin{minipage}[c]{0.32\textwidth}
		\centering
		\textbf{(b)} NF200 channel

		\includegraphics[width = \textwidth]{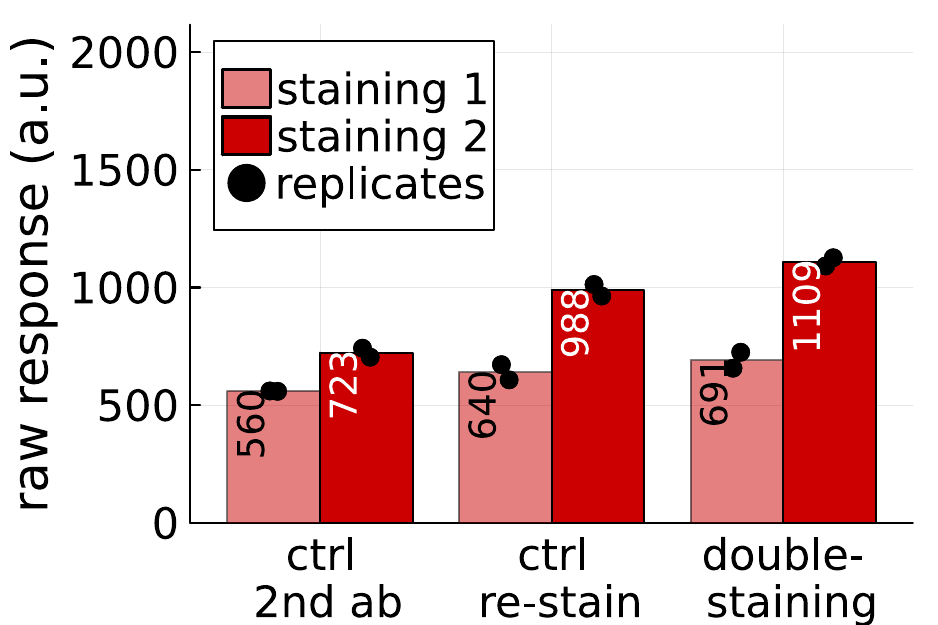}
	\end{minipage}
	\begin{minipage}[c]{0.32\textwidth}
		\centering
		\textbf{(c)} RPS11 channel

		\includegraphics[width = \textwidth]{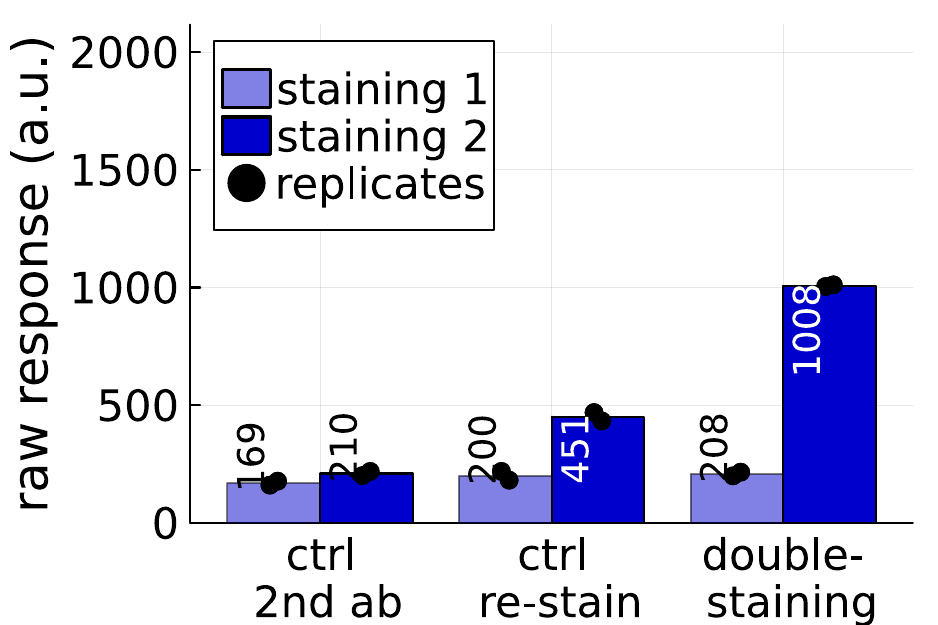}
	\end{minipage}\vspace{0.5cm}

	\begin{minipage}[c]{0.35\textwidth}
		\centering
		\textbf{(d)} Computational composite

		\includegraphics[width = 0.85\textwidth]{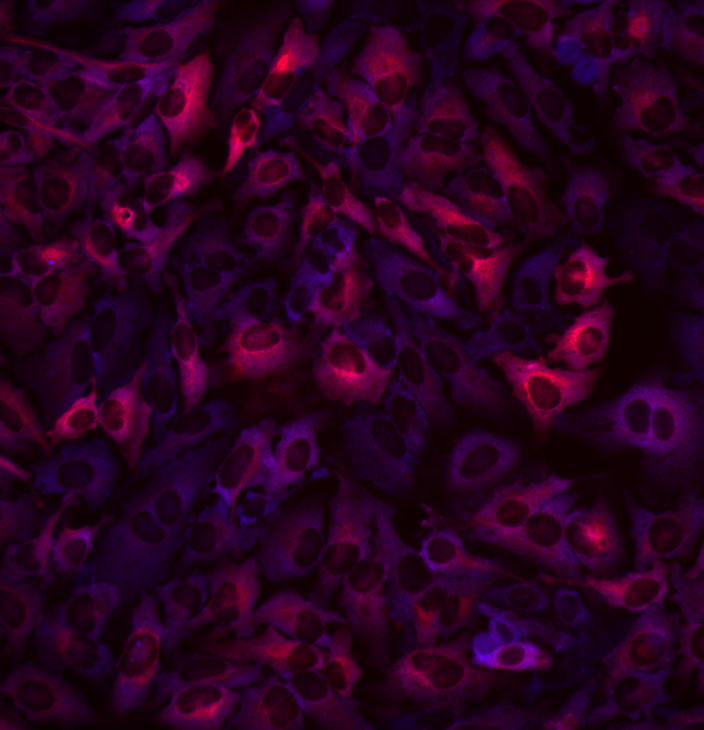}
	\end{minipage}
	\begin{minipage}[c]{0.35\textwidth}
		\centering
		\textbf{(e)} Channel composite

		\includegraphics[width = 0.85\textwidth]{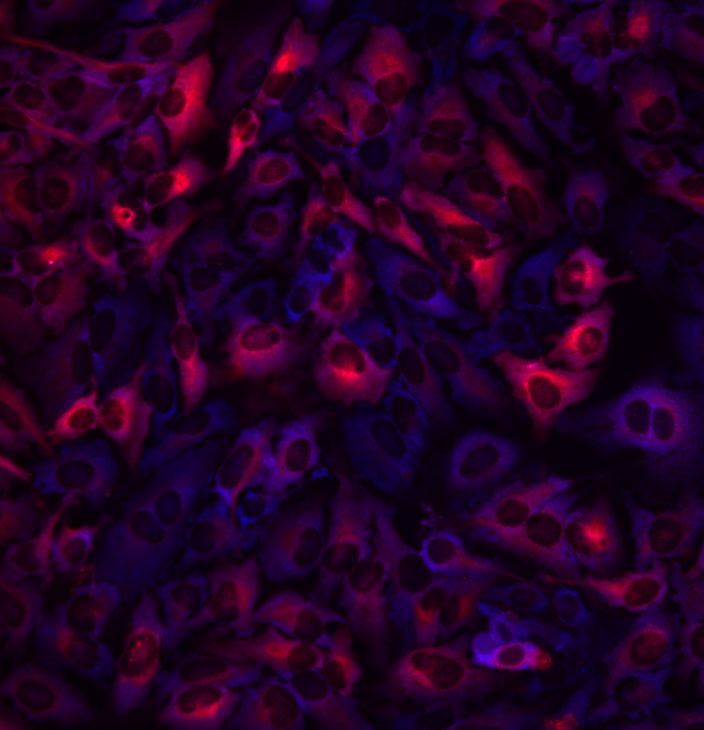}
	\end{minipage}

	\vspace{0.5cm}
	\rule{\textwidth}{2pt}
	\vspace{0.2cm}

	{\large\bfseries Plate 2: too high concentration for the 1st staining}\vspace{0.5cm}

	\begin{minipage}[c]{0.32\textwidth}
		\centering
		\textbf{(f)} Channel sum

		\includegraphics[width = \textwidth]{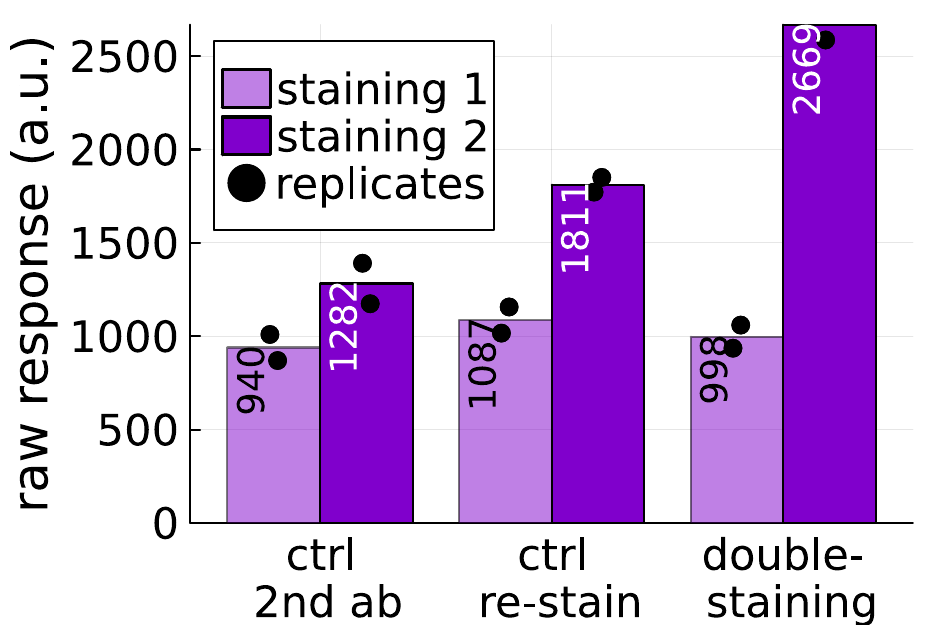}
	\end{minipage}
	\begin{minipage}[c]{0.32\textwidth}
		\centering
		\textbf{(g)} NF200 channel

		\includegraphics[width = \textwidth]{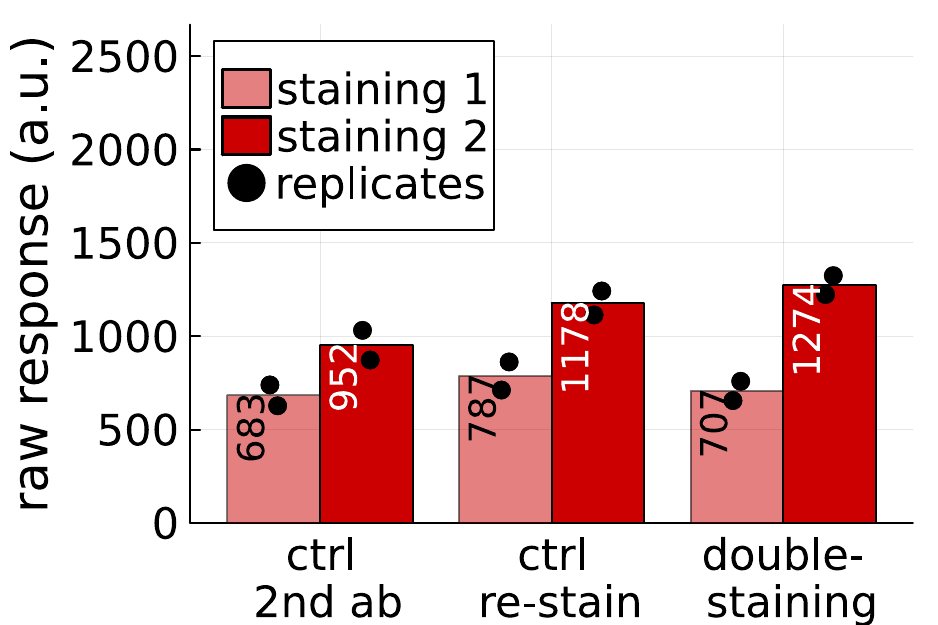}
	\end{minipage}
	\begin{minipage}[c]{0.32\textwidth}
		\centering
		\textbf{(h)} RPS11 channel

		\includegraphics[width = \textwidth]{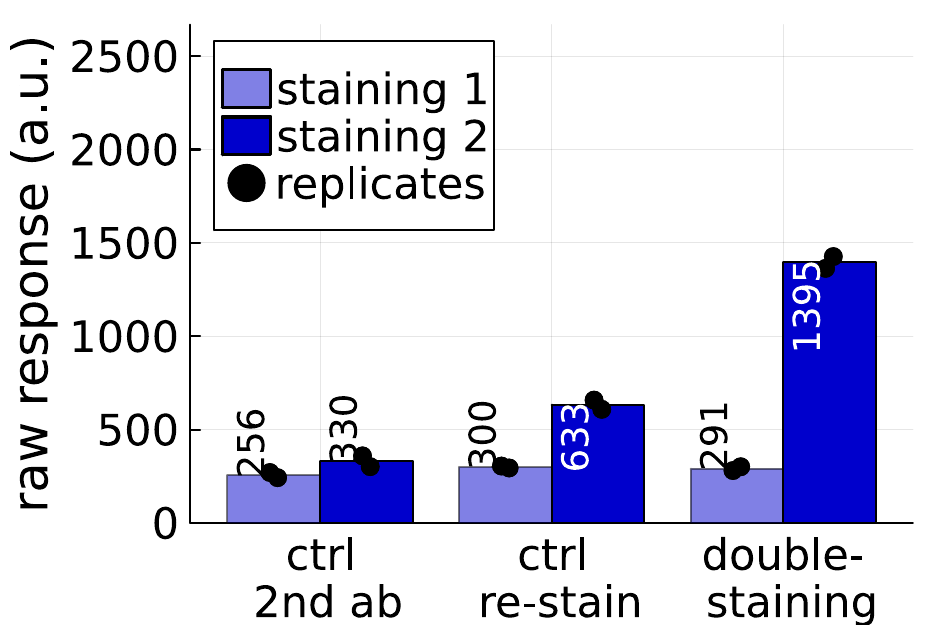}
	\end{minipage}\vspace{0.5cm}

	\begin{minipage}[c]{0.35\textwidth}
		\centering
		\textbf{(i)} Computational composite

		\includegraphics[width = 0.85\textwidth]{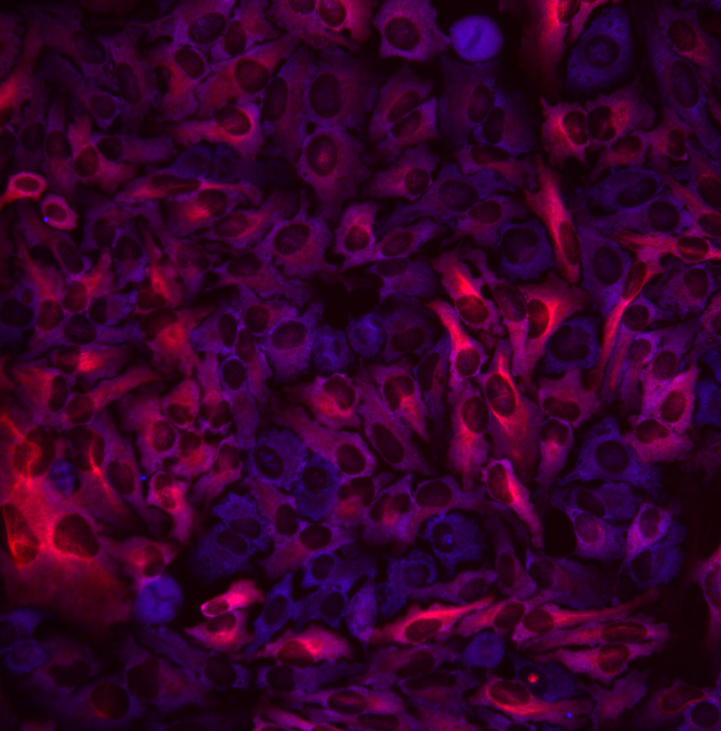}
	\end{minipage}
	\begin{minipage}[c]{0.35\textwidth}
		\centering
		\textbf{(j)} Channel composite

		\includegraphics[width = 0.85\textwidth]{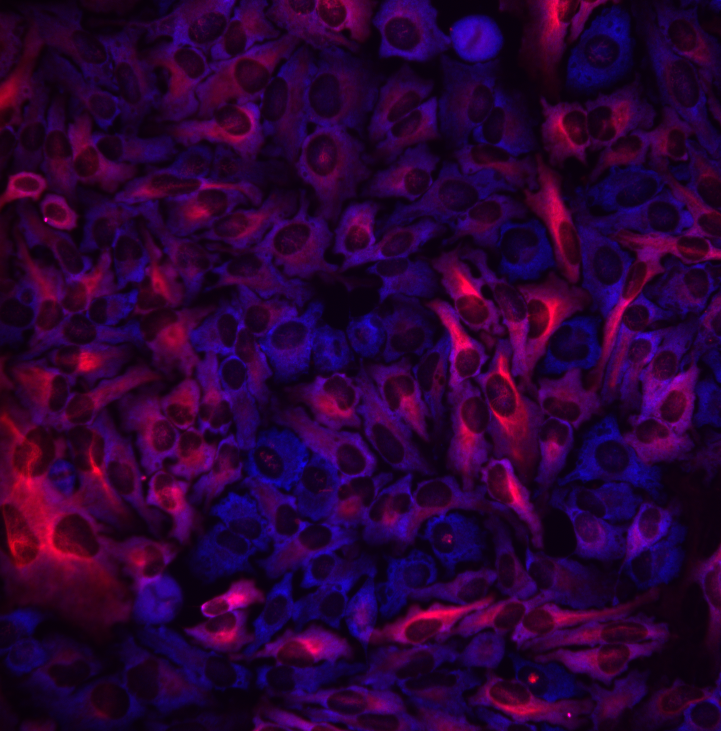}
	\end{minipage}
	\captionof{figure}{Double-staining experiment, using a too high dilution quotient ($d_{8} = 1:25600$ instead of $d_{10} = 1:102400$) for the 1st staining.}
	\label{sup-fig: too high}
\end{minipage}

\printbibliography[title= Supplement references]

\end{refsection}

\end{document}